\documentclass[a4paper,11pt]{article}
\pdfoutput=1
\usepackage[margin=1in]{geometry}
\usepackage[utf8]{inputenc}
\usepackage[T1]{fontenc}
\usepackage{microtype}
\usepackage{amsfonts}
\usepackage{amssymb}
\usepackage{amsmath}
\usepackage{amsthm}
\usepackage{fullpage}
\usepackage{enumitem}
\usepackage{mathtools}
\usepackage[dvipsnames]{xcolor}
\usepackage[hidelinks,bookmarksopen=true]{hyperref}
\usepackage[capitalize]{cleveref}
\usepackage{soul}
\usepackage{framed}
\usepackage{setspace}

\newtheorem{theorem}{Theorem}[section]
\newtheorem{corollary}[theorem]{Corollary}
\newtheorem{lemma}[theorem]{Lemma}

\newtheorem{observation}[theorem]{Observation}
\newtheorem{proposition}[theorem]{Proposition}

\theoremstyle{definition}
\newtheorem{definition}[theorem]{Definition}
\newtheorem{conjecture}[theorem]{Conjecture}
\theoremstyle{remark}
\newtheorem{remark}[theorem]{Remark}

\definecolor{mypink}{RGB}{199, 21 133}
\hypersetup{
  colorlinks,
  urlcolor = blue,
  linkcolor = mypink,
  citecolor = blue
}

\setul{2pt}{0.6pt}

\allowdisplaybreaks

\let\oldabstract\abstract
\let\oldendabstract\endabstract
\makeatletter
\renewenvironment{abstract}
{%
  {\list{}{\addtolength{\leftmargin}{-0.4em}%
    \listparindent 1.5em%
     \setstretch{0.95}
     \itemindent    \listparindent%
     \rightmargin   \leftmargin%
     \parsep        \z@ \@plus\p@}%
     \item\relax}%
  {\endlist}%
\oldabstract}
{\oldendabstract}
\makeatother

\newcommand{\emptystring}{\varepsilon}
\newcommand{\revstr}[1]{#1^{R}}
\newcommand{\dol}{{\rm \$}}
\newcommand{\hash}{\text{\#}}
\newcommand{\Occ}{{\rm Occ}}
\newcommand{\match}[1]{\overline{#1}}

\newcommand{\rhs}[1]{{\rm rhs}(#1)}
\newcommand{\rhsgen}[2]{{\rm rhs}_{#1}(#2)}
\renewcommand{\exp}[1]{{\rm exp}(#1)}
\newcommand{\expgen}[2]{{\rm exp}_{#1}(#2)}
\newcommand{\langgen}[2]{L_{#1}(#2)}

\newcommand{\bigO}{\mathcal{O}}
\newcommand{\dd}{\mathinner{.\,.}}
\newcommand{\algname}[1]{\text{\sc #1}}
\newcommand{\Zp}{\mathbb{Z}_{>0}}
\newcommand{\Zn}{\mathbb{Z}_{\geq 0}}
\newcommand{\Pts}{\mathcal{P}}
\newcommand{\RNA}{\mathrm{RNA}}
\newcommand{\WRNA}{\mathrm{WRNA}}
\DeclareMathOperator{\polylog}{polylog}
\DeclareMathOperator{\poly}{poly}

\begin{document}

\title{Grammar Boosting: A New Technique for Proving Lower Bounds\\
  for Computation over Compressed Data
}
\author{
  \normalsize Rajat De\\[-0.2ex]
  \normalsize Stony Brook University\\[-0.2ex]
  \normalsize \texttt{rde@cs.stonybrook.edu}
  \and
  \normalsize Dominik Kempa\\[-0.2ex]
  \normalsize Stony Brook University\\[-0.2ex]
  \normalsize \texttt{kempa@cs.stonybrook.edu}
}

\date{\vspace{-0.5cm}}
\maketitle

\begin{abstract}
  Computation over compressed data is a new paradigm in the design of
  algorithms and data structures, that can reduce the space usage
  and speed up the computation by orders of magnitude. One of the
  most frequently employed compression frameworks, capturing many
  practical compression methods (such as the Lempel--Ziv family,
  dictionary methods, and others) is grammar compression. In
  this framework, a string $T$ of length $N$ is represented as a
  context-free grammar of size $n$ whose language contains only
  the string $T$. In this paper, we focus on studying the limitations of
  these techniques.
  Previous work focused on proving lower bounds for algorithms
  and data structures operating over grammars constructed using
  algorithms that achieve the approximation ratio $\rho =
  \bigO(\polylog N)$ (since finding the smallest such grammar is
  NP-hard, every polynomial time grammar compressor can be viewed as
  an approximation algorithm).  Unfortunately, for the majority of
  grammar compressors, $\rho$ is either unknown or satisfies $\rho =
  \omega(\polylog N)$: In their seminal paper, Charikar et al. [IEEE
  Trans. Inf. Theory 2005] studied seven popular grammar compression
  algorithms: \algname{RePair}, \algname{Greedy},
  \algname{LongestMatch}, \algname{Sequential}, \algname{Bisection},
  \algname{LZ78}, and \algname{$\alpha$-Balanced}. Only one of them
  (\algname{$\alpha$-Balanced}) is known to achieve $\rho =
  \bigO(\polylog N)$.

  In this paper, we develop the first technique for proving lower
  bounds for data structures and algorithms on grammars
  that is fully general and does not depend on the approximation
  ratio $\rho$ of the used grammar compressor.
  Our first set of results concerns \emph{compressed data
  structures}. In 2013, Verbin and Yu proved that implementing
  random access to $T$ using a grammar constructed by an algorithm
  with $\rho = \bigO(\polylog N)$ requires $\Omega(\log N / \log \log
  N)$ time in the worst case. This lower bound applies to any
  structure using $\bigO(n \polylog N)$ space and matches the existing
  upper bounds. We prove that this lower bound holds also for
  \algname{RePair}, \algname{Greedy}, \algname{LongestMatch},
  \algname{Sequential}, and \algname{Bisection}, while $\Omega(\log
  \log N)$ time is required for random access to \algname{LZ78}. Our
  lower bounds apply to any structure using $\bigO(n \polylog N)$
  space and match the existing upper bounds. Moreover, our technique
  generalizes to
  classes of algorithms. Most notably, we tackle the notoriously hard
  to analyze class of \emph{global} algorithms (that includes, e.g.,
  the \algname{RePair} algorithm) and show that the lower bound
  $\Omega(\log N / \log \log N)$ applies to the whole class. This
  makes a significant step forward in a long-standing open problem of
  analyzing global algorithms; in the words of Charikar et al.:
  ``Because they are so natural and our understanding is so
  incomplete, global algorithms are one of the most interesting topics
  related to the smallest grammar problem that deserves further
  investigation.''

  Our second set of results concerns \emph{compressed computation},
  i.e., computation that runs in time that depends on the size of the
  input in compressed form. Recently, Abboud, Backurs, Bringmann, and
  Künnemann [FOCS 2017 and NeurIPS 2020] proved numerous limitations
  of compressed computation under popular conjectures (such as SETH,
  $k$-Clique, $k$-OV, and $k$-SUM). Similarly as above, however, their
  framework also displays a dependence on $\rho$. For example, their
  results imply that, assuming the Combinatorial $k$-Clique
  Conjecture, there is no combinatorial algorithm to solve CFG Parsing
  (for which the best classical algorithm has a time complexity of
  $\bigO(N^3)$) on grammars constructed using \algname{Bisection}
  (which satisfies $\rho = \widetilde{\Theta}(N^{1/2})$) that runs in
  $\bigO(n^3 \cdot N)$ or $\bigO(n^{3/2} \cdot N^2)$ time. The same
  is not known, however, for an algorithm running in $\bigO(n^5 \cdot N)$
  or $\bigO(n^3 \cdot N^2)$ time. Using our new techniques, we improve
  these and other conditional lower bounds. For example, for the CFG
  parsing on \algname{Bisection}, we rule out an algorithms with runtime
  $\bigO(n^c \cdot N^{3-\epsilon})$ for \emph{all}
  constants~$c>0$~and~$\epsilon>0$.
\end{abstract}

\thispagestyle{empty}
\clearpage
\pagenumbering{arabic}

\section{Introduction}\label{sec:intro}

Modern applications produce textual data at a rate not seen before.
During 2004--2015, the cost of sequencing the DNA of a single person
has decreased from \$20 million to around \$1000, i.e., by a factor of
$2 \times 10^{5}$~\cite{hernaez2019genomic}. This resulted in projects
like 100,000 Genome Project~\cite{100k}, that during 2013-2018,
produced around 75 terabytes of text. The efforts to sequence even
larger populations are now underway, e.g., in 2018, 26 countries
started the ongoing 1+ Million Genomes Initiative~\cite{mg}.
It is predicted that genomics research will generate between 2 and 40
exabytes of data within the next decade~\cite{estimate,stephens2015big}.
Other sources of massive textual datasets include versioned text documents
(such as Wikipedia) and source code
repositories~(such~as~Github)~\cite{NavarroIndexes,resolution}.

This explosion of data has not been matched by the corresponding
increase in computational power. One ray of hope in being able to
handle such massive datasets is that they are highly
repetitive~\cite{Przeworski2000,AGC,BergerDY16,
GreenfieldWH2019,NavarroIndexes}. This has been the driving force
behind the development of \emph{compressed algorithms and data
structures}~\cite{navarrobook,NavarroMeasures,NavarroIndexes}, which
combine aspects of information theory, lossless data compression, and
combinatorial pattern matching, to perform various queries or even run
complex computation directly on data in compressed
form~\cite{Gagie2020,resolution,attractors,AbboudBBK17}.

One of the most general frameworks for storing highly repetitive
strings is \emph{grammar compression}~\cite{KiefferY00,Charikar05,Rytter03}, in
which we represent a string using a \emph{straight-line program
(SLP)}, i.e., a context-free grammar, whose language contains only
the input string. On the one hand, this framework is easy to work
with and can succinctly encode even complex structure of
repetitions. On the other hand, it comes with
solid mathematical foundations -- as shown
in~\cite{Charikar05,Rytter03,GNPlatin18,attractors,resolution,%
delta,RaskhodnikovaRRS13,KempaS22}, grammar compression is up to logarithmic factors
equivalent to LZ77~\cite{LZ77}, LZ-End~\cite{kreft2010navarro},
RLBWT~\cite{BWT}, macro schemes~\cite{storer1978macro}, collage
systems~\cite{collage}, string attractors~\cite{attractors}, and
substring complexity~\cite{RaskhodnikovaRRS13,delta}, and it is at least as powerful as
automata~\cite{blumer1987complete}, Byte-Pair~\cite{gage1994new}, and
several other LZ-type compressors~\cite{LZ78,LZW,LZSS,LZFG}. For this
reason, grammar compression has been a very popular framework in
numerous previous studies. This includes pattern matching~\cite{Jez2015,GanardiG22,AbboudBBK17,Gawrychowski11,%
  BringmannWK19,CKW20,AbboudBBK17},
sequence similarity~\cite{Tiskin15,HermelinLLW13,AbboudBBK17,GaneshKLS22},
context-free grammar (CFG) parsing, RNA folding, disjointness~\cite{AbboudBBK17}, and
compressed linear algebra~\cite{AbboudBBK20,FerraginaMGKNST22}.
Grammars are also the
key component in algorithms converting between different compressed
representations~\cite{resolution,avl}. We refer to
surveys in~\cite{NavarroIndexes,NavarroMeasures,Lohrey12} and
discussion~in~\cite{AbboudBBK17,Gagie2020,attractors,resolution}~for~more~details.

The central component in many of the above applications, and a useful
structure on its own, is a \emph{compressed index} -- a data structure
requiring small space (close to the size of SLP representing the text)
that supports various queries over the underlying (uncompressed)
text. Nowadays, SLP-based indexes supporting
random-access~\cite{BLRSRW15,balancing,blocktree,attractors},
rank/select~\cite{Prezza19,PereiraNB17,DCC2015,blocktree},
LCE~\cite{tomohiro-lce,dynstr,NishimotoMFCS}, pattern
matching~\cite{ClaudeN11,ClaudeN12a,ClaudeNP21,
GagieGKNP14,GagieGKNP12,Diaz-DominguezN21,%
ChristiansenEKN21,KociumakaNO22,dynstr}, as well as various
spatio-temporal geometric queries~\cite{BrisaboaGNP16,BrisaboaGMP18}
are available. Despite these advances, our understanding of lower bounds
on SLP-compressed indexes remains an open problem, \ul{even for random
access} (the most basic query and a building block for more complex
queries).
\begin{itemize}[itemsep=0pt,parsep=0pt]
\item On the one hand, Bille et al.~\cite{BLRSRW15} proved that for
  any grammar compression algorithm that reduces a length-$N$ string
  $T$ into a representation of size $n$, we can build a structure of
  size $\bigO(n)$ allowing decoding of any symbol of $T$ in
  $\bigO(\log N)$ time. The latter result has recently been
  generalized by Ganardi et al.~\cite{balancing}, who proved that for
  any size-$n$ SLP encoding a length-$N$ string, there exists a
  size-$\bigO(n)$ SLP encoding the same string but with height
  $\bigO(\log N)$.  At the cost of increasing the space by a
  $\bigO(\log^{\epsilon} n)$ factor, it is possible to reduce the
  query time to $\bigO(\log N / \log \log N)$~\cite{blocktree}. To
  complement this, Verbin and Yu, proved that for every algorithm that
  achieves an $\bigO(\polylog N)$ approximation ratio\footnote{Finding
  the smallest grammar encoding of a given string is
  NP-hard~\cite{Charikar05}. Thus, every polynomial time grammar compression
  algorithm can be viewed as the approximation algorithm for the
  \emph{smallest grammar problem}~\cite{Charikar05}.}, one cannot
  access symbols of $T$ faster than $\bigO(\log N / \log \log N)$ time
  using a representation of size $\bigO(n \polylog
  N)$~\cite{VerbinY13}.\footnote{Verbin and Yu~\cite{VerbinY13}
  formulate this equivalently as stating that for any \emph{universal}
  data structure (i.e., working for every grammar compressor),
  one cannot achieve $o(\log N / \log \log N)$ query time
  in $\bigO(n \polylog N)$ space.}
\item On the other hand, for grammar compression algorithms like
  \algname{LZ78}~\cite{LZ78}, in $\bigO(n)$ space it is possible to
  implement random access in $\bigO(\log \log N)$~\cite{DuttaLRR13}.
  Consistent with the bound of Verbin and Yu, LZ78 achieves an
  approximation ratio of
  $\widetilde{\Omega}(n^{2/3})$~\cite{Charikar05,BannaiHHIJLR21}.
\end{itemize}

The above situation suggests that there exists a trade-off between the
approximation ratio and the time required for random access. This,
however, leads to two very serious issues:
\begin{enumerate}[itemsep=0pt,parsep=0pt]
\item
  The above techniques do not say anything about lower bounds on
  random access to grammars computed using algorithms with
  $\omega(\polylog N)$ approximation factor. This is problematic,
  since the majority of
  practical grammar compressors are in this category: Charikar et
  al.~\cite{Charikar05} prove that
  \algname{Sequitur}~\cite{sequitur},
  \algname{Sequential}~\cite{sequential},
  \algname{Bisection}~\cite{bisection1,bisection2},
  \algname{LZ78}~\cite{LZ78},
  \algname{LZW}~\cite{LZW} all achieve $\Omega(N^{\epsilon})$ ratio
  (for some constant $\epsilon > 0$). Badkobeh et
  al.~\cite{BadkobehGIKKP17}
  prove~analogous~bound~for~\algname{LZD}~\cite{lzd}.
\item
  Even worse, for many grammar compressors, we do not know their
  approximation ratio. This includes \algname{Greedy}~\cite{greedy1,greedy2,greedy3},
  \algname{LongestMatch}~\cite{longestmatch}, and 
  \algname{RePair}~\cite{repair} -- the last being one of the most practical and
  widely studied compressors~\cite{BilleGP17,%
  OchoaN19,GagieIMNST19,LohreyMM13,MienoIH22,FuruyaTNIBK19,GanczorzJ17}
  which ``($\dots$) consistently outperforms other grammar-based
  compressors, including those that offer theoretical guarantees of
  approximation.''~\cite{NavarroMeasures}
\end{enumerate}

With current techniques, proving lower bounds for algorithms like
\algname{RePair} appears to be a hopeless task, since \algname{RePair}
has resisted all attempts to prove an upper bound on its approximation
ratio for over 20 years. Given this situation, we ask:

\vspace{-0.0ex}
\setlength{\FrameSep}{1.7ex}
\begin{framed}
\noindent
\textbf{Problem 1.} \emph{Can we prove lower bounds for data structures based on
  grammar compressors without first establishing their approximation ratio?}
\end{framed}
\vspace{-0.0ex}

Another application of data compression in the design of algorithm is
\emph{compressed computation}, where the goal is to develop algorithms
whose runtime depends on the size of the input in the compressed
form. For example, Tiskin~\cite{Tiskin15} developed an algorithm that,
given two strings $S_1 \in \Sigma^{N}$ and $S_2 \in \Sigma^{N}$, both
in grammar-compressed form of total size $n$, computes the longest
common subsequence $LCS(S_1, S_2)$ in $\widetilde{\bigO}(n \cdot N)$
time.  For highly compressible strings (e.g., when $n =
\bigO(N^{1/2})$), this is a significant improvement over the currently
best general algorithm for LCS that runs in $\bigO(N^2)$
time (and is unlikely to be improved to
$\bigO(N^{2-\epsilon})$ due to the recent conditional lower
bound~\cite{AbboudBW15} based on the Strong Exponential Time
Hypothesis (SETH)).

Abboud, Backurs, Bringmann, and
Künnemann~\cite{AbboudBBK17,AbboudBBK20} recently asked whether
algorithms like the above LCS algorithm can be improved, i.e., is the
algorithm running in $\bigO((n \cdot N)^{1-\epsilon})$ time
achievable. They proved that under popular
hardness assumptions such as SETH, $k$-OV, $k$-Clique, or $k$-SUM,
the currently best compressed algorithms for several problems are optimal.
In particular, they showed that unless SETH fails, there is no algorithm for LCS that runs
in $\bigO((n \cdot N)^{1 - \epsilon})$ time, for any constant
$\epsilon > 0$. They proved similar conditional lower bounds for CFG
parsing, RNA folding, matrix-vector multiplication, inner product,
and~several~other~problems.

Similarly as the lower bound of Verbin and Yu, however, the techniques
in~\cite{AbboudBBK17,AbboudBBK20} exhibit a dependence on the
approximation ratio $\rho$ of the grammar compressor used to obtain
the input grammar. For example, for the CFG parsing problem, where
given a CFG $\Gamma$ and a string $S \in \Sigma^{N}$, the goal is the
check if $S \in L(\Gamma)$, the currently best combinatorial algorithm
runs in $\widetilde{\bigO}(N^3)$
time~\cite{cocke1969programming,younger1967recognition,kasami1966efficient}
(for simplicity, we assume $|\Gamma| = \widetilde{\bigO}(1)$).
Abboud, Backurs, Bringmann, and Künnemann~\cite{AbboudBBK17} proved
that unless the Combinatorial $k$-Clique Conjecture fails,
for $S$ constructed using grammar compressors with $\rho =
\bigO(\polylog N)$, there is no algorithm running
$\bigO(\poly(n) \cdot N^{3 - \epsilon})$ time. However, for larger
$\rho$, e.g., $\rho = \Theta(N^{\alpha})$, this technique
excludes only the algorithms running in $\bigO(n^c \cdot N^{3 -
  \epsilon})$, where $c < \epsilon / \alpha$
(\cref{cor:cfg-lower-bound-simple-generalization}). For example, if
$\alpha = 1/2$, then there is no combinatorial algorithm running in
$\bigO(n^3 \cdot N)$ or $\bigO(n^{3/2} \cdot N^2)$, but it leaves open
whether there is an algorithm running in
$\bigO(n^5 \cdot N)$ of $\bigO(n^3 \cdot N^2)$ time.
We~thus~ask:

\vspace{-1.0ex}
\setlength{\FrameSep}{1.7ex}
\begin{framed}
\noindent
\textbf{Problem 2.}
  \emph{What are the limitations for compressed computation on grammars
  obtained using algorithms with large or unknown approximation
  ratios? Can we prove such lower bounds without first establishing
  those approximation ratios?}
\end{framed}
\vspace{-3.0ex}

\paragraph{Our Results}

We present a new technique for proving lower bound on
grammar-compressed strings called \emph{Grammar Boosting}, that does
not require any knowledge about the approximation ratio of the
algorithm, and lets us answer both of the above questions.

\vspace{1ex}
\noindent
\emph{New Lower Bounds for Data Structures.} We prove
that the lower bound of $\Omega(\log N / \log \log N)$ applies to
nearly all of the classical and commonly used grammar compressors,
including: \algname{RePair}~\cite{repair},
\algname{Greedy}~\cite{greedy3,greedy1,greedy2},
\algname{LongestMatch}~\cite{longestmatch},
\algname{Sequitur}~\cite{sequitur},
\algname{Sequential}~\cite{sequential},
\algname{Bisection}~\cite{bisection1,bisection2}, and
\algname{LZD}~\cite{lzd}. No lower bounds for random access on either
of these grammars were known before. Our bound applies to any structure whose
space is $\bigO(n \polylog N)$ (where $n$ is the output size of any of
the above algorithms), i.e., it is always as strong as
the bound of Verbin and Yu~\cite{VerbinY13}.
This proves that there exist
algorithms with $\omega(\polylog N)$ approximation ratio that require
$\Omega(\log N / \log \log N)$ time for random access. This establishes
the first query separation between algorithms like \algname{Sequitur}
or \algname{LZD}, and \algname{LZ78} (which admits a random access
solution with $\bigO(\log \log N)$ query time~\cite{DuttaLRR13}). As an
auxiliary result, we show (via a reduction from the \emph{colored predecessor
problem}~\cite{PatrascuT06}) that random access to \algname{LZ78} in
$\bigO(\log \log N)$ time is in fact optimal within near-linear
space (i.e., $\bigO(n \polylog N)$), which is the case
in~\cite{DuttaLRR13}.

Our technique applies not only to individual algorithms, but is able
to capture an entire class. Specifically, we show that the lower bound
$\Omega(\log N / \log \log N)$ holds for all \emph{global
algorithms}~\cite{Charikar05} (which includes \algname{RePair},
\algname{Greedy}, and \algname{LongestMatch}).
This makes a significant step forward in a long-standing open problem of
postulated by Charikar et al.~\cite{Charikar05}:
``Because they are so natural and our understanding is so
incomplete, global algorithms are one of the most interesting topics
related to the smallest grammar problem
that deserves further investigation.''

The key idea in the framework and Verbin and Yu~\cite{VerbinY13} is to
prove that given any collection $\Pts$ of $n$ points on an $n \times
n$ grid, we can construct a string $A(\Pts)$ (called the
\emph{answer string}; see \cref{def:answer-string}) of length
$|A(\Pts)| = \Theta(n^2)$
that encodes
answers to all possible parity range counting queries~\cite{Patrascu07} on
$\Pts$, and has a grammar of size $\bigO(n \polylog n)$.  Since
answering %
such
queries in $\bigO(n \polylog n)$ space
requires $\Omega(\log n / \log \log n)$ time~\cite{Patrascu07}
(\cref{th:parity-range-counting-lower-bound}), any universal %
structure that for a grammar $G$ encoding $T \in \Sigma^N$ takes
$\bigO(n \polylog N)$ space, must thus take $\Omega(\log N / \log \log
N)$ time~for~access.

Let \algname{Alg} be some fixed grammar compression algorithm. The
issue with applying the above idea to \algname{Alg}
is that it would require proving a
bound the size of the output of \algname{Alg} on $A(\Pts)$, which for,
e.g., \algname{RePair} can be very difficult. We instead prove
that for any string $T$ for which there \emph{exists} an SLP of size
$n$ that encodes $T$, we can construct a string $T'$ such that:
\begin{enumerate}[label=(\alph*),itemsep=0pt,parsep=0pt]
\item\label{intro:it-1}
  The length of $T'$ is polynomial in the length of $T$,
\item\label{intro:it-2}
  $T$ can be quickly identified within $T'$, e.g., $T[j] =
  T'[\alpha + \beta \cdot j]$ for some $\alpha \geq 0$ and $\beta > 0$,
\item \algname{Alg} compresses $T'$ into size $\bigO(n)$. Note that
  this does not require proving anything about the compression of $T$
  by \algname{Alg}.
\end{enumerate}

In other words, we take a string $T$ having a small grammar $G$
and ``boost'' the performance of \algname{Alg} by presenting $T$ in a
well-structured form of $T'$ so that $\algname{Alg}$ compresses the
string $T'$ into size $\bigO(|G|)$.  This still lets us utilize the
reduction from parity range counting queries~\cite{Patrascu07}, since
by \cref{intro:it-1}, $\log |T'| = \Theta(\log |T|)$, and by
\cref{intro:it-2}, accessing symbols of $T$ via $T'$ does not incur
any time penalty. To construct $T'$, we typically first define an
auxiliary grammar $G'$ with the set of nonterminals similar to $G$,
but including special sentinel symbols identifying the nonterminals.
The string $T'$ is then defined by listing expansions of all nonterminals
of $G'$ in the order of nondecreasing length, repeating each expansion
twice, optionally separating with additional sentinel symbols
(\cref{def:alpha,def:beta}). The crux of the analysis is to show that for any
of the algorithms we studied, such structuring forces the algorithm to
compress the string in a specific way. For global algorithms
(\cref{sec:global}) this is particularly hard, since
their behavior is not very well understood.
We manage, however, to
fully characterize a class of all intermediate grammars that global
algorithms can reach during processing of $T'$
(\cref{def:bexp,def:bigG}), and in a series of lemmas prove that there
is only one possible final grammar (\cref{lm:exp-bigG,lm:dollars,%
lm:one-dollar,lm:dollar-on-rhs,lm:bexp-occ,lm:maximal-string,lm:nonov-s,%
lm:sub,lm:maximal-substring-existence,lm:one-step,lm:iso,%
lm:global-output-grammar,lm:global-output-size}). We present more
details in the Technical Overview (\cref{sec:technical-overview}).  As
a result, we obtain a series of lower bounds stated in
\cref{th:global,th:sequitur,th:sequential,th:lzd,th:bisection}.  In a
single theorem, we can summarize it as follows.

\begin{theorem}\label{th:main}
  Let \algname{Alg} be any global algorithm (e.g., \algname{RePair},
  \algname{Greedy}, or \algname{LongestMatch}), or one of the
  following algorithms: \algname{Sequitur}, \algname{Sequential},
  \algname{Bisection}, or \algname{LZD}. For any string $T$, let
  $\algname{Alg}(T)$ denote the output of \algname{Alg} on $T$.  In
  the cell-probe model, there is no data structure that, for every
  string $T$ of length $N$, achieves $\bigO(|\algname{Alg}(T)|
  \log^{c} N)$ space $($where $c = \bigO(1))$ and implements random
  access queries to $T$ in $o(\log N / \log \log N)$ time.
\end{theorem}

\noindent
\emph{New Lower Bounds for Compressed Computation.} Our second result
is to demonstrate that the grammar boosting technique complements the
framework of Abboud, Backurs, Bringmann, and
Künnemann~\cite{AbboudBBK17,AbboudBBK20} for proving conditional lower
bounds for compressed computation.  We consider two problems: CFG
Parsing (defined above), and Weighted RNA Folding
(\cref{sec:rna-problem}). Among other results, the authors
of~\cite{AbboudBBK17} prove that unless the \emph{Combinatorial
$k$-Clique Conjecture} (asserting that a $k$-clique cannot be
combinatorially found in $\bigO(|V|^{k(1-\epsilon)})$ time for any
$\epsilon>0$) fails, these problems essentially require $\Omega(N^3)$
time, even for highly compressible inputs.  Here we extend these
results to numerous grammar compressors with large or unknown
approximation ratio $\rho$.  In particular, this includes
\algname{Sequential}, \algname{LZD}, and the class of global
algorithms.

Consider the CFG Parsing problem.  The key idea in the hardness proof
presented in~\cite{AbboudBBK17} is, given $k \geq 3$ and an undirected
graph $G = (V, E)$, to construct a small CFG $\Gamma$ and a highly
compressible string $S$ of length $|S| = \bigO(|V|^{k+2})$
(specifically, $S$ has a grammar of size $\bigO(|V|^3)$), such that
$G$ has $3k$-clique if and only if $S \in L(\Gamma)$. To adapt this
reduction to a grammar compression algorithm \algname{Alg}, we
construct a ``well structured'' string $S'$ and a CFG $\Gamma'$ such
that:
\begin{enumerate}[label=(\alph*),itemsep=0pt,parsep=0pt]
\item $|S'| = |V|^{k+\bigO(1)}$,
\item $S' \in L(\Gamma')$ if and only if $S \in L(\Gamma)$,
\item \algname{Alg} compresses $S'$ into size $|V|^{\bigO(1)}$.
\end{enumerate}
The construction of $S'$ for \algname{Sequential} is similar as for
the random-access problem above (the construction for
RNA folding is slightly more involved). The CFG $\Gamma'$ is obtained
from $\Gamma$ by ensuring that it ignores every second symbol as well as some
sufficiently long prefix of $S'$.  The details of these reductions are
presented in \cref{sec:cfg-parsing,sec:rna-folding}.  The nice
property of these extensions of~\cite{AbboudBBK17} is that they can
augment a conditional hardness proof essentially in a black-box manner
(in the case of CFG parsing, we only need
\cref{lm:kclique-small-slg}).

\vspace{-1ex}
\paragraph{Related Work}

While LZ77 and grammar-based compressed indexes are capable of
supporting fundamental queries such as random access, longest common
extension (LCE), rank/select, or even pattern matching, often required
is the more powerful functionality of suffix arrays~\cite{sa} and
suffix trees~\cite{st}. The underlying compression method that
supports such powerful queries is Run-Length Burrows--Wheeler
Transform (RLBWT)~\cite{BWT}. Gagie et al.~\cite{Gagie2020} recently
demonstrated that it is possible to efficiently support all these
queries using $\bigO(r \log n)$ space, where $r$ is the size of
RLBWT. On the other hand, Kempa and Kociumaka proved that for all
strings, it always holds $r = \bigO(\delta \log \delta \max(1,
\log \tfrac{n}{\delta \log \delta}))$~\cite{resolution} (where
$\delta$ is the \emph{substring complexity}~\cite{delta},
a measure closely related to Lempel-Ziv and grammar compression),
establishing a link between RLBWT-based
indexes, and LZ and grammar-based indexes. Nishimoto et
al.~\cite{nishimototabei} showed how to reduce the space
of~\cite{Gagie2020} for $LF$ and $\Phi^{-1}$ queries to
$\bigO(r)$. Significant efforts also went into the construction of
compressed indexes~\cite{OhnoSTIS18,nishimototabei2,resolution,%
breaking,PolicritiP17,dcc2017,Kempa19}, as well as making text
indexes dynamic~\cite{dynstr,dynsa,NishimotoDAM,NishimotoMFCS}.

\section{Preliminaries}\label{sec:prelim}

Let $w \in \Sigma^n$ be a \emph{string} (or \emph{text}) of length $n$
over \emph{alphabet} $\Sigma$. Denote $\sigma = |\Sigma|$. We index
strings starting from 1, i.e., $w = w[1]w[2] \cdots w[n]$. For $1 \leq
i \leq j \leq n$ we denote substrings of $w$ as $w[i \dd j]$ and by
$[i \dd j)$ we mean $[i \dd j-1]$. We denote the length of string $w$
as $|w|$. The concatenation of strings $u$ and $v$ is written as
$u\cdot v$ or $uv$, and the empty string is denoted $\emptystring$.
For every $u, v \in \Sigma^{*}$, we denote $\Occ(u, v) = \{i \in [1
\dd |v|] : i + |u| \leq |v| + 1 \text{ and }v[i \dd i + |u|) = u\}$.

A \emph{context-free grammar} (CFG) is a tuple $G = (V, \Sigma, R,
S)$, where $V$ is a finite set of \emph{nonterminals} (or
\emph{variables}), $\Sigma$ is a finite set of \emph{terminals}, and
$R \subseteq V \times (V \cup \Sigma)^*$ is a set of
\emph{productions} (or \emph{rules}). We assume $V \cap \Sigma =
\emptyset$ and $S \in V$. The nonterminal $S$ is called the
\emph{starting nonterminal}. Nonterminals in $V \setminus \{S\}$ are
called \emph{secondary}. If $(N, \gamma) \in R$ then we write $N
\rightarrow \gamma$. For $u, v \in (V \cup \Sigma)^*$ we write $u
\Rightarrow v$ if there exist $u_1, u_2 \in (V\cup \Sigma)^*$ and a
rule $N \rightarrow \gamma$ such that $u=u_1 N u_2$ and $v = u_1
\gamma u_2$. We say that $u$ \emph{derives} $v$ and write $u
\Rightarrow^* v$, if there exists a sequence $u_1, \ldots, u_k$, $k
\geq 1$ such that $u = u_1$, $v = u_k$, and $u_i \Rightarrow u_{i+1}$
for $1 \leq i < k$.  The \emph{language} of grammar $G$ is the set
$L(G) := \{w \in \Sigma^* \mid S \Rightarrow^* w\}$.

A grammar $G = (V, \Sigma, R, S)$ is called a \emph{straight-line
grammar} (SLG) if for any $N \in V$ there is exactly one production
with $N$ on the left side, and there exists a linear order $\prec$ on $V$
such that for every $X, Y \in V$, $Y$ occurring in $\rhsgen{G}{X}$
implies $X \prec Y$.
The unique
$\gamma$ such that $N \rightarrow \gamma$ is called the
\emph{definition} of $N$ and denoted $\rhsgen{G}{N}$. If $G$ is clear
from the context, we simply write $\rhs{N}$.  In any SLG, for any $u
\in (V \cup \Sigma)^*$ there exists exactly one $w \in \Sigma^*$ such
that $u \Rightarrow^* w$. We call such $w$ the \emph{expansion} of
$u$, and denote $\expgen{G}{u}$ (or simply $\exp{u}$ when $G$ is
clear).
Note that for any
SLG $G$, $L(G) = \{\expgen{G}{S}\}$.

We say that two SLGs $G_1 = (V_1, \Sigma, R_1, S_1)$ and $G_2 = (V_2,
\Sigma, R_2, S_2)$ are \emph{isomorphic}, if there exists a bijection
$f: V_1 \cup \Sigma \rightarrow V_2 \cup \Sigma$ such that
\vspace{-0.5ex}
\begin{itemize}[itemsep=0pt,parsep=0pt]
\item $f(S_1) = S_2$,
\item For every $c \in \Sigma$, $f(c) = c$,
\item For every $N_1 \in V_1$, letting $N_2 = f(N_1)$, $S_1 =
  \rhsgen{G_1}{N_1}$, and $S_2 = \rhsgen{G_2}{N_2}$, it holds that
  $|S_1| = |S_2|$, and for every $j \in [1 \dd |S_1|]$,
  $S_2[j] = f(S_1[j])$.
\end{itemize}
If $G_1$ is isomorphic to $G_2$, then $|V_1| = |V_2|$.
Moreover, for every $N_1 \in V_1$,
$\expgen{G_1}{N_1} = \expgen{G_2}{f(N_1)}$. In particular, $L(G_1) =
\{\expgen{G_1}{S_1}\} = \{\expgen{G_2}{f(S_1)}\} =
\{\expgen{G_2}{S_2}\} = L(G_2)$.

Let $G = (V, \Sigma, R, S)$ be an SLG. We define the \emph{parse tree}
of $A \in V \cup \Sigma$ as a rooted ordered tree $\mathcal{T}_{G}(A)$
(we omit $G$, whenever it is clear from the context), where each node
$v$ is associated to a symbol $s(v) \in V \cup \Sigma$. The root of
$\mathcal{T}(A)$ is a node $\rho$ such that $s(\rho) = A$. If $A \in
\Sigma$ then $\rho$ has no children. If $A \in V$ and $\rhs{A} = B_1
\cdots B_k$, then $\rho$ has $k$ children and the subtree rooted at
the $i$th child is (a copy of) $\mathcal{T}(B_i)$. The parse tree
$\mathcal{T}(G)$ of $G$ is defined as the parse tree $\mathcal{T}(S)$
of $S$.

The idea of \emph{grammar compression} is, given a string $w$, to
compute a small SLG $G$ such that $L(G) = \{w\}$. The \emph{size} of
the grammar is defined as $|G| := \sum_{N \in V}|\rhs{N}|$.  Clearly,
it is easy to encode any $G$ in $\bigO(|G|)$ space: pick an ordering
of nonterminals and write down the definitions of all variables with
nonterminals replaced by their number in the order. For any grammar
compression algorithm \algname{Alg}, we denote the output of
\algname{Alg} on a string $u$ by $\algname{Alg}(u)$.

The size of the smallest SLG generating $w$ is denoted $g^*(w)$. The
decision problem \textsc{SmallestGrammar} of determining whether for a
given string $w$ it holds $g^*(w) \leq t$ is NP-hard~\cite{Charikar05}
(or even APX-hard~\cite{Charikar05}), but $\bigO(\log
(n/g^*))$-approximations are known~\cite{Rytter03, Charikar05, Jez16}.

\begin{definition}\label{def:admissible}
  An SLG $G = (V, \Sigma, R, S)$ is \emph{admissible} if for
  every $X \in V$, it holds $|\rhsgen{G}{X}| = 2$ and $X$ occurs in
  the parse tree $\mathcal{T}(S)$.
\end{definition}

\subsection{Hardness Assumptions}

\begin{conjecture}[$k$-Clique]
  For every $k \geq 3$ and $\epsilon > 0$,
  there is no algorithm checking if an undirected graph $G = (V, E)$ has
  a $k$-clique
  that runs
  in $\bigO(|V|^{(k\omega/3)(1 - \epsilon)})$ time.
\end{conjecture}

\begin{conjecture}[Combinatorial $k$-Clique]
  For every $k \geq 3$ and $\epsilon > 0$,
  there is no combinatorial algorithm checking if an undirected graph $G = (V, E)$ has
  a $k$-clique
  that runs
  in $\bigO(|V|^{k(1 - \epsilon)})$ time.
\end{conjecture}

\section{Grammar Compression Algorithms}\label{sec:algorithms}

\subsection{Global Algorithms}\label{sec:algs-global}

\begin{definition}\label{def:maximal-string}
  Let $G = (V, \Sigma, R, S)$ be an SLG. A string $s \in (\Sigma \cup
  V)^{+}$ is called \emph{maximal} (with respect to $G$) if it
  satisfies the following conditions:
  \begin{enumerate}[itemsep=0pt,parsep=0pt]
  \item $|s| \geq 2$,
  \item The string $s$ has at least two non-overlapping occurrences on
    the right-hand side of $G$ (i.e., in the definitions of
    nonterminals of $G$),
  \item There is no string $s' \in (\Sigma \cup V)^{+}$ such that
    $|s'| > |s|$ and $s'$ has at least as many non-overlapping
    occurrences on the right-hand side of $G$ as $s$.
  \end{enumerate}
\end{definition}

\begin{remark}\label{rm:nonov}
  The number of non-overlapping occurrences of $u$ on the
  right-hand side of the grammar is defined with a greedy search,
  i.e., we scan the definition of each nonterminal left-to-right, and
  once an occurrence of $u$ is found at position $i$, we restart the
  search at position~$i+|u|$.
\end{remark}

Charikar et al.~\cite{Charikar05} defines the class of ``global
algorithms'' as all grammar compression algorithms operating according
to the following principle. Begin with a grammar having a single
starting nonterminal $S$ whose
definition is $w$. Each iteration of the algorithm: (1) chooses a
maximal substring $s$ (\cref{def:maximal-string}), (2) creates a new
non-terminal $N$ and sets $s$ as its definition, and (3) scans (left to
right) the definition of every nonterminal, replacing every encountered
occurrence of $s$ with $N$. Note that all the affected occurrences of
$s$ are nonoverlapping. Global algorithms differ only in the choice of
the maximal substring $s$ at each step.

Charikar et al.~\cite{Charikar05} lists the following global algorithms:

\begin{description}[style=sameline,itemsep=0.5ex]
\item[\algname{RePair}]\cite{repair}: In each round, the algorithm
  selects the maximal string $s$ with the highest number of
  non-overlapping occurrences on the right-hand side of the current
  grammar. We remark that the original
  formulation of \algname{RePair}~\cite{repair} additionally requires
  $|s| = 2$.

\item[\algname{Greedy}]\cite{greedy3,greedy1,greedy2}: In each round,
  the algorithm selects the maximal string that results in the largest
  reduction in the grammar size.

\item[\algname{LongestMatch}]\cite{longestmatch}: In each round, the
  algorithm selects the longest maximal string.
\end{description}

\subsection{Nonglobal Algorithms}\label{sec:algs-nonglobal}

\begin{description}[style=sameline,itemsep=0.2ex]
\item[\algname{Sequential}]\cite{sequential}: Process the input
  left-to-right.  In each step, first compute the longest prefix of the
  remaining suffix of the input
  that is equal to $\exp{N}$ for
  some secondary nonterminal $N$ existing in the grammar, and append
  $N$ to the definition of the start rule.  If there is no such
  prefix, append the next symbol from the input. If now there
  exists a pair of symbols $AB$ on the right-hand side of the grammar
  with two non-overlapping occurrences (in~\cite{sequential}, it is proved
  that there cannot be more such occurrences; see also
  \cref{lm:seq-irreducible}), create a new nonterminal $M$
  with $\rhs{M} = AB$, and replace both the occurrences with
  $M$. Finally, if after this update there exists a nonterminal that
  is only used once on the right-hand side of the grammar, remove it
  from the grammar, replacing its occurrence with its definition.

\item[\algname{Sequitur}]\cite{sequitur}:
  Process input string left-to-right.  In each step, we first append the
  next symbol from the input into the definition of the start rule. We
  then apply the following reductions to the grammar as long as
  possible, each time choosing the reduction earliest in the list:
  \begin{enumerate}[itemsep=0pt,parsep=0pt]
  \item If the length-2 suffix $AB$ of the definition of the starting
    nonterminal $S$ is equal to the definition of some other
    nonterminal $N$, replace this length-$2$ suffix with $N$.
  \item If the length-2 suffix $AB$ of the definition of the starting
    nonterminal has another non-overlapping occurrence on the
    right-hand side of the grammar (in~\cite{sequitur}, it is proved that
    there cannot be more than one such occurrences), create a new
    nonterminal $M$ with $\rhs{M} = AB$, and replace both occurrences
    with $M$.
  \item If there exists a nonterminal that is only used once on the
    right-hand side of the grammar, remove it from the grammar,
    replacing its only occurrence with its definition.
  \end{enumerate}

\item[\algname{Bisection}]\cite{bisection1,bisection2}: Let $w \in
  \Sigma^{n}$ be the input string. The first step of the algorithm
  computes a set $\mathcal{S}$ of substrings of $w$ as follows.
  First, we insert the string $w$ itself into the set. If $n > 1$, we
  then compute the largest $k \geq 0$ such that $2^k < n$ and
  recursively insert into $\mathcal{S}$ substrings of $w[1 \dd 2^k]$
  and $w(2^k \dd n]$. After the enumeration is complete, we create a
  nonterminal for every element $s$ of $\mathcal{S}$. The definition
  for a nonterminal corresponding to every $s \in \mathcal{S}$ such
  that $|s| > 1$ is set to be the two nonterminals corresponding to
  the initial two substrings of $s$ computed during the enumeration
  phase.

\item[\algname{LZ78}]\cite{LZ78}: Let $w \in \Sigma^{n}$ be the input
  string. The \algname{LZ78} algorithm computes the factorization $w =
  f_1 f_2 \dots f_{z_{78}}$ (the elements of which %
  are called \emph{phrases}) such that for every $i \in [1 \dd z_{78}]$,
  it holds either that $f_{i} \in \Sigma$ (if $w[|f_1 f_2 \dots
  f_{i-1}| + 1]$ is the leftmost occurrence of that symbol in $w$), or
  $f_i$ is the longest prefix of $f_i \dots f_{z_{78}}$ such that
  there exists $i' \in [1 \dd i)$ satisfying $f_{i'} c = f_i$ for some
  $c \in \Sigma$. This parsing can be easily encoded as an SLG of
  size $3z_{78}$.

\item[\algname{LZD}]\cite{lzd}: Let $w \in \Sigma^n$ be the input
  string. The \algname{LZD} algorithm factorizes $w$ into $f_1 f_2
  \cdots f_{m}$ such that $f_0 = \emptystring$, and for $1 \leq i \leq
  m$, $f_i = f_{i_1}f_{i_2}$ where $f_{i_1}$ is the longest prefix of
  $w[k \dd n]$ with $f_{i_1} \in \{f_j : 1 \leq j < i\} \cup \Sigma$,
  $f_{i_2}$ is the longest prefix of $w[k + |f_{i_1}| \dd n]$ with
  $f_{i_2} \in \{f_j : 0 \leq j < i\} \cup \Sigma$, and $k = |f_1
  \cdots f_{i-1}| + 1$. Intuitively, at step $i$ , $1 \leq i \leq m$,
  \algname{LZD} computes $f_{i_1}$ as the longest prefix of the
  unprocessed string among $f_1,...,f_{i-1}$ or $\Sigma$. It then
  analogously computes $f_{i_2}$ for remaining suffix of $w$ (or sets
  $f_{i_2} = \emptystring$ if the remaining suffix is empty). The
  $i$th phrase is then defined as $f_{i} = f_{i_1} f_{i_2}$. Note,
  that we can represent this factorization as an SLG by creating a
  nonterminal $N_i$ for each factor $f_i$, and then creating the
  starting nonterminal $S$ with $N_1 \cdots N_{m}$ as the definition.
  The size of this SLG is $3m$.
\end{description}

\section{Technical Overview}\label{sec:technical-overview}

\vspace{-1.0ex}
Due to space constraints, here we present the overview of the
basic grammar boosting (for data structures), and defer
further generalizations to
\cref{sec:cfg-prior,sec:cfg-challenge,sec:rna-prior,sec:rna-challenge}.

\subsection{The Framework of Verbin and Yu}\label{sec:overview-verbin-yu}

The study of data structure lower bounds on grammar-compressed strings
was pioneered by Verbin and Yu~\cite{VerbinY13}. Below we provide the
summary of their techniques.

\begin{definition}[Verbin and Yu~\cite{VerbinY13}]\label{def:answer-string}
  Let $\Pts \subseteq [1 \dd m]^2$ be a set of $|\Pts| = m$ points on
  an $m \times m$ grid. By $A(\Pts)$ we denote a binary string of
  length $m^2$ defined such that for every $x, y \in [1 \dd m]$,
  \[
    A(\Pts)[x + (y-1)m] = |\{(x',y') \in \Pts : x' \leq x\text{ and
    }y' \leq y\}| \bmod 2.
  \]
\end{definition}

Verbin and Yu called $A(\Pts)$ the \emph{answer string}, as it encodes
the answers for all possible parity range counting queries on the set
$\Pts$. Any such query, given $(x,y) \in [1 \dd m]^2$, returns the
parity of the number of points from $\Pts$ in the range $[1 \dd x]
\times [1 \dd y]$ (note that rows are enumerated bottom to top).
Verbin and
Yu proved the following result.

\begin{lemma}[Verbin and Yu~\cite{VerbinY13}]\label{lm:answer-string-grammar-size}
  Assume that $m$ is a power of two.  Let $\Pts \subseteq [1 \dd m]^2$
  be a set of $|\Pts| = m$ points on an $m \times m$ grid. There
  exists an admissible SLG $G = (V, \Sigma, R, S)$
  (\cref{def:admissible}) of height $\bigO(\log m)$ such that
  $L(G) = \{A(\Pts)\}$ and $|G| = \bigO(m \log m)$.
\end{lemma}

The main idea in the above proof is as follows. Let $M$ be such that
for $x, y \in [1 \dd m]$, it holds $M[y,x] = A(\Pts)[(y-1)m + x]$.
Let $\mathcal{T}_y$ be a perfect binary tree on $m$ leafs, such that
the $x$th leftmost leaf of $\mathcal{T}_y$ is associated with the
symbol $M[y,x]$, and each internal node is associated with a substring
obtained by concatenating the substrings of its left and right child.
We thus have $m$ trees, each corresponding to a row in $M$. The
grammar for $A(\Pts)$ is constructed as follows.

We process the rows of $M$ bottom-up, maintaining the invariant that
after the first $k$ rows are processed, for every nonterminal $X$ in
the current grammar $G$ it holds $\expgen{G}{X} = s$ for some $s$ that
is either a substring or a negation of a substring represented by one
of the internal nodes in $\{\mathcal{T}_1, \dots, \mathcal{T}_k\}$.
Conversely, for every internal node $v$ in one of the trees in
$\{\mathcal{T}_1, \ldots, \mathcal{T}_k\}$, letting $s$ be the
substring associated with $v$, there exist nonterminals in the current
grammar expanding to $s$ and a negation of $s$.  With such structure
we can clearly ensure that $|\rhsgen{G}{X}| = 2$.  The $(k+1)$st row
of $M$ is processed as follows. Suppose that there exist $q$ points
$(x_1, y_1), \dots, (x_q, y_q)$ in $\Pts$ with $y$-coordinate $k +
1$. Assume $x_1 < \dots < x_q$. Observe that $M[k\,{+}\,1,\cdot\ ]$
can be obtained from $M[k,\cdot\ ]$ by first negating all bits in
$M[k,x_1 \dd m]$, then negating all bits in $M[k,x_2 \dd m]$, and so
on.  Consider updating $\mathcal{T}_{k+1}$ to represent $M[k+1,\dots]$
as we add $(x_1,y_1), \ldots, (x_q,y_q)$ to $\Pts$. Let $i \in [1 \dd
q]$. First, we negate the symbol in the $x_i$th leaf $v$ of
$\mathcal{T}_{k+1}$.  We then perform a traversal from $v$ to the
root.  Whenever we arrive at node $v$ from its right child, we check
if the current grammar contains a nonterminal expanding to the
substring represented by $v$. If not, we introduce two new
nonterminals: one for the substring and one for its negation. If we
reach $v$ from its left child, we proceed analogously, except we first
negate the right child of $v$. We add at most $2\log m$
nonterminals. Over all rows of $M$, this amounts to $2m \log m$
nonterminals. The height of every nonterminal in the grammar is at
most $\log m$.

We now ensure that there exists a single nonterminal $S$ whose
expansion is equal to the entire string $A(\Pts)$. For $i \in [m \dd
2m)$, let $R_i$ be a nonterminal whose expansion is the $i$th lowest
row of $M$. For $i = m - 1, \dots, 1$, we add a nonterminal $R_i$ to a
grammar with a definition $\rhs{R_i} = R_{2i} R_{2i + 1}$. It is easy
to see that the nonterminal $R_1$ then satisfies $\exp{R_1} =
A(\Pts)$.  The nonterminals created in this process again form a
perfect binary tree of height $\log m$. We add exactly $m - 1$
nonterminals, and hence the total number of nonterminals in the output
grammar is $\bigO(m \log m)$. The above step increases the height of
the grammar by $\log m$. Thus, the final height is still $\bigO(\log
m)$.

\begin{theorem}[P\u{a}tra\c{s}cu~\cite{Patrascu07}]\label{th:parity-range-counting-lower-bound}
  In the cell-probe model, there is no data structure that, for every
  set $\Pts$ of $|\Pts| = m$ points on an $m \times m$ grid, achieves
  $\bigO(m \log^{c} m)$ space $($where $c = \bigO(1))$ and implements
  parity range counting queries on $\Pts$ in $o(\log m / \log \log m)$
  time.
\end{theorem}

\begin{theorem}[Verbin and Yu~\cite{VerbinY13}]\label{th:slp-random-access-lower-bound}
  In the cell-probe model, there is no data structure that, for every
  string $T$ of length $N$ and every SLG $G$ of $T$ such that $L(G) =
  \{T\}$, achieves $\bigO(|G| \log^c N)$ space $($where $c =
  \bigO(1))$ and implements random access queries to $T$ in $o(\log N
  / \log \log N)$ time.
\end{theorem}

The key idea in the proof of the above fact is as follows. Suppose
that there exists a data structure $D$ that, given any SLG $G$ for a
string $T \in \Sigma^N$, uses $\bigO(|G| \log^c N)$ space (where $c =
\bigO(1)$) and answers random access queries on $T$ in $o(\log N /
\log \log N)$ time. Let $\Pts \subseteq [1 \dd m]^2$ be any set of
$|\Pts| = m$ points on an $m \times m$ grid. Assume for simplicity
that $m$ is a power of two (otherwise, letting $m'$ be the smallest
power of two satisfying $m' \geq m$, we apply the proof for $\Pts' =
\Pts \cup \{(p, p)\}_{p \in (m \dd m']}$; note that the answer to any
range counting query on $\Pts$ is equal to the answer on $\Pts'$). By
\cref{lm:answer-string-grammar-size}, there exists an admissible SLG
$G_{\Pts} = (V_{\Pts}, \{0, 1\}, R_{\Pts}, S_{\Pts})$ such that
$L(G_{\Pts}) = \{A(\Pts)\}$ is an answer string for $\Pts$
(\cref{def:answer-string}), and it holds $|G_{\Pts}| = \bigO(m \log
m)$. Recall that $|A(\Pts)| = m^2$. Let $D'$ denote the structure $D$
for $G_{\Pts}$. By $|G_{\Pts}| = \bigO(m \log m)$ and the assumption,
$D$ uses $\bigO(|G_{\Pts}| \log^c |A(\Pts)|) = \bigO(m \log m
\log^c(m^2)) = \bigO(m \log^{1 + c} m)$ space, and implements random
access to $A(\Pts)$ in $o(\log(m^2) / \log \log(m^2)) = o(\log m /
\log \log m)$ time. Given $D'$ and any $(x, y) \in [1 \dd m]^2$, we
can thus answer in $o(\log m / \log \log m)$ the parity range query on
$\Pts$ with arguments $(x, y)$ by issuing a random access query on
$A(\Pts)$ with position $j = x + (y-1)m$. Thus, the existence of $D'$
contradicts \cref{th:parity-range-counting-lower-bound}.

\subsection{Grammar Boosting}\label{sec:overview-boosting}

\paragraph{The Main Idea}

We now describe our new technique.  Consider any $T \in \Sigma^N$ and
assume that there exists a grammar $G = (V, \Sigma, R, S)$ such that
$L(G) = \{T\}$ and $|G| = n$.  Assume that all nonterminals in $V$
occur in the parse tree of $G$, and that $G$ is \emph{admissible}
(\cref{def:admissible}).  All grammars that we start with will
satisfy these properties, but it is easy to see that any grammar can
be transformed into an admissible grammar generating the same string
without asymptotically increasing its size, and all nonterminals that
do not appear in the parse tree can be removed. Let $\Sigma' = \Sigma
\cup \{\dol_{i} : i \in [1 \dd |V|]\}$ and $G' = (V, \Sigma', R', S)$ be a grammar
with the same set of
nonterminals and the starting nonterminal as $G$, but with a unique
sentinel symbol in every definition, i.e., such that for every
$i \in [1 \dd |V|]$, it holds $\rhsgen{G'}{N_i} = A \dol_i B$, where
$A, B \in V \cup \Sigma$ are such that $\rhsgen{G}{N_i} = AB$.
Consider any ordering $N_1, \dots, N_{|V|}$ of nonterminals in $V$ such that
$|\expgen{G}{N_1}| \leq \dots \leq |\expgen{G}{N_{|V|}}|$ and let $T'
= \bigodot_{j = 1, \dots, |V|} \expgen{G'}{N_j} \cdot \hash_{2j-1}
\cdot \expgen{G'}{N_j} \cdot \hash_{2j}$.

\begin{description}[style=sameline,itemsep=1ex,font={\normalfont\itshape}]
\item[Observation: $|T'| = \bigO(|T|^2)$ and $T$ is an easily
  identifiable subsequence of $T'$.]  Note that every $X \in V$, we
  have $|\expgen{G'}{X}| = 2|\expgen{G}{X}| - 1$ and, letting $m =
  |\expgen{G}{X}|$, it holds $\expgen{G'}{X}[2j - 1] =
  \expgen{G}{X}[j]$ for every $j \in [1 \dd m]$. By definition of
  $T'$, we therefore have $|T'| = 4 \sum_{X \in V}|\expgen{G}{X}|$
  (\cref{lm:alpha-length}). Consequently, since for every $X \in V$ it
  holds $|\expgen{G}{X}| \leq |T|$, and $|V| \leq |T|$, we obtain
  $|T'| \leq 4|T|^2$. For the second claim, note that since there
  exists $j \in [1 \dd |V|]$ such that $S = N_j$, it follows by the
  above that for some $\delta \geq 0$, we have $T[j] = T'[\delta + 2j
  - 1]$, where $j \in [1 \dd |T|]$ (\cref{lm:delta}).
\end{description}

By the above observation, $T'$ is plain enough that we can use it to
access symbols of $T$ without incurring any penalty in the runtime.
We now outline how to prove that $T'$ is simultaneously structured
strongly enough, so that the algorithms studied in the paper
compress it to size $\bigO(n)$.

\vspace{-1ex}
\paragraph{Analysis of Nonglobal Algorithms}

As an illustration, we consider the processing of $T'$ using
\algname{Sequential} (see \cref{sec:algs-nonglobal}).  Denote
$\Sigma'' = \Sigma' \cup \{\hash_{j} : j \in [1 \dd 2|V|]\}$. For
every $k \in [1 \dd |V|]$, let $G_k = (V_k, \Sigma'', R_k, S_k)$ be
such that $V_k = \{N_1, \dots, N_k, S_k\}$, for every $i \in [1 \dd
k]$, $\rhsgen{G_k}{N_i} = \rhsgen{G'}{N_i}$, and $\rhsgen{G_k}{S_k}
= \bigodot_{i=1,\dots,k} N_i \cdot \hash_{2i-1} \cdot N_i \cdot
\hash_{2i}$.  We claim that after $8k$ steps of \algname{Sequential},
the algorithm has processed the prefix $\bigodot_{j = 1, \dots, k}
\expgen{G'}{N_j} \cdot \hash_{2j-1} \cdot \expgen{G'}{N_j} \cdot
\hash_{2j}$ of $T'$, and the resulting grammar is isomorphic to
$G_k$. We proceed by induction on $k$. The inductive base is easily
verified. To show the inductive step, let $A, B \in V \cup \Sigma$ be
such that $\rhsgen{G'}{N_k} = A \cdot \dol_{k} \cdot B$ and assume
that \algname{Sequential} processed $\bigodot_{j=1,\dots,k-1}
\expgen{G'}{N_j} \cdot \hash_{2j-1} \cdot \expgen{G'}{N_j} \cdot
\hash_{2j}$.  Thus, $\expgen{G'}{A} \cdot \dol_{k} \cdot
\expgen{G'}{B} \cdot \hash_{2k-1} \cdot \expgen{G'}{A} \cdot
\dol_{k} \cdot \expgen{G'}{B} \cdot \hash_{2k}$ is the prefix of the
remaining suffix. Note that during the next five steps, we process
$\expgen{G'}{A} \cdot \dol_{k} \cdot \expgen{G'}{B} \cdot
\hash_{2k-1} \cdot \expgen{G'}{A}$, and simply append five symbols $A'
\dol_{k} B' \hash_{2k-1} A'$ (where $A'$ and $B'$, respectively,
correspond to $\expgen{G'}{A}$ and $\expgen{G'}{B}$) to the definition
of the starting nonterminal.  Next, we create a new nonterminal $X$
capturing the repetition of $A'\dol_{k}$. In the seventh step, we
again create a new nonterminal $X'$ corresponding to the repetition of
$XB'$, and then remove $X$ (it now occurs only once). Finally, we
append $\hash_{2k}$. The result is isomorphic to $G_{k}$. We refer to
\cref{lm:sequential} for details.  The high-level analysis of
\algname{Sequitur} (\cref{sec:sequitur}) is similar, except each step
involves many smaller substeps (in which intermediate
grammars are partially completed versions of $G_k$).

\vspace{-1ex}
\paragraph{Analysis of Global Algorithms}

We now outline the %
proof that all global
algorithms on $T'$ output the same grammar as nonglobal algorithms
(\cref{sec:global}). The key difficulty in the analysis,
compared to nonglobal algorithms, is that replacements
leading up to the grammar isomorphic with $G_{|V|}$ do not occur in
order. We thus need to generalize the class of intermediate
grammars. We show that it suffices to consider $2^{|V|}$ grammars. We
define them as follows.
\begin{itemize}[itemsep=0pt,parsep=0pt]
\item Let $\{M_i : i \in [1 \dd |V|]\}$ be a set of fresh
  variables, i.e., such that $\{M_i\}_{i \in [1 \dd |V|]} \cap
  \{N_i\}_{i \in [1 \dd |V|]} = \emptyset$.  For every $X \in V \cup
  \Sigma$ and $\mathcal{I} \subseteq \{1, \dots, |V|\}$, by ${\rm
  bexp}_{\mathcal{I}}(X)$ we denote a string obtained by starting
  with $X$ and repeatedly expanding the nonterminals
  (according to their definition in $G'$) until only symbols in
  $\Sigma' \cup \{N_i\}_{i \in \mathcal{I}}$ are left.  Each
  occurrence of the remaining nonterminal from $\{N_i\}_{i \in
  \mathcal{I}}$ is then replaced with the matching symbol from
  $\{M_i\}_{i \in \mathcal{I}}$ (see \cref{def:bexp}).
\item For every $\mathcal{I} \subseteq \{1, \dots, |V|\}$, we let
  $G_{\mathcal{I}} = (V_{\mathcal{I}}, \Sigma'', R_{\mathcal{I}}, S)$,
  where $V_{\mathcal{I}} = \{S\} \cup \{M_i\}_{i \in
  \mathcal{I}}$. For any $i \in \mathcal{I}$, we let
  $\rhsgen{G_{\mathcal{I}}}{M_i} = {\rm bexp}_{\mathcal{I}}(X) \cdot
  \dol_{i} \cdot {\rm bexp}_{\mathcal{I}}(Y)$, where $X, Y \in V
  \cup \Sigma$ are such that $\rhsgen{G}{N_i} = XY$.  We also set
  $\rhsgen{G_{\mathcal{I}}}{S} = \bigodot_{i=1,\dots,|V|} {\rm
  bexp}_{\mathcal{I}}(N_i) \cdot \hash_{2i-1} \cdot {\rm
  bexp}_{\mathcal{I}}(N_i) \cdot \hash_{2i}$.  We denote $\mathbb{G}
  = \{G_{\mathcal{I}} : \mathcal{I} \subseteq \{1, \dots,
  |V|\}\}$. Note that $|\mathbb{G}| = 2^{|V|}$ and $L(G_{\mathcal{I}})
  = \{T'\}$ holds for every $\mathcal{I} \subseteq \{1, \dots, |V|\}$
  (\cref{lm:exp-bigG}).  Observe also that the initial grammar that
  every global algorithm starts with when processing $T'$ (see
  \cref{sec:algs-global}), is isomorphic with $G_{\emptyset}$.
\end{itemize}

\begin{description}[style=sameline,itemsep=0.8ex,font={\normalfont\itshape}]
\item[Observation 1: If $s$ is maximal with respect to
  $G_{\mathcal{I}}$, then $s = {\rm bexp}_{\mathcal{I}}(N_i)$ for some
  $i \in \{1, \dots, |V|\} \setminus \mathcal{I}$.]  Each symbol in
  $\{\hash_i\}_{i \in [1 \dd 2|V|]}$ occurs once on the right-hand
  side of $G_{\mathcal{I}}$. On the other hand, every second symbol
  in the remaining substrings of $G_{\mathcal{I}}$ belongs to
  $\{\dol_{i}\}_{i \in [1 \dd |V|]}$. Thus, by $|s| \geq 2$, one of
  them occurs in $s$. It is easy to check that every symbol in
  $\{\dol_{i}\}_{i \in \mathcal{I}}$ occurs once on the right-hand
  side of $G_{\mathcal{I}}$ (\cref{lm:dollar-on-rhs}). Thus, $s$
  contains $\dol_{i}$ for some $i \in \{1, \dots, |V|\} \setminus
  \mathcal{I}$.  An inductive argument shows that every occurrence of
  $\dol_{i}$ for such $i$ can be extended into an occurrence of
  ${\rm bexp}_{\mathcal{I}}(N_i)$ (\cref{lm:bexp-occ}). String ${\rm
  bexp}_{\mathcal{I}}(N_i)$ must therefore be a substring of
  $s$. Choosing as $i$ the maximal $t \in \{1, \dots, |V|\} \setminus
  \mathcal{I}$ such that $\dol_{t}$ occurs in $s$, we thus have $s =
  {\rm bexp}_{\mathcal{I}}(N_i)$, since every occurrence of ${\rm
  bexp}_{\mathcal{I}}(N_i)$ on the right-hand side of
  $G_{\mathcal{I}}$ is surrounded by either $\dol_{i'}$ with $i' >
  i$, or a symbol from $\{\hash_{i}\}_{i \in [1 \dd 2|V|]}$ (see
  \cref{lm:maximal-string} for details).

\item[Observation 2: The output of one step of every global algorithm
  on $G_{\mathcal{I}}$ with ${\rm bexp}_{\mathcal{I}}(N_i)$ as a
  maximal string is isomorphic to $G_{\mathcal{I} \cup \{i\}}$.]
  Denote $\mathcal{I}' = \mathcal{I} \cup \{i\}$ and $s = {\rm
  bexp}_{\mathcal{I}}(N_i)$.  First, we observe that by an inductive
  argument it follows that replacing every occurrence of $M_i$ on the
  right-hand size of $G_{\mathcal{I}'}$ with $s$, and removing
  nonterminal $M_i$, results in $G_{\mathcal{I}}$ (see
  \cref{cor:sub}). On the other hand, no two
  occurrences of $s$ on the right-hand side of $G_{\mathcal{I}}$ are
  overlapping (\cref{lm:bexp-occ}). These two together imply
  the claim, since the output of a single step of a global algorithm is
  then not determined by the order of replacements; see \cref{rm:nonov}
  and \cref{lm:nonov-s}.

\item[Observation 3: There exists a maximal string with respect to
  $G_{\mathcal{I}}$ if and only if $\mathcal{I} \neq \{1, \dots,
  |V|\}$.]  The first implication follows from above. For the second
  implication, observe that if $\mathcal{I} \neq \{1, \dots, |V|\}$,
  then, letting $i \in \{1, \dots, |V|\} \setminus \mathcal{I}$ and $s
  = {\rm bexp}_{\mathcal{I}}(N_i)$, we have $|s| \geq 2$, and $s$ has
  at least two non-overlapping occurrences on the right-hand side of
  $G_{\mathcal{I}}$ (in the definition of $S$).  Thus, either $s$ is
  maximal, or it can be extended into a maximal string
  (\cref{def:maximal-string}); see
  \cref{lm:maximal-substring-existence}.
\end{description}

By the above, the intermediate grammars computed by every global
algorithm on $T'$ are isomorphic to a chain $G_{\emptyset},
G_{\mathcal{I}_1}, \dots, G_{\mathcal{I}_{|V|}}$ such that
$\mathcal{I}_1 \subsetneq \mathcal{I}_2 \subsetneq \dots \subsetneq
\mathcal{I}_{|V|}$.  Thus, %
$\mathcal{I}_{|V|} = \{1,
\dots, |V|\}$, and hence the final grammar is isomorphic to $G_{\{1,
\dots, |V|\}}$, which has size $\bigO(|V|) = \bigO(n)$.

\vspace{-1ex}
\paragraph{Putting Everything Together}

\cref{th:main} follows from the above analysis as follows.  Suppose
that for some \algname{Alg} as in \cref{th:main}, there exists a
structure $D$ that for any $T \in \Sigma^N$ uses
$\bigO(|\algname{Alg}(T)|\log^{c} N)$ space (where $c = \bigO(1)$) and
answers random access queries on $T$ in $o(\log N / \log \log N)$
time.  Let $\Pts \subseteq [1 \dd m]^2$ be any set of $|\Pts| = m$
points on an $m \times m$ grid.  By
\cref{lm:answer-string-grammar-size}, there exists an admissible
grammar $G_{\Pts} = (V_{\Pts}, \{0, 1\}, R_{\Pts}, S_{\Pts})$ such
that $L(G_{\Pts}) = \{A(\Pts)\}$ is the answer string for $\Pts$
(\cref{def:answer-string}), and it holds $|G_{\Pts}| = \bigO(m \log
m)$. Let us now consider the string $T'$ (defined as in the beginning
of \cref{sec:overview-boosting}) for $T = A(\Pts)$ and $G = G_{\Pts}$.
As noted in the
initial observation of \cref{sec:overview-boosting}, it holds $|T'| =
\bigO(|T|^2) = \bigO(m^4)$, and there exists $\delta \geq 0$, such
that $T[j] = T'[\delta + 2j - 1]$, for every $j \in [1 \dd |T|]$.  Let
$G_{T'} = \algname{Alg}(T')$ be the output of \algname{Alg} on
$T'$. By the above discussion, we have $|G_{T'}| = \bigO(|G_{\Pts}|) =
\bigO(m \log m)$. Let $D'$ denote a data structure consisting of the
following two components:
\begin{enumerate}[itemsep=0pt,parsep=0pt]
\item The structure $D$ for string $T'$. By $|G_{T'}| = \bigO(m
  \log m)$ and the above assumption, $D$ uses
  $\bigO(|\algname{Alg}(T')| \log^c |T'|) = \bigO(m \log m \log^c
  (m^4)) = \bigO(m \log^{1 + c} m)$ space, and implements access to
  $T'$ in $o(\log |T'| / \log \log |T'|) \,{=}\, o(\log (m^4) / \log
  \log (m^4)) \allowbreak = o(\log m / \log \log m)$ time,
\item The position $\delta \geq 0$, as defined above.
\end{enumerate}
Observe that given the structure $D'$ and any $(x, y) \in [1 \dd
m]^2$, we can answer in $o(\log m / \log \log m)$ the parity range
query on $\Pts$ with arguments $(x, y)$ by %
a random access
query to $T'$ with position $j = \delta + 2j' - 1$, where $j' = x +
(y-1)m$. Thus, the existence of $D'$ contradicts
\cref{th:parity-range-counting-lower-bound}.

\section{Random Access}\label{sec:random-access}

\subsection{Analysis of Global Algorithms}\label{sec:global}

\begin{definition}\label{def:alpha}
  Let $G = (V, \Sigma, R, S)$ be an admissible SLG. Assume that the
  sets $\Sigma$, $\{\dol_i : i \in [1 \dd |V|]\}$, and $\{\hash_i : i
  \in [1 \dd 2|V|]\}$ are pairwise disjoint. Denote $\Sigma' = \Sigma
  \cup \{\dol_i : i \in [1 \dd |V|]\}$ and $\Sigma'' = \Sigma' \cup
  \{\hash_i : i \in [1 \dd 2|V|]\}$.  By $\alpha(G)$, we denote the
  subset of $\Sigma''^{*}$ such that for every $w \in \Sigma''^{*}$,
  $w \in \alpha(G)$ holds if and only if there exists a sequence
  $(N_i)_{i \in [1 \dd |V|]}$ such that:
  \begin{itemize}[itemsep=0pt,parsep=0pt]
    \item $\{N_i : i \in [1 \dd |V|]\} = V$,
    \item $|\expgen{G}{N_i}| \leq |\expgen{G}{N_{i+1}}|$ holds for $i
      \in [1 \dd |V|)$, and
    \item $w = \bigodot_{i=1,\dots,|V|} \expgen{G'}{N_i} \cdot
      \hash_{2i-1} \cdot \expgen{G'}{N_i} \cdot \hash_{2i}$,
  \end{itemize}
  where $G' = (V, \Sigma', R', S)$ is defined so that for every $i \in
  [1 \dd |V|]$, it holds $\rhsgen{G'}{N_i} = A \cdot \dol_{i} \cdot B$
  (where $A, B \in V \cup \Sigma$ are such that $\rhsgen{G}{N_i} =
  AB$).
\end{definition}

Informally, in the above construction, given an admissible SLG $G$, we
first create an auxiliary SLG $G'$ which augments each nonterminal so
that its expansion contains a sentinel symbol unique to that
nonterminal. We then let $\alpha(G)$ be the set of all strings
obtained by first ordering all nonterminals of $G'$ according to the
length of their expansion (resolving the ties arbitrarily), and then
concatenating their expansions (each repeated twice) in this order,
with additional sentinel symbols inserted in between.

\begin{observation}\label{obs:exp}
  Let $G$ be an admissible SLG and let $G'$ be as in \cref{def:alpha}.
  For every $i \in \{1, \dots, |V|\}$, it holds
  $$\expgen{G'}{N_i} = \expgen{G'}{X} \cdot \dol_{i} \cdot
  \expgen{G'}{Y},$$ where $X, Y \in V \cup \Sigma$ are such
  that $\rhsgen{G}{N_i} = XY$.
\end{observation}

For the duration of this section, let us fix some admissible SLG $G =
(V, \Sigma, R, S)$ and some ordering $(N_i)_{i \in [1 \dd |V|]}$ of
$V$ satisfying $|\expgen{G}{N_i}| \leq |\expgen{G}{N_{i+1}}|$ for $i
\in [1 \dd |V|)$.  We let $G' = (V, \Sigma', R', S)$ be the
corresponding SLG defined as in \cref{def:alpha}, i.e., we have
$\Sigma' = \Sigma \cup \{\dol_{i} : i \in [1 \dd |V|]\}$, and the set
of rules $R'$ is defined so that for every $i \in [1 \dd |V|]$,
$\rhsgen{G'}{N_i} = A \cdot \dol_i \cdot B$, where $A, B \in V \cup
\Sigma$ are such that $\rhsgen{G}{N_i} = AB$.  We also denote
$\Sigma'' = \Sigma' \cup \{\hash_i : i \in [1 \dd 2|V|]\}$ and let $w
\in \alpha(G)$ be the string corresponding to the above ordering
$(N_i)_{i \in [1 \dd |V|]}$, i.e., $w = \bigodot_{i=1,\dots,|V|}
\expgen{G'}{N_i} \cdot \hash_{2i-1} \cdot \expgen{G'}{N_i} \cdot
\hash_{2i}$. Let $\{M_1, M_2, \dots, M_{|V|}\}$ be a set of $|V|$
elements such that $\{M_1, \dots, M_{|V|}\} \cap (V \cup \Sigma'') =
\emptyset$.

\begin{observation}\label{ob:interleave}
  Let $X \in \Sigma \cup V$ and $m = |\expgen{G}{X}|$. Then,
  $|\expgen{G'}{X}| = 2m - 1$, and for every $j \in [1 \dd m]$, it
  holds $\expgen{G}{X}[j] = \expgen{G'}{X}[2j - 1]$.
\end{observation}

\begin{lemma}\label{lm:delta}
  Let $L(G) = \{u\}$ and $m = |u|$.  There exists $\delta \geq 0$ such
  that for every $j \in [1 \dd m]$, it holds $u[j] = w[\delta + 2j -
  1]$.
\end{lemma}
\begin{proof}
  Recall that $u = \expgen{G}{S}$.  By \cref{ob:interleave},
  $|\expgen{G'}{S}| = 2m - 1$, and for every $j \in [1 \dd m]$, it
  holds $u[j] = \expgen{G'}{S}[2j - 1]$. Observe now that by
  \cref{def:alpha}, the string $\expgen{G'}{S}$ occurs in $w$, i.e.,
  there exists $\delta \geq 0$ such that $w[\delta + 1 \dd \delta + 2m
  - 1] = \expgen{G'}{S}$.  By the earlier observation we thus have,
  for every $j \in [1 \dd m]$, $u[j] = \expgen{G'}{S}[2j - 1] =
  w[\delta + 2j - 1]$.
\end{proof}

\begin{lemma}\label{lm:alpha-length}
  It holds $|w| = 4 \sum_{N \in V}|\expgen{G}{N}|$.
\end{lemma}
\begin{proof}
  By \cref{ob:interleave}, for every $N \in V$, it holds
  $|\expgen{G'}{N}| = 2|\expgen{G}{N}| - 1$. Thus,
  \begin{align*}
    |w|
      &= |\textstyle\bigodot_{i=1,\dots,|V|}
         \expgen{G'}{N_i} \cdot \hash_{2i-1} \cdot
         \expgen{G'}{N_i} \cdot \hash_{2i}|\\
      &= \textstyle\sum_{N \in V} 2 + 2|\expgen{G'}{N}|\\
      &= \textstyle\sum_{N \in V} 2 + 2(2|\expgen{G}{N}| - 1)\\
      &= 4\textstyle\sum_{N \in V} |\expgen{G}{N}|.
      \qedhere
  \end{align*}
\end{proof}

\begin{definition}\label{def:bexp}
  Consider any $\mathcal{I} \subseteq \{1, \dots, |V|\}$.  Let ${\rm
  bexp}_{\mathcal{I}} : \Sigma \cup V \rightarrow (\Sigma' \cup
  \{M_i\}_{i \in \mathcal{I}})^{+}$ be an auxiliary function defined
  so that for every $c \in \Sigma$, it holds ${\rm
  bexp}_{\mathcal{I}}(c) = c$, and for every $i \in \{1, \dots,
  |V|\}$, we have
  \vspace{2ex}
  \[
    {\rm bexp}_{\mathcal{I}}(N_i) =
    \begin{cases}
      M_i & \text{if }i \in \mathcal{I}, \\
      {\rm bexp}_{\mathcal{I}}(X) \cdot \dol_{i}
      \cdot {\rm bexp}_{\mathcal{I}}(Y) & \text{otherwise},
    \end{cases}
    \vspace{2ex}
  \]
  where $X, Y \in V \cup \Sigma$ are such that $\rhsgen{G}{N_i} = XY$.
\end{definition}

\begin{remark}\label{rm:bexp}
  To see the motivation for its name, observe that the function ${\rm
  bexp}_{\mathcal{I}}$ performs a ``bounded expansion'' of any
  element of $V \cup \Sigma$ (compare to \cref{obs:exp}).  More
  precisely, for every $i \in [1 \dd |V|]$,
  ${\rm bexp}_{\mathcal{I}}(N_i)$ returns a string obtained by
  initializing the output string to $N_i$, and then
  performing a minimal number of operations that replace a nonterminal
  with its definition in $G'$ so that the resulting string contains only
  symbols in $\Sigma' \cup \{N_i\}_{i \in \mathcal{I}}$.  Each
  occurrence of the remaining nonterminal from $\{N_i\}_{i \in
  \mathcal{I}}$ is then replaced with the matching symbol from
  $\{M_i\}_{i \in \mathcal{I}}$.  For example, if $\mathcal{I} =
  \emptyset$, then for every $i \in \{1, \dots, |V|\}$, we have ${\rm
  bexp}_{\mathcal{I}}(N_i) = \expgen{G'}{N_i}$. On the other hand,
  if $\mathcal{I} = \{1, 2, \dots, |V|\}$, then for every $i \in \{1,
  \dots, |V|\}$, it holds ${\rm bexp}_{\mathcal{I}}(N_i) = M_i$.
\end{remark}

\begin{definition}\label{def:bigG}
  For any $\mathcal{I} \subseteq \{1, \dots, |V|\}$, we define a
  grammar $G_{\mathcal{I}} = (V_{\mathcal{I}}, \Sigma'',
  R_{\mathcal{I}}, S)$, where $V_{\mathcal{I}} = \{S\} \cup \{M_i\}_{i
  \in \mathcal{I}}$.  For any $i \in \mathcal{I}$, we let
  \[
    \rhsgen{G_{\mathcal{I}}}{M_i} = {\rm bexp}_{\mathcal{I}}(X) \cdot
    \dol_{i} \cdot {\rm bexp}_{\mathcal{I}}(Y),
  \]
  where $X, Y \in V \cup \Sigma$ are such that $\rhsgen{G}{N_i} = XY$.
  We also set
  \[
    \rhsgen{G_{\mathcal{I}}}{S} = \bigodot_{i=1,\dots,|V|} {\rm
    bexp}_{\mathcal{I}}(N_i) \cdot \hash_{2i-1} \cdot {\rm
    bexp}_{\mathcal{I}}(N_i) \cdot \hash_{2i}.
  \]
  We then denote $\mathbb{G} = \{G_{\mathcal{I}} : \mathcal{I}
  \subseteq \{1, \dots, |V|\}\}$.
\end{definition}

\begin{remark}
  In a series of lemmas we will now establish that every intermediate
  grammar occurring during the processing of $w$ using a global
  algorithm is isomorphic to one of the grammars in
  $\mathbb{G}$. To this end, we will establish the characterization of
  maximal strings in elements of $\mathbb{G}$; in particular, that
  there exists precisely one element of $\mathbb{G}$ that does not
  contain any maximal strings.  This will yield the unique grammar in
  $\mathbb{G}$ that is isomorphic to the output of every global
  algorithm on $w$.
\end{remark}

\begin{lemma}\label{lm:exp-bigG}
  For every $\mathcal{I} \subseteq \{1, \dots, |V|\}$, it holds
  $L(G_{\mathcal{I}}) = \{w\}$.
\end{lemma}
\begin{proof}

  First, we prove by induction on $i$, that for every $i \in \{1,
  \dots, |V|\}$, it holds $\expgen{G_{\mathcal{I}}}{{\rm
  bexp}_{\mathcal{I}}(N_i)} \allowbreak = \expgen{G'}{N_i}$.  Note that
  $\expgen{G_{\mathcal{I}}}{{\rm bexp}_{\mathcal{I}}(N_i)}$ is well
  defined by ${\rm bexp}_{\mathcal{I}}(N_i) \,{\in}\, (\Sigma' \cup
  \{M_i\}_{i \in \mathcal{I}})^{+} \,{\subseteq}\, (\Sigma'' \cup
  V_{\mathcal{I}})^{+}$.

  We first prove the induction base. Let $X, Y \in V \cup \Sigma$ be
  such that $\rhsgen{G}{N_1} = XY$. The condition $|\expgen{G}{N_i}|
  \leq \dots \leq |\expgen{G}{N_{|V|}}|$ implies that $X, Y \in
  \Sigma$.  By $\rhsgen{G'}{N_1} = X \cdot \dol_{1} \cdot Y$, we
  then have $\expgen{G'}{N_1} = \expgen{G'}{X} \cdot
  \expgen{G'}{\dol_{1}} \cdot \expgen{G'}{Y} = X \cdot \dol_{1}
  \cdot Y$. Consider now two cases:
  \begin{itemize}[itemsep=0pt,parsep=0pt]
  \item First, assume that $1 \not\in \mathcal{I}$. By
    \cref{def:bexp}, we then immediately obtain ${\rm
    bexp}_{\mathcal{I}}(N_1) = {\rm bexp}_{\mathcal{I}}(X) \cdot
    \dol_{1} \cdot {\rm bexp}_{\mathcal{I}}(Y) = X \cdot \dol_{1}
    \cdot Y$.
  \item Let us now assume that $1 \in \mathcal{I}$. Applying
    \cref{def:bexp} yields ${\rm bexp}_{\mathcal{I}}(N_1) = M_1$. On
    the other hand, by \cref{def:bigG}, $\rhsgen{G_{\mathcal{I}}}{M_1}
    = {\rm bexp}_{\mathcal{I}}(X) \cdot \dol_{1} \cdot {\rm
    bexp}_{\mathcal{I}}(Y) = X \cdot \dol_{1} \cdot Y$.
    Consequently, $\expgen{G_{\mathcal{I}}}{{\rm
    bexp}_{\mathcal{I}}(N_1)} = \expgen{G_{\mathcal{I}}}{M_1} =
    \expgen{G_{\mathcal{I}}}{X} \cdot
    \expgen{G_{\mathcal{I}}}{\dol_{1}} \cdot
    \expgen{G_{\mathcal{I}}}{Y} = X \cdot \dol_{1} \cdot Y$.
  \end{itemize}

  We now prove the induction step. Consider $i > 1$.  Let $X, Y \in V
  \cup \Sigma$ be such that $\rhsgen{G}{N_i} = XY$. Then,
  $\rhsgen{G'}{N_i} = X \cdot \dol_{i} \cdot Y$, and hence
  $\expgen{G'}{N_i} = \expgen{G'}{X} \cdot \dol_{i} \cdot
  \expgen{G'}{Y}$. Observe also that we then have
  $\expgen{G_{\mathcal{I}}}{{\rm bexp}_{\mathcal{I}}(X)} =
  \expgen{G'}{X}$.  To see this, note that if $X \in V$, then
  $\rhsgen{G'}{N_i} = X \cdot \dol_{i} \cdot Y$ implies that
  $|\expgen{G}{X}| < |\expgen{G}{N_i}|$. Hence, there exists $j \in
  \{1, \dots, i - 1\}$ such that $X = N_j$. By the inductive
  assumption, we then have $\expgen{G_{\mathcal{I}}}{{\rm
  bexp}_{\mathcal{I}}(X)} = \expgen{G_{\mathcal{I}}}{{\rm
  bexp}_{\mathcal{I}}(N_j)} = \expgen{G'}{N_j} =
  \expgen{G'}{X}$. Otherwise (i.e., if $X \in \Sigma$), it follows by
  definition that $\expgen{G_{\mathcal{I}}}{{\rm
  bexp}_{\mathcal{I}}(X)} = \expgen{G_{\mathcal{I}}}{X} = X =
  \expgen{G'}{X}$. We have thus proved $\expgen{G_{\mathcal{I}}}{{\rm
  bexp}_{\mathcal{I}}(X)} = \expgen{G'}{X}$.  Analogously, it
  holds $\expgen{G_{\mathcal{I}}}{{\rm bexp}_{\mathcal{I}}(Y)} =
  \expgen{G'}{Y}$.  We are now ready to prove the induction
  step. Consider two cases:
  \begin{itemize}[itemsep=0pt,parsep=0pt]
  \item First, assume $i \not\in \mathcal{I}$. By \cref{def:bexp}, we
    then have ${\rm bexp}_{\mathcal{I}}(N_i) = {\rm
    bexp}_{\mathcal{I}}(X) \cdot \dol_{i} \cdot {\rm
    bexp}_{\mathcal{I}}(Y)$. By combining this with the above
    properties of $X$ and $Y$, we thus have
    $\expgen{G_{\mathcal{I}}}{{\rm bexp}_{\mathcal{I}}(N_i)} =
    \expgen{G_{\mathcal{I}}}{{\rm bexp}_{\mathcal{I}}(X)} \cdot
    \expgen{G_{\mathcal{I}}}{\dol_{i}} \cdot
    \expgen{G_{\mathcal{I}}}{{\rm bexp}_{\mathcal{I}}(Y)} =
    \expgen{G'}{X} \cdot \dol_{i} \cdot \expgen{G'}{Y} =
    \expgen{G'}{N_i}$.
  \item Let us now assume $i \in \mathcal{I}$. Applying
    \cref{def:bexp} yields ${\rm bexp}_{\mathcal{I}}(N_i) = M_i$. On
    the other hand, by \cref{def:bigG}, $\rhsgen{G_{\mathcal{I}}}{M_i}
    = {\rm bexp}_{\mathcal{I}}(X) \cdot \dol_{i} \cdot {\rm
    bexp}_{\mathcal{I}}(Y)$.  Combining again with the above
    observations about $X$ and $Y$ we thus obtain
    $\expgen{G_{\mathcal{I}}}{{\rm bexp}_{\mathcal{I}}(N_i)} =
    \expgen{G_{\mathcal{I}}}{M_i} = \expgen{G_{\mathcal{I}}}{{\rm
    bexp}_{\mathcal{I}}(X)} \cdot
    \expgen{G_{\mathcal{I}}}{\dol_{1}} \cdot
    \expgen{G_{\mathcal{I}}}{{\rm bexp}_{\mathcal{I}}(Y)} =
    \expgen{G'}{X} \cdot \dol_{i} \cdot \expgen{G'}{Y} =
    \expgen{G'}{N_i}$.
  \end{itemize}

  Utilizing the above property, and applying
  \cref{def:bigG,def:alpha}, we thus obtain:
  \begin{align*}
    L(G_{\mathcal{I}})
      &= \{\expgen{G_{\mathcal{I}}}{S}\}\\
      &= \{\textstyle\bigodot_{i=1,\dots,|V|}
         \expgen{G_{\mathcal{I}}}{{\rm bexp}_{\mathcal{I}}(N_i)} \cdot
         \expgen{G_{\mathcal{I}}}{\hash_{2i-1}} \cdot 
         \expgen{G_{\mathcal{I}}}{{\rm bexp}_{\mathcal{I}}(N_i)} \cdot
         \expgen{G_{\mathcal{I}}}{\hash_{2i}}\}\\
      &= \{\textstyle\bigodot_{i=1,\dots,|V|}
         \expgen{G'}{N_i} \cdot
         \hash_{2i-1} \cdot 
         \expgen{G'}{N_i} \cdot
         \hash_{2i}\} = \{w\}. \qedhere
  \end{align*}
\end{proof}

\begin{lemma}\label{lm:dollars}
  Let $\mathcal{I} \subseteq \{1, \dots, |V|\}$, $i \in \{1, \dots,
  |V|\}$, and $s = {\rm bexp}_{\mathcal{I}}(N_i)$.  Then, $|s|$ is odd
  and:
  \begin{itemize}[itemsep=0pt,parsep=0.0ex]
  \item For every $p \in [1 \dd \lceil \tfrac{|s|}{2} \rceil)$, it
    holds $s[2p] \in \{\dol_{j}\}_{j \in \{1, \dots, i\} \setminus
    \mathcal{I}}$, and
  \item For every $p \in [1 \dd \lceil \tfrac{|s|}{2} \rceil]$, it
    holds $s[2p - 1] \in \Sigma \cup \{M_j\}_{j \in \{1,\dots,i\} \cap
    \mathcal{I}}$.
  \end{itemize} 
\end{lemma}
\begin{proof}

  We proceed by induction on $i$. Let $i = 1$.  If $1 \in
  \mathcal{I}$, then by \cref{def:bexp}, ${\rm
  bexp}_{\mathcal{I}}(N_1) = M_1$.  The first claim is vacuously
  true, while the second holds by $\{1, \dots, i\} \cap \mathcal{I} =
  \{1\}$. Let us thus assume $1 \not\in \mathcal{I}$. The assumption
  $|\expgen{G}{N_1}| \leq \dots \leq |\expgen{G}{N_{|V|}}|$ then
  yields $\rhsgen{G}{N_1} \in \Sigma^{*}$. Let $X, Y \in \Sigma$ be
  such that $\rhsgen{G}{N_1} = XY$.  Applying \cref{def:bexp} then
  yields ${\rm bexp}_{\mathcal{I}}(N_1) = {\rm bexp}_{\mathcal{I}}(X)
  \cdot \dol_{1} \cdot {\rm bexp}_{\mathcal{I}}(Y) = X \cdot
  \dol_{1} \cdot Y$.  The first claim is then satisfied since,
  $\{\dol_{j}\}_{j \in \{1, \dots, i\} \setminus \mathcal{I}} =
  \{\dol_{1}\}$.  The second claim also holds since, letting $s =
  {\rm bexp}_{\mathcal{I}}(N_1)$, we then have $\{s[1], s[3]\} = \{X,
  Y\} \subseteq \Sigma$.

  Let us now assume $i > 1$. If $i \in \mathcal{I}$, then by
  \cref{def:bexp}, ${\rm bexp}_{\mathcal{I}}(N_i) = M_i$. As above,
  the first claim then holds vacuously, while the second claim holds
  by $i \in \{1, \dots, i\} \cap \mathcal{I}$. Let us thus assume $i
  \not\in \mathcal{I}$. Let $X, Y \in \Sigma \cup V$ be such that
  $\rhsgen{G}{N_i} = XY$. Applying \cref{def:bexp} then yields ${\rm
  bexp}_{\mathcal{I}}(N_i) = s_1 \cdot \dol_{i} \cdot s_2$, where
  $s_1 = {\rm bexp}_{\mathcal{I}}(X)$ and $s_2 = {\rm
  bexp}_{\mathcal{I}}(Y)$. Denote $n_1 = |s_1|$ and $n_2 = |s_2|$.
  Observe now that if $X \in \Sigma$, then $s_1 = X$, and thus $n_1$
  is odd.  Otherwise (i.e., $X \not\in\Sigma$), $|\expgen{G}{X}| <
  |\expgen{G}{N_i}|$ implies that there exists $l \in [1 \dd i)$ such
  that $X = N_l$.  By the inductive assumption, $|{\rm
  bexp}_{\mathcal{I}}(N_l)| = |s_1| = n_1$ is odd. Thus, in both cases
  $n_1$ is odd.  Analogously, $n_2$ is odd. Consequently, $|s| = n_1 +
  n_2 + 1$ is odd too. We now show the second claim.
  Let us now take $p \in [1 \dd \lceil \tfrac{|s|}{2} \rceil)$.
  Consider three cases:
  \begin{itemize}[itemsep=0pt,parsep=0pt]
  \item First, assume $p < \lceil \tfrac{n_1}{2} \rceil$. Note that
    then we must have $n_1 > 1$, which in turn implies $X \in V$. By
    $|\expgen{G}{X}| < |\expgen{G}{N_i}|$, there exists $l \in [1 \dd
    i)$ such that $X = N_l$. By the inductive assumption, we thus have
    $s[2p] = s_1[2p] \in \{\dol_{j}\}_{j \in \{1,\dots,l\} \setminus
    \mathcal{I}} \subseteq \{\dol_{j}\}_{j \in \{1, \dots, i\}
    \setminus \mathcal{I}}$.
  \item Next, assume $p = \lceil \tfrac{n_1}{2} \rceil$.  Since $n_1$
    is odd, we have $\lceil \tfrac{n_2}{2} \rceil = \tfrac{n_1 +
    1}{2}$. Thus, $2p = n_1 + 1$ and hence $s[2p] = s[n_1 + 1] =
    \dol_{i}$. We thus have $s[2p] \in \{\dol_{j}\}_{j \in \{1,
    \dots, i\} \setminus \mathcal{I}}$ (note that we used $i \not\in
    \mathcal{I}$).
  \item Finally, assume $p > \lceil \tfrac{n_1}{2} \rceil$.  Note that
    this implies $n_2 > 1$, since otherwise we cannot have $p \in
    (\lceil \tfrac{n_1}{2} \rceil \dd \lceil \tfrac{|s|}{2} \rceil)$.
    This implies $Y \in V$. By $|\expgen{G}{Y}| < |\expgen{G}{N_i}|$,
    there exists $r \in [1 \dd i)$ such that $X = N_r$. Denote $p' =
    \lceil \tfrac{n_1}{2} \rceil$ and $p'' = p - p'$.  As noted
    above, $n_1$ being odd implies $p' = \tfrac{n_1 + 1}{2}$ and hence
    $2p' = n_1 + 1$. Consequently, $s[2p] = s[2p' + 2p''] =
    s_2[2p'']$.  By the inductive assumption we thus have $s[2p] =
    s_2[2p''] \in \{\dol_{j}\}_{j \in \{1, \dots, r\} \setminus
    \mathcal{I}} \subseteq \{\dol_{j}\}_{j \in \{1, \dots, i\}
    \setminus \mathcal{I}}$.
  \end{itemize}
  This concludes the proof of the second claim. To show the third claim,
  let $p \in [1 \dd \lceil \tfrac{|s|}{2} \rceil]$. Then:
  \begin{itemize}[itemsep=0pt,parsep=0pt]
  \item First, assume $p \leq \lceil \tfrac{n_1}{2} \rceil$.  If $X
    \in \Sigma$, then $n_1 = 1$ and $p = 1$. We then indeed have $s[2p
    - 1] = s[1] = X \in \Sigma \subseteq \Sigma \cup \{M_j\}_{j \in
    \{1, \dots, i\} \cap \mathcal{I}}$.  Let us now assume $X \in
    V$. Then, by $|\expgen{G}{X}| < |\expgen{G}{N_i}|$ there exists $l
    \in [1 \dd i)$ such that $X = N_l$. By the inductive assumption,
    we thus have $s[2p - 1] = s_1[2p - 1] \in \Sigma \cup \{M_j\}_{j
    \in \{1, \dots, l\} \cap \mathcal{I}} \subseteq \Sigma \cup
    \{M_j\}_{j \in \{1, \dots, i\} \cap \mathcal{I}}$.
  \item Let us now assume $p > \lceil \tfrac{n_1}{2} \rceil$. Denote
    $p' = \lceil \tfrac{n_1}{2} \rceil$ and $p'' = p - p'$.  Since
    $n_1$ is odd, it holds $p' = \tfrac{n_1 + 1}{2}$. Thus, $2p' = n_1
    + 1$ and hence $s[2p - 1] = s[2p' + 2p'' - 1] = s[n_1 + 1 + 2p'' -
    1] = s_2[2p'' - 1]$.  Consider now two cases. If $Y \in \Sigma$,
    then $n_2 = 1$ and $p'' = 1$. Thus, we indeed have $s[2p - 1] =
    s_2[2p'' - 1] = s_2[1] = Y \in \Sigma \subseteq \Sigma \cup
    \{M_j\}_{j \in \{1, \dots, i\} \cap \mathcal{I}}$. Let us thus
    assume $Y \in V$. By $|\expgen{G}{Y}| < |\expgen{G}{N_i}|$ there
    exists $r \in [1 \dd i)$ such that $Y = N_r$. By the inductive
    assumption we then have $s[2p - 1] = s_2[2p'' - 1] \in \Sigma \cup
    \{M_j\}_{j \in \{1, \dots, r\} \cap \mathcal{I}} \subseteq \Sigma
    \cup \{M_j\}_{j \in \{1, \dots, i\} \cap \mathcal{I}}$. \qedhere
  \end{itemize}
\end{proof}

\begin{lemma}\label{lm:one-dollar}
  Let $\mathcal{I} \subseteq \{1, \dots, |V|\}$ and $i \in \{1, \dots,
  |V|\} \setminus \mathcal{I}$. Then, $|\Occ(\dol_{i}, {\rm
  bexp}_{\mathcal{I}}(N_i))| = 1$.
\end{lemma}
\begin{proof}
  Let $X, Y \in \Sigma \cup V$ be such that $\rhsgen{G}{N_i} = XY$.
  By $j \not\in \mathcal{I}$ and \cref{def:bexp}, it holds ${\rm
  bexp}_{\mathcal{I}}(N_i) = {\rm bexp}_{\mathcal{I}}(X) \cdot
  \dol_{i} \cdot {\rm bexp}_{\mathcal{I}}(Y)$.  Thus,
  $|\Occ(\dol_{i}, {\rm bexp}_{\mathcal{I}}(N_i))| \geq 1$. It
  remains to show that $\Occ(\dol_{i}, {\rm bexp}_{\mathcal{I}}(X))
  = \emptyset$ and $\Occ(\dol_{i}, {\rm bexp}_{\mathcal{I}}(Y)) =
  \emptyset$.  We only show the first equation (the other follows
  analogously). If $X \in \Sigma$, then ${\rm bexp}_{\mathcal{I}}(X) =
  X$, and hence we immediately obtain $\Occ(\dol_{i}, {\rm
  bexp}_{\mathcal{I}}(X)) = \emptyset$.  Let us thus assume $X \in
  V$. By $|\expgen{G}{X}| < |\expgen{G}{N_i}|$, there exists $l \in [1
  \dd i)$ such that $X = N_l$.  By \cref{lm:dollars}, ${\rm
  bexp}_{\mathcal{I}}(N_l) \in (\Sigma \cup \{M_j\}_{j \in \{1, \dots,
  l\} \cap \mathcal{I}} \cup \{\dol_j\}_{j \in \{1, \dots, l\}
  \setminus \mathcal{I}})^{*}$. Since $l < i$, $\dol_{i}$ thus does
  not occur in ${\rm bexp}_{\mathcal{I}}(N_l) = {\rm
  bexp}_{\mathcal{I}}(X)$, i.e., $\Occ(\dol_{i}, {\rm
  bexp}_{\mathcal{I}}(X)) = \emptyset$.
\end{proof}

\begin{lemma}\label{lm:dollar-on-rhs}
  Let $\mathcal{I} \subseteq \{1, \dots, |V|\}$. If $i \in
  \mathcal{I}$, then $\dol_{i}$ occurs exactly once on the
  right-hand side of $G_{\mathcal{I}}$.
\end{lemma}
\begin{proof}
  Recall, that $V_{\mathcal{I}} = \{S\} \cup \{M_i\} \cup \{M_j\}_{j
  \in \mathcal{I} \setminus \{i\}}$ (\cref{def:bigG}).  We consider
  each of the elements in this partition separately:
  \begin{itemize}[itemsep=0pt,parsep=0pt]
  \item By $\rhsgen{G_{\mathcal{I}}}{S} = \bigodot_{t=1,\dots,|V|}
    {\rm bexp}_{\mathcal{I}}(N_t) \cdot \hash_{2t-1} \cdot {\rm
    bexp}_{\mathcal{I}}(N_t) \cdot \hash_{2t}$ and
    \cref{lm:dollars}, it follows that $\rhsgen{G_{\mathcal{I}}}{S}
    \in (\Sigma \cup \{M_t\}_{t \in \mathcal{I}} \cup
    \{\dol_{t}\}_{t \in \{1, \dots, |V|\} \setminus
    \mathcal{I}})^{*}$.  In particular, $\dol_{i}$ does not occur
    in $\rhsgen{G_{\mathcal{I}}}{S}$.
  \item Let $j \in \mathcal{I} \setminus \{i\}$ and let $X, Y \in V
    \cup \Sigma$ be such that $\rhsgen{G}{N_j} = XY$. By
    \cref{def:bigG}, we then have $\rhsgen{G_{\mathcal{I}}}{M_j} =
    {\rm bexp}_{\mathcal{I}}(X) \cdot \dol_{j} \cdot {\rm
    bexp}_{\mathcal{I}}(Y)$. If $X \in \Sigma$, then ${\rm
    bexp}_{\mathcal{I}}(X) \in \Sigma$. Otherwise (i.e., $X \in V$),
    it follows by \cref{lm:dollars}, that ${\rm bexp}_{\mathcal{I}}(X)
    \in (\Sigma \cup \{M_t\}_{t \in \mathcal{I}} \cup
    \{\dol_{t}\}_{t \in \{1, \dots, |V|\} \setminus
    \mathcal{I}})^{*}$.  In both cases, we thus obtain that
    $\dol_{i}$ does not occur in ${\rm bexp}_{\mathcal{I}}(X)$.
    Analogously, $\dol_{i}$ does not occur in ${\rm
    bexp}_{\mathcal{I}}(Y)$. Finally, note that $j \in \mathcal{I}
    \setminus \{i\}$ implies $j \neq i$. Thus, $\dol_{j} \neq
    \dol_{i}$.  Consequently, $\dol_{i}$ does not occur in
    $\rhsgen{G_{\mathcal{I}}}{M_j}$.
  \item Finally, we consider $\rhsgen{G_{\mathcal{I}}}{M_i}$. Let $X,
    Y \in V \cup \Sigma$ be such that $\rhsgen{G}{N_i} = XY$. Then, it
    holds $\rhsgen{G_{\mathcal{I}}}{M_i} = {\rm bexp}_{\mathcal{I}}(X)
    \cdot \dol_{i} \cdot {\rm bexp}_{\mathcal{I}}(Y)$. By the same
    argument as above, it follows that $\dol_{i}$ does not occur in
    ${\rm bexp}_{\mathcal{I}}(X)$ and ${\rm
    bexp}_{\mathcal{I}}(Y)$. We thus obtain that $\dol_{i}$ has
    exactly one occurrence in $\rhsgen{G_{\mathcal{I}}}{M_i}$.
    \qedhere
  \end{itemize}
\end{proof}

\begin{lemma}\label{lm:bexp-occ}
  Let $\mathcal{I} \subseteq \{1, \dots, |V|\}$ and $i \in \{1, \dots,
  |V|\} \setminus \mathcal{I}$.  Then, for every $j \in \{1, \dots,
  |V|\}$, it holds $|\Occ(\dol_{i}, {\rm bexp}_{\mathcal{I}}(N_j))|
  = |\Occ({\rm bexp}_{\mathcal{I}}(N_i), {\rm
  bexp}_{\mathcal{I}}(N_j))|$. Moreover, any two occurrences
  of ${\rm bexp}_{\mathcal{I}}(N_i)$ in ${\rm bexp}_{\mathcal{I}}(N_j)$
  are non-overlapping.
\end{lemma}
\begin{proof}

  Let $X, Y \in \Sigma \cup V$ be such that $\rhsgen{G}{N_i} = XY$.
  By $i \not\in \mathcal{I}$, we then have ${\rm bexp}_{\mathcal{I}}(N_i)
  = {\rm bexp}_{\mathcal{I}}(X) \cdot \dol_{i} \cdot {\rm
  bexp}_{\mathcal{I}}(Y)$. In particular, $\dol_{i}$ occurs in
  ${\rm bexp}_{\mathcal{I}}(N_i)$. We proceed by induction on $j$.

  To prove the induction base, let $j = 1$. If $1 \in \mathcal{I}$,
  then ${\rm bexp}_{\mathcal{I}}(N_1) = M_1$.  In this case, by $M_1
  \not\in V \cup \Sigma''$ we immediately obtain $|\Occ(\dol_{i},
  {\rm bexp}_{\mathcal{I}}(N_1))| = 0 \leq |\Occ({\rm
  bexp}_{\mathcal{I}}(N_i), {\rm bexp}_{\mathcal{I}}(N_1))|$. Let us
  now assume $1 \not\in \mathcal{I}$.  The assumption
  $|\expgen{G}{N_1}| \leq \dots \leq |\expgen{G}{N_{|V|}}|$ then
  implies that there exists $X', Y' \in \Sigma$ such that
  $\rhsgen{G}{N_1} = X' Y'$.  Then, ${\rm bexp}_{\mathcal{I}}(N_1) =
  {\rm bexp}_{\mathcal{I}}(X') \cdot \dol_{1} \cdot {\rm
  bexp}_{\mathcal{I}}(Y') = X' \cdot \dol_{1} \cdot Y'$. We then
  consider two subcases. If $i \,{=}\, 1$, then $|\Occ({\rm
  bexp}_{\mathcal{I}}(N_i), {\rm bexp}_{\mathcal{I}}(N_j)| =
  |\Occ({\rm bexp}_{\mathcal{I}}(N_1), {\rm bexp}_{\mathcal{I}}(N_1))|
  = 1$.  On the other hand, by the fact that $X', Y' \in \Sigma$, we
  then obtain $|\Occ(\dol_{i}, {\rm bexp}_{\mathcal{I}}(N_j))| =
  |\Occ(\dol_{1}, X' \cdot \dol_{1} \cdot Y')| = 1$.  Thus, the
  claim holds. Let us now consider the second subcase, i.e., $i \neq
  1$.  Then, $\dol_{i} \neq \dol_{1}$. Combining with $X', Y' \in
  \Sigma$, we thus obtain $\dol_{i} \not\in \{X', \dol_{1}, Y'\}$,
  and hence $|\Occ(\dol_{i}, {\rm bexp}_{\mathcal{I}}(N_j))| =
  |\Occ(\dol_{i}, X' \cdot \dol_{1} \cdot Y')| = 0$. Since
  $\dol_{i}$ occurs in ${\rm bexp}_{\mathcal{I}}(N_i)$ (see above),
  we thus also have $|\Occ({\rm
  bexp}_{\mathcal{I}}(N_i), {\rm bexp}_{\mathcal{I}}(N_j))| = 0$. Since
  for $j = 1$, there is always at most one occurrence of
  ${\rm bexp}_{\mathcal{I}}(N_i)$ in ${\rm bexp}_{\mathcal{I}}(N_j)$,
  the claim about disjoint occurrences therefore holds vacuously.

  We now prove the induction step. Let $j > 1$. If $j \in
  \mathcal{I}$, then ${\rm bexp}_{\mathcal{I}}(N_j) = M_j$.  By $M_j
  \not\in V \cup \Sigma''$, we thus immediately obtain
  $|\Occ(\dol_{i}, {\rm bexp}_{\mathcal{I}}(N_j))| = 0 \leq
  |\Occ(\dol_{i}, {\rm bexp}_{\mathcal{I}}(N_j))|$.  Let us thus
  assume $j \not\in \mathcal{I}$.  Let $X', Y' \in \Sigma \cup V$ be
  such that $\rhsgen{G}{N_j} = X'Y'$. Then, ${\rm
  bexp}_{\mathcal{I}}(N_j) = {\rm bexp}_{\mathcal{I}}(X') \cdot
  \dol_{j} \cdot {\rm bexp}_{\mathcal{I}}(Y')$. We now consider three
  subcases:
  \begin{itemize}[itemsep=0pt,parsep=0pt]
  \item First, assume $j = i$. Then, ${\rm bexp}_{\mathcal{I}}(N_i) =
    {\rm bexp}_{\mathcal{I}}(N_j)$, and hence we immediately obtain
    $|\Occ({\rm bexp}_{\mathcal{I}}(N_i), {\rm
    bexp}_{\mathcal{I}}(N_j))| = 1$. On the other hand, by
    \cref{lm:one-dollar}, we then obtain $|\Occ(\dol_{i}, {\rm
    bexp}_{\mathcal{I}}(N_j))| = |\Occ(\dol_{i}, {\rm
    bexp}_{\mathcal{I}}(N_i))| = 1$. Thus, the claim holds. Note
    that the claim about disjoint occurrences again holds vacuously
    here.
  \item Next, assume $j < i$. By \cref{lm:dollars}, ${\rm
    bexp}_{\mathcal{I}}(N_j) \in (\Sigma \cup \{M_t\}_{t \in \{1,
    \dots, j\} \cap \mathcal{I}} \cup \{\dol_t\}_{t \in \{1, \dots,
    j\} \setminus \mathcal{I}})^{*}$. Since $j < i$, $\dol_{i}$ thus
    does not occur in ${\rm bexp}_{\mathcal{I}}(N_j)$, and hence
    $|\Occ(\dol_{i}, {\rm bexp}_{\mathcal{I}}(N_j))| = 0$. Since
    $\dol_{i}$ occurs in ${\rm bexp}_{\mathcal{I}}(N_i)$, this also
    implies $|\Occ({\rm bexp}_{\mathcal{I}}(N_i), {\rm
    bexp}_{\mathcal{I}}(N_j))| = 0$.
  \item Finally, let us assume $j > i$. Recall, that above we observed
    that ${\rm bexp}_{\mathcal{I}}(N_j) = {\rm bexp}_{\mathcal{I}}(X')
    \cdot \dol_{j} \cdot {\rm bexp}_{\mathcal{I}}(Y')$. Note,
    however, that \cref{lm:dollars} implies that ${\rm
    bexp}_{\mathcal{I}}(N_i) \in (\Sigma \cup \{M_t\}_{t \in \{1,
    \dots, i\} \cap \mathcal{I}} \cup \{\dol_t\}_{t \in \{1, \dots,
    i\} \setminus \mathcal{I}})^{*}$. Thus, by $i < j$, every
    occurrence of ${\rm bexp}_{\mathcal{I}}(N_i)$ in ${\rm
    bexp}_{\mathcal{I}}(N_j)$ is entirely contained in either ${\rm
    bexp}_{\mathcal{I}}(X')$ or ${\rm bexp}_{\mathcal{I}}(Y')$. The
    same holds for any occurrence of $\dol_{i}$. Consequently,
    \begin{align*}
      |\Occ(\dol_{i}, {\rm bexp}_{\mathcal{I}}(N_j))|
         =&\ |\Occ(\dol_{i}, {\rm bexp}_{\mathcal{I}}(X'))| +
             |\Occ(\dol_{i}, {\rm bexp}_{\mathcal{I}}(Y'))|,\\
      |\Occ({\rm bexp}_{\mathcal{I}}(N_i), {\rm bexp}_{\mathcal{I}}(N_j))|
         =&\ |\Occ({\rm bexp}_{\mathcal{I}}(N_i), {\rm bexp}_{\mathcal{I}}(X'))| +\\
          &\ |\Occ({\rm bexp}_{\mathcal{I}}(N_i), {\rm bexp}_{\mathcal{I}}(Y'))|.
    \end{align*}
    If $X' \in \Sigma$, then ${\rm bexp}_{\mathcal{I}}(X') = X'$, and
    hence we immediately obtain $|\Occ(\dol_{i}, {\rm
    bexp}_{\mathcal{I}}(X'))| = |\Occ({\rm bexp}_{\mathcal{I}}(N_i),
    {\rm bexp}_{\mathcal{I}}(X'))| = 0$, since $\dol_{i}$ occurs
    both in $\dol_{i}$ and ${\rm bexp}_{\mathcal{I}}(N_i)$.  Let us
    thus assume $X' \in V$. By $|\expgen{G}{X'}| < |\expgen{G}{N_j}|$,
    there exists $l \in [1 \dd j)$ such that $X' = N_l$. By the
    inductive assumption, we thus have $|\Occ(\dol_{i}, {\rm
    bexp}_{\mathcal{I}}(X'))| = |\Occ(\dol_{i}, {\rm
    bexp}_{\mathcal{I}}(N_l))| = |\Occ({\rm
    bexp}_{\mathcal{I}}(N_i), {\rm bexp}_{\mathcal{I}}(N_l))| =
    |\Occ({\rm bexp}_{\mathcal{I}}(N_i), {\rm
    bexp}_{\mathcal{I}}(X'))|$. We also obtain that any two
    occurrences of ${\rm bexp}_{\mathcal{I}}(N_i)$ in ${\rm
    bexp}_{\mathcal{I}}(X')$ are non-overlapping.  In both cases, we
    have thus proved that it holds $|\Occ(\dol_{i}, {\rm
    bexp}_{\mathcal{I}}(X'))| = |\Occ({\rm bexp}_{\mathcal{I}}(N_i),
    {\rm bexp}_{\mathcal{I}}(X'))|$, and any two occurrences of ${\rm
    bexp}_{\mathcal{I}}(N_i)$ in ${\rm bexp}_{\mathcal{I}}(X')$ are
    non-overlapping.  Analogously, it holds $|\Occ(\dol_{i}, {\rm
    bexp}_{\mathcal{I}}(Y'))| = |\Occ({\rm bexp}_{\mathcal{I}}(N_i),
    {\rm bexp}_{\mathcal{I}}(Y'))|$, and any two occurrences of ${\rm
    bexp}_{\mathcal{I}}(N_i)$ in ${\rm bexp}_{\mathcal{I}}(Y')$ are
    non-overlapping.  Putting these equalities with the earlier two,
    we thus obtain
    \begin{align*}
      |\Occ(\dol_{i}, {\rm bexp}_{\mathcal{I}}(N_j))|
         &= |\Occ(\dol_{i}, {\rm bexp}_{\mathcal{I}}(X'))| +
            |\Occ(\dol_{i}, {\rm bexp}_{\mathcal{I}}(Y'))|,\\ 
         &= |\Occ({\rm bexp}_{\mathcal{I}}(N_i), {\rm bexp}_{\mathcal{I}}(X'))| +
            |\Occ({\rm bexp}_{\mathcal{I}}(N_i), {\rm bexp}_{\mathcal{I}}(Y'))|,\\
         &= |\Occ({\rm bexp}_{\mathcal{I}}(N_i), {\rm bexp}_{\mathcal{I}}(N_j))|.
    \end{align*}
    It remains to observe that above we have also ruled out the
    possibility of two occurrences of ${\rm bexp}_{\mathcal{I}}(N_i)$
    in ${\rm bexp}_{\mathcal{I}}(N_j)$ overlapping each other.
    \qedhere
  \end{itemize}
\end{proof}

\begin{lemma}\label{lm:maximal-string}
  Let $\mathcal{I} \subseteq \{1, \dots, |V|\}$ and let $s$ be a
  maximal string with respect to $G_{\mathcal{I}}$. Then, there exists
  $i \in \{1, \dots, |V|\} \setminus \mathcal{I}$ such that $s = {\rm
  bexp}_{\mathcal{I}}(N_i)$.
\end{lemma}
\begin{proof}

  For any $j \in \{1, \dots, |V|\}$, let $X_j, Y_j \in \Sigma \cup V$
  be such that $\rhsgen{G}{N_j} = X_j Y_j$. Recall that
  $V_{\mathcal{I}} = \{S\} \cup \{M_j\}_{j \in \mathcal{I}}$ and hence
  the right-hand side of $G_{\mathcal{I}}$ contains the following
  strings:
  \begin{itemize}
  \item $\rhsgen{G_{\mathcal{I}}}{S} = \bigodot_{j=1,\dots,|V|}
    {\rm bexp}_{\mathcal{I}}(N_j) \cdot \hash_{2j-1} \cdot
    {\rm bexp}_{\mathcal{I}}(N_j) \cdot \hash_{2j}$,
  \item $\rhsgen{G_{\mathcal{I}}}{M_j} =
    {\rm bexp}_{\mathcal{I}}(X_j) \cdot \dol_{j} \cdot
    {\rm bexp}_{\mathcal{I}}(Y_j)$, where $j \in \mathcal{I}$.
  \end{itemize}
  By \cref{lm:dollars}, $\{{\rm bexp}_{\mathcal{I}}(N_t)\}_{t \in \{1,
  \dots, |V|\}} \subseteq (\Sigma \cup \{M_t\}_{t \in \mathcal{I}}
  \cup \{\dol_{t}\}_{t \in \{1, \dots, |V|\} \setminus
  \mathcal{I}})^{*}$. In particular, none of the strings contain any
  of the symbols in the set $\{\dol_{t}\}_{t \in \mathcal{I}} \cup
  \{\hash_t\}_{t \in [1 \dd 2|V|]}$. Moreover, note that $\{X_t,
  Y_t\}_{t \in \{1, \dots, |V|\}} \subseteq \Sigma \cup V$.  For any
  $c \in \Sigma$, ${\rm bexp}_{\mathcal{I}}(c) \in \Sigma$.
  Consequently, we also have $\{{\rm bexp}_{\mathcal{I}}(X_t), {\rm
  bexp}_{\mathcal{I}}(Y_t)\} \subseteq (\Sigma \cup \{M_t\}_{t \in
  \mathcal{I}} \cup \{\dol_{t}\}_{t \in \{1, \dots, |V|\} \setminus
  \mathcal{I}})^{*}$.  This implies that each of the characters in the
  set $\{\hash_t\}_{t \in [1 \dd 2|V|]} \cup \{\dol_{t}\}_{t \in
  \mathcal{I}}$ occurs on the right-hand side of $G_{\mathcal{I}}$
  exactly once. Thus, since by \cref{def:maximal-string}, $s$
  occurs on the right-hand size of $G_{\mathcal{I}}$ at least twice,
  it does not contain any of these symbols. By the above
  characterization of the right-hand side of $G_{\mathcal{I}}$, $s$ is
  therefore a substring of one of the strings in the collection
  $\{{\rm bexp}_{\mathcal{I}}(N_j)\}_{j \in \{1, \dots, |V|\}} \cup
  \{{\rm bexp}_{\mathcal{I}}(X_j), {\rm bexp}_{\mathcal{I}}(Y_j)\}_{j
  \in \mathcal{I}} \subseteq \{{\rm bexp}_{\mathcal{I}}(N_j)\}_{j \in
  \{1, \dots, |V|\}} \cup \{{\rm bexp}_{\mathcal{I}}(c)\}_{c \in
  \Sigma}$.  Note, however, that \cref{def:maximal-string} requires
  that $|s| \geq 2$.  Since $j \in \mathcal{I}$ implies that $|{\rm
  bexp}_{\mathcal{I}}(N_j)| = 1$, and for every $c \in \Sigma$, $|{\rm
  bexp}_{\mathcal{I}}(c)| = 1$, we thus obtain that $s$ is a substring
  of one of the elements in $\{{\rm bexp}_{\mathcal{I}}(N_j)\}_{j \in
  \{1, \dots, |V|\} \setminus \mathcal{I}}$.  This also implies that
  letting $\alpha_j = 2 + |\{t \in \mathcal{I} : X_t = N_j\}| +
  |\{t \in \mathcal{I} : Y_t = N_j\}|$ for $j \in \{1, \dots, |V|\}
  \setminus \mathcal{I}$, the number $f_s$ of occurrences of $s$
  on the right-hand side of $G_{\mathcal{I}}$ satisfies:
  \[
    f_s = \sum_{j \in \{1, \dots, |V|\} \setminus \mathcal{I}}
    \alpha_j \cdot |\Occ(s, {\rm bexp}_{\mathcal{I}}(N_j))|
  \]
  Note that the number $f^{\rm nonov}_{s}$ of non-overlapping
  occurrences (see \cref{rm:nonov}) of $s$ on the right-hand side of
  $G_{\mathcal{I}}$ then satisfies $f^{\rm nonov}_s \leq f_s$.

  Observe now that by \cref{lm:dollars}, every second symbol in each
  of the strings in the collection $\{{\rm
  bexp}_{\mathcal{I}}(N_j)\}_{j \in \{1, \dots, |V|\}}$ belongs to
  the set $\{\dol_{j}\}_{j \in \{1, \dots, |V|\} \setminus
  \mathcal{I}}$. Since $|s| \geq 2$, the string $s$ thus contains
  one of those symbols (note that this in particular implies that
  $\{1, \dots, |V|\} \setminus \mathcal{I} \neq \emptyset$).  Let
  \[
    i = \max\{t \in \{1, \dots, |V|\} \setminus \mathcal{I} :
    \dol_{t}\text{ occurs in }s\}.
  \]
  In three steps, we will prove that $s = {\rm bexp}_{\mathcal{I}}(N_i)$:

  \begin{enumerate}
  \item First, we show that $s$ is a substring of
    ${\rm bexp}_{\mathcal{I}}(N_i)$. Let
    \[
      i' = \min \{t \in \{1, \dots, |V|\} \setminus \mathcal{I} :
      s\text{ is a substring of }{\rm bexp}_{\mathcal{I}}(N_t)\}.
    \]
    Note, that by the above discussion, $i'$ is well defined.  Observe
    that we cannot have $i' < i$, since $\dol_{i}$ occurs in $s$,
    and by \cref{lm:dollars}, ${\rm bexp}_{\mathcal{I}}(N_{i'}) \in
    (\Sigma \cup \{M_t\}_{t \in \{1, \dots, i'\}} \cup
    \{\dol_{t}\}_{t \in \{1, \dots, i'\}})^{*}$. Thus, $i' \geq
    i$. Suppose that $i' > i$. Since $i' \not\in \mathcal{I}$, by
    \cref{def:bexp}, we have ${\rm bexp}_{\mathcal{I}}(N_{i'}) = {\rm
    bexp}_{\mathcal{I}}(X) \cdot \dol_{i'} \cdot {\rm
    bexp}_{\mathcal{I}}(Y)$, where $X, Y \in \Sigma \cup V$ are such
    that $\rhsgen{G}{N_{i'}} = XY$. By definition of $i$,
    $\dol_{i'}$ does not occur in $s$.  Thus, $s$ is a substring
    of either ${\rm bexp}_{\mathcal{I}}(X)$ or ${\rm
    bexp}_{\mathcal{I}}(Y)$. Assume without the loss of generality
    that it is a substring of ${\rm bexp}_{\mathcal{I}}(X)$. By $|s|
    \geq 2$, we then must have $X \in V$. By $|\expgen{G}{X}| <
    |\expgen{G}{N_{i'}}|$, there exists $l \in [1 \dd i')$ such that
    $X = N_{l}$.  In other words, $s$ is a substring of ${\rm
    bexp}_{\mathcal{I}}(N_{l})$. By $l < i'$, this contradicts the
    definition of $i'$. We have thus proved that $i = i'$. Hence,
    $s$ is a substring of ${\rm bexp}_{\mathcal{I}}(N_{i})$.

  \item Denote $s' = {\rm bexp}_{\mathcal{I}}(N_i)$. In this step, we
    prove that the number $f^{\rm nonov}_{s'}$ of non-overlapping
    occurrences of $s'$ on the right-hand side of $G_{\mathcal{I}}$
    satisfies $f^{\rm nonov}_{s'} \geq f^{\rm nonov}_{s}$. Consider
    some $j \in \{1, \dots, |V|\} \setminus \mathcal{I}$. By
    definition of $\alpha_j$ and the above explicit listing of strings
    occurring on the right-hand side of $G_{\mathcal{I}}$, there exist
    $\alpha_j$ non-overlapping occurrences of the string ${\rm
    bexp}_{\mathcal{I}}(N_j)$ on the right-hand side of $G_{\mathcal{I}}$.
    Observe now that, on the one hand, since $s$ contains symbol
    $\dol_{i}$, it holds $|\Occ(s, {\rm bexp}_{\mathcal{I}}(N_j))|
    \leq |\Occ(\dol_{i}, {\rm bexp}_{\mathcal{I}}(N_j))|$.  On the
    other hand, by \cref{lm:bexp-occ}, we obtain $|\Occ(\dol_{i},
    {\rm bexp}_{\mathcal{I}}(N_j))| = |\Occ({\rm
    bexp}_{\mathcal{I}}(N_i), {\rm bexp}_{\mathcal{I}}(N_j))| =
    |\Occ(s', {\rm bexp}_{\mathcal{I}}(N_j))|$. Moreover, by
    \cref{lm:bexp-occ}, all occurrences of $s'$ in ${\rm
    bexp}_{\mathcal{I}}(N_j)$ are non-overlapping.  This implies
    that the number $f^{\rm nonov}_{s'}$ of non-overlapping
    occurrences of $s'$ on the right-hand side of $G_{\mathcal{I}}$
    satisfies
    \begin{align*}
      f^{\rm nonov}_{s'}
        &\geq
          \sum_{j \in \{1, \dots, |V|\} \setminus \mathcal{I}}
          \alpha_j \cdot |\Occ(s', {\rm bexp}_{\mathcal{I}}(N_j))|
        =
          \sum_{j \in \{1, \dots, |V|\} \setminus \mathcal{I}}
          \alpha_j \cdot |\Occ(\dol_{i}, {\rm bexp}_{\mathcal{I}}(N_j))|\\
        &\geq
          \sum_{j \in \{1, \dots, |v|\} \setminus \mathcal{I}}
          \alpha_j \cdot |\Occ(s, {\rm bexp}_{\mathcal{I}}(N_j))|
        = f_{s} \geq f^{\rm nonov}_{s}.
    \end{align*}

    \item In the first step, we proved that $s$ is a substring of $s'$.
      Suppose that $s \neq s'$. Note that then $|s'| > |s|$ and, by
      the second step, $s'$ has at least as many non-overlapping
      occurrences on the right-hand side of $G_{\mathcal{I}}$ as $s$.
      By the third condition in \cref{def:maximal-string}, this
      contradicts that $s$ is maximal. Thus, we must have
      $s = s' = {\rm bexp}_{\mathcal{I}}(N_i)$. \qedhere
  \end{enumerate}
\end{proof}

\newcommand{\substitute}[3]{{\rm sub}(#1,#2,#3)}

\begin{definition}
  Let $u, v \in \Sigma^{+}$ and $c \in \Sigma$. We define
  $\substitute{u}{c}{v}$ as a string obtained by replacing all occurrences
  of $c$ in $u$ with $v$. Formally, $\substitute{u}{c}{v} =
  \bigodot_{i=1,\dots,|u|} f(u[i])$, where $f: \Sigma \rightarrow
  \Sigma^{+}$ is such that for every $a \in \Sigma$:
  \vspace{2ex}
  \[
    f(a) =
      \begin{cases}
        a & \text{if }a \neq c, \\
        v & \text{otherwise}.
      \end{cases}
    \vspace{2ex}
  \]
\end{definition}

\newcommand{\Gin}{G_{\rm in}}
\newcommand{\Rulesin}{R_{\rm in}}
\newcommand{\Startin}{S_{\rm in}}
\newcommand{\Vin}{V_{\rm in}}
\newcommand{\Sigmain}{\Sigma_{\rm in}}
\newcommand{\Gaux}{G_{\rm aux}}
\newcommand{\Xaux}{X_{\rm aux}}
\newcommand{\Vaux}{V_{\rm aux}}
\newcommand{\Rulesaux}{R_{\rm aux}}
\newcommand{\taux}{t_{\rm aux}}
\newcommand{\Gout}{G_{\rm out}}
\newcommand{\Vout}{V_{\rm out}}
\newcommand{\Rulesout}{R_{\rm out}}
\newcommand{\Xout}{X_{\rm out}}
\newcommand{\tout}{t_{\rm out}}

\begin{lemma}\label{lm:nonov-s}
  Consider an SLG $\Gin = (\Vin, \Sigmain, \Rulesin, \Startin)$.
  Let $s \in (\Vin \cup \Sigmain)^{+}$ be maximal with
  respect to $\Gin$ (\cref{def:maximal-string}).  Assume that
  any two occurrences of $s$ on the right-hand side of $\Gin$
  are non-overlapping.  Let $\Gaux = (\Vaux, \Sigmain, \Rulesaux,
  \Startin)$, where $\Vaux = \Vin \cup \{\Xaux\}$,
  $\Xaux \not\in \Vin \cup \Sigmain$, and:
  \begin{enumerate}[itemsep=0pt,parsep=0pt]
  \item It holds $\rhsgen{\Gaux}{\Xaux} = s$, and
  \item For every $Y \in \Vin$, it holds
  $\substitute{\rhsgen{\Gaux}{Y}}{\Xaux}{s} =
  \rhsgen{\Gin}{Y}$ and $\Occ(s,\rhsgen{\Gaux}{Y}) = \emptyset$.
  \end{enumerate}
  Then, the output of a single step of any global algorithm on $\Gin$,
  using $s$ as a maximal substring, is isomorphic to $\Gaux$.
\end{lemma}
\begin{proof}
  Let $\Gout = (\Vout, \Sigmain, \Rulesout, \Startin)$ be
  the output of a single step of a global algorithm with SLG $\Gin$ as
  input and $s$ as the maximal string. By definition, we
  then have $\Vout = \Vin \cup \{\Xout\}$, where $\Xout \not\in
  \Vin \cup \Sigmain$ is a new nonterminal.  Consider $g : \Vout
  \cup \Sigmain \rightarrow \Vaux \cup \Sigmain$ defined so that
  for every $c \in \Sigmain$, it holds $g(c) = c$, for every $Y \in
  \Vin$, it holds $g(Y) = Y$, and finally, that $g(\Xout) = \Xaux$.
  The function $g$ is clearly a bijection.

  By definition, it holds $\rhsgen{\Gout}{\Xout} = s$. This
  immediately implies that $|\rhsgen{\Gaux}{g(\Xout)}| =
  |\rhsgen{\Gaux}{\Xaux}| = |s| = |\rhsgen{\Gout}{\Xout}|$, and that
  for every $j \in [1 \dd |s|]$, $\rhsgen{\Gaux}{g(\Xout)}[j] =
  \rhsgen{\Gaux}{\Xaux}[j] = s[j] = g(s[j]) =
  g(\rhsgen{\Gout}{\Xout}[j])$.

  Let us now consider $Y \in \Vin$. Denote $t = \rhsgen{\Gin}{Y}$,
  $\tout = \rhsgen{\Gout}{Y}$, and $\taux = \rhsgen{\Gaux}{g(Y)}
  \allowbreak = \rhsgen{\Gaux}{Y}$. Let also $k = |\Occ(s,
  t)|$.  By the assumption about non-overlapping occurrences of $s$ on
  the right-hand side of $\Gin$, we can uniquely write $t = y_0
  s y_1 s \dots s y_{k-1} s y_{k}$, where for every $j \in [0 \dd k]$,
  we have $y_j \in (\Vin \cup \Sigmain)^{*}$.  By definition of the
  global algorithm, this implies that $\tout = y_0 \Xout y_1 \Xout \dots
  \Xout y_{k-1} \Xout y_{k}$.  Let $\taux = z_0 \Xaux z_1 \Xaux \dots
  \Xaux z_{k'-1} \Xaux z_{k'}$, where $k' \geq 0$ and for every $j \in
  [0 \dd k']$, it holds $z_j \in (\Vin \cup \Sigmain)^{*}$, i.e.,
  $\Xaux$ does not occur in any of the strings in $\{z_j\}_{j \in [0
  \dd k']}$. Note that we then have $\substitute{\taux}{\Xaux}{s} = z_0 s
  z_1 s \dots s z_{k-1} s z_{k} = t$.  We will prove that $k' = k$ and
  for every $j \in [0 \dd k]$, $z_j = y_j$. First, observe that $k'
  \leq k$, since otherwise the assumption $\substitute{\taux}{\Xaux}{s} =
  t$ implies $|\Occ(s, t)| = k' > k$. Suppose $k' < k$.  By the
  assumption about non-overlapping occurrences of $s$ in $\taux$,
  $\substitute{\taux}{\Xaux}{s} = t$ then implies that $s$ is a substring of
  one of the elements of $\{z_j\}_{j \in [0 \dd k']}$.  But then
  $\Occ(s, \taux) \neq \emptyset$, contradicting the assumption. We thus
  have $k' = k$.  Observe now that by definition of global algorithms,
  we have $\substitute{\tout}{\Xout}{s} = t$. Thus, $|t| = |\tout| + |s|
  \cdot (k-1)$. On the other hand, we also have
  $\substitute{\taux}{\Xaux}{s} = t$, and hence $|t| = |\taux| + |s| (k-1)$.
  We thus obtain $|\tout| = |t| - |s| (k-1) = |\taux|$.
  Suppose that there exists $j \in [0 \dd
  k]$ such that $y_j \neq z_j$.  Let us take the smallest such
  $j$. Observe, that we cannot have $j = k$, since then we would have
  $|z_k| = |\taux| - k - \sum_{q\in[0 \dd k)} |z_q| = |\tout| - k -
  \sum_{q \in [0 \dd k)} |y_q| = |y_k|$, and both $y_k$ and $z_k$ are
  suffixes of $t$. Thus, $j \in [0 \dd k)$. Consider now two cases:
  \begin{itemize}[itemsep=0pt,parsep=0pt]
  \item First, assume $|y_j| \neq |z_j|$. Denote $i_y = |y_0 s y_1 s
    \dots s y_j|$ and $i_z = |z_0 s z_1 s \dots s z_j|$.  Since $|y_0
    s \dots s y_{j-1} s| = |z_0 s \dots s z_{j-1} s|$, we have $i_y
    \neq i_z$. Recall now that $y_0 s \dots s y_k = z_0 s \dots s z_k
    \allowbreak = t$. Thus, by $j < k$, we have $i_y, i_z \in \Occ(s, t)$.
    By the assumption, there are no two occurrences of $s$ in $t$ that
    overlap each other. Thus, $|i_y - i_z| \geq |s|$. Consider now two
    subcases.  If $i_y < i_z$, then $s$ is a substring of
    $z_j$. Since, however, $z_j$ is a substring of $\taux$, this
    contradicts the assumption $\Occ(s, \taux) = \emptyset$. If $i_z <
    i_y$, then $s$ is a substring of $y_j$.  By $t = y_0 s y_1 s \dots
    s y_{k-1} s y_k$, this implies $|\Occ(s, t)| > k$, which
    contradicts the definition of $k$.
  \item Let us now assume $|y_j| = |z_j|$. By $|y_0 s \dots s y_{j-1}
    s| = |z_0 s \dots s z_{j-1} s|$, this implies $y_0 s \dots s y_{k}
    \allowbreak \neq z_0 s \dots s z_{k}$, which contradicts the assumption
    $\substitute{\taux}{\Xaux}{s} = t$, since $\substitute{\taux}{\Xaux}{s}
    = z_0 s \dots s z_k$ and $t = y_0 s \dots s y_{k}$.
  \end{itemize}
  We have thus proved that $k' = k$ and that for every $j \in [0 \dd
  k]$, it holds $y_j = z_j$.  Consequently, $\taux = y_0 \Xaux y_1
  \Xaux \dots \Xaux y_{k-1} \Xaux y_{k}$.  By $g(\Xout) = \Xaux$ and
  $\{y_j\}_{j \in [0 \dd k]} \subseteq (\Vin \cup \Sigmain)^{*}$, this
  implies that for every $j \in [1 \dd |\taux|]$, it holds $\taux[j] =
  g(\tout[j])$, which concludes the proof that $\Gout$ is isomorphic to
  $\Gaux$.
\end{proof}

\begin{remark}
  To see an example, where the assumption about $s$ not having two
  overlapping occurrences on the right-hand side of $\Gin$ is
  needed in \cref{lm:nonov-s}, let $\Gin = (\{\Startin\},
  \{\texttt{a}\}, \Rulesin, \Startin)$, where
  $\rhsgen{\Gin}{\Startin} = \texttt{aaaa}$. Let also $s = \texttt{aa}$
  and $\Gaux = (\{\Startin, \Xaux\}, \{\texttt{a}\}, \Rulesaux,
  \Startin)$ be such that $\rhsgen{\Gaux}{\Startin}
  = a \Xaux a$ and $\rhsgen{\Gaux}{\Xaux} = \texttt{aa}$.  Then,
  the conditions in \cref{lm:nonov-s} are satisfied for $\Gaux$,
  but the output $\Gout$ of a global algorithm on $\Gin$
  using $s$ as the maximal substring, is not isomorphic to $\Gaux$.
\end{remark}

\begin{lemma}\label{lm:sub}
  Let $\mathcal{I} \subseteq \{1, \dots, |V|\}$ and $i \in \{1, \dots,
  |V|\} \setminus \mathcal{I}$. Denote $\mathcal{I}' = \mathcal{I}
  \cup \{i\}$ and $s = {\rm bexp}_{\mathcal{I}}(N_i)$.  Then, for
  every $X \in V \cup \Sigma$, it holds
  \[
    {\rm bexp}_{\mathcal{I}}(X) = \substitute{{\rm
    bexp}_{\mathcal{I}'}(X)}{M_i}{s}.
  \]
\end{lemma}
\begin{proof}
  If $X \in \Sigma$ then, by definition, for every $\mathcal{J}
  \subseteq \{1, \dots, |V|\}$, it holds ${\rm bexp}_{\mathcal{J}}(X)
  = X$. Thus, by $M_i \not\in \Sigma$, we obtain ${\rm
  bexp}_{\mathcal{I}}(X) = X = \substitute{X}{M_i}{s} =
  \substitute{{\rm bexp}_{\mathcal{I}'}(X)}{M_i}{s}$.

  Let us thus assume that $X \in V$, i.e., that for some $j \in \{1,
  \dots, |V|\}$, it holds $X = N_j$. We prove the claim by induction
  on $j$. To prove the induction base, let $j = 1$.  Let us first
  assume that $i = j = 1$. On the one hand, ${\rm
  bexp}_{\mathcal{I}}(X) = {\rm bexp}_{\mathcal{I}}(N_1) = s$.  On
  the other hand, by $1 \in \mathcal{I}'$ we have ${\rm
  bexp}_{\mathcal{I}'}(X) = {\rm bexp}_{\mathcal{I}}(N_i) = M_1$.
  Consequently, $\substitute{{\rm bexp}_{\mathcal{I}'}(X)}{M_1}{s} =
  \substitute{M_1}{M_1}{s} = s$. Let us now assume $i \neq 1$. If $1
  \in \mathcal{I}$, then $1 \in \mathcal{I}'$ and hence ${\rm
  bexp}_{\mathcal{I}}(X) = {\rm bexp}_{\mathcal{I}'}(X) = M_1$.  By
  $i \neq 1$, we then have $\substitute{{\rm
  bexp}_{\mathcal{I}'}(X)}{M_i}{s} = \substitute{M_1}{M_i}{s} =
  M_1$. Thus, we proved the claim. It remains to consider the case $1
  \not\in \mathcal{I}$. Note that by $1 \neq i$, we then also have $1
  \not\in \mathcal{I}'$. By the assumption $|\expgen{G}{N_1}| \leq
  |\expgen{G}{N_1}| \leq \dots \leq |\expgen{G}{N_{|V|}}|$, it follows
  that there exist $A, B \in \Sigma$ such that $\rhsgen{G}{N_1} = AB$.
  We then have ${\rm bexp}_{\mathcal{I}}(X) = {\rm
  bexp}_{\mathcal{I}}(A) \cdot \dol_{j} \cdot {\rm
  bexp}_{\mathcal{I}}(B) = A \cdot \dol_{1} \cdot B$ and ${\rm
  bexp}_{\mathcal{I}'}(X) = {\rm bexp}_{\mathcal{I}'}(A) \cdot
  \dol_{1} \cdot {\rm bexp}_{\mathcal{I}'}(B) = A \cdot \dol_{1}
  \cdot B$.  By $M_i \not\in \Sigma$, it thus follows that
  $\substitute{{\rm bexp}_{\mathcal{I}'}(X)}{M_i}{s} = \substitute{A
  \cdot \dol_{1} \cdot B}{M_i}{s} = A \cdot \dol_{1} \cdot
  B$. Thus, we proved the claim.

  Let us now assume $j > 1$.  Let us first assume that $i = j$. On the
  one hand, ${\rm bexp}_{\mathcal{I}}(X) = {\rm
  bexp}_{\mathcal{I}}(N_i) = s$.  On the other hand, by $i \in
  \mathcal{I}'$ we have ${\rm bexp}_{\mathcal{I}'}(X) = {\rm
  bexp}_{\mathcal{I}}(N_i) = M_i$.  Consequently, $\substitute{{\rm
  bexp}_{\mathcal{I}'}(X)}{M_i}{s} = \substitute{M_i}{M_i}{s} =
  s$.  Let us now assume $i \neq j$. Consider two subcases.  If $j \in
  \mathcal{I}$, then $j \in \mathcal{I}'$ and hence ${\rm
  bexp}_{\mathcal{I}}(X) = {\rm bexp}_{\mathcal{I}'}(X) = M_j$.  By
  $i \neq j$, we then have $\substitute{{\rm
  bexp}_{\mathcal{I}'}(X)}{M_i}{s} = \substitute{M_j}{M_i}{s} =
  M_j$. Thus, we proved the claim.  It remains to consider the case $j
  \not\in \mathcal{I}$. Note that by $j \neq i$, we then also have $j
  \not\in \mathcal{I}'$. Thus, letting $A, B \in V \cup \Sigma$ be
  such that $\rhsgen{G}{N_j} = AB$, we have ${\rm
  bexp}_{\mathcal{I}}(X) = {\rm bexp}_{\mathcal{I}}(A) \cdot
  \dol_{j} \cdot {\rm bexp}_{\mathcal{I}}(B)$ and ${\rm
  bexp}_{\mathcal{I}'}(X) = {\rm bexp}_{\mathcal{I}'}(A) \cdot
  \dol_{j} \cdot {\rm bexp}_{\mathcal{I}'}(B)$.  Consider now two
  cases:
  \begin{itemize}[itemsep=0pt,parsep=0pt]
  \item First, assume $A \in \Sigma$. On the one hand, we then have
    ${\rm bexp}_{\mathcal{I}}(A) = A$. On the other hand, ${\rm
    bexp}_{\mathcal{I}'}(A) = A$ and by $M_i \not\in \Sigma$, we have
    $\substitute{{\rm bexp}_{\mathcal{I}'}(A)}{M_i}{s} =
    \substitute{A}{M_i}{s} = A$. Putting these together, we thus
    obtain ${\rm bexp}_{\mathcal{I}}(A) = \substitute{{\rm
    bexp}_{\mathcal{I}'}(A)}{M_i}{s}$.
  \item Next, assume $A \in V$. By $|\expgen{G}{A}| <
    |\expgen{G}{N_j}|$, there exists $l \in [1 \dd j)$ such that $A =
    N_l$. By the inductive assumption we then have ${\rm
    bexp}_{\mathcal{I}}(A) = {\rm bexp}_{\mathcal{I}}(N_l) =
    \substitute{{\rm bexp}_{\mathcal{I}'}(N_l)}{M_i}{s} \allowbreak =
    \substitute{{\rm bexp}_{\mathcal{I}'}(A)}{M_i}{s}$.
  \end{itemize}
  In both cases, we thus proved that ${\rm bexp}_{\mathcal{I}}(A) =
  \substitute{{\rm bexp}_{\mathcal{I}'}(A)}{M_i}{s}$.  Analogously, it
  holds ${\rm bexp}_{\mathcal{I}}(B) = \substitute{{\rm
  bexp}_{\mathcal{I}'}(B)}{M_i}{s}$. Noting that $\dol_{j} =
  \substitute{\dol_{j}}{M_i}{s}$, we thus have
  \begin{align*}
    {\rm bexp}_{\mathcal{I}}(X)
      &= {\rm bexp}_{\mathcal{I}}(A) \cdot \dol_{j} \cdot
         {\rm bexp}_{\mathcal{I}}(B)\\
      &= \substitute{{\rm bexp}_{\mathcal{I}'}(A)}{M_i}{s} \cdot
         \substitute{\dol_{j}}{M_i}{s} \cdot
         \substitute{{\rm bexp}_{\mathcal{I}'}(B)}{M_i}{s}\\
      &= \substitute{
           {\rm bexp}_{\mathcal{I}'}(A) \cdot \dol_{j} \cdot
           {\rm bexp}_{\mathcal{I}'}(B)}{M_i}{s}\\
      &= \substitute{{\rm bexp}_{\mathcal{I'}}(X)}{M_i}{s}.
      \qedhere
  \end{align*}
\end{proof}

\begin{corollary}\label{cor:sub}
  Let $\mathcal{I} \subseteq \{1, \dots, |V|\}$ and $i \in \{1, \dots,
  |V|\} \setminus \mathcal{I}$. Denote $\mathcal{I}' = \mathcal{I}
  \cup \{i\}$ and
  $s = {\rm bexp}_{\mathcal{I}}(N_i)$.
  Then, for
  every $X \in V_{\mathcal{I}}$, it holds
  \[
    \rhsgen{G_{\mathcal{I}}}{X} =
    \substitute{\rhsgen{G_{\mathcal{I'}}}{X}}{M_i}{s}.
  \]
\end{corollary}
\begin{proof}
  Recall that $V_{\mathcal{I}} = \{S\} \cup \{M_j\}_{j \in \mathcal{I}}$.
  By \cref{lm:sub} and \cref{def:bigG}, it holds:
  \begin{align*}
    \rhsgen{G_{\mathcal{I}}}{S}
      &= \textstyle\bigodot_{i=1,\dots,|V|}
         {\rm bexp}_{\mathcal{I}}(N_i) \cdot \hash_{2i-1} \cdot
         {\rm bexp}_{\mathcal{I}}(N_i) \cdot \hash_{2i}\\
      &= \textstyle\bigodot_{i=1,\dots,|V|}
         \substitute{{\rm bexp}_{\mathcal{I}'}(N_i)}{M_i}{s} \cdot
         \substitute{\hash_{2i-1}}{M_i}{s} \cdot\\
      &  \hspace{2.35cm}
         \substitute{{\rm bexp}_{\mathcal{I}'}(N_i)}{M_i}{s} \cdot
         \substitute{\hash_{2i}}{M_i}{s}\\
      &= \substitute{\textstyle\bigodot_{i=1,\dots,|V|}
         {\rm bexp}_{\mathcal{I}'}(N_i) \cdot \hash_{2i-1} \cdot
         {\rm bexp}_{\mathcal{I}'}(N_i) \cdot \hash_{2i}}{M_i}{s}\\
      &= \substitute{\rhsgen{G_{\mathcal{I}'}}{S}}{M_i}{s}.
  \end{align*}
  Let us now consider $j \in \mathcal{I}$. Let $X, Y \in V \cup \Sigma$
  be such that $\rhsgen{G}{N_j} = XY$. By \cref{lm:sub} and
  \cref{def:bigG} it then follows that:
  \begin{align*}
    \rhsgen{G_{\mathcal{I}}}{M_j}
      &= {\rm bexp}_{\mathcal{I}}(X) \cdot \dol_{j} \cdot
         {\rm bexp}_{\mathcal{I}}(Y)\\
      &= \substitute{{\rm bexp}_{\mathcal{I}'}(X)}{M_i}{s} \cdot
         \substitute{\dol_{j}}{M_i}{s} \cdot
         \substitute{{\rm bexp}_{\mathcal{I}'}(Y)}{M_i}{s}\\
      &= \substitute{{\rm bexp}_{\mathcal{I}'}(X) \cdot \dol_{j} \cdot
                     {\rm bexp}_{\mathcal{I}'}(Y)}{M_i}{s}\\
      &= \substitute{\rhsgen{G_{\mathcal{I}'}}{M_j}}{M_i}{s}.
      \qedhere
  \end{align*}
\end{proof}

\begin{lemma}\label{lm:maximal-substring-existence}
  Let $\mathcal{I} \subseteq \{1, \dots, |V|\}$. There exists a
  maximal string with respect to $G_{\mathcal{I}}$ if and only if
  $\mathcal{I} \neq \{1, \dots, |V|\}$.
\end{lemma}
\begin{proof}

  Let us first assume that there exists $s$ that is maximal with
  respect to $G_{\mathcal{I}}$. By \cref{lm:maximal-string}, it
  follows that there exists $i \in \{1, \dots, |V|\} \setminus
  \mathcal{I}$ such that $s = {\rm bexp}_{\mathcal{I}}(N_i)$. Thus,
  $\mathcal{I} \neq \{1, \dots, |V|\}$.

  Let us now assume that $\mathcal{I} \neq \{1, \dots, |V|\}$. Observe
  (see \cref{def:bigG}) that each of the symbols in the set
  $\{\hash_j\}_{j \in [1 \dd 2|V|]} \cup \{\dol_{j}\}_{j \in
  \mathcal{I}}$ has only one occurrence on the right-hand side of
  $G_{\mathcal{I}}$; see also the proof of
  \cref{lm:maximal-string}. This implies that any substring with at
  least two non-overlapping occurrences on the right-hand side of
  $G_{\mathcal{I}}$ must be a substring of one of the string in
  $\{{\rm bexp}_{\mathcal{I}}(N_j)\}_{j \in \{1, \dots, |V|\}
  \setminus \mathcal{I}}$. On the other hand, each of the strings in
  $\{{\rm bexp}_{\mathcal{I}}(N_j)\}_{j \in \{1, \dots, |V|\}
  \setminus \mathcal{I}}$ has at least two occurrences on the
  right-hand side of $G_{\mathcal{I}}$ (in the definition of
  $S$). Thus, a string has at least two occurrences on the right-hand
  side of $G_{\mathcal{I}}$ if and only if it is a substring of one of
  the elements in $\{{\rm bexp}_{\mathcal{I}}(N_j)\}_{j \in \{1,
  \dots, |V|\} \setminus \mathcal{I}}$. Let $i = \arg\max_{j \in \{1,
  \dots, |V|\} \setminus \mathcal{I}} \{|{\rm
  bexp}_{\mathcal{I}}(N_j)|\}$ (with ties resolved arbitrarily) and
  $s = {\rm bexp}_{\mathcal{I}}(N_i)$. By the above discussion, $s$
  has at least two non-overlapping occurrences on the right-hand side
  of $G_{\mathcal{I}}$, and by definition of $i$, there are no longer
  strings with the above property. Since $i \not\in \mathcal{I}$, it
  also holds by \cref{def:bexp}, that $|s| \geq 2$. Thus, by
  \cref{def:maximal-string}, $s$ is maximal with respect to
  $G_{\mathcal{I}}$.
\end{proof}

\begin{lemma}\label{lm:one-step}
  Let $\mathcal{I} \subsetneq \{1, \dots, |V|\}$. Let $i \in
  \{1, \dots, |V|\} \setminus \mathcal{I}$ be such that
  ${\rm bexp}_{\mathcal{I}}(N_i)$ is maximal with respect to
  $G_{\mathcal{I}}$. Denote $s = {\rm bexp}_{\mathcal{I}}(N_i)$.
  Let $\Gout$ be the output of a single step of a global algorithm
  on $G_{\mathcal{I}}$ with $s$ as the maximal string. Then, $\Gout$
  is isomorphic with $G_{\mathcal{I} \cup \{i\}}$.
\end{lemma}
\begin{proof}
  Let $\mathcal{I}' = \mathcal{I} \cup
  \{i\}$.  Denote $\Gin = (\Vin, \Sigmain, \Rulesin, \Startin) =
  G_{\mathcal{I}}$ and $\Gaux = (\Vaux, \Sigmain, \Rulesaux, \Startin)
  \allowbreak = G_{\mathcal{I}'}$. Observe that letting $\Xaux = M_i$,
  it holds $\Vaux = \Vin \cup \{\Xaux\}$ and $\Xaux \not\in \Vin \cup
  \Sigmain$. By \cref{def:bigG}, we also have $\rhsgen{\Gaux}{\Xaux} =
  {\rm bexp}_{\mathcal{I}}(N_i) = s$.  Lastly, by \cref{cor:sub}, for
  every $Y \in \Vin$, it holds
  $\substitute{\rhsgen{\Gaux}{Y}}{\Xaux}{s} = \rhsgen{\Gin}{Y}$.
  Observe that we also have $\Occ(s, \rhsgen{\Gaux}{Y}) = \emptyset$,
  since by $i \in \mathcal{I}'$ and \cref{lm:dollar-on-rhs},
  $\dol_{i}$ has only one occurrence on the right-hand side of
  $\Gaux$. Since by definition it occurs in $\rhsgen{\Gaux}{\Xaux}$,
  it thus cannot occur in the definition of any other nonterminal
  (i.e., any nonterminal in $\Vin$).  Consequently, it follows by
  \cref{lm:nonov-s}, that $\Gout$ is isomorphic to $\Gaux =
  G_{\mathcal{I'}'} = G_{\mathcal{I} \cup \{i\}}$.
\end{proof}

\begin{lemma}\label{lm:iso}
  Let $G_{\rm in} = (V_{\rm in}, \Sigma, R_{\rm in}, S_{\rm in})$ and
  $G'_{\rm in} = (V'_{\rm in}, \Sigma, R'_{\rm in}, S'_{\rm in})$ be
  isomorphic SLGs via bijection $g : V_{\rm in} \cup \Sigma
  \rightarrow V'_{\rm in} \cup \Sigma$.  Let $s \in (V_{\rm in} \cup
  \Sigma)^{+}$ and $s' = \bigodot_{i=1,\dots,|s|}g(s[i])$. Then:
  \begin{enumerate}[itemsep=0pt,parsep=0pt]
  \item $s$ is maximal with respect to $G_{\rm in}$ if and only if
    $s'$ is maximal with respect to $G'_{\rm in}$,
  \item If $s$ is maximal, then the output of a single step of a
    global algorithm on $G_{\rm in}$ and $s$ is isomorphic to that on
    $G'_{\rm in}$ and $s'$.
  \end{enumerate}
\end{lemma}
\begin{proof}
  1. Let $X \in V_{\rm in}$ and $w \in (V_{\rm in} \cup \Sigma)^{*}$.
  Denote $t = \rhsgen{G_{\rm in}}{X}$, $X' = g(X)$, $t' =
  \rhsgen{G'_{\rm in}}{X'}$, and $w' = \bigodot_{i=1,\dots,|w|}g(w[i])$.
  By definition of $g$, it then holds $|t| = |t'|$ and $|w| = |w'|$.
  Moreover, for every $i \in [1 \dd |t|]$, $i \in \Occ(w, t)$ holds if
  and only if $i \in \Occ(w', t')$. Hence, the number of
  (non-overlapping) occurrences of $w$ on the right-hand side of
  $G_{\rm in}$ is equal to the number of (non-overlapping) occurrences
  of $w'$ on the right-hand side of $G'_{\rm in}$. This implies that
  $|s'| = |s| \geq 2$, and that the second and third condition in
  \cref{def:maximal-string} holds for $G_{\rm in}$ and $s$ if and
  only if it holds for $G'_{\rm in}$ and $s'$.

  2. Let $G_{\rm new}$ (resp.\ $G'_{\rm new}$) be the output of a
  single step of global algorithm with $G_{\rm in}$ (resp.\ $G'_{\rm
  in}$) as input and $s$ (resp. $s'$) as a maximal string.  Let
  $N_{\rm new}$ (resp.\ $N'_{\rm new}$) be the newly introduced
  nonterminal.  We then have $\rhsgen{G_{\rm in}}{N_{\rm new}} = s$
  (resp.\ $\rhsgen{G'_{\rm in}}{N'_{\rm in}} = s'$). Let $g' : V_{\rm
  in} \cup \{N_{\rm new}\} \cup \Sigma \rightarrow V'_{\rm in} \cup
  \{N'_{\rm new}\} \cup \Sigma$ be defined so that $g'(N_{\rm new}) =
  N'_{\rm new}$, and for the remaining arguments, $g'$ matches $g$.
  Clearly, $g'$ is a bijection and it holds $s' =
  \bigodot_{i=1,\dots,|s|} g'(s[i])$, since $s' \in (V'_{\rm in} \cup
  \Sigma)^{*}$.  To see that the condition for isomorphism is
  satisfied for the remaining nonterminals, consider $X \in V_{\rm
  in}$ and let $t, t'$, and $X'$ be defined as above.  We observe
  that by the above characterization of $\Occ(s, t)$ and $\Occ(s',
  t')$, it follows that if $t = y_0 s y_1 s \dots s y_{k-1} s y_{k}$
  is a factorization including all non-overlapping occurrences of $s$
  in $t$ obtained by a left-to-right greedy search (see
  \cref{rm:nonov}), then the corresponding factorization for $t'$ and
  $s'$ is $t' = y'_0 s' y'_1 s' \dots s' y'_{k-1} s' y'_{k}$, where
  $y'_i = \bigodot_{j \in [1 \dd |y_i|]}g(y_i[j])$. This implies
  that, letting $z = \rhsgen{G_{\rm new}}{X}$ and $z' =
  \rhsgen{G'_{\rm new}}{X'}$, it holds that $z = y_0 N_{\rm new} y_1
  N_{\rm new} \dots N_{\rm new} y_{k-1} N_{\rm new} y_{k}$ and $z' =
  y'_0 N'_{\rm new} y'_1 N'_{\rm new} \dots N'_{\rm new} y'_{k-1} N'_{\rm
  new} y'_{k}$.  By definition of $g'$, we thus immediately obtain
  that for every $j \in [1 \dd |z|]$, $z'[j] = g'(z[j])$.
\end{proof}

\begin{lemma}\label{lm:global-output-grammar}
  The output of every global algorithm on $w$ is isomorphic with
  $G_{\{1, \dots, |V|\}}$.
\end{lemma}
\begin{proof}
  Recall, that $w \in \alpha(G)$ has been fixed for the duration of
  this section.  Let $k \in [0 \dd |V|]$. We prove by induction on $k$
  that after $k$ steps of a global algorithms on $w$, there exists
  $\mathcal{I} \subseteq \{1, \dots, |V|\}$ such that $|\mathcal{I}| =
  k$ and the resulting grammar is isomorphic with
  $G_{\mathcal{I}}$. For $k = |V|$ this yields the claim, since then
  we must have $\mathcal{I} = \{1, \dots, |V|\}$.

  To show the induction base, observe that every global algorithm
  given the string $w$ as input starts with a grammar $G_{\rm beg} =
  (V_{\rm beg}, \Sigma'', R_{\rm beg}, S_{\rm beg})$ such that $V_{\rm
  beg} = \{S_{\rm beg}\}$ and $\rhsgen{G_{\rm beg}}{S_{\rm beg}} = w
  = \bigodot_{i=1,\dots,|V|} \expgen{G'}{N_i} \cdot \hash_{2i-1} \cdot
  \expgen{G'}{N_i} \cdot \hash_{2i}$.  Observe now that by
  \cref{def:bigG}, it holds $V_{\emptyset} = \{S\}$ and
  $\rhsgen{G_{\emptyset}}{S} = \bigodot_{i=1,\dots,|V|} {\rm
  bexp}_{\emptyset}(N_i) \cdot \hash_{2i-1} \cdot {\rm
  bexp}_{\emptyset}(N_i) \cdot \hash_{2i}$. It remains to observe
  that when $\mathcal{I} = \emptyset$, the definition of ${\rm
  bexp}_{\mathcal{I}}(N_i)$ matches that of $\expgen{G'}{N_i}$ (see
  also \cref{rm:bexp}). For every $i \in \{1, \dots, |V|\}$, we thus
  have ${\rm bexp}_{\emptyset}(N_i) = \expgen{G'}{N_i}$, and hence
  $\rhsgen{G_{\emptyset}}{S} = \rhsgen{G_{\rm beg}}{S_{\rm beg}}$.
  Consequently, the bijection $g : \Sigma'' \cup \{S_{\rm beg}\}
  \rightarrow \Sigma'' \cup \{S\}$ such that $g(c) = c$ holds for $c
  \in \Sigma''$, and $g(S_{\rm beg}) = S$, establishes the isomorphism
  of $G_{\rm beg}$ and $G_{\emptyset}$.

  Let us now consider $k > 0$. By the inductive assumption, there
  exists $\mathcal{I} \subseteq \{1, \dots, |V|\}$ such that
  $|\mathcal{I}| = k - 1$ and the grammar $G_{\rm in} = (V_{\rm in},
  \Sigma'', R_{\rm in}, S_{\rm in})$ resulting from the first $k-1$
  steps of the algorithm is isomorphic with $G_{\mathcal{I}}$.  Let $g
  : V_{\rm in} \cup \Sigma'' \rightarrow V_{\mathcal{I}} \cup
  \Sigma''$ be the corresponding bijection.  Observe now that $k - 1 <
  |V|$ implies $\mathcal{I} \neq \{1, \dots, |V|\}$. Thus, by
  \cref{lm:maximal-substring-existence}, there exists a maximal string
  with respect to $G_{\mathcal{I}}$. By \cref{lm:iso}, this implies
  there also exists a maximal string with respect to $G_{\rm in}$. Let
  $s$ be the maximal string with respect to $G_{\rm in}$ that was
  chosen by the algorithm, and let $G_{\rm out} = (V_{\rm out},
  \Sigma'', R_{\rm out}, S_{\rm out})$ be the result of one step of
  the algorithm one $G_{\rm in}$ and $s$. Let $s' =
  \bigodot_{i=1,\dots,|s|}g(s[i])$. By \cref{lm:iso}, $s'$ is maximal
  with respect to $G_{\mathcal{I}}$.  By \cref{lm:maximal-string},
  there exists $i \in \{1, \dots, |V|\} \setminus \mathcal{I}$ such
  that $s' = {\rm bexp}_{\mathcal{I}}(N_i)$. Let $G'_{\rm out} =
  (V'_{\rm out}, \Sigma'', R'_{\rm out}, S)$ be the result of one step
  of the global algorithm on $G_{\mathcal{I}}$ as input with $s'$ as
  the maximal string. By \cref{lm:one-step}, $G'_{\rm out}$ is
  isomorphic to $G_{\mathcal{I} \cup \{i\}}$. On the other hand, by
  \cref{lm:iso}, $G_{\rm out}$ is isomorphic with $G'_{\rm
  out}$. Since isomorphism is transitive, it thus follows that
  $G_{\rm out}$ is isomorphic with $G_{\mathcal{I} \cup \{i\}}$. This
  concludes the proof of the inductive step.
\end{proof}

\begin{lemma}\label{lm:global-output-size}
  Every global algorithm with $w$ as input outputs a grammar of size
  $\tfrac{7}{2}|G|$.
\end{lemma}
\begin{proof}
  Denote $\mathcal{I} = \{1, \dots, |V|\}$.  By
  \cref{lm:global-output-grammar}, the output of every global
  algorithm on $w$ is isomorphic with $G_{\mathcal{I}}$, and hence
  also of equal size (see \cref{sec:prelim}).  By \cref{def:alpha}
  (see also \cref{rm:bexp}) it follows that $|G_{\mathcal{I}}| =
  |\rhsgen{G_{\mathcal{I}}}{S}| +
  \sum_{i=1}^{|V|}|\rhsgen{G_{\mathcal{I}}}{N_i}| = 4|V| + 3|V| =
  7|V|$. On the other hand, since $G$ is admissible, it holds $|G| =
  2|V|$. Consequently, the size of $G_{\mathcal{I}}$ is $7|V| =
  \tfrac{7}{2}|G|$.
\end{proof}

\begin{theorem}\label{th:global}
  Let \algname{Global} be any global grammar compression algorithm
  (\cref{sec:global}). For any $u \in \Sigma^{*}$, let
  $\algname{Global}(u)$ denote the output of \algname{Global} on
  $u$. In the cell-probe model, every static data structure that for
  $u \in \Sigma^n$ uses $\bigO(|\algname{Global}(u)| \log^{c} n)$
  space $($where $c = \bigO(1))$, requires $\Omega(\log n / \log \log
  n)$ time to answer random access queries on $u$.
\end{theorem}
\begin{proof}
  Suppose that there exists a structure $D$ that for any $u \in
  \Sigma^n$ uses $\bigO(|\algname{Global}(u)|\log^{c} n)$ space
  (where $c = \bigO(1)$) and answers random access queries on $u$ in
  $o(\log n / \log \log n)$ time.  Let $\Pts \subseteq [1 \dd m]^2$ be
  any set of $|\Pts| = m$ points on an $m \times m$ grid. Assume for
  simplicity that $m$ is a power of two (otherwise, letting $m'$ be
  the smallest power of two satisfying $m' \geq m$, we apply the proof
  for $\Pts' = \Pts \cup \{(p, p)\}_{p \in (m \dd m']}$; note that the
  answer to any range counting query on $\Pts$ is equal to the answer
  on $\Pts'$). By
  \cref{lm:answer-string-grammar-size}, there exists an admissible SLG
  $G_{\Pts} = (V_{\Pts}, \{0, 1\}, R_{\Pts}, S_{\Pts})$ of height $h = \bigO(\log m)$
  such that $L(G_{\Pts}) = \{A(\Pts)\}$ is the answer string for $\Pts$
  (\cref{def:answer-string}), and it holds $|G_{\Pts}| = \bigO(m \log
  m)$. Note that since $G_{\Pts}$ is admissible, it holds
  $\sum_{X \in V_{\Pts}}|\expgen{G_{\Pts}}{X}| \leq |A(\Pts)| \cdot (h + 1)
  = \bigO(m^2 \log m)$. Let $w \in \alpha(G_{\Pts})$ (\cref{def:alpha}).
  Note that $w$ has the following two properties:
  \begin{itemize}
  \item By \cref{lm:alpha-length} and the above observation, $|w| =
    4\sum_{X \in V_{\Pts}}|\expgen{G_{\Pts}}{X}| = \bigO(m^2 \log m)$,
  \item By \cref{lm:delta}, there exists $\delta \geq 0$, such that
    for every $j \in [1 \dd m^2]$, it holds $A(\Pts)[j] = w[\delta +
    2j - 1]$.
  \end{itemize}
  Let $G_w = \algname{Global}(w)$ be the output of \algname{Global}
  on $w$. By \cref{lm:global-output-size}, we have $|G_w| =
  \tfrac{7}{2}|G_{\Pts}| = \bigO(m \log m)$. Let $D'$ denote a data
  structure consisting of the following two components:
  \begin{enumerate}
  \item The data structure $D$ for string $w$. By $|G_w| = \bigO(m
    \log m)$ and the above assumption, $D$ uses
    $\bigO(|\algname{Global}(w)| \log^c |w|) = \bigO(m \log m \log^c
    (m^2 \log m)) = \bigO(m \log^{1 + c} m)$ space, and implements random
    access to $w$ in $o(\log |w| / \log \log |w|) \,{=}\, o(\log
    (m^2 \log m) / \log \log (m^2 \log m)) \allowbreak =
    o(\log m / \log \log m)$ time,
  \item The position $\delta \geq 0$, as defined above.
  \end{enumerate}
  Observe that given the structure $D'$ and any $(x, y) \in [1 \dd
  m]^2$, we can answer in $o(\log m / \log \log m)$ the parity range
  query on $\Pts$ with arguments $(x, y)$ by issuing a random access
  query to $w$ with position $j = \delta + 2j' - 1$, where $j' = x +
  (y-1)m$. Thus, the existence of $D'$ contradicts
  \cref{th:parity-range-counting-lower-bound}.
\end{proof}

\subsection{Analysis of \algname{Sequential}}\label{sec:sequential}

\begin{lemma}\label{lm:sequential}
  Let $G$ be an admissible SLG. For every $w \in \alpha(G)$,
  \algname{Sequential} outputs a grammar of size
  $\tfrac{7}{2}|G|$.
\end{lemma}
\begin{proof}

  Denote $G = (V, \Sigma, R, S)$ and let $G' = (V, \Sigma', R', S)$ be
  the auxiliary grammar constructed as in \cref{def:alpha}.  Let also
  $(N_i)_{i \in [1 \dd |V|]}$ be the corresponding sequence from
  \cref{def:alpha}, i.e., a sequence of all nonterminals of $G'$
  ordered by the expansion length in $G'$ and such that $w =
  \bigodot_{i=1,\dots,|V|} \expgen{G'}{N_i} \cdot \hash_{2i-1} \cdot
  \expgen{G'}{N_i} \cdot \hash_{2i}$. Finally, let $\Sigma'' = \Sigma'
  \cup \{\hash_i : i \in [1 \dd 2|V|]\}$ be as in \cref{def:alpha}. We
  prove by induction that for every $k \geq 0$, the grammar produced
  by \algname{Sequential} after $8k$ steps is isomorphic with $G_k$,
  where $G_k = (V_k, \Sigma'', R_k, S_k)$ is such that:
  \begin{itemize}[itemsep=0pt,parsep=0pt]
  \item $V_k = \{S_k, N_1, \dots, N_k\}$,
  \item For every $i \in [1 \dd k]$, $\rhsgen{G_k}{N_i} =
    \rhsgen{G'}{N_i}$,
  \item $\rhsgen{G_k}{S_k} = \bigodot_{i=1,\dots,k}
    N_i \cdot \hash_{2i-1} \cdot N_i \cdot \hash_{2i}$.
  \end{itemize}

  The base case follows immediately, since $G_0$ is the grammar
  containing only the starting nonterminal $S_0$ whose definition
  is the empty string. Such grammar is isomorphic with the starting
  grammar of \algname{Sequential}.

  Let us now consider the algorithm after the first $8(k-1)$ steps.
  By the inductive assumption, the current grammar $G_{\rm cur} =
  (V_{\rm cur}, \Sigma'', R_{\rm cur}, S_{\rm cur})$ is isomorphic
  with $G_{k-1}$. Let $f: V_{\rm cur} \cup \Sigma'' \rightarrow
  V_{k-1} \cup \Sigma''$ be the corresponding bijection (see
  \cref{sec:prelim}).  By $|\expgen{G_{k-1}}{S_{k-1}}| =
  2\sum_{i=1}^{k-1} (|\expgen{G'}{N_i}| + 1)$, the remaining
  unprocessed suffix of $w$ is thus $w' := \bigodot_{i=k,\dots,|V|}
  \expgen{G'}{N_i} \cdot \hash_{2i-1} \cdot \expgen{G'}{N_i} \cdot
  \hash_{2i}$.  Let $A, B \in V \cup \Sigma$ be such that
  $\rhsgen{G'}{N_k} = A \cdot \dol_{k} \cdot B$.  Consider the next
  8 steps of \algname{Sequential}:
  \begin{enumerate}
  \item\label{it:lm-sequential-step-1} Since $\dol_{k}$ occurs only
    on the right-hand side of $N_k$, the occurrence of $\dol_{k}$ at
    position $|\expgen{G'}{A}| + 1$ in $w'$ is thus the leftmost in
    $w$.  Therefore, the longest prefix of $w'$ equal to the expansion
    of some nonterminal in $G_{\rm cur}$ is not longer than
    $|\expgen{G'}{A}|$. If $A \in \Sigma$, then by definition we
    append $A$ to $\rhsgen{G_{\rm cur}}{S_{\rm cur}}$. Otherwise, by
    $|\expgen{G'}{A}| < |\expgen{G'}{A} \dol_{k} \expgen{G'}{B}| =
    |\expgen{G'}{N_k}|$, there exists $k' < k$ such that $A =
    N_{k'}$. By the inductive assumption, we thus have $\expgen{G'}{A}
    = \expgen{G_{\rm cur}}{f^{-1}(N_{k'})}$, and hence we append
    $f^{-1}(N_{k'})$ to $\rhsgen{G_{\rm cur}}{S_{\rm cur}}$.
    Appending a symbol to $\rhsgen{G_{\rm cur}}{S_{\rm cur}}$ cannot
    create a repeating pair, since the last symbol of
    $\rhsgen{G_{k-1}}{S_{k-1}}$ is $\hash_{2(k-1)}$.
  \item\label{it:lm-sequential-step-2} The next unprocessed symbol of
    $w'$ is $\dol_{k}$. Since as noted above this is its leftmost
    occurrence in $w$, in this step we append $\dol_{k}$ to
    $\rhsgen{G_{\rm cur}}{S_{\rm cur}}$, and cannot create a repeating
    pair.
  \item\label{it:lm-sequential-step-3} Similarly as in
    Step~\ref{it:lm-sequential-step-1}, if $B \in
    \Sigma$ then we append $B$ to $\rhsgen{G_{\rm cur}}{S_{\rm
    cur}}$. Otherwise, by $|\expgen{G'}{B}| < |\expgen{G'}{N_k}|$,
    there exists $k'' < k$ such that $B = N_{k''}$. In this case, we
    append $f^{-1}(N_{k''})$ to $\rhsgen{G_{\rm cur}}{S_{\rm cur}}$.
    Since we previously appended $\dol_{k}$, we do not create a
    repeating pair.
  \item\label{it:lm-sequential-step-4} The next unprocessed symbol
    of $w'$ is $\hash_{2k - 1}$, which is handled as in
    Step~\ref{it:lm-sequential-step-2}.
  \item\label{it:lm-sequential-step-5} The remaining unprocessed
    suffix starts with $\expgen{G'}{N_k} \cdot \hash_{2k} =
    \expgen{G'}{A} \cdot \dol_{k} \cdot \expgen{G'}{B} \cdot
    \hash_{2k}$. The only nonterminal in $G_{\rm cur}$ whose expansion
    contain $\dol_{k}$ is $S_{\rm cur}$. However, since
    $\expgen{G_{\rm cur}}{S_{\rm cur}}$ also contains $\hash_{2k-1}$,
    we cannot match its expansion with any prefix. Consequently, the
    length of the longest prefix matching the expansion of some
    existing nonterminal does not contain $\dol_{k}$. On the other
    hand, as observed in Step~\ref{it:lm-sequential-step-1}, if $A
    \not\in \Sigma$, then we have $\expgen{G'}{A} = \expgen{G_{\rm
    cur}}{f^{-1}(N_{k'})}$, where $k' < k$. We thus append
    $f^{-1}(N_{k'})$ to $\rhsgen{G_{\rm cur}}{S_{\rm cur}}$. Otherwise
    ($A \in \Sigma$), we append $A$. Since in both cases we append a
    symbol right after $\hash_{2k - 1}$, we do not create a repeating
    pair.
  \item Next, by a similar argument as above, we append $\dol_{k}$
    to $\rhsgen{G_{\rm cur}}{S_{\rm cur}}$. Note that after this,
    letting $A'$ be the symbol appended to $\rhsgen{G_{\rm
    cur}}{S_{\rm cur}}$ in Steps~\ref{it:lm-sequential-step-1}
    and~\ref{it:lm-sequential-step-5}, we have a repetition of a pair
    $A' \cdot \dol_{k}$ in $\rhsgen{G_{\rm cur}}{S_{\rm cur}}$. We
    thus add a new nonterminal $N_{\rm tmp}$ with $\rhsgen{G_{\rm
    cur}}{N_{\rm tmp}} = A' \cdot \dol_{k}$ to $G_{\rm cur}$,
    and replace both occurrences of $A' \cdot \dol_{k}$ in
    $\rhsgen{G_{\rm cur}}{S_{\rm cur}}$ with $N_{\rm tmp}$. Observe
    that after this replacement, $N_{\rm tmp}$ occurs in
    $\rhsgen{G_{\rm cur}}{S_{\rm cur}}$ twice.  Note also that if $A'
    \not\in \Sigma$ then, in addition to the occurrence in
    $\rhsgen{G_{\rm cur}}{N_{\rm tmp}}$, the nonterminal $A'$ has at
    least two more occurrences (by the inductive assumption) in
    $\rhsgen{G_{\rm cur}}{S_{\rm cur}}$.
  \item Next, using the analogous argument as in
    Step~\ref{it:lm-sequential-step-3}, we append either $B$ or
    $f^{-1}(N_{k''})$ to $\rhsgen{G_{\rm cur}}{S_{\rm cur}}$. Denote
    the appended symbol by $B'$. Note that after this, we have a
    repetition of a pair $N_{\rm tmp} \cdot B'$ in
    $\rhsgen{G_{\rm cur}}{S_{\rm cur}}$. Thus, we add a new
    nonterminal $N'_{\rm tmp}$ with
    $\rhsgen{G_{\rm cur}}{N'_{\rm tmp}} = N_{\rm tmp} \cdot B'$ to
    $G_{\rm cur}$, and replace both occurrences of $N_{\rm tmp} \cdot B'$
    in $\rhsgen{G_{\rm cur}}{S_{\rm cur}}$ with $N'_{\rm tmp}$. Observe
    that after this replacement, the nonterminal $N_{\rm tmp}$ occurs
    only once on the right-hand side of $G_{\rm cur}$. We thus
    replace its only occurrence in $\rhsgen{G_{\rm cur}}{N'_{\rm tmp}}$
    with its definition (i.e., with $A' \cdot \dol_{k}$). This
    results in $\rhsgen{G_{\rm cur}}{N'_{\rm tmp}} = A' \cdot
    \dol_{k} \cdot B'$. The nonterminal $N'_{\rm tmp}$ occurs twice
    in $\rhsgen{G_{\rm cur}}{S_{\rm cur}}$, and hence no further
    modifications are needed.
  \item The next unprocessed symbol of $w'$ is $\hash_{2k}$, which
    is handled as in Step~\ref{it:lm-sequential-step-4}.
  \end{enumerate}
  Observe now that the only new nonterminal created during the above
  eight steps that was not subsequently deleted is $N'_{\rm tmp}$, and
  $\rhsgen{G_{\rm cur}}{S_{\rm cur}}$ was updated by appending the
  string $N'_{\rm tmp} \cdot \hash_{2k-1} \cdot N'_{\rm tmp} \cdot
  \hash_{2k}$. Let us now extend $f$ so that $f(N'_{\rm tmp}) =
  N_k$. Note that if $A' \in \Sigma$ then $A = A'$ and hence the first
  symbols in $\rhsgen{G_{k}}{N_k}$ and $\rhsgen{G_{\rm cur}}{N'_{\rm
  tmp}}$ are equal. On the other hand, if $A \not\in \Sigma$, then
  $A = N_{k'} \in V_{k-1}$ and $A' = f^{-1}(N_{k'})$. Consequently,
  $f(A') = N_{k'}$, and hence letting $S = \rhsgen{G_{k}}{N_k}$ and
  $S' = \rhsgen{G_{\rm cur}}{N'_{\rm tmp}}$, we have $S[1] = N_{k'} =
  f(f^{-1}(N_{k'})) = f(A') = f(S'[1])$. Analogously, we either have
  $S[3] = S'[3]$ or $S[3] = f(S'[3])$. We also have $S[2] =
  \dol_{k} = S'[2]$. The grammar $G_{\rm cur}$ is thus isomorphic
  with $G_{k}$. This concludes the proof of the inductive step.

  By the above, the final grammar computed by \algname{Sequential} is
  isomorphic to $G_{|V|}$. As noted in \cref{sec:prelim}, this implies
  that its size is equal to $|G_{|V|}| = \sum_{N \in
  V}|\rhsgen{G_{|V|}}{N}| = \sum_{i=1}^{|V|} |\rhsgen{G_{|V|}}{N_i}|
  + |\rhsgen{G_{|V|}}{S_{|V|}}| = 3|V| + 4|V| = 7|V|$. On the other
  hand, since $G$ is admissible, we have $|G| = 2|V|$. Consequently,
  the final grammar computed by \algname{Sequential} has size $7|V| =
  \tfrac{7}{2}|G|$.
\end{proof}

\begin{theorem}\label{th:sequential} 
  For any $u \in \Sigma^{*}$, let $\algname{Sequential}(u)$ denote the
  output of the \algname{Sequential} algorithm on $u$. In the cell-probe
  model, every static data structure that for a string $u \in \Sigma^n$ uses
  $\bigO(|\algname{Sequential}(u)| \log^{c} n)$ space $($where $c =
  \bigO(1))$, requires $\Omega(\log n / \log \log n)$ time to answer
  random access queries on $u$.
\end{theorem}
\begin{proof}
  The proof proceeds analogously as in \cref{th:global}, except we
  observe that the grammar $G_w = \algname{Sequential}(w)$ satisfies
  $|G_w| = \tfrac{7}{2}|G_{\Pts}|$ by \cref{lm:sequential}.
\end{proof}

\subsection{Analysis of \algname{Sequitur}}\label{sec:sequitur}

\begin{lemma}\label{lm:sequitur}
  Let $G$ be an admissible SLG. For every $w \in \alpha(G)$
  (\cref{def:alpha}), \algname{Sequitur} outputs a grammar of
  size $\tfrac{7}{2}|G|$.
\end{lemma}
\begin{proof}

  We start similarly as in the proof of \cref{lm:sequential}.  Denote
  $G = (V, \Sigma, R, S)$ and let $G' = (V, \Sigma', R', S)$ be the
  auxiliary grammar constructed as in \cref{def:alpha}.  Let $(N_i)_{i
  \in [1 \dd |V|]}$ be the corresponding sequence from
  \cref{def:alpha}.  For any $i \in [1 \dd |V|]$, let $L_i, R_i \in V
  \cup \Sigma$ be such that $\rhsgen{G}{N_i} = L_i R_i$.  Finally, let
  $\Sigma'' = \Sigma' \cup \{\hash_i : i \in [1 \dd 2|V|]\}$ be as in
  \cref{def:alpha}.  For any $k \in [0 \dd |V|]$, let $s_k =
  \sum_{i=1}^{k} |\expgen{G'}{N_i \cdot \hash_{2i-1} \cdot N_i \cdot
  \hash_{2i}}|$.  We prove by induction that for every $k \in [0 \dd
  |V|]$, the grammar produced by \algname{Sequitur} after $s_k$
  steps~\footnote{Here by a \emph{step} we consider the execution of
  the algorithm between the processing of two input symbols. Note that
  unlike for \algname{Sequential}, this may take more than $\bigO(1)$
  time, since in \algname{Sequitur} appending a single symbol may
  spawn a chain of $\omega(1)$ reductions, as defined in
  \cref{sec:algorithms}.} is isomorphic with $G_k = (V_k, \Sigma'',
  R_k, S_k)$ defined as in the proof of \cref{lm:sequential}, i.e.,
  such that $V_k = \{S_k, N_1, \dots, N_k\}$, for every $i \in [1
  \dd k]$, $\rhsgen{G_k}{N_i} = \rhsgen{G'}{N_i}$, and
  $\rhsgen{G_k}{S_k} = \bigodot_{i=1,\dots,k} N_i \cdot \hash_{2i-1}
  \cdot N_i \cdot \hash_{2i}$.

  We introduce the following notation:
  \begin{itemize}
  \item For any $A \in V \cup \Sigma$ and $\ell \in [1 \dd
    |\expgen{G}{A}|]$, by ${\rm seq}_{\ell}(A)$ denote a sequence of
    nonterminals from $V$ (i.e., a string over alphabet $V$) defined
    recursively as follows.  If $A \in \Sigma$, then ${\rm
    seq}_{\ell}(A) = \emptystring$.  Otherwise, letting $X, Y \in V
    \cup \Sigma$ be such that $\rhsgen{G}{A} = XY$ and $\ell' =
    |\expgen{G}{X}|$, we define
    \vspace{2ex}
    \[
      {\rm seq}_{\ell}(A) =
      \begin{dcases}
        {\rm seq}_{\ell}(X) & \text{if }\ell \leq \ell',\\
        A \cdot {\rm seq}_{\ell - \ell'}(Y) & \text{otherwise}.
      \end{dcases}
      \vspace{2ex}
    \]
    In other words, ${\rm seq}_{\ell}(A)$ is the sequence of
    nonterminals obtained by traversing the parse tree
    $\mathcal{T}_G(A)$ from the root to the leftmost $\ell$th leaf and
    including every nonterminal from which the path goes to the right
    child.
  \item For any $A \in V \cup \Sigma$ and $\ell \in [1 \dd
    |\expgen{G}{A}|]$, we define ${\rm ids}_{\ell}(A) = (a_1, \dots,
    a_{q})$ such that, letting $s = {\rm seq}_{\ell}(A)$, we have $q =
    |s|$, and for every $i \in [1 \dd q]$ it holds $N_{a_i} =
    s[i]$. In other words, the sequence ${\rm ids}_{\ell}(A)$ contains
    the indices of the consecutive variables from the sequence ${\rm
    seq}_{\ell}(A)$.
  \end{itemize}

  To prove the base case, note that $s_0 = 0$, and $G_0$ contains
  only the starting nonterminal $S_0$ whose definition is the empty
  string. Such SLG is isomorphic with the starting grammar of
  \algname{Sequitur}.

  The proof of the inductive step for \algname{Sequitur} is more
  involved than for \algname{Sequential}. By the inductive assumption,
  grammar $G_{\rm cur} = (V_{\rm cur}, \Sigma'', R_{\rm cur}, S_{\rm
  cur})$ after $s_{k-1}$ steps of \algname{Sequitur} is isomorphic
  with $G_{k-1}$. Let $f: V_{\rm cur} \cup \Sigma'' \rightarrow
  V_{k-1} \cup \Sigma''$ be the corresponding bijection (see
  \cref{sec:prelim}).  By $|\expgen{G_{k-1}}{S_{k-1}}| =
  2\sum_{i=1}^{k-1} (|\expgen{G'}{N_i}| + 1)$, the remaining
  unprocessed suffix of $w$ is thus $w' := \bigodot_{i=k,\dots,|V|}
  \expgen{G'}{N_i} \cdot \hash_{2i-1} \cdot \expgen{G'}{N_i} \cdot
  \hash_{2i}$.  Let $A, B \in V \cup \Sigma$ be such that
  $\rhsgen{G'}{N_k} = A \cdot \dol_{k} \cdot B$.  We split the next
  $s_{k} - s_{k-1}$ steps, i.e., the processing of the substring
  $\expgen{G'}{N_k} \cdot \hash_{2k-1} \cdot \expgen{G'}{N_k} \cdot
  \hash_{2k} = \expgen{G'}{A} \cdot \dol_{k} \cdot \expgen{G'}{B}
  \cdot \hash_{2k-1} \cdot \expgen{G'}{A} \cdot \dol_{k} \cdot
  \expgen{G'}{B} \cdot \hash_{2k}$, into 8 \emph{phases} (where each
  phase consists of some number of consecutive steps), and analyze
  each phase separately:

  \vspace{2ex}
  \noindent
  \emph{Phase 1}
  Denote $\ell_{A} = |\expgen{G}{A}|$ and note that
  $|\expgen{G'}{A}| = 2\ell_{A} - 1$. We define the first phase as
  the processing of the leftmost $2(\ell_{A} - 1)$ symbols of $w'$.
  For any $\ell \in [1 \dd \ell_{A})$, we define the SLG
  $G_{k-1,\ell} = (V_{k-1,\ell}, \Sigma'', R_{k-1,\ell}, S_{k-1})$,
  where, letting $(a_1, \dots, a_q) = {\rm ids}_{\ell + 1}(A)$, it
  holds $V_{k-1,\ell} = V_{k-1} \cup \{T_{a_1}, \dots, T_{a_q}\} =
  \{N_1, \dots, N_{k-1}, S_{k-1}, T_{a_1}, \dots, T_{a_q}\}$
  (nonterminals $T_i$ are named so, since they are temporary) and
  the rules of the grammar are defined as follows:
  \begin{itemize}
  \item For $i \in \{a_1, \dots, a_q\}$,
    $\rhsgen{G_{k-1,\ell}}{N_{i}} = T_{i} R_{i}$ and
    $\rhsgen{G_{k-1,\ell}}{T_{i}} = L_{i} \dol_{i}$,
  \item For $i \in [1 \dd |V|] \setminus \{a_1, \dots, a_q\}$,
    $\rhsgen{G_{k-1,\ell}}{N_i} = \rhsgen{G'}{N_i}$,
  \item $\rhsgen{G_{k-1,\ell}}{S_{k-1}} = (\bigodot_{i=1,\dots,k-1}
    N_i \cdot \hash_{2i-1} \cdot N_i \cdot \hash_{2i}) \cdot T_{a_1}
    \cdot T_{a_2} \cdot \ldots \cdot T_{a_q}$.
  \end{itemize}
  In other words, the grammar $G_{k-1,\ell}$ is obtained by modifying
  $G_{k-1}$ in two ways. The first modification is
  rather cosmetic, and it is to alter the definition of the rules
  $\{N_i : i \in \mathcal{I}\}$ for $\mathcal{I}$ containing
  identifiers of all nonterminals occurring in ${\rm seq}_{\ell +
  1}(A)$. The alteration is to change each nonterminal $N_i$
  (where $i \in \mathcal{I}$) so that rather than $L_i \cdot
  \dol_{i} \cdot R_i$, its definition is of length two, i.e., $T_i
  \cdot R_i$. We then set the definition of the auxiliary variable
  $T_i$ to $L_i \cdot \dol_{i}$.  Therefore, the expansions of
  nonterminals $\{N_1, \dots, N_{k-1}\}$ in both grammars
  $G_{k-1,\ell}$ and $G_{k-1}$ are equal.  The second difference of
  $G_{k-1,\ell}$ compared to $G_{k-1}$ is that the definition of the
  starting nonterminal contains all the temporary variables at the
  end, i.e., $\rhsgen{G_{k-1,\ell}}{S_{k-1}} = \rhsgen{G_{k-1}}{S_{k-1}}
  \cdot T_{a_1} \cdots \ldots \cdot T_{a_{q}}$, where $q =
  |\mathcal{I}|$. We now prove that the grammars $G_{k-1,\ell}$ (where
  $\ell \in [1 \dd \ell_A)$) describe the behavior of the
  \algname{Sequitur} algorithm as it processes the prefix of $w'$ of
  length $2(\ell_A - 1)$. More precisely, we show that for every $\ell
  \in [1 \dd \ell_A)$, the grammar computed by \algname{Sequitur} after
  processing the leftmost $2\ell$ symbols of $w'$ is isomorphic with
  $G_{k-1,\ell}$. The proof is by induction on $\ell$. To prove the
  base case, consider the execution of \algname{Sequitur} on the
  first two symbols on $w'$.  Let $a_1 \in [1 \dd |V|]$ be such that
  $N_{a_1}$ is the symbol of the penultimate node on the leftmost
  root-to-leaf path in the parse tree $\mathcal{T}_G(A)$. On the one
  hand, since $\expgen{G'}{A}$ is a prefix of $w'$, it follows by
  $\rhsgen{G'}{N_{a_1}} = L_{a_1} \cdot \dol_{a_1} \cdot
  R_{a_1}$ that $w'[1] = \expgen{G'}{A}[1] = \expgen{G'}{N_{a_1}}[1]
  = L_{a_1}$ and $w'[2] = \expgen{G'}{A}[2] =
  \expgen{G'}{N_{a_1}}[2] = \dol_{a_1}$. On the other hand, note
  that since $G$ is admissible, $N_{a_1}$ is the only node in which
  the path from the root of $\mathcal{T}_G(A)$ to its 2nd leftmost
  leaf turns right. Consequently, ${\rm seq}_{2}(A) = (N_{a_1})$ and
  ${\rm ids}_{2}(A) = a_1$.  Consider now the execution of
  \algname{Sequitur} when processing the first two symbols of $w'$:
  \begin{itemize}
  \item We begin by appending $w'[1] = L_{a_1}$ to the definition of
    the starting nonterminal of $G_{\rm cur} = (V_{\rm cur},
    \Sigma'', R_{\rm cur}, S_{\rm cur})$.  Let $G'_{\rm cur} =
    (V'_{\rm cur}, \Sigma'', R'_{\rm cur}, S'_{\rm cur})$ denote the
    resulting grammar. We have $V'_{\rm cur} = V_{\rm cur}$,
    $S'_{\rm cur} = S_{\rm cur}$.  Definitions of all nonterminals
    in $G'_{\rm cur}$ are the same as in $G_{\rm cur}$, except that
    $\rhsgen{G'_{\rm cur}}{S'_{\rm cur}} = \rhsgen{G_{\rm
    cur}}{S_{\rm cur}} \cdot L_{a_1}$.  By the inductive
    assumption, $G_{\rm cur}$ is isomorphic to $G_{k-1}$.
    Consequently, the last symbol of $\rhsgen{G_{\rm cur}}{S_{k-1}}$ is
    $\hash_{2(k-1)}$. Thus, since $\hash_{2(k-1)}$ occurs only once
    on the right-hand side of $G_{k-1}$, \algname{Sequitur} invokes
    no reductions after this step.
  \item We next process the symbol $w'[2] = \dol_{a_1}$.  Let
    $G''_{\rm cur} = (V''_{\rm cur}, \Sigma'', R''_{\rm cur}, S''_{\rm
    cur})$ denote the initial new grammar obtained by appending the
    symbol.  We have $V''_{\rm cur} = V'_{\rm cur} = V_{\rm cur}$ and
    $S''_{\rm cur} = S'_{\rm cur} = S_{\rm cur}$. We also have
    $\rhsgen{G''_{\rm cur}}{S''_{\rm cur}} = \rhsgen{G'_{\rm
    cur}}{S'_{\rm cur}} \cdot \dol_{a_1} = \rhsgen{G_{\rm
    cur}}{S_{\rm cur}} \cdot L_{a_1} \dol_{a_1}$. Definitions
    of all other nonterminals in $V''_{\rm cur}$ are the same as in
    $G_{\rm cur}$.  Recall now that $G_{\rm cur}$ is isomorphic with
    $G_{k-1}$.  Thus, letting $X = f^{-1}(N_{a_1})$, since the height
    of the parse-tree $\mathcal{T}_G(N_{a_1})$ is one, we have
    $\rhsgen{G''_{\rm cur}}{X} = \rhsgen{G_{k-1}}{N_{a_1}} = L_{a_1}
    \dol_{a_1} R_{a_1}$.  Thus, \algname{Sequitur} invokes the
    second reduction in the list, which creates a new nonterminal
    $Q_{a_1}$ with the definition $L_{a_1} \dol_{a_1}$ and
    replaces both occurrences of $L_{a_1} \dol_{a_1}$ on the
    right-hand side of $G''_{\rm cur}$ with $Q_{a_1}$. In other words,
    letting $G'''_{\rm cur} = (V'''_{\rm cur}, \Sigma'', R'''_{\rm
    cur}, S'''_{\rm cur})$ be the resulting grammar, we have
    $V'''_{\rm cur} = V''_{\rm cur} \cup \{Q_{a_1}\}$, $S'''_{\rm cur}
    = S''_{\rm cur}$, $\rhsgen{G'''_{\rm cur}}{Q_{a_1}} = L_{a_1}
    \dol_{a_1}$, $\rhsgen{G'''_{\rm cur}}{f^{-1}(N_{a_1})} =
    Q_{a_1} R_{a_1}$, and $\rhsgen{G'''_{\rm cur}}{S'''_{\rm cur}} =
    \rhsgen{G_{\rm cur}}{S_{\rm cur}} \cdot Q_{a_1}$.
  \item It remains to observe that $G'''_{\rm cur}$ is isomorphic with
    $G_{k-1,1}$ (the bijection is obtained by extending the bijection
    $f : V_{k-1} \cup \Sigma'' \rightarrow V_{\rm cur} \cup \Sigma''$
    so that $f(T_{a_1}) = Q_{a_1}$). This concludes the proof of the
    induction base case.
  \end{itemize}

  Let us now assume that the algorithm has processed the leftmost
  $2(\ell - 1)$ symbols of $w'$, where $\ell \geq 2$, and let $G_{\rm
  cur} = (V_{\rm cur}, \Sigma'', R_{\rm cur}, S_{\rm cur})$ be the
  resulting grammar. By the inductive assumption, $G_{\rm cur}$ is
  isomorphic to $G_{k-1,\ell-1}$. Let $g : V_{k-1,\ell-1} \cup
  \Sigma'' \rightarrow V_{\rm cur} \cup \Sigma''$ denote the
  corresponding bijection. Denote $(v_1, \dots, v_q) = {\rm
  seq}_{\ell}(A)$ and $(a_1, \dots, a_q) = {\rm ids}_{\ell}(A)$.
  Using this notation, we thus have $V_{\rm cur} = \{g(N_1), \dots,
  g(N_{k-1}), g(S_{k-1}), g(T_{a_1}), \dots, g(T_{a_q})\}$ and
  \begin{itemize}
  \item For $i \in \{a_1, \dots, a_q\}$, $\rhsgen{G_{\rm
    cur}}{g(N_{i})} = g(T_{i}) \cdot g(R_{i})$ and $\rhsgen{G_{\rm
    cur}}{g(T_{i})} = g(L_{i}) \cdot \dol_{i}$,
  \item For $i \in [1 \dd |V|] \setminus \{a_1, \dots, a_q\}$,
    $\rhsgen{G_{\rm cur}}{g(N_i)} = g(L_i) \cdot \dol_{i} \cdot
    g(R_i)$,
  \item $\rhsgen{G_{\rm cur}}{S_{\rm cur}} = (\bigodot_{i=1,\dots,k-1}
    g(N_i) \cdot \hash_{2i-1} \cdot g(N_i) \cdot \hash_{2i}) \cdot
    g(T_{a_1}) \cdot \ldots \cdot g(T_{a_q})$.
  \end{itemize}

  Note that by $\ell \geq 2$, we have $q \geq 1$, since unless the
  root-to-leaf path traverses to the leftmost leaf in the tree, it
  must have a node in which we turn right.  Denote the leftmost
  $\ell$th leaf in $\mathcal{T}_G(A)$ by $v$, and let $v'$ be the parent
  of $v$ in $\mathcal{T}_G(A)$.  We have $v' \in V = \{N_1, \dots,
  N_{|V|}\}$. Let thus $j \in [1 \dd |V|]$ be such that $s(v') =
  N_j$. We consider two cases:
  \begin{enumerate}[label=(\roman*)]
  \item First, assume that $v$ is the left child of $v'$. Since
    $\rhsgen{G_{k-1}}{N_{j}} = L_{j} R_{j}$, this implies that $s(v) =
    L_{j}$.  Thus, since the symbol $w'[2\ell - 1]$ is equal to the
    symbol of the $\ell$th leftmost leaf in $\mathcal{T}_G(A)$, we have
    $w'[2\ell - 1] = s(v) = L_{N_j}$. We then also have $w'[2\ell] =
    \rhsgen{G'}{N_j}[2] = \dol_{j}$. We now analyze the
    \algname{Sequitur} algorithm as it processes the symbols
    $w'[2\ell - 1] = L_{j}$ and $w'[2\ell] = \dol_{j}$.
    \begin{itemize}
    \item Let $G'_{\rm cur} = (V'_{\rm cur}, \Sigma'', R'_{\rm cur},
      S'_{\rm cur})$ denote the initial grammar right after appending
      $w'[2\ell - 1]$ to the definition of the start rule of $G_{\rm
      cur}$. We have $V'_{\rm cur} = V_{\rm cur}$ and $S'_{\rm cur} =
      S_{\rm cur}$. We also have $\rhsgen{G'_{\rm cur}}{S'_{\rm cur}}
      = \rhsgen{G_{\rm cur}}{S_{\rm cur}} \cdot L_{j}$.  The
      definitions of other variables are as in $G_{\rm cur}$. Recall
      now that by the inductive assumption, $G_{\rm cur}$ is
      isomorphic with $G_{k-1,\ell-1}$. Consequently, the last two
      symbols of $\rhsgen{G'_{\rm cur}}{S'_{\rm cur}}$ are $g(T_{a_q})
      \cdot L_{j}$. We claim that this pair does not occur anywhere
      else on the right-hand side of $G'_{\rm cur}$. To see this,
      observe that $T_{a_q}$ occurs only twice on the right-hand side
      of $G_{k-1,\ell-1}$: first as the last symbol of
      $\rhsgen{G_{k-1,\ell-1}}{S_{k-1}}$, and second in the
      definition of $N_{a_q}$ (recall that
      $\rhsgen{G_{k-1,\ell-1}}{N_{a_q}} = T_{a_q} \cdot R_{a_q}$).  Thus,
      $g(T_{a_q})$ occurs only twice on the right-hand side of
      $G'_{\rm cur}$, with the second occurrence in the definition of
      $g(N_{a_q})$ (since $\rhsgen{G'_{\rm cur}}{g(N_{a_q})) =
      g(T_{a_q}} \cdot g(R_{a_q})$). Consequently, to show that
      $g(T_{a_q}) L_j$ does not occur twice on the right-hand side of
      $G'_{\rm cur}$, it suffices to prove $g(R_{a_q}) \neq L_j$.  To
      show this, observe that $v_q$ is an ancestor of $v'$ since both
      are on the path from the root of $\mathcal{T}_G(A)$ to its
      leftmost $\ell$th leaf.  Since, however, $v$ is the left child
      of $v'$, the node $v'$ does not occur in ${\rm
      seq}_{\ell}(A)$. Consequently, $s(v') \neq N_{a_q}$, and hence
      $v' \neq v_q$.  By definition of ${\rm seq}_{\ell}(A)$, to reach
      $v'$ from $v_q$, we have to first descend to the right child of
      $v_q$, and then keep following the left child until we reach
      $v'$.  Note that this implies that the height of the right child
      of $v_q$ is equal to at least the height of $v'$. This implies
      $R_{a_q} \not\in \Sigma$. Consequently, we also have $g(R_{a_q})
      \not\in \Sigma$.  By $L_j \in \Sigma$, we thus obtain
      $g(R_{a_q}) \neq L_j$. We have thus proved that $g(T_{a_q}) L_j$
      occurs only once on the right-hand side of $G'_{\rm
      cur}$. Consequently, \algname{Sequitur} performs no reductions.
    \item Let $G''_{\rm cur} = (V''_{\rm cur}, \Sigma'', R''_{\rm
      cur}, S''_{\rm cur})$ denote the grammar immediately after
      appending $w'[2\ell] = \dol_{j}$ to the definition of
      $\rhsgen{G'_{\rm cur}}{S'_{\rm cur}}$. We have $V''_{\rm cur} =
      V'_{\rm cur}$, $S''_{\rm cur} = S'_{\rm cur}$, $\rhsgen{G''_{\rm
      cur}}{S''_{\rm cur}} = \rhsgen{G'_{\rm cur}}{S'_{\rm cur}} \cdot
      \dol_{j} = \rhsgen{G_{\rm cur}}{S_{\rm cur}} \cdot L_{j} \cdot
      \dol_{j}$.  Observe now that since $j$ does not occur in
      $\{a_1, \dots, a_q\}$ (since as noted above $v'$ does not occur
      in $\{v_1, \dots, v_q\}$ and $s(v') = N_j$), the substring $L_j
      \cdot \dol_{j}$ has another occurrence on the right-hand side
      of $G_{k-1,\ell-1}$, namely, in the definition of $N_j$ (recall
      that $\rhsgen{G_{k-1,\ell-1}}{N_j} = L_j \dol_{j} R_j$).  By
      $L_j, \dol_{j} \in \Sigma''$, the substring $L_j \cdot
      \dol_{j}$ also has another occurrence on the right-hand side
      of $G''_{\rm cur}$ in the definition of $g(N_j)$. The algorithm
      thus applies the second reduction in the list. Let $X$ denotes
      the newly created nonterminal. Letting $G'''_{\rm cur} =
      (V'''_{\rm cur}, \Sigma'', R'''_{\rm cur}, S'''_{\rm cur})$ be
      the resulting grammar, we have $V'''_{\rm cur} = V''_{\rm cur}
      \cup \{X\} = V_{\rm cur} \cup \{X\}$, $S'''_{\rm cur} = S''_{\rm
      cur}$, $\rhsgen{G'''_{\rm cur}}{X} = L_j \dol_{j}$,
      $\rhsgen{G'''_{\rm cur}}{g(N_j)} = X \cdot g(R_j)$,
      $\rhsgen{G'''_{\rm cur}}{S'''_{\rm cur}} = \rhsgen{G_{\rm
      cur}}{S_{\rm cur}} \cdot X$, and definitions for all other
      elements of $V'''_{\rm cur}$ are as in $G''_{\rm cur}$.  Since
      $X$ was a newly created nonterminal, the last two symbols of
      $\rhsgen{G'''_{\rm cur}}{S'''_{\rm cur}}$ must only occur once
      on the right-hand side of $G'''_{\rm cur}$. Note also that
      by replacing $L_j \cdot \dol_{j}$ with $X$, we removed
      some occurrences of $L_j$, but since $L_j \in \Sigma$, this
      does not cause any nonterminal to only have a single occurrence
      on the right-hand side of the grammar. Thus, the algorithm
      performs no further reductions.
    \item It remains to show that $G'''_{\rm cur}$ is isomorphic to
      $G_{k-1,\ell}$.  First, we observe that ${\rm ids}_{\ell + 1}(A)
      = (a_1, \dots, a_q, j)$.  This follows since the $\ell$th
      leftmost leaf in $\mathcal{T}_G(A)$ (i.e., $v$) is the left child
      of its parent $v'$. This implies that the path from the root to
      the ($\ell+1$)st leftmost leaf in $\mathcal{T}_G(A)$ first reaches
      $v'$ and then goes right. Thus, to obtain ${\rm
      seq}_{\ell+1}(A)$ we only need to append $s(v') = N_j$ to
      ${\rm seq}_{\ell}(A)$, and hence ${\rm ids}_{\ell+1}(A) = (a_1,
      \dots, a_q, j)$.  By definition $G_{k-1,\ell}$, we thus have
      $\rhsgen{G_{k-1,\ell}}{S_{k-1}} = \rhsgen{G_{k-1}}{S_{k-1}}
      \cdot T_{a_1} \cdot \ldots \cdot T_{a_q} \cdot T_{j} =
      \rhsgen{G_{k-1,\ell-1}}{S_{k-1}} \cdot T_j$.  This immediately
      implies that $G'''_{\rm cur}$ is isomorphic with $G_{k-1,\ell}$,
      as it suffices to take the bijection $g$ and extend it by
      setting $g(T_j) = X$.
    \end{itemize}
  \item Let us now assume that $v$ is a right child of $v'$, i.e.,
    $s(v) = R_j$. Note that this implies that $v_q = v'$, i.e., $a_q =
    j$.  Let $q' \in [1 \dd q]$ be the smallest integer such for for
    every $t \in [q' \dd q)$, $v_{t+1}$ is the right child of $v_t$ in
    $\mathcal{T}_G(A)$. Observe also that for every $t \in [q' \dd q)$
    we then have $R_{a_t} = N_{a_{t+1}}$. Note that also $R_{a_q} =
    R_j = s(v)$. By the above, we have $w'[2\ell - 1] = s(v) = R_j$.
    Let $\widehat{v}$ be the parent of $v_{q'}$ in $\mathcal{T}_G(A)$
    (the node $\widehat{v}$ is well-defined by $q' > 1$) and let
    $\widehat{a} \in [1 \dd |V|]$ be such that $s(\widehat{v}) =
    N_{\widehat{a}}$. Observe that if for some node $u$ of
    $\mathcal{T}_G(A)$ such that $s(u) = N_h$ (where $h \in [1 \dd
    |V|]$) it holds $\max {\rm int}(u') = p$, where $u'$ is the left
    child of $u$, then $w'[2p] = \dol_{h}$.  By applying this
    observation for $u = \widehat{v}$ and $u' = v_{q'}$, we obtain
    from $\max {\rm int}(u') = \max {\rm int}(v_{q'}) = \max {\rm
    int}(v_{q'+1}) = \dots = \max {\rm int}(v_q) = \max {\rm int}(v')
    = \max {\rm int}(v) = \ell$ (the equalities follow since they
    corresponding to the path in $\mathcal{T}_G(A)$ that always descend
    right) that $w'[2\ell] = \dol_{\widehat{a}}$. We now analyze
    the \algname{Sequitur} algorithm as it processes the symbols
    $w'[2\ell - 1] = R_{j}$ and $w'[2\ell] = \dol_{\widehat{a}}$.
    \begin{itemize}
    \item Let $G'_{\rm cur} = (V'_{\rm cur}, \Sigma'', R'_{\rm cur},
      S'_{\rm cur})$ denote the initial grammar after appending
      $w'[2\ell - 1] = R_{j}$ to the definition of the start rule of
      $G_{\rm cur}$. We have $V'_{\rm cur} = V_{\rm cur}$ and $S'_{\rm
      cur} = S_{\rm cur}$. We also have $\rhsgen{G'_{\rm
      cur}}{S'_{\rm cur}} = \rhsgen{G_{\rm cur}}{S_{\rm cur}}
      \cdot R_{j}$.  The definitions of other variables are as in
      $G_{\rm cur}$. Recall now that by the inductive assumption,
      $G_{\rm cur}$ is isomorphic with $G_{k-1,\ell-1}$. Consequently,
      the last $q$ symbols of $\rhsgen{G'_{\rm cur}}{S'_{\rm cur}}$
      are $g(T_{a_1}) \cdot \ldots \cdot g(T_{a_{q-1}}) \cdot
      g(T_{a_q}) \cdot R_j = g(T_{a_1}) \cdot \ldots \cdot
      g(T_{a_{q-1}}) \cdot g(T_{a_q}) \cdot R_{a_q} = g(T_{a_1}) \cdot
      \ldots \cdot g(T_{a_{q-1}}) \cdot g(T_{a_q}) \cdot g(R_{a_q})$.
      The last equality follows by $R_{a_q} \in \Sigma$, and the fact
      that for every $c \in \Sigma$, we have $g(c) = c$.  Recall now
      that $G_{\rm cur}$ is isomorphic with $G_{k-1,\ell-1}$.  Thus,
      $g(N_{a_q}) \in V_{\rm cur} = V'_{\rm cur}$ and $\rhsgen{G'_{\rm
      cur}}{g(N_{a_q})} = \rhsgen{G_{\rm cur}}{g(N_{a_q})} =
      g(T_{a_q}) \cdot g(R_{a_q})$, i.e., the length-2 suffix
      $g(T_{a_q}) \cdot g(R_{a_q})$ of $\rhsgen{G'_{\rm cur}}{S'_{\rm
      cur}}$ is equal to the definition of the nonterminal
      $g(N_{a_q})$ from $G'_{\rm cur}$.  Consequently,
      \algname{Sequitur} invokes the first reduction, replacing the
      length-2 suffix with $g(N_{a_q})$.  After this replacement, the
      right-hand side of the starting rule of the resulting grammar
      has $g(T_{a_1}) \cdot \ldots \cdot g(T_{a_{q-1}}) \cdot
      g(N_{a_q})$ as a suffix. Recall now that above we observed that
      for every $t \in [q' \dd q)$, it holds $R_{a_t} = N_{a_{t+1}}$.
      Consequently, if $q' < q$, we have $g(N_{a_q}) = g(R_{a_{q-1}})$
      and the suffix of the current starting rule can be rewritten as
      $g(T_{a_1}) \cdot \ldots \cdot g(T_{a_{q-1}}) \cdot
      g(R_{a_{q-1}})$.  Using the same argument as above, this means
      that \algname{Sequitur} will invoke the first reduction rule,
      replacing the length-2 suffix $g(T_{a_{q-1}}) \cdot
      g(R_{a_{q-1}})$ of the current start rule with
      $g(N_{a_{q-1}})$. Observe that this process will repeat exactly
      $q-q'+1$ times, reducing the initial length-$(q+1)$ suffix of
      the starting rule to the length-$q'$ suffix $g(T_{a_1}) \cdot
      \ldots \cdot g(T_{a_{q'-1}}) \cdot g(N_{a_{q'}})$. Observe now
      that if $q' > 1$, then the pair $g(T_{a_{q'-1}}) \cdot
      g(N_{a_{q'}})$ has only a single occurrence on the right
      hand-side of the current grammar. To see this, note that
      $g(T_{a_{q'-1}})$ has only one other occurrence in addition to
      the one in the start rule, and this occurrence is in the
      definition $\rhsgen{G'_{\rm cur}}{g(N_{a_{q'-1}})} =
      g(T_{a_{q'-1}}) \cdot g(R_{a_{q'-1}})$.  Thus, it suffices to
      prove that $g(N_{a_{q'}}) \neq g(R_{a_{q'-1}})$ which is
      equivalent to $N_{a_{q'}} \neq R_{a_{q'-1}}$. To show this, let
      $v''$ be the right child of $v_{q'-1}$. By definition, we have
      $s(v'') = R_{a_{q'-1}}$. Observe that $v''$ does not occur in
      $\{v_1, \dots, v_q\}$ since this would mean $v'' = v_{q'}$ which
      would contradict the definition of $q'$ (since then $v_{q'-1},
      v_{q'}, \dots, v_{q}$ forms the path where we always turn right
      and end at $v$).  Since, however, $v_{q'-1}$ occurs in $\{v_1,
      \dots, v_{q}\}$, the node $v''$ lies on the path from the root
      of $\mathcal{T}_G(A)$ to $v$. Consequently, $v''$ is an ancestor
      of $v_{q'}$. Since, however, to reach $v$ from $v''$ we have to
      first turn left, this implies that $v'' \neq v_{q'}$, i.e.,
      $v''$ is a proper ancestor of $v_{q'}$. This implies
      $R_{a_{q'-1}} = s(v'') \neq s(v_{q'}) = N_{a_{q'}}$. We have
      thus proved that if $q' > 1$, then the pair $g(T_{a_{q'-1}})
      \cdot g(N_{a_{q'}})$ has only a single occurrence on the right
      hand-side of the current grammar. Thus, \algname{Sequitur} will
      not invoke either of the first two reductions. Observe, however,
      that after performing $q-q'+1$ reductions according to the first
      rule, we have removed one occurrence of each variable in the set
      $\{g(T_{a_{q'}}), \dots, g(T_{a_q})\}$.  Since they all
      initially only had two occurrences on the right-hand side of
      $G'_{\rm cur}$, the \algname{Sequitur} now invokes the third
      reduction rule for all of them, replacing their only occurrence
      with their definition, removing the variables from the
      grammar. Consequently, for every $i \in \{a_{q'}, \dots,
      a_{q}\}$, the definition $\rhsgen{G'_{\rm cur}}{g(N_{i})} =
      g(T_{i}) \cdot g(R_{i})$ becomes $\rhsgen{G'_{\rm
      cur}}{g(N_{i})} = g(L_{i}) \cdot \dol_{i} \cdot
      g(R_{i})$. Observe that this introduces a new pair $\dol_{i}
      \cdot g(R_{i})$ on the right-hand side of the grammar. However,
      this pair does not repeat, since the grammar before the
      modification contains the only occurrence of $\dol_{i}$ on
      the right-hand side. Consequently, during or after the
      application of the third reduction rule, the algorithm will not
      apply either of the first two.  Recall now that $v_{q'}$ is the
      left child of its parent $\widehat{v}$, and we denoted
      $s(\widehat{v}) = N_{\widehat{a}}$.  Therefore, by
      $\rhsgen{G_{k-1}}{N_{\widehat{a}}} = L_{\widehat{a}}
      R_{\widehat{a}}$, we have $N_{a_{q'}} = s(v_{q'}) =
      L_{\widehat{a}}$. Consequently, $g(T_{a_1}) \cdot \ldots \cdot
      g(T_{a_{q'-1}}) \cdot g(N_{a_{q'}}) = g(T_{a_1}) \cdot \ldots
      \cdot g(T_{a_{q'-1}}) \cdot g(L_{\widehat{a}})$.  To sum up,
      letting $G''_{\rm cur} = (V''_{\rm cur}, \Sigma'', R''_{\rm
        cur}, S''_{\rm cur})$ be the grammar resulting from the above
      $2(q-q'+1)$ reductions, we have $V''_{\rm cur} = V'_{\rm cur}
      \setminus \{g(T_{a_{q'}}), \dots, g(T_{a_q})\} = V_{\rm cur}
      \setminus \{g(T_{a_{q'}}), \dots, g(T_{a_q})\} = \{g(N_1),
      \dots, g(N_{k-1}), g(S_{k-1}), g(T_{a_1}), \dots,
      g(T_{a_{q'-1}})\}$, $S''_{\rm cur} = S'_{\rm cur} = S_{\rm
      cur}$, and
      \begin{itemize}
      \item For $i \in \{a_1, \dots, a_{q'-1}\}$, $\rhsgen{G''_{\rm
        cur}}{g(N_{i})} = g(T_{i}) \cdot g(R_{i})$ and
        $\rhsgen{G''_{\rm cur}}{g(T_{i})} = g(L_{i}) \cdot
        \dol_{i}$,
      \item For $i \in [1 \dd |V|] \setminus \{a_1, \dots,
        a_{q'-1}\}$, $\rhsgen{G''_{\rm cur}}{g(N_i)} = g(L_i) \cdot
        \dol_{i} \cdot g(R_i)$,
      \item $\rhsgen{G''_{\rm cur}}{S''_{\rm cur}} =
        (\bigodot_{i=1,\dots,k-1} g(N_i) \cdot \hash_{2i-1} \cdot
        g(N_i) \cdot \hash_{2i}) \cdot g(T_{a_1}) \cdot \ldots \cdot
        g(T_{a_{q'-1}}) \cdot g(L_{\widehat{a}})$.
      \end{itemize}
    \item Let us now consider the grammar obtained from $G''_{\rm
      cur}$ by appending $w'[2\ell] = \dol_{\widehat{a}}$ to the
      definition of the start rule. By $\widehat{v} \not\in \{v_1,
      \dots, v_q\}$, we have $\widehat{a} \not\in \{a_1, \dots,
      a_{q}\}$. Thus, there exist two occurrences of
      $g(L_{\widehat{a}}) \cdot \dol_{\widehat{a}}$ on the
      right-hand side of the current grammar: at the end of the
      definition of the start rule, and in the definition
      $\rhsgen{G''_{\rm cur}}{g(N_{\widehat{a}})} = g(L_{\widehat{a}})
      \cdot \dol_{\widehat{a}} \cdot g(R_{\widehat{a}})$.
      Consequently, \algname{Sequitur} invokes the second reduction
      rule, creating a new nonterminal $X$ with the definition
      $g(L_{\widehat{a}}) \cdot \dol_{\widehat{a}}$, and replaces
      both occurrences of this definition with $X$.  Let $G'''_{\rm
      cur} = (V'''_{\rm cur}, \Sigma'', R'''_{\rm cur}, S'''_{\rm
      cur})$ be the resulting grammar. We have $V'''_{\rm cur} =
      V''_{\rm cur} \cup \{X\}$, $S'''_{\rm cur} = S''_{\rm cur}$.  As
      for the definitions, we have $\rhsgen{G'''_{\rm
      cur}}{N_{\widehat{a}}} = X \cdot g(R_{\widehat{a}})$,
      $\rhsgen{G'''_{\rm cur}}{X} = g(L_{\widehat{a}}) \cdot
      \dol_{\widehat{a}}$, and $\rhsgen{G'''_{\rm cur}}{S'''_{\rm
      cur}} = (\bigodot_{i=1,\dots,k-1} g(N_i) \cdot \hash_{2i-1}
      \cdot g(N_i) \cdot \hash_{2i}) \cdot g(T_{a_1}) \cdot \ldots
      \cdot g(T_{a_{q'-1}}) \cdot X$.  The remaining definitions are
      as in $G''_{\rm cur}$. Since $X$ was a newly created
      nonterminal, the last two symbols of $\rhsgen{G'''_{\rm
      cur}}{S'''_{\rm cur}}$ occur only once on the right-hand
      side of $G'''_{\rm cur}$, and hence the algorithm does not
      invoke any of the first two reduction rules.  Note also that by
      replacing $g(L_{\widehat{a}}) \cdot \dol_{\widehat{a}}$ with
      $X$, we removed some occurrences of $g(L_{\widehat{a}})$, but
      since $g(L_{\widehat{a}}) \in \{g(N_1), \dots, g(N_{k-1})\}$, it
      follows that $g(L_{\widehat{a}})$ has at least two remaining
      occurrences in the prefix $\bigodot_{i=1,\dots,k-1} g(N_i) \cdot
      \hash_{2i-1} \cdot g(N_i) \cdot \hash_{2i}$ of
      $\rhsgen{G'''_{\rm cur}}{S'''_{\rm cur}}$. Thus,
      \algname{Sequitur} does not invoke the third reduction rule.
    \item It remains to show that $G'''_{\rm cur}$ is isomorphic to
      $G_{k-1,\ell}$.  Let $v_{\min}$ denote the leftmost leaf in the
      subtree rooted in the right child of $\widehat{v}$. Since $v$ is
      the rightmost leaf in the subtree rooted in the left child of
      $\widehat{v}$ (i.e., $v_{q'}$), we thus obtain that $v_{\min}$
      is the ($\ell+1$)st leftmost leaf in $\mathcal{T}_G(A)$.  This
      implies that ${\rm seq}_{\ell+1}(A)$ contains vertices $\{v_1,
      \dots, v_{q'}\}$ and then the vertex $\widehat{v}$ as the last
      vertex where the path to $v_{\min}$ turns right, i.e., ${\rm
      seq}_{\ell+1}(A) = (v_1, \dots, v_{q'-1}, \widehat{v})$.
      Consequently, ${\rm ids}_{\ell+1}(A) = (a_1, \dots, a_{q'-1},
      \widehat{a})$ and $\rhsgen{G_{k-1,\ell}}{S_{k-1}} =
      (\bigodot_{i=1,\dots,k-1} N_i \cdot \hash_{2i-1} \cdot N_i \cdot
      \hash_{2i}) \cdot T_{a_1} \cdot \ldots \cdot T_{a_{q'-1}} \cdot
      T_{\widehat{a}}$.  It immediately follows that by extending $g$
      so that $g(T_{\widehat{a}}) = X$, we obtain a bijection from
      $G_{k-1,\ell}$ to $G'''_{\rm cur}$. By
      $\rhsgen{G_{k-1,\ell}}{T_{\widehat{a}}} = L_{\widehat{a}} \cdot
      \dol_{\widehat{a}}$ and $\rhsgen{G'''_{\rm
      cur}}{g(T_{\widehat{a}})} = \rhsgen{G'''_{\rm cur}}{X} =
      g(L_{\widehat{a}}) \cdot g(\dol_{\widehat{a}})$ we thus
      obtain that $G'''_{\rm cur}$ is isomorphic to $G_{k-1,\ell}$.
      Thus concludes the proof of the induction step.
    \end{itemize}
  \end{enumerate}
  We have thus proved that for every $\ell \in [1 \dd \ell_A)$, the
  grammar computed by \algname{Sequitur} after processing the leftmost
  $2\ell$ symbols of $w'$ is isomorphic with $G_{k-1,\ell}$. Let us
  denote the grammar computed after $2(\ell_A - 1)$ steps by $G^{a} =
  (V^a, \Sigma'', R^a, S^a)$. By the above, $G^a$ is isomorphic with
  $G_{k-1,\ell_A-1}$.

  \vspace{2ex}
  \noindent
  \emph{Phase 2}
  The second phase is the processing of the next two symbols of $w'$,
  i.e., $w'[2\ell_A - 1]$ and $w'[2\ell_A]$.  Let $(v_1, \dots, v_q) =
  {\rm seq}_{\ell_A}(A)$ and $(a_1, \dots, a_q) = {\rm
  ids}_{\ell_A}(A)$.  By the analysis of phase 1, we have
  $g(S_{k-1}) = S^a$, $V^a = \{g(N_1), \dots, g(N_{k-1}), g(S_{k-1}),
  g(T_{a_1}), \dots, g(T_{a_q})\}$ and $\rhsgen{G^a}{g(S_{k-1})} =
  (\bigodot_{i=1,\dots,k-1} g(N_i) \cdot \hash_{2i-1} \cdot g(N_i)
  \cdot \hash_{2i}) \cdot g(T_{a_1}) \cdot \ldots \cdot g(T_{a_q})$,
  where $g$ is the bijection establishing isomorphism of $G^a$ and
  $G_{k-1,\ell_A-1}$. Note that for every $i \in [1 \dd q)$, $v_{i+1}$
  is the right child of $v_i$.  Observe also that by $2\ell_A - 1 =
  |\expgen{G'}{A}|$, $w[2\ell_A - 1]$ is the symbol corresponding to
  the rightmost leaf in $\mathcal{T}_{G}(A)$.  Hence, by the above
  observation, $w'[2\ell_A - 1] = R_{a_q}$.  Consequently, by the same
  argument as in the analysis of phase 1, appending $R_{a_q}$ to
  $\rhsgen{G^a}{S^a}$ causes a chain of $q$ reductions, which replace
  the suffix $g(T_{a_1}) \cdot \ldots \cdot g(T_{a_q})$ of
  $\rhsgen{G^a}{S^a}$ by a single symbol $g(N_{a_1})$, while also
  removing variables $\{g(T_{a_1}), \dots, g(T_{a_q})\}$ from
  $V^a$. Observe now that $v_1$ is the root of
  $\mathcal{T}_{G}(A)$. Thus, $N_{a_1} = s(v_1) = A$, and hence
  $g(N_{a_1}) = g(A)$. After these $q$ reductions, the algorithm does
  not perform any more reductions, since the symbol preceding $g(A)$
  is $\hash_{2k-2}$, which is unique in $w$. In the next step, we
  append $w'[2\ell_A] = \dol_{k}$ to the definition of the current
  start rule. Since this is the leftmost occurrence of this symbol in
  $w$, we do not perform any reductions. Therefore, letting $G^b =
  (V^b, \Sigma'', R^b, S^b)$ be the grammar resulting from the
  operations in the second phase, we have $S^b = S^a$, $V^b =
  \{g(N_1), \dots, g(N_{k-1}), g(S_{k-1})\}$ and
  $\rhsgen{G^b}{g(S_{k-1})} = (\bigodot_{i=1,\dots,k-1} g(N_i) \cdot
  \hash_{2i-1} \cdot g(N_i) \cdot \hash_{2i}) \cdot g(A) \cdot
  \dol_{k}$.

  \vspace{2ex}
  \noindent
  \emph{Phases 3-4}
  We define phase 3 as the processing of the next $2(\ell_B-1)$
  symbols $w'$, where $\ell_B = |\expgen{G}{B}|$, and phase 4 as the
  processing of the following two symbols of $w'$. The analysis is
  nearly identical as for phases 1-2, except the last processed symbol
  is $w'[2\ell_A + 2\ell_B] = \hash_{2k-1}$.  Letting $G^d = (V^d,
  \Sigma'', R^d, S^d)$ be the grammar resulting from operations in
  phase 3 and 4, we have $S^d = S^a$, $V^d = \{g(N_1), \dots,
  g(N_{k-1}), g(S_{k-1})\}$ and $\rhsgen{G^d}{g(S_{k-1})} =
  (\bigodot_{i=1,\dots,k-1} g(N_i) \cdot \hash_{2i-1} \cdot g(N_i)
  \cdot \hash_{2i}) \cdot g(A) \cdot \dol_{k} \cdot g(B) \cdot
  \hash_{2k-1}$.

  \vspace{2ex}
  \noindent
  \emph{Phases 5-6}
  We define phase 5 as the processing of the next $2(\ell_A - 1)$
  symbols of $w'$, and phase 6 as the processing of the next two
  symbols of $w'$. By nearly identical analysis as for phases 1-2,
  this results in appending symbols $g(A) \cdot \dol_{k}$ to the
  definition of the start rule of $G^d$. As a last step, however, we
  eliminate the repeating pair $g(A) \cdot \dol_{k}$ by first
  creating a new nonterminal $T$, and replacing both occurrences of
  $g(A) \cdot \dol_{k}$ with $T$.  Letting $G^f = (V^f, \Sigma'',
  R^f, S^f)$ be the resulting grammar, we have $S^f = S^a$, $V_f =
  \{g(N_1), \dots, g(N_{k-1}), g(S_{k-1}), T\}$, $\rhsgen{G^f}{T} =
  g(A) \cdot \dol_{k}$, and $\rhsgen{G^f}{g(S_{k-1})} =
  (\bigodot_{i=1,\dots,k-1} g(N_i) \cdot \hash_{2i-1} \cdot g(N_i)
  \cdot \hash_{2i}) \cdot T \cdot g(B) \cdot \hash_{2k-1} \cdot T$.

  \vspace{2ex}
  \noindent
  \emph{Phases 7-8}
  We define phase 7 as the processing of the next $2(\ell_B - 1)$
  symbols of $w'$, and phase 8 as the processing of the next two
  symbols of $w'$. By the same analysis as for phase 3, this first
  results in appending the symbol $g(B)$ to the definition of the
  staring rule. This holds since when phase 7 starts, the last symbol
  of the start rule is $T$, and its only other occurrence is followed
  by $g(B)$. Since the expansion length of all intermediate
  nonterminals created during phase 7 is shorter than
  $|\expgen{G'}{B}|$, no repetitions involving $g(B)$ will be
  discovered. Only after the second occurrence of $g(B)$ is
  discovered, \algname{Sequitur} eliminates the repetition of $T \cdot
  g(B)$ by introducing a new nonterminal $T'$ and replacing both
  occurrences of $T \cdot g(B)$ with $T'$. After that, $T$ has only a
  single occurrence on the right-hand side, and hence its is deleted
  and both its occurrences are replaced with its definition, i.e.,
  $g(A) \cdot \dol_{k}$.  Finally, we append $\hash_{2k}$ to the
  definition of the start rule, which does not invoke any
  reductions. Letting $G^h = (V^h, \Sigma'', R^h, S^h)$ be the final
  grammar, we have $S^h = S^a$, $V_h = \{g(N_1), \dots, g(N_{k-1}),
  g(S_{k-1}), T'\}$, $\rhsgen{G^h}{T'} = g(A) \cdot \dol_{k} \cdot
  g(B)$, and $\rhsgen{G^h}{g(G_{k-1})} = (\bigodot_{i=1,\dots,k-1}
  g(N_i) \cdot \hash_{2i-1} \cdot g(N_i) \cdot \hash_{2i}) \cdot T'
  \cdot \hash_{2k-1} \cdot T' \cdot \hash_{2k}$. Consequently,
  extending the bijection by setting $g(N_{k}) = T'$ establishes that
  $G^h$ is isomorphic with $G_{k}$. This concludes the proof of the
  induction step.

  \vspace{2ex}
  \noindent
  \emph{Summary}
  We have thus proved that for every $k \in [0 \dd |V|]$, the grammar
  produced by \algname{Sequitur} after $s_k$ steps is isomorphic with
  $G_k = (V_k, \Sigma'', R_k, S_k)$.  In articular, the final grammar
  is isomorphic to $G_{|V|}$. By the same analysis as at the end of
  the proof of \cref{lm:sequential}, this implies that the final
  grammar computed by \algname{Sequitur} has size $7|V| =
  \tfrac{7}{2}|G|$.
\end{proof}

\begin{theorem}\label{th:sequitur}
  For any $u \in \Sigma^{*}$, let $\algname{Sequitur}(u)$ denote the
  output of the \algname{Sequitur} algorithm on $u$. In the cell-probe
  model, every static data structure that for $u \in \Sigma^n$ uses
  $\bigO(|\algname{Sequitur}(u)| \log^{c} n)$ space $($where $c =
  \bigO(1))$, requires $\Omega(\log n / \log \log n)$ time to answer
  random access queries on $u$.
\end{theorem}
\begin{proof}
  The proof proceeds analogously as in \cref{th:global}, except we
  observe that the grammar $G_w = \algname{Sequitur}(w)$ satisfies
  $|G_w| = \tfrac{7}{2}|G_{\Pts}|$ by \cref{lm:sequitur}.
\end{proof}

\subsection{Analysis of \algname{LZD}}\label{sec:lzd}

\begin{definition}\label{def:beta}
  Let $G = (V, \Sigma, R, S)$ be an admissible SLG. Assume $\Sigma
  \cap \{\dol_i : i \in [1 \dd 2|V|]\} = \emptyset$ and let $\Sigma' =
  \Sigma \cup \{\dol_i : i \in [1 \dd 2|V|]\}$. By $\beta(G)$ we
  denote the subset of $\Sigma'^{*}$ such that for every $w \in
  \Sigma'^{*}$, $w \in \beta(G)$ holds if and only if there exists a
  sequence $(N_i)_{i \in [1 \dd |V|]}$ such that:
  \begin{itemize}[itemsep=0pt,parsep=0pt]
    \item $\{N_i : i \in [1 \dd |V|]\} = V$,
    \item $|\expgen{G}{N_i}| \leq |\expgen{G}{N_{i+1}}|$ holds for $i
      \in [1 \dd |V|)$,
    \item $w = \bigodot_{i=1,\dots,|V|} \expgen{G'}{N_{i,0}}
      \expgen{G'}{N_{i,0}}$,
  \end{itemize}
  where $G' = (V', \Sigma', R', S')$ is defined so that:
  \begin{itemize}[itemsep=0pt,parsep=0pt]
  \item $V' = \{S'\} \cup \bigcup_{i \in [1 \dd |V|]} \{N_{i,0},
    N_{i,1}, N_{i,2}\}$ is a set of $|V'| = 3|V| + 1$ variables,
  \item For every $i \in [1 \dd |V|]$,
    \begin{align*}
      \rhsgen{G'}{N_{i,1}} &=
        \begin{cases}
          N_{j,0} \dol_{2i-1}
            & \text{if }A = N_j\text{ for }j \in [1 \dd |V|],\\
          A \dol_{2i-1}
            & \text{otherwise},\\
        \end{cases}\\
      \rhsgen{G'}{N_{i,2}} &=
        \begin{cases}
          N_{k,0} \dol_{2i}
           & \text{if }B = N_k\text{ for }k \in [1 \dd |V|],\\
          B \dol_{2i}
           & \text{otherwise},\\
        \end{cases}\\[3ex]
      \rhsgen{G'}{N_{i,0}} &= N_{i,1} N_{i,2},
    \end{align*}
    where $A, B \in V \cup \Sigma$ are such that $\rhsgen{G}{N_i} =
    AB$,
  \item $\rhsgen{G'}{S'} = \bigodot_{i=1,\dots,|V|} N_{i,1} N_{i,2}
    N_{i,0}$.
  \end{itemize}
\end{definition}

\begin{lemma}\label{lm:beta-length}
  Let $G = (V, \Sigma, R, S)$ be an admissible SLG and $w \in
  \beta(G)$.  Let $(N_i)_{i \in [1 \dd |V|]}$ be the sequence and $G'
  = (V', \Sigma', R', S')$ be the SLG corresponding to $w$ in
  \cref{def:beta}. Then:
  \begin{enumerate}[itemsep=0pt,parsep=0pt]
  \item\label{lm:beta-length-it-1}
    For every $i \in [1 \dd |V|]$, letting $A, B \in V \cup \Sigma$ be
    such that $\rhsgen{G}{N_i} = AB$, it holds:
    \begin{itemize}[itemsep=0pt,parsep=0pt]
    \item $|\expgen{G'}{N_{i,1}}| = 3|\expgen{G}{A}| - 1$,
    \item $|\expgen{G'}{N_{i,2}}| = 3|\expgen{G}{B}| - 1$,
    \item $|\expgen{G'}{N_{i,0}}| = 3|\expgen{G}{N_i}| - 2$,
    \end{itemize}
  \item\label{lm:beta-length-it-2}
    $|w| = 6\sum_{X \in V} |\expgen{G}{X}| - 4|V|$.
  \end{enumerate}
\end{lemma}
\begin{proof}
  1. We proceed by induction on $|\expgen{G}{N_i}|$.  Let us thus
  assume $|\expgen{G}{N_i}| = 2$. Note that then $A, B \in \Sigma$.
  By \cref{def:beta}, we then have
  \begin{itemize}[itemsep=0pt,parsep=0pt]
    \item $\expgen{G'}{N_{i,1}} = a\dol_{2i-1}$,
    \item $\expgen{G'}{N_{i,2}} = b\dol_{2i}$,
    \item $\expgen{G'}{N_{i,0}} = a\dol_{2i-1}b\dol_{2i}$.
  \end{itemize}
  Thus, the claim holds. Let us now assume $|\expgen{G}{N_i}| >
  2$. Consider two cases. If $A \in \Sigma$, then
  $|\expgen{G'}{N_{i,1}}| = 2 = 3|\expgen{G}{A}| - 1$ follows as in
  the induction base. Let us thus assume $A \in V$. By
  $\rhsgen{G}{N_i} = AB$, it follows that $|\expgen{G}{A}| <
  |\expgen{G}{N_i}|$. This implies that $A = N_j$ for some $j \in [1
  \dd i)$.  By the inductive assumption $|\expgen{G'}{N_{j,0}}| =
  3|\expgen{G}{N_j}| - 2 = 3|\expgen{G}{A}| - 2$. On the other hand,
  by \cref{def:beta}, $\rhsgen{G'}{N_{i,1}} =
  N_{j,0}\dol_{2i-1}$. Thus, $|\expgen{G'}{N_{i,1}}| =
  |\expgen{G'}{N_{j,0}}| + 1 = 3|\expgen{G}{A}| - 1$.  The proof of
  $|\expgen{G'}{N_{i,2}}| = 3|\expgen{G}{B}| - 1$ is analogous. By
  $\rhsgen{G'}{N_{i,0}} = N_{i,1}N_{i,2}$, we thus obtain
  \begin{align*}
    |\expgen{G'}{N_{i,0}}|
      &= |\expgen{G'}{N_{i,1}}| + |\expgen{G'}{N_{i,2}}|\\
      &= 3|\expgen{G}{A}| - 1 + |\expgen{G}{B}| - 1\\
      &= 3|\expgen{G}{A}\expgen{G}{B}| - 2\\
      &= 3|\expgen{G}{N_i}| - 2.
  \end{align*}

  2. By \cref{lm:beta-length}\eqref{lm:beta-length-it-2},
  \begin{align*}
    |w|
      &= |\textstyle\bigodot_{i=1,\dots,|V|}
        \expgen{G'}{N_{i,0}}\expgen{G'}{N_{i,0}}\\
      &= 2\textstyle\sum_{i \in [1 \dd |V|]}(3|\expgen{G}{N_i}| - 2)\\
      &= 6\textstyle\sum_{X \in V}|\expgen{G}{X}| - 4|V|.
      \qedhere
  \end{align*}
\end{proof}

\begin{lemma}\label{lm:lzd}
  Let $G$ be an admissible SLG. For every string $w \in \beta(G)$,
  \algname{LZD} outputs a grammar of size $\tfrac{9}{2}|G|$.
\end{lemma}
\begin{proof}

  Denote $G = (V, \Sigma, R, S)$. Let $\Sigma'$ be as in
  \cref{def:beta}.  Let also $(N_i)_{i \in [1 \dd |V|]}$ and $G' =
  (V', \Sigma', R', S')$ be the sequence and the SLG corresponding to
  $w$ in \cref{def:beta}. Recall that $V' = \{S'\} \cup \bigcup_{i \in
  [1 \dd |V|]} \{N_{i,0}, N_{i,1}, N_{i,2}\}$.  We prove by
  induction that for every $i \in [1 \dd |V|]$, after $3i$ steps,
  \algname{LZD} processed the prefix $\bigodot_{j=1,\dots,i}
  \expgen{G'}{N_{j,0}} \expgen{G'}{N_{j,0}}$ of $w$ and the produced
  grammar is isomorphic to $G'_i$, where $G'_i$ is defined as $G'$
  restricted to the first $3i + 1$ variables, i.e., $G'_i = (V'_i,
  \Sigma', R'_i, S'_i)$, where:
  \begin{itemize}
  \item $V'_i = \{S'_i\} \cup \bigcup_{i \in [1 \dd i]} \{N_{i,0},
    N_{i,1}, N_{i,2}\}$,
  \item For every $j \in [1 \dd i]$ and $b \in \{0, 1, 2\}$, it holds
    $\rhsgen{G'_i}{N_{j,b}} = \rhsgen{G'}{N_{j,b}}$,
  \item $\rhsgen{G'_i}{S'_i} = \bigodot_{j=1,\dots,i} N_{j,1} N_{j,2}
    N_{j,0}$.
  \end{itemize}

  To prove the base case of $i = 1$, observe that $\rhsgen{G}{N_1} =
  ab$, where $a, b \in \Sigma$.  Thus, $a \dol_1 b \dol_2 a \dol_1 b
  \dol_2$ is a prefix of $w$. When running \algname{LZD} on $w$, the
  first step creates $f_1 = a \dol_1$, the second step introduces $f_2
  = a \dol_2$, and in the third step, we obtain $f_3 = f_1 f_2$, since
  $f_1 = a \dol_1$ is the longest prefix of the remaining string equal
  to one of the earlier phrases. Similarly, $b \dol_2$ is the longest
  prefix after that with the corresponding phrase.  This parsing
  corresponds to an SLG $G_{\rm cur} = (V_{\rm cur}, \Sigma', R_{\rm
  cur}, S_{\rm cur})$ such that $V_{\rm cur} = \{S_{\rm cur},
  M_{1,0}, M_{1,1}, M_{1,2}\}$, and it holds $\rhsgen{G_{\rm
  cur}}{M_{1,1}} = a \dol_1$, $\rhsgen{G_{\rm cur}}{M_{1,2}} = b
  \dol_2$, $\rhsgen{G_{\rm cur}}{M_{1,0}} = M_{1,1} M_{1,2}$, and
  $\rhsgen{G_{\rm cur}}{S_{\rm cur}} = M_{1,1} M_{1,2}
  M_{1,0}$. Recall now that $G'_1 = (V'_1, \Sigma', R'_1, S'_1)$,
  where $V'_1 = \{S'_1, N_{1,0}, N_{1,1}, N_{1,2}\}$,
  $\rhsgen{G'_1}{N_{1,1}} = \rhsgen{G'}{N_{1,1}} = a \dol_1$,
  $\rhsgen{G'_1}{N_{1,2}} = \rhsgen{G'}{N_{1,2}} = b \dol_2$,
  $\rhsgen{G'_1}{N_{1,0}} = \rhsgen{G'}{N_{1,0}} = N_{1,1} N_{1,2}$,
  and $\rhsgen{G'_1}{S'_1} = N_{1,1} N_{1,2} N_{1,0}$. Thus, $G_{\rm
  cur}$ is clearly isomorphic with $G'_1$, concluding the proof of
  the induction base.

  We now prove the induction step. Assume $i > 1$. Let $G_{\rm cur} =
  (V_{\rm cur}, \Sigma', R_{\rm cur}, S_{\rm cur})$ be the grammar
  corresponding to the parsing computed by \algname{LZD} after the
  first $3(i-1)$ steps.  By the inductive assumption, \algname{LZD}
  has processed the prefix $w' := \bigodot_{j=1,\dots,i-1}
  \expgen{G'}{N_{j,0}} \expgen{G'}{N_{j,0}}$ of $w$, resulting in the
  parsing $f_1 \cdots f_{3(i-1)}$, and $G_{\rm cur}$ is isomorphic to
  $G'_{i-1}$.  Let $g : V'_{i-1} \cup \Sigma' \rightarrow
  V_{\rm cur} \cup \Sigma'$ be the corresponding bijection (see
  \cref{sec:prelim}). Let $w''$ be such that $w' w'' = w$.
  Then, $\expgen{G'}{N_{i,0}}
  \expgen{G'}{N_{i,0}} = \expgen{G'}{N_{i,1}} \expgen{G'}{N_{i,2}}
  \expgen{G'}{N_{i,1}} \expgen{G'}{N_{i,2}}$ is a prefix of $w''$.
  First, observe that
  $\dol_{2i-1}$ and $\dol_{2i}$ do not occur in $w'$, since they do
  not appear in $\rhsgen{G'}{N_{j,b}}$ for $j \in [1 \dd i)$ and $b
  \in \{0, 1, 2\}$.  Let $A, B \in V \cup \Sigma$ be such that
  $\rhsgen{G}{N_i} = AB$.  Consider two cases:
  \begin{itemize}
  \item First, assume that there exists $j \in [1 \dd |V|]$ such that
    $A = N_j$. Then, $\rhsgen{G'}{N_{i,1}} = N_{j,0} \dol_{2i-1}$, and
    hence $\expgen{G'}{N_{i,1}} = \expgen{G'}{N_{j,0}} \dol_{2i-1}$ is
    a prefix of $w''$.  Note that $\rhsgen{G}{N_i} = AB$ implies
    $|\expgen{G}{A}| < |\expgen{G}{N_i}|$. Thus, $j < i$. Recall now
    that $\expgen{G'_{i-1}}{N_{j,0}} = \expgen{G'}{N_{j,0}}$. Thus,
    letting $X = g(N_{j,0})$, it holds $\expgen{G_{\rm cur}}{X} =
    \expgen{G'}{N_{j,0}}$. Since nonterminals $V_{\rm cur} \setminus
    \{S_{\rm cur}\}$ of $G_{\rm cur}$ correspond to the phrases in the
    current parsing $f_1 \cdots f_{3(i-1)}$, we thus obtain that there
    exists $j' \in [1 \dd 3(i-1)]$ such that $f_{j'} = \expgen{G_{\rm
    cur}}{X} = \expgen{G'}{N_{j,0}}$. Since, as noted above,
    $\dol_{2i-1}$ does not occur in $w'$, by definition of
    \algname{LZD} (see \cref{sec:algs-nonglobal}) it follows that the
    next phrase is $f_{3i - 2} = f_{j'} \dol_{2i-1} = \expgen{G_{\rm
    cur}}{X} \dol_{2i-1} = \expgen{G'}{N_{j,0}} \dol_{2i-1} =
    \expgen{G'}{N_{i,1}}$.  This corresponds to adding a nonterminal
    $X_1$ into the current grammar with the definition $g(N_{j,0})
    \dol_{2i-1} = g(\rhsgen{G'}{N_{i,1}}[1]) \cdot
    g(\rhsgen{G'}{N_{i,1}}[2]) = g(\rhsgen{G'_i}{N_{i,1}}[1]) \cdot
    g(\rhsgen{G'_i}{N_{i,1}}[2])$.
  \item Let us now assume that $A \in \Sigma$. Then,
    $\rhsgen{G'}{N_{i,1}} = A \dol_{2i-1}$, and hence
    $\expgen{G'}{N_{i,1}} = A \dol_{2i-1}$ is a prefix of
    $w''$. Since, as noted above, $\dol_{2i-1}$ does not occur in
    $w'$, by definition of \algname{LZD} it follows that the next
    phrase is $f_{3i-2} = A \dol_{2i-1} = \expgen{G'}{N_{i,1}}$.  This
    corresponds to adding a nonterminal $X_1$ into the current grammar
    with the definition $A \dol_{2i-1} = g(\rhsgen{G'}{N_{i,1}}[1])
    \cdot g(\rhsgen{G'}{N_{i,1}}[2]) = g(\rhsgen{G'_i}{N_{i,1}}[1])
    \cdot g(\rhsgen{G'_i}{N_{i,1}}[2])$.
  \end{itemize}
  In both cases, we obtain that $f_{3i-2} = \expgen{G'}{N_{i,1}}$, and
  adding this factor corresponds to adding a nonterminal $X_1$ with
  the definition $g(\rhsgen{G'_i}{N_{i,1}}[1]) \cdot
  g(\rhsgen{G'_i}{N_{i,1}}[2])$ to the current grammar. Analogously,
  since $\dol_{2i}$ does not occur in $w'$ or in
  $\expgen{G'}{N_{i,1}}$, it holds $f_{3i-1} = \expgen{G'}{N_{i,2}}$,
  and adding this factor corresponds to adding a nonterminal $X_2$
  with the definition $g(\rhsgen{G'_i}{N_{i,2}}[1]) \cdot
  g(\rhsgen{G'_i}{N_{i,2}}[2])$ to the current grammar.  The remaining
  unprocessed suffix of $w$ starts with
  $\expgen{G'}{N_{i,1}}\expgen{G'}{N_{i,2}}$. Since $\dol_{2i-1}$
  (resp.\ $\dol_{2i}$) is the last symbol in $\expgen{G'}{N_{i,1}}$
  (resp.\ $\expgen{G'}{N_{i,2}}$), and the current grammar has exactly
  one nonterminal containing $\dol_{2i-1}$ (resp.\ $\dol_{2i}$), i.e.,
  $X_1$ (resp.\ $X_2$), whose expansion is $\expgen{G'}{N_{i,1}} =
  f_{3i-2}$ (resp.\ $\expgen{G'}{N_{i,2}} = f_{3i-1}$), it follows
  that $f_{3i} = f_{3i-2}f_{3i-1}$. Adding this factor corresponds to
  adding a nonterminal $X_0$ with the definition $X_1 X_2$ to the
  current grammar. We have thus proved that after $3i$ steps, the
  algorithm processed the prefix $\bigodot_{j=1,\dots,i}
  \expgen{G'}{N_{j,0}}\expgen{G'}{N_{j,0}}$ of $w$, i.e., the first
  part of the induction claim. To show the second part, let $G_{\rm
  new} = (V_{\rm new}, \Sigma', R_{\rm new}, S_{\rm new})$ where
  $V_{\rm new} = V_{\rm cur} \cup \{X_0, X_1, X_2\}$, $S_{\rm new} =
  S_{\rm cur}$, for every $X \in V_{\rm cur}$, $\rhsgen{G_{\rm
  new}}{X} = \rhsgen{G_{\rm cur}}{X}$, and the definitions of
  $X_0$, $X_1$, and $X_2$ in $G_{\rm new}$ are as above.  Letting $g'
  : V_{\rm new} \cup \Sigma' \rightarrow G'_i \cup \Sigma'$ be defined
  in the same way as $g$ on $V_{\rm cur} \cup \Sigma'$, and $g'(X_0) =
  N_{i,0}$, $g'(X_1) = N_{i,1}$, and $g'(X_2) = N_{i,2}$, it
  immediately follows by the above that $g'$ is a bijection
  establishing the isomorphism of $G_{\rm new}$ and $G'_i$. This
  concludes the proof of the inductive step.

  By the above, the final grammar $G_{\rm final}$ computed by
  \algname{LZD} on $w$ is isomorphic to $G'_{|V|}$. As noted in
  \cref{sec:prelim}, this implies that $|G_{\rm final}| = |G'_{|V|}| =
  \sum_{i=1}^{|V|} \sum_{b \in \{0,1,2\}} |\rhsgen{G'_{|V|}}{N_{i,b}}|
  + |\rhsgen{G'_{|V|}}{S'_{|V|}}| = 6|V| + 3|V| = 9|V|$.  On the other
  hand, since $G$ is admissible, we have $|G| = 2|V|$. Consequently,
  $|G_{\rm final}| = 9|V| = \tfrac{9}{2}|G|$.
\end{proof}

\begin{theorem}\label{th:lzd}
  For any $u \in \Sigma^{*}$, let $\algname{LZD}(u)$ denote the output
  of the \algname{LZD} algorithm on $u$. In the cell-probe model,
  every static data structure that for $u \in \Sigma^n$ uses
  $\bigO(|\algname{LZD}(u)| \log^{c} n)$ space $($where $c =
  \bigO(1))$, requires $\Omega(\log n / \log \log n)$ time to answer
  random access queries on $u$.
\end{theorem}
\begin{proof}
  Suppose that there exists a structure $D$ that for any $u \in
  \Sigma^n$ uses $\bigO(|\algname{LZD}(u)| \log^c n)$ space (where $c
  = \bigO(1)$) and answers random access queries on $u$ in $o(\log n /
  \log \log n)$ time.  Let $\Pts \subseteq [1 \dd m]^2$ be any set of
  $|\Pts| = m$ points on an $m \times m$ grid. Assume for simplicity
  that $m = 2^k$ (otherwise, we pad $\Pts$ as in the proof of
  \cref{th:global}). Observe that then $m^2$ is a power of two too. By
  \cref{lm:answer-string-grammar-size}, there exists an admissible SLG
  $G_{\Pts} = (V_{\Pts}, \{0, 1\}, R_{\Pts}, S_{\Pts})$ such that
  $L(G_{\Pts}) = \{A(\Pts)\}$ is an answer string for $\Pts$
  (\cref{def:answer-string}), and it holds $|G_{\Pts}| = \bigO(m \log
  m)$ and $\sum_{X \in V_{\Pts}}|\expgen{G_{\Pts}}{X}| = \bigO(m^2
  \log m)$.  Moreover, observe that by the construction of $G_{\Pts}$
  (see \cref{sec:overview-verbin-yu}):
  \begin{itemize}[itemsep=0pt,parsep=0pt]
  \item For every $X \in V_{\Pts}$, $|\expgen{G_{\Pts}}{X}|$ is a
    power of two, and the parse tree $\mathcal{T}_{G_{\Pts}}(X)$ is a
    perfect binary tree of height $\log |\expgen{G_{\Pts}}{X}|$,
  \item For every $y \in [1 \dd m]$, there exists $X \in
    V_{\Pts}$ such that $\expgen{G_{\Pts}}{X} = A(\Pts)((y-1)m \dd
    ym]$. In other words, for every row in the matrix $M$ (see
    \cref{sec:overview-verbin-yu}), there exists a nonterminal in
    $V_{\Pts}$ with the corresponding expansion.
  \end{itemize}
  
  Observe now that there exists a sequence $(N_i)_{i \in [1 \dd
  |V_{\Pts}|]}$ that simultaneously satisfies the following
  conditions:
  \begin{itemize}[itemsep=0pt,parsep=0pt]
  \item $\{N_i : i \in [1 \dd |V_{\Pts}|]\} = V_{\Pts}$,
  \item $|\expgen{G_{\Pts}}{N_i}| \leq |\expgen{G_{\Pts}}{N_{i+1}}|$
    holds for every $i \in [1 \dd |V_{\Pts}|)$,
  \item Nonterminals expanding to consecutive rows of $M$
    (\cref{sec:overview-verbin-yu}) occur consecutively in the
    sequence $(N_i)_{i \in [1 \dd |V_{\Pts}|]}$, i.e., there exists
    $i_0 \in [0 \dd |V_{\Pts}|)$ such that for every $y \in [1 \dd
    m]$, $\expgen{G_{\Pts}}{N_{i_0+y}} = A(\Pts)((y-1)m \dd ym]$.
  \end{itemize}

  Let $w \in \beta(G_{\Pts})$ and $G'_{\Pts} = (V'_{\Pts}, \Sigma',
  R'_{\Pts}, S'_{\Pts})$ be the string and the SLG corresponding to
  $(N_i)_{i \in [1 \dd |V_{\Pts}|]}$ in \cref{def:beta}.  By
  \cref{lm:beta-length}\eqref{lm:beta-length-it-2}, it holds $|w| \leq
  6\sum_{X \in V_{\Pts}} |\expgen{G_{\Pts}}{X}| = \bigO(m^2 \log m)$.
  Observe that by \cref{def:beta}, for every $i \in [1 \dd
  |V_{\Pts}|]$, $\expgen{G_{\Pts}}{N_i}$ is a subsequence of
  $\expgen{G'_{\Pts}}{N_{i,0}}$. Observe also that the positions of symbols
  from $\expgen{G_{\Pts}}{N_i}$ in $\expgen{G'_{\Pts}}{N_{i,0}}$
  depend only on the shape of the parse tree
  $\mathcal{T}_{G_{\Pts}}(N_i)$, and not on the string
  $\expgen{G_{\Pts}}{N_i}$ itself.  Since for every $y_1, y_2 \in [1
  \dd m]$, the parse trees $\mathcal{T}_{G_{\Pts}}(N_{i_0 + y_1})$
  and $\mathcal{T}_{G_{\Pts}}(N_{i_0 + y_2})$ are both perfect binary
  trees of height $\log m$, it follows that there exists a mapping $f
  : [1 \dd m] \rightarrow \Zp$ such that for every $x,y \in [1 \dd
  m]$, it holds $A(\Pts)[(y-1)m + x] = \expgen{G_{\Pts}}{N_{i_0 +
  y}}[x] = \expgen{G'_{\Pts}}{N_{i_0 + y,0}}[f(x)]$. By combining
  this with \cref{lm:beta-length}\eqref{lm:beta-length-it-1}, and letting
  $\delta = |\bigodot_{i=1,\dots,i_0} \expgen{G'_{\Pts}}{N_{i,0}}
  \expgen{G'_{\Pts}}{N_{i,0}}|$, for every $x,y \in [1 \dd m]$, it
  holds:
  \[
    A(\Pts)[(y-1)m + x] = w[\delta + 2(y-1)(3m-2) + f(x)].
  \]

  Let $G_w = \algname{LZD}(w)$ be the output of \algname{LZD} on
  $w$. By \cref{lm:lzd}, we have $|G_w| = \frac{9}{2}|G_{\Pts}| =
  \bigO(m \log m)$.  Let $D'$ denote a data structure consisting of
  the following three components:
  \begin{enumerate}[itemsep=0pt,parsep=0pt]
  \item The data structure $D$ for string $w$. By $|G_w| = \bigO(m
    \log m)$ and the above assumption, $D$ uses
    $\bigO(|\algname{LZD}(w)| \log^c |w|) = \bigO(m \log m \log^c (m^2
    \log m)) = \bigO(m \log^{1 + c} m)$ space, and implements random
    access to $w$ in $o(\log |w| / \log \log |w|) \,{=}\, o(\log (m^2
    \log m) / \log \log (m^2 \log m)) \allowbreak = o(\log m / \log
    \log m)$ time,
  \item The array $F[1 \dd m]$ defined by $F[i] = f(i)$,
  \item The position $\delta \geq 0$, as defined above.
  \end{enumerate}
  In total, $D'$ needs $\bigO(m \log^{1 + c} m)$ space.  Observe that
  given the structure $D'$ and any $(x, y) \in [1 \dd m]^2$, we can
  answer in $o(\log m / \log \log m)$ the parity range query on $\Pts$
  with arguments $(x, y)$ by issuing a random access query to $w$ with
  position $j = \delta + 2(y-1)(3m-2) + F[x]$.  Thus, the existence of
  $D'$ contradicts
  \cref{th:parity-range-counting-lower-bound}.
\end{proof}

\subsection{Analysis of \algname{Bisection}}\label{sec:bisection}

\begin{definition}\label{def:dyadic-interval}
  Let $a, b \in \Zn$ be such that $a < b$. We call the interval $(a
  \dd b]$ \emph{dyadic} if there exists $k \geq 0$ such that $b - a =
  2^k$ and $a$ is a multiple of $2^k$ (in particular, if $a = 0$).
\end{definition}

The following observation follows directly from the definition of
\algname{Bisection} (see \cref{sec:algs-nonglobal}).

\begin{observation}\label{ob:bisection}
  Let $u \in \Sigma^n$, where $n$ is a power of two. Then,
  \algname{Bisection} applied to $u$ outputs an SLG of size $2 \cdot
  |\{u(a \dd b] : (a \dd b]\text{ is dyadic}, (a \dd b] \subseteq (0
  \dd n],\text{ and }b-a > 1\}|$.
\end{observation}

\begin{theorem}\label{th:bisection}
  For any $u \in \Sigma^{*}$, let $\algname{Bisection}(u)$ denote the
  output of the \algname{Bisection} algorithm on $u$. In the
  cell-probe model, every static data structure that for $u \in
  \Sigma^n$ uses $\bigO(|\algname{Bisection}(u)| \log^{c} n)$ space
  $($where $c = \bigO(1))$, requires $\Omega(\log n / \log \log n)$
  time to answer random access queries on $u$.
\end{theorem}
\begin{proof}
  Suppose that there exists a structure $D$ that for any $u \in
  \Sigma^n$ uses $\bigO(|\algname{Bisection}(u)| \log^c n)$ space
  (where $c = \bigO(1)$) and answers random access queries on $u$ in
  $o(\log n / \log \log n)$ time.  Let $\Pts \subseteq [1 \dd m]^2$ be
  any set of $|\Pts| = m$ points on an $m \times m$ grid. Assume for
  simplicity that $m$ is a power of two (otherwise, we pad $\Pts$ as
  in the proof of \cref{th:global}).
  Observe that then $m^2$ is a power of two too. By
  \cref{lm:answer-string-grammar-size}, there exists an admissible SLG
  $G_{\Pts} = (V_{\Pts}, \{0, 1\}, R_{\Pts}, S_{\Pts})$ such that
  $L(G_{\Pts}) = \{A(\Pts)\}$ is an answer string for $\Pts$
  (\cref{def:answer-string}), and it holds $|G_{\Pts}| = \bigO(m \log
  m)$. Observe that by the construction of $G_{\Pts}$ (see
  \cref{sec:overview-verbin-yu}), for every dyadic interval $(a \dd b]
  \subseteq (0 \dd m^2]$ satisfying $b - a > 1$, there exists a
  nonterminal $X_{a,b} \in V_{\Pts}$ satisfying
  $|\rhsgen{G_{\Pts}}{X_{a,b}}| = 2$ and $\expgen{G_{\Pts}}{X_{a,b}} =
  A(\Pts)(a \dd b]$. By \cref{ob:bisection}, this implies that
  $|\algname{Bisection}(A(\Pts))| \leq |G_{\Pts}| = \bigO(m \log m)$.
  Let $D'$ denote the structure $D$ for string $A(\Pts)$. It needs
  $\bigO(|\algname{Bisection}(A(\Pts))| \log^c |A(\Pts)|) = \bigO(m
  \log m \log^c (m^2)) = \bigO(m \log^{1 + c} m)$ space and implements
  random access to $A(\Pts)$ in $o(\log (m^2) / \log \log (m^2)) =
  o(\log m / \log \log m)$ time.  Given $D'$ and any $(x, y) \in [1
  \dd n]^2$, we can thus answer in $o(\log m / \log \log m)$ the
  parity range query on $\Pts$ with arguments $(x, y)$ by issuing a
  random access query on $A(\Pts)$ with position $j = x +
  (y-1)n$. Thus, the existence of $D'$ contradicts
  \cref{th:parity-range-counting-lower-bound}.
\end{proof}

\subsection{Analysis of \algname{LZ78}}\label{sec:lz78}

The upper bound for the random-access problem on LZ78-compressed text was
established in~\cite{DuttaLRR13}.

\begin{theorem}[Dutta, Levi, Ron, Rubinfeld~\cite{DuttaLRR13}]
  For any $u \in \Sigma^{*}$, let $\algname{LZ78}(u)$ denote the
  output of the \algname{LZ78} algorithm on $u$. There exits a data
  structure that, for any $u \in \Sigma^{n}$, uses
  $\bigO(|\algname{LZ78}(u)|)$ space and answers random access queries
  on $u$ in $\bigO(\log \log n)$ time.
\end{theorem}

The key idea of the above solution is as follows. We store boundaries
of all phrases in the LZ78 parsing in a predecessor data structure.
Since for a length-$n$ string, the number of phrases $z_{78}$
satisfies $z_{78} = \Omega(\sqrt{n})$, using~\cite{PatrascuT06}, we
obtain $\bigO(\log \log n)$ predecessor query time and $\bigO(z_{78})$
space.  We also store all LZ78 phrases in a trie augmented with the
support for level ancestor queries
(using~\cite{BenderF04,BerkmanV94,Dietz91}, we obtain linear space and
$\bigO(1)$ time). At query time, we first locate the phrase containing
the queried position, and then use the level ancestor query to obtain
the symbol.

We prove that the above structure is optimal. More precisely, we show
(more generally) that as long as the space of the structure is
near-linear in the size of the LZ78 parsing, i.e.,
$\bigO(|\algname{LZ78}(u)| \log^{c} n)$, where $c = \bigO(1)$, the
query time must be $\Omega(\log \log n)$.

\begin{lemma}\label{lm:lz78}
  Let $u = c_1^{p_1} c_2^{p_2} \dots c_k^{p_k}$, where for every $i
  \in [1 \dd k]$, it holds $p_i > 0$ and $c_i \in \{{\tt 0}, {\tt 1}\}$,
  and $|u| = \sum_{i=1}^{k} p_i = k^2$. The \algname{LZ78} algorithm
  factorizes $u$ into at most $6k$ phrases.
\end{lemma}
\begin{proof}
  For every $i \in [1 \dd k]$, let $s_i = \sum_{j=1}^{i} p_j$. To
  streamline the formulae, we also set $s_0 = 0$.  Let $z_{78}$ denote
  the number of phrases in the LZ78 parsing of $u$, and let $(e_i)_{i
  \in [0 \dd z_{78}]}$ be a sequence such $e_0 = 0$ and for $i \in [1
  \dd z_{78}]$, $e_i$ is the last position of the $i$th leftmost
  phrase.  For any $i \in [1 \dd z_{78}]$, we call the $i$th phrase
  $u(e_{i-1} \dd e_{i}]$ \emph{internal}, if there exists $i' \in [1
  \dd k]$ such that $(e_{i-1} \dd e_{i} + 1] \subseteq (s_{i'-1} \dd
  s_{i'}]$. Otherwise, the phrase is \emph{external}. Note that every
  internal phrase is either a substring of $\texttt{0}^{\infty}$ or
  $\texttt{1}^{\infty}$. We call those internal phrases \emph{type-0}
  and \emph{type-1}, respectively. We bound the number of phrases of
  each type as follows:
  \begin{itemize}
  \item First, we show that the number of external phrases is at most
    $2k$. To this end, we show that for every $i \in [1 \dd k]$, the
    substring $u(s_{i-1} \dd s_i]$ it overlapped by at most two
    external phrases. Suppose that this does not hold, i.e., there
    exist $i_1, i_2, i_3 \in [1 \dd z_{78}]$ such that phrases
    $u(e_{i_1-1} \dd e_{i_1}]$, $u(e_{i_2-1} \dd e_{i_2}]$, and
    $u(e_{i_3-1} \dd e_{i_3}]$ are external and overlap $u(s_{i-1} \dd
    s_i]$. Assume without loss of generality that $i_1 < i_2 <
    i_3$. Note that then $e_{i_1-1} < e_{i_1} \leq e_{i_2-1} < e_{i_2}
    \leq e_{i_3-1} < e_{i_3}$.  Note also that our assumption about
    overlapping applied to phrases $u(e_{i_1-1} \dd e_{i_1}]$ and
    $u(e_{i_3-1} \dd e_{i_3}]$ then implies $s_{i-1} < e_{i_1}$ and
    $e_{i_3 - 1} + 1 \leq s_{i}$. By $e_{i_1} \leq e_{i_2-1}$ and
    $e_{i_2} \leq e_{i_3-1}$, we thus have $s_{i-1} < e_{i_2-1} <
    e_{i_2} + 1 \leq s_{i}$, which implies $(e_{i_2-1} \dd e_{i_2} +
    1] \subseteq (s_{i-1} \dd s_i]$.  By definition, the phrase
    $u(e_{i_2-1} \dd e_{i_2}]$ is therefore internal, a contradiction.
  \item Next, we show that the number of internal type-$0$ phrases is
    at most $2k - 1$. Denote their number by $q$. Let $\{i_1, i_2,
    \dots, i_q\} \subseteq [1 \dd z_{78}]$ be such that $i_1 < \dots <
    i_q$ and for every $j \in [1 \dd q]$, $u(e_{i_j-1} \dd e_{i_j}]$
    is a type-$0$ internal phrase. We show by induction that for every
    $j \in [1 \dd q]$, it holds $e_{i_j} - e_{i_j-1} \geq j$.  The
    induction base case holds trivially, since every phrase is a
    nonempty string. Let us now assume $j > 1$ and suppose $e_{i_j} -
    e_{i_j-1} = j' < j$.  Let $i' \in [1 \dd k]$ be such that
    $(e_{i_j-1} \dd e_{i_j} + 1] \subseteq (s_{i'-1} \dd s_{i'}]$
    (such $i'$ exists since $u(e_{i_j-1} \dd e_{i_j}]$ is an internal
    phrase). Note that then $u[e_{i_j} + 1] = \texttt{0}$.  By the
    inductive assumption, $\texttt{0}^{j-1}$ is a substring of
    $u(e_{i_{j-1}-1} \dd e_{i_{j-1}}]$.  Since for every phrase in the
    LZ78 parsing, all its proper prefixes must occur earlier as
    phrases, we thus obtain that there exists $i'' \in [1 \dd i_j)$
    such that $u(e_{i''-1} \dd e_{i''}] = \texttt{0}^{j'}$.  By
    $u[e_{i_j} + 1] = \texttt{0}$, this implies $e_{i_j} - e_{i_j-1} >
    j'$, since $u(e_{i_j-1} \dd e_{i_j}+1] = \texttt{0}^{j' + 1}$ is a
    valid candidate for the phrase, a contradiction. We have thus
    proved the inductive step. Suppose now that $q \geq 2k$.  The
    total length of internal type-0 phrases would then be $\geq
    \tfrac{2k(2k+1)}{2} = 2k^2 + k > |u|$.  Thus, $q < 2k$.
  \item Analogously as above, the number of internal type-1 phrases is
    at most $2k - 1$.
  \end{itemize}
  By the above, we thus obtain $z_{78} \leq 6k - 2 \leq 6k$.
\end{proof}

The following lower bound is a special case of the general tight
tradeoff for the colored predecessor problem established by
P\u{a}tra\c{s}cu and Thorup. In the colored predecessor problem, we
are given a collection of integers $Y = \{(x_1, c_1), \dots, (x_m,
c_m)\} \subseteq [0 \dd u) \times \{{\tt 0}, {\tt 1}\}$ (where $u$ is
the size of the universe), each augmented with a bit (a ``color''). We
assume that there exists $j \in [1 \dd m]$ such that $x_j = 0$. Given
any $x \in [0 \dd u)$, the query asks to return the color of its
predecessor in $\{x_1, \dots, x_m\}$, i.e., the value $c_i$, where $i
= \max \{j \in [1 \dd m] : x_j \leq x\}$.

\begin{theorem}[P\u{a}tra\c{s}cu and
    Thorup~\cite{PatrascuT06}]\label{th:predecessor-lower-bound}
  Let $m \geq 1$ and $Y = \{(x_1, c_1), \dots, (x_m, c_m)\} \subseteq
  [0 \dd m^2) \times \{{\tt 0}, {\tt 1}\}$ be such that $0 = x_1 <
  \dots < x_m$.  In the cell-probe model, every static data structure
  that for a set $Y$ uses $\bigO(m \log^{c} m)$ space $($where $c =
  \bigO(1))$, requires $\Omega(\log \log m)$ time to answer colored
  predecessor queries on $Y$.
\end{theorem}

\begin{theorem}\label{th:lz78-lower-bound}
  For any $u \in \Sigma^{*}$, let $\algname{LZ78}(u)$ denote the
  output of the \algname{LZ78} algorithm on $u$. In the cell-probe
  model, every static data structure that for $u \in \Sigma^n$ uses
  $\bigO(|\algname{LZ78}(u)| \log^{c} n)$ space $($where $c =
  \bigO(1))$, requires $\Omega(\log \log n)$ time to answer random
  access queries on $u$.
\end{theorem}
\begin{proof}
  Suppose that there exists a data structure $D$ that for any $u \in
  \Sigma^n$ uses $\bigO(|\algname{LZ78}(u)| \log^c n)$ space (where $c
  = \bigO(1)$) and answers random access queries on $u$ in $o(\log
  \log n)$ time. Let $m \geq 1$ and $Y = \{(x_1, c_1), \dots, (x_m,
  c_m)\} \subseteq [0 \dd m^2) \times \{\texttt{0}, \texttt{1}\}$ be
  such that $0 = x_1 < \dots < x_m$.  Let
  \[
    w_Y = c_1^{x_2-x_1} \cdot c_2^{x_3-x_2} \cdot \ldots \cdot
    c_{m-1}^{x_m - x_{m-1}} \cdot c_m^{m^2 - x_{m}}.
  \]
  Observe that for every $x \in [0 \dd m^2)$, it holds $w_Y[x] =
  c_{i}$, where $i = \max \{j \in [1 \dd m] : x_j \leq x\}$.  In other
  words, $w_Y[x]$ is the answer to the colored predecessor problem on
  $Y$ with argument $x$ (see~\cite{PatrascuT06}). By $|w_Y| = m^2$ and
  \cref{lm:lz78}, it holds $|\algname{LZ78}(w_Y)| \leq 6m$. Let $D'$
  denote the structure $D$ for string $w_Y$. It needs
  $\bigO(|\algname{LZ78}(w_Y)| \log^c |w_Y|) = \bigO(m \log^c (m^2)) =
  \bigO(m \log^c m)$ space, and implements random access queries to
  $w_Y$ in $o(\log \log (m^2)) = o(\log \log m)$ time. Given $D'$ and
  any $x \in [0 \dd m^2)$, we can thus answer in $o(\log \log m)$ time
  a colored predecessor query on $Y$ with argument $x$ by issuing a
  random access query on $w_Y$ with position $x$. Thus, the existence
  of $D'$ contradicts \cref{th:predecessor-lower-bound}.
\end{proof}

\section{Parsing Context-Free Grammars}\label{sec:cfg-parsing}

\subsection{Problem Definition}\label{sec:cfg-problem}

\vspace{-2ex}
\setlength{\FrameSep}{1.7ex}
\begin{framed}
  \noindent
  \textsc{Context-Free Grammar (CFG) Parsing}\\
  \textbf{Input}: A string $u \in \Sigma^n$ and a context-free grammar
  $\Gamma$.\\
  \textbf{Output:} Decide, whether $u \in L(\Gamma)$, i.e., whether
  $u$ is in the language of $\Gamma$.
\end{framed}

\subsection{Prior Work}\label{sec:cfg-prior}

Abboud, Backurs, Bringmann, and K{\"{u}}nnemann
developed a new technique for proving the conditional hardness of
CFG parsing on grammar-compressed strings, establishing the following
result.

\begin{theorem}[\cite{AbboudBBK17}]\label{th:cfg-main-lower-bound}
  Let $\delta \in (0, 1]$. Assuming the $k$-Clique Conjecture
  (resp.\ Combinatorial $k$-Clique Conjecture), there is no algorithm
  (resp.\ combinatorial algorithm) that, given $u \in \Sigma^{n}$ and
  a CFG $\Gamma$ of size $|\Gamma| = \bigO(\polylog n)$ satisfying
  $|\Sigma| = \bigO(1)$ and $g^{*}(u) = \bigO(n^{\delta})$,
  determines if $u \in L(\Gamma)$ in $\bigO(n^{\omega-\epsilon})$
  $($resp.\ $\bigO(n^{3 - \epsilon})$$)$ time, for any $\epsilon > 0$.
\end{theorem}

In other words, assuming the $k$-Clique Conjecture
(resp.\ Combinatorial $k$-Clique Conjecture), the CFG parsing for any
length-$n$ string $u$ cannot be performed in
$\bigO(n^{\omega-\epsilon})$ (resp.\ $\bigO(n^{3-\epsilon})$) time,
even restricted to highly compressible $u$. This implies the following
result.

\begin{corollary}[\cite{AbboudBBK17}]
  Assuming the $k$-Clique Conjecture (resp.\ Combinatorial $k$-Clique
  Conjecture), there is no algorithm (resp.\ combinatorial algorithm)
  that, given any SLG $G$ and a CFG $\Gamma$ such that $L(G) = \{u\}$
  for some $u \in \Sigma^n$ $($where $|\Sigma| = \bigO(1)$$)$ and
  $|\Gamma| = \bigO(\polylog n)$, determines if $u \in L(\Gamma)$ in
  $\bigO(\poly(|G|) \cdot n^{\omega-\epsilon})$
  $($resp.\ $\bigO(\poly(|G|) \cdot n^{3 - \epsilon})$$)$ time, for
  any $\epsilon > 0$.
\end{corollary}

The key idea in the proof of the above result is as follows.
Suppose that the algorithm in
question exists and runs in $\bigO(|G|^c \cdot n^{\omega - \epsilon})$
(resp.\ $\bigO(|G|^c \cdot n^{3 - \epsilon})$) time, where $\epsilon >
0$ and $c > 0$. Let $\delta = \tfrac{\epsilon}{3c}$. Consider any $u
\in \Sigma^{n}$ such that $|\Sigma| = \bigO(1)$ and $g^{*}(u) =
\bigO(n^{\delta})$, and some CFG $\Gamma$ satisfying $|\Gamma| =
\bigO(\polylog n)$. Run the following algorithm:

\begin{enumerate}
\item First, using any grammar compression algorithm with an
  approximation ratio of $\bigO(\log n)$ (such
  as~\cite{Rytter03,Charikar05,Jez16}), compute an SLG $G$ such that
  $L(G) = \{u\}$ and $|G| = \bigO(g^{*}(u) \log n) =
  \bigO(n^{\delta} \log n)$. Using for example~\cite{Jez16}, this
  takes $\bigO(n)$ time.
\item Second, run the above hypothetical algorithm algorithm for CFG
  parsing. This takes $\bigO(|G|^c \cdot n^{\omega-\epsilon}) =
  \bigO((n^{\delta} \log n)^c \cdot n^{\omega-\epsilon}) =
  \bigO(n^{\omega - 2\epsilon/3} \log^{c} n) = \bigO(n^{\omega -
  \epsilon/3})$ (resp. $\bigO(|G|^{c} \cdot n^{3-\epsilon}) =
  \bigO(n^{3-\epsilon/3})$) time.
\end{enumerate}
In total, we have thus spent $\bigO(n^{\omega - \epsilon/3})$ (resp.\
$\bigO(n^{3-\epsilon/3})$) time checking if $u \in L(\Gamma)$. By
\cref{th:cfg-main-lower-bound}, this implies that the $k$-Clique
Conjecture (resp.\ Combinatorial $k$-Clique Conjecture) is false.

\subsection{The Main Challenge}\label{sec:cfg-challenge}

The above result shows that it is hard to solve CFG parsing in
compressed time assuming the input has been constructed using a
grammar compression algorithm with a small approximation
factor.

This raises the question about the role of the approximation ratio in
the hardness. For the majority of the known grammar compressors, their
approximation ratio is either $\omega(\polylog n)$ or
unknown~\cite{repair,greedy1,longestmatch,bisection1,sequential,sequitur,LZ78}.
In such scenario, the techniques of Abboud, Backurs, Bringmann, and
K{\"{u}}nnemann give weaker lower bounds than above. To illustrate
the problem, let us redo the above analysis for an algorithm with
the approximation factor as a parameter. We first first recall the
crucial technical lemma proved by Abboud, Backurs, Bringmann,
and K{\"{u}}nnemann.

\begin{lemma}[\cite{AbboudBBK17}]\label{lm:kclique-small-slg}
  Let $k \geq 1$ be a constant. For every undirected graph $G = (V,
  E)$, there exists a string $u \in \Sigma^{*}$ of length $|u| =
  \Theta(|V|^{k+2})$ over alphabet $\Sigma = \{{\tt 0}, {\tt 1}, {\rm
  \#}, {\rm \$}, x, y, z\}$, and a CFG $\Gamma$ of size $|\Gamma| =
  \bigO(\log |V|)$ such that $u \in L(\Gamma)$ holds if and only if
  $G$ has a $3k$-clique. Moreover, given $G$, in $\bigO(|V|^3)$ time
  we can compute $\Gamma$ and an SLG $H = (V_H, \Sigma, R_H, S_H)$
  such that $L(H) = \{u\}$, $|H| = \bigO(|V|^3)$, and
  $\sum_{X \in V_H}|\expgen{H}{X}| = \bigO(|V|^{k+5})$.
\end{lemma}

Using the above lemma, we can prove the following result.

\begin{theorem}[Based
    on~\cite{AbboudBBK17}]\label{th:cfg-lower-bound-simple-generalization}
  Consider a grammar compression algorithm \algname{Alg} with an
  approximation ratio $\bigO(n^{\alpha})$ (where $\alpha \in (0, 1)$).
  Let $\delta \in (\alpha, 1]$.
  Assuming the $k$-Clique Conjecture (resp.\ Combinatorial
  $k$-Clique Conjecture), there is no algorithm (resp.\ combinatorial
  algorithm) that, given $u \in \Sigma^{n}$ and a CFG $\Gamma$
  satisfying $|\Gamma| = \bigO(\polylog n)$, $|\Sigma| = \bigO(1)$,
  and $|\algname{Alg}(u)| = \bigO(n^{\delta})$, determines if $u \in
  L(\Gamma)$ in $\bigO(n^{\omega-\epsilon})$ $($resp.\ $\bigO(n^{3 -
  \epsilon})$$)$ time, for any $\epsilon > 0$.
\end{theorem}
\begin{proof}
  Suppose the above theorem is not true, and let $\epsilon > 0$ be
  such that the algorithm in question runs in
  $\bigO(n^{\omega-\epsilon})$ (resp.\ $\bigO(n^{3-\epsilon})$)
  time. Assume that we are given an undirected graph $G = (V, E)$. Let
  $k = \max(\lceil \tfrac{3}{\delta - \alpha} \rceil, \lceil
  \tfrac{12}{\epsilon} \rceil)$.  We execute the following algorithm:
  \begin{enumerate}
  \item First, using \cref{lm:kclique-small-slg}, in $\bigO(|V|^3)$
    time we build a CFG $\Gamma$ and an SLG $H = (V_{H}, \Sigma,
    R_{H}, S_{H})$ such that letting $L(H) = \{u\}$, it holds
    $|\Sigma| = \bigO(1)$, $|H| = \bigO(|V|^3)$, $|u| =
    \Theta(|V|^{k+2})$, $|\Gamma| = \bigO(\log |V|) = \bigO(\log
    |u|)$, and $u \in L(\Gamma)$ holds if and only if $G$ has a
    $3k$-clique. Using $H$, we generate $u$ in $\bigO(|V|^{k+2})$
    time.  Observe that the existence of $H$ implies that $g^{*}(u) =
    \bigO(|V|^3)$. Consequently, $|\algname{Alg}(u)| = \bigO(g^{*}(u)
    \cdot |u|^{\alpha}) = \bigO(|V|^3 \cdot |u|^{\alpha}) =
    \bigO(|u|^{3/(k+2) + \alpha}) = \bigO(|u|^{\delta})$, where we
    used that $|u| = \Theta(|V|^{k+2})$ and $\tfrac{3}{\delta -
    \alpha} \leq k$ (which implies $3/(k+2) \leq \delta - \alpha$).

  \item Second, we apply the above hypothetical CFG
    parsing to $u$ and $\Gamma$.  More precisely, we check if $u \in
    L(\Gamma)$ in $\bigO(|u|^{\omega-\epsilon}) =
    \bigO(|V|^{(k+2)(\omega-\epsilon)}) =
    \bigO(|V|^{k(\omega-\epsilon) + 2\omega}) =\allowbreak
    \bigO(|V|^{\omega k(1 - \epsilon/(2\omega)) - (k\epsilon/2 -
    2\omega)}) = \bigO(|V|^{\omega k(1 - \epsilon/(2\omega))})$
    (resp.\ $\bigO(|u|^{3 - \epsilon}) \,{=}\, \bigO(|V|^{(k+2)(3 -
    \epsilon)}) \allowbreak= \bigO(|V|^{k(3-\epsilon) + 6)})
    \allowbreak= \bigO(|V|^{3k(1 - \epsilon/6)})$) time, where
    the last equality follows by $k \geq \lceil \tfrac{12}{\epsilon}
    \rceil$ (which implies $k\epsilon/2 \geq 6$).
  \end{enumerate}
  We have thus checked if $G$ contains a $3k$-clique in $\bigO(|V|^3 +
  |V|^{k+2} + |V|^{\omega k(1 - \epsilon/(2\omega))}) =
  \bigO(|V|^{\omega k(1 - \epsilon')})$ (resp. $\bigO(|V|^3 +
  |V|^{k+2} + |V|^{3k(1 - \epsilon/6)}) = \bigO(|V|^{3k(1 -
  \epsilon')})$) time, where $\epsilon' > 0$ is some constant (and
  we used that $k \geq 3$). This implies that the $k$-Clique Conjecture
  (resp.\ Combinatorial $k$-Clique Conjecture) is false.
\end{proof}

\begin{corollary}[Based
    on~\cite{AbboudBBK17}]\label{cor:cfg-lower-bound-simple-generalization}
  Consider a grammar compression algorithm \algname{Alg} that runs in
  $\bigO(n^q)$ time, where $q < \omega$ (resp.\ $q < 3$) is a constant,
  and has an approximation
  ratio $\bigO(n^{\alpha})$ (for a constant $\alpha \in (0, 1)$).
  Assuming the $k$-Clique Conjecture (resp.\ Combinatorial $k$-Clique
  Conjecture), there is no algorithm (resp.\ combinatorial algorithm)
  that, given the SLG $G = \algname{Alg}(u)$ of a string $u \in
  \Sigma^{n}$ $($where $|\Sigma| = \bigO(1)$$)$ and a CFG $\Gamma$ such
  that $|\Gamma| = \bigO(\polylog n)$, determines if $u \in L(\Gamma)$ in
  $\bigO(|G|^c \cdot n^{\omega-\epsilon})$ $($resp.\ $\bigO(|G|^c
  \cdot n^{3 - \epsilon})$$)$ time, for any constants $\epsilon > 0$
  and $c \in (0, \epsilon/\alpha)$.
\end{corollary}
\begin{proof}
  Suppose that such algorithm exists. Let $\delta = \alpha/2 +
  \epsilon/(2c)$.  Observe that $c \in (0, \epsilon/\alpha)$ implies
  $\epsilon/(2c) > \alpha/2$.  Thus, $\delta > \alpha$. On the other
  hand, we also have $c \delta = c\alpha/2 + \epsilon/2 < \epsilon$.
  Suppose now that we are given a string $u \in \Sigma$ and a CFG
  $\Gamma$ such that $|\Gamma| = \bigO(\polylog n)$, $|\Sigma| =
  \bigO(1)$, and $|\algname{Alg}(u)| = \bigO(n^{\delta})$.  We execute
  the following algorithm:
  \begin{enumerate}
  \item First, we compute $G = \algname{Alg}(u)$. This takes
    $\bigO(n^q)$ time.  By the assumption, we have $|G| =
    \bigO(n^{\delta})$.
  \item Next, we apply the above hypothetical algorithm to $\Gamma$
    and $G$, i.e., we check if $u \in L(\Gamma)$ in $\bigO(|G|^c
    n^{\omega-\epsilon}) = \bigO(n^{\delta c} n^{\omega-\epsilon}) =
    \bigO(n^{\omega - (\epsilon - \delta c)})$ (resp.\ $\bigO(|G|^c
    n^{3-\epsilon}) = \bigO(n^{\delta c} n^{3-\epsilon}) = \bigO(n^{3
    - (\epsilon - \delta c)})$) time.
  \end{enumerate}
  In total, we have thus checked if $u \in L(\Gamma)$ in $\bigO(n^q +
  n^{\omega - (\epsilon - \delta c)}) = \bigO(n^{\omega - \epsilon'})$
  (resp.\ $\bigO(n^q + n^{3 - (\epsilon - \delta c)}) = \bigO(n^{3 -
  \epsilon'})$) time, where $\epsilon' > 0$ is some constant By
  \cref{th:cfg-lower-bound-simple-generalization}, this implies that
  the $k$-Clique Conjecture (resp.\ Combinatorial $k$-Clique
  Conjecture) is false. Note that
  \cref{th:cfg-lower-bound-simple-generalization} requires that
  $\delta > \alpha$, which we showed above.
\end{proof}

\begin{remark}
  The above result thus shows that applying the techniques
  from~\cite{AbboudBBK17} prevents some compressed algorithms
  (assuming the $k$-Clique Conjecture or the Combinatorial $k$-Clique
  Conjecture), e.g., if $\alpha = 1/3$, then we cannot combinatorially
  solve CFG parsing on grammars produced by \algname{Alg} in
  $\bigO(|G|^5 \cdot n)$ or $\bigO(|G|^2 \cdot n^2)$ time, but the
  lower bound does not prevent an algorithm running in $\bigO(|G|^7
  \cdot n)$ or $\bigO(|G|^4 \cdot n^2)$ time due to the dependence of the
  lower bound on $\alpha$. In the following sections, we describe
  methods for eliminating this dependance from the analysis.
\end{remark}

\subsection{Preliminaries}\label{sec:cfg-prelim}

\begin{lemma}\label{lm:interleave}
  Let $\Gamma = (V, \Sigma, R, S)$ be a CFG. Assume that the sets
  $\Sigma$, $\{\dol_i\}_{i \in [1 \dd |V|]}$, and $\{\hash_i\}_{i \in
  [1 \dd 2|V|]}$ are pairwise disjoint.  Denote $\Sigma'' = \Sigma
  \cup \{\dol_i\}_{i \in [1 \dd |V|]} \cup \{\hash_i\}_{i \in [1 \dd
  2|V|]}$. There exists a CFG $\Gamma'$ such that $L(\Gamma') =
  \{a_1 a_2 \cdots a_{2k} \in \Sigma''^{2k} : k \in \Zp\text{ and
  }a_1a_3a_5 \cdots a_{2k-3}a_{2k-1} \in L(\Gamma)\}$ and $|\Gamma'| =
  \bigO(|\Gamma| + |\Sigma''|)$. Moreover, given $\Gamma$ and $|V|$,
  we can construct $\Gamma'$ in $\bigO(|\Gamma| + |\Sigma''|)$ time.
\end{lemma}
\begin{proof}
  Let $\Gamma' = (V \cup \{X\}, \Sigma'', R', S)$. The set $R'$ is
  defined as the smallest set satisfying the following conditions:
  \begin{itemize}[itemsep=0pt,parsep=0pt]
  \item For every $c \in \Sigma''$, $R'$ contains the rule $X
    \rightarrow c$,
  \item For every rule $Y \rightarrow b_1 b_2 \cdots b_q$ from $R$,
    the set $R'$ contains the rule $Y \rightarrow b_1 X b_2 X \cdots
    b_q X$.
  \end{itemize}
  The correctness of this construction follows immediately. The total
  size of the rules is $R'$ is twice the total size of the rules in
  $R$ plus $|\Sigma''|$. Thus, $|\Gamma'| = \bigO(|\Gamma| +
  |\Sigma''|)$.  Given $\Gamma$ and $|V|$, the CFG $\Gamma'$ is easily
  constructed in $\bigO(|\Gamma| + |\Sigma''|)$ time.
\end{proof}

\begin{lemma}\label{lm:add-prefix}
  Let $\Gamma = (V, \Sigma, R, S)$ be a CFG. For every $k \geq 1$,
  there exists a CFG $\Gamma'$ such that $L(\Gamma') = \{xy : x \in
  \Sigma^k\text{ and }y \in L(\Gamma)\}$ and $|\Gamma'| =
  \bigO(|\Gamma| + |\Sigma| + \log k)$. Moreover, given $\Gamma$ and
  $k$, we can construct $\Gamma'$ in $\bigO(|\Gamma| + |\Sigma| + \log
  k)$ time.
\end{lemma}
\begin{proof}
  Denote $m = \lfloor \log k \rfloor$. Let $\Gamma' = (V \cup \{S',
  X_0, X_1, \dots, X_{m}\}, \Sigma, R', S')$.  The set $R'$ is defined
  as the minimal superset of $R$ satisfying the following conditions:
  \begin{itemize}[itemsep=0pt,parsep=0pt]
  \item For every $c \in \Sigma$, $R'$ contains the rule $X_0
    \rightarrow c$,
  \item For every $i \in [1 \dd m]$, $R'$ contains the rule $X_i
    \rightarrow X_{i-1} X_{i-1}$,
  \item $R'$ contains the rule $S' \rightarrow X_{b_1} \cdot X_{b_2}
    \cdots X_{b_q} \cdot S$, where $b_1, \dots, b_q \in \Zn$ are such
    that $b_1 < \dots < b_q$ and $k = \prod_{i=1}^{q}2^{b_i}$. Note
    that $b_q \leq m$, and hence $X_{b_i}$ is defined for every $i \in
    [1 \dd q]$.
  \end{itemize}
  First, observe that for every $i \in [0 \dd m]$, we have
  $\langgen{\Gamma'}{X_i} = \Sigma^{2^i}$. Thus, prepending the
  starting rule $S$ with $X_{b_1} \cdots X_{b_q}$ ensures that any
  string in $L(\Gamma')$ is of the form $xy$, where $x \in
  \Sigma^{2^{b_1} \cdots 2^{b_q}} = \Sigma^{k}$ and $y \in
  L(\Gamma)$. The total size of the newly introduced rules is
  $\bigO(|\Sigma| + m) = \bigO(|\Sigma| + \log k)$, and hence
  $|\Gamma'| = \bigO(|\Gamma| + |\Sigma| + \log k)$. Given $\Gamma$
  and $k$, adding the new rules takes $\bigO(|\Sigma| + \log k)$ time,
  and thus the construction takes $\bigO(|\Gamma| + |\Sigma| + \log
  k)$ time.
\end{proof}

\begin{lemma}\label{lm:extra-dollars}
  Let $G = (V, \Sigma, R, S)$ be a CFG such that $\Sigma \cap
  \{\dol_i\}_{i \in [1 \dd 2|V|]} = \emptyset$. Denote $\Sigma' =
  \Sigma \cup \{\dol_i\}_{i \in [1 \dd 2|V|]}$.  There
  exists a CFG $\Gamma'$ of size $|\Gamma'| = \bigO(|\Gamma| + |V|)$
  such that $L(\Gamma') = \{u \in \Sigma'^{*} : {\rm
  erase}(u, \{\dol_i\}_{i \in [1 \dd 2|V|]}) \in L(\Gamma)\}$,
  where ${\rm erase}(u, S)$ is the subsequence of $u$ resulting from
  erasing all occurrences of symbols from $S$ in $u$. Moreover, given
  $\Gamma$ and $|V|$, we can construct $\Gamma'$ in $\bigO(|\Gamma|
  + |V|)$ time.
\end{lemma}
\begin{proof}
  Let $\Gamma' = (V \cup \{N_{\dol}\}, \Sigma', R', S)$. The set $R'$
  is defined as the smallest set satisfying the following conditions:
  \begin{itemize}[itemsep=0pt,parsep=0pt]
  \item $R'$ contains rules $N_{\dol} \rightarrow N_{\dol} N_{\dol}$
    and $N_{\dol} \rightarrow \emptystring$,
  \item For every $i \in [1 \dd 2|V|]$, $R'$ contains the rule
    $N_{\dol} \rightarrow \dol_i$,
  \item For every rule $Y \rightarrow b_1 b_2 \dots b_q$ in $R$, the
    set $R'$ contains the rule $Y \rightarrow N_{\dol} b_1 N_{\dol}
    b_2 N_{\dol} \dots N_{\dol} b_q N_{\dol}$.
  \end{itemize}
  The correctness of this construction follows immediately. The total
  size of the rules in $R'$ is at most three times the size of rules
  in $R$ plus $2|V| + 2$. Thus, $|\Gamma'| = \bigO(|\Gamma| + |V|)$.
  Given $\Gamma$ and $|V|$, the CFG $\Gamma'$ is easily constructed in
  $\bigO(|\Gamma| + |V|)$ time.
\end{proof}

\subsection{Analysis of \algname{Sequitur}, \algname{Sequential}, and
  Global Algorithms}\label{sec:cfg-sequential}

\begin{observation}\label{ob:alpha}
  Let $G = (V, \Sigma, R, S)$ be an admissible SLG. Let $u \in
  \Sigma^{+}$ be such that $L(G) = \{u\}$ and let $w \in \alpha(G)$
  (\cref{def:alpha}).
  Then, for every $j \in [1 \dd |u|]$, it holds $u[j] = u'[2j - 1]$,
  where $u'$ is the length-$2|u|$ suffix of $w$.
\end{observation}
\begin{proof}
  Let $G'$ be the SLG and $(N_i)_{i \in [1 \dd |V|]}$ be a sequence of
  nonterminals corresponding to $w$ in \cref{def:alpha}. The
  assumption about $\mathcal{T}(S)$ implies that for every $X \in V
  \setminus \{S\}$, it holds $|\expgen{G}{X}| <
  |\expgen{G}{S}|$. Thus, $N_{|V|} = S$.  By an easy induction, it
  follows that for every $X \in V$, letting $m = |\expgen{G}{X}|$, it
  holds $|\expgen{G'}{X}| = 2m - 1$, and for every $i \in [1 \dd m]$,
  we have $\expgen{G}{X}[i] = \expgen{G'}{X}[2i - 1]$. By $w =
  \bigodot_{i=1,\dots,|V|} \expgen{G'}{N_i} \cdot \hash_{2i-1} \cdot
  \expgen{G'}{N_i} \cdot \hash_{2i}$, for every $j \in [1 \dd |u|]$,
  we thus obtain that $u' := \expgen{G'}{N_{|V|}} \hash_{2|V|} =
  \expgen{G'}{S} \hash_{2|V|}$ is a suffix of $w$. By the above
  discussion, we thus obtain the claim.
\end{proof}

\begin{lemma}\label{lm:construction-alpha}
  Let $G = (V, \Sigma, R, S)$ be an admissible SLG and assume that
  $4\sum_{X \in V}|\expgen{G}{X}|$ fits into a machine word.  Given
  $G$, we can compute some $w \in \alpha(G)$ (\cref{def:alpha}) in
  $\bigO(\sum_{X \in V}|\expgen{G}{X}|)$ time.
\end{lemma}
\begin{proof}
  Assume that $G$ is given using an encoding in which nonterminals are
  identified with consecutive positive integers. The construction of
  $w \in \alpha(G)$ consist of three steps:
  \begin{enumerate}[itemsep=0pt,parsep=0pt]
  \item First, in $\bigO(|G|)$ time we sort the nonterminals of the
    implicit grammar DAG (defined so that there is an edge connecting
    nonterminals $X$ and $Y$ when $Y$ appears in $\rhsgen{G}{X}$)
    topologically. In $\bigO(|G|)$ time we then compute
    $|\expgen{G}{X}|$ for every $X \in V$. We then sort all $X \in V$
    using $|\expgen{G}{X}|$ as the key (with ties resolved
    arbitrarily).  Using radix sort, this can be done in
    $\bigO(\sum_{X \in V}|\expgen{G}{X}|)$ time. Let $\{N_i\}_{i \in
    [1 \dd |V|]}$ denote the resulting sequence.
  \item Second, we construct $G'$ defined as in \cref{def:alpha}. This
    is easily done in $\bigO(|G|)$ time.
  \item Lastly, we compute the output string $w =
    \bigodot_{i=1,\dots,|V|} \expgen{G'}{N_i} \cdot \hash_{2i-1} \cdot
    \expgen{G'}{N_i} \cdot \hash_{2i}$.  Using $G'$ this is easily
    done in $\bigO(|w|)$ time, which is $\bigO(\sum_{X \in
    V}|\expgen{G}{X}|)$ by \cref{lm:alpha-length}.  \qedhere
  \end{enumerate}
\end{proof}

\begin{lemma}\label{lm:alpha-construction-of-gamma-prim}
  Let $u \in \Sigma^{+}$ and $G = (V, \Sigma, R, S)$ be an admissible
  SLG such that $L(G) = \{u\}$. Assume that the sets $\Sigma$,
  $\{\dol_i\}_{i \in [1 \dd |V|]}$, and $\{\hash_i\}_{i \in [1 \dd
  2|V|]}$ are pairwise disjoint.  Denote $\Sigma'' = \Sigma \cup
  \{\dol_i\}_{i \in [1 \dd |V|]} \cup \{\hash_i\}_{i \in [1 \dd
  2|V|]}$. Let $\Gamma$ be a CFG such that $L(\Gamma) \subseteq
  \Sigma^{*}$. There exists a CFG $\Gamma'$ such that $|\Gamma'| =
  \bigO(|\Gamma| + |\Sigma''| + \log |u|)$, and for every $w \in
  \alpha(G)$ (\cref{def:alpha}), $u \in L(\Gamma)$ holds if and only if $w \in
  L(\Gamma')$. Moreover, given $G$ and $\Gamma$, and assuming that
  $|u|$ fits in $\bigO(1)$ machine words, we can construct $\Gamma'$
  in $\bigO(|\Gamma| + |\Sigma| + |G|)$ time.
\end{lemma}
\begin{proof}
  Let $w \in \alpha(G)$ and $k = |w| - 2|u|$. We define
  \begin{itemize}[itemsep=0pt,parsep=0pt]
  \item $L_1 = L(\Gamma)$,
  \item $L_2 = \{a_1 a_2 \cdots a_{2k} \in \Sigma''^{2k} : k \in \Zp
    \text{ and }a_1a_3a_5 \cdots a_{2k-3}a_{2k-1} \in L_1\}$,
  \item $L_3 = \{xy : x \in \Sigma''^k\text{ and }y \in L_2\}$.
  \end{itemize}

  By definition of $L_2$, for every pair of strings $x \in \Sigma^{k}$
  and $y \in \Sigma''^{2k}$ (where $k \in \Zp$) such that $x =
  y[1]y[3] \cdots y[2k-3]y[2k-1]$, $x \in L_1$ holds if and only if $y
  \in L_2$.  By \cref{ob:alpha}, letting $u'$ be the length-$2|u|$
  suffix of $w$, we thus obtain that $u \in L_1$ if and only if $u'
  \in L_2$. On the other hand, by definition of $L_3$ and $k = |w| -
  2|u|$, $u' \in L_2$ holds if and only if $w \in L_3$. We have thus
  proved that $u \in L(\Gamma)$ holds if and only if $w \in L_3$.

  Our next goal is therefore to prove that there exists an CFG
  $\Gamma'$ of the desired size such that $L(\Gamma') = L_3$.  First,
  by \cref{lm:interleave} applied to $\Gamma$ and $|V|$, there exists
  a CFG $\Gamma_2$ satisfying $|\Gamma_2| = \bigO(|\Gamma| +
  |\Sigma''|)$ such that $L(\Gamma_2) = L_2$.  Second, by
  \cref{lm:add-prefix} applied to $\Gamma_2$ and $k$, there exists a
  CFG $\Gamma_3$ satisfying $|\Gamma_3| = \bigO(|\Gamma_2| +
  |\Sigma''| + \log k) = \bigO(|\Gamma| + |\Sigma''| + \log |u|)$ such
  that $L(\Gamma_3) = L_3$. Letting $\Gamma' = \Gamma_3$, we thus
  obtain the sought CFG. Note that, by \cref{def:alpha}, it holds $|w|
  \geq 2|\expgen{G}{S}| + 2 = 2|u| + 2$, and hence $k \geq 1$.  To
  show that $\log k = \bigO(\log |u|)$, first note that since we
  assumed that $G$ is admissible, it follows that $\sum_{X \in
  V}|\expgen{G}{X}| \leq |V| \cdot |u| \leq |u|^2$. By
  \cref{lm:alpha-length}, we thus obtain $k = |w| - 2|u| = 4\sum_{X \in
  V}|\expgen{G}{X}| - 2|u| \leq 4|u|^2$, and hence $\log k =
  \bigO(\log |u|)$.

  We construct $\Gamma'$ as follows. First, applying
  \cref{lm:interleave} to $\Gamma$ and $|V|$, we construct $\Gamma_2$
  in $\bigO(|\Gamma| + |\Sigma''|) = \bigO(|\Gamma| + |G|)$ time.
  Next, we compute $k$. This requires computing $|w| = 4\sum_{X \in
  V}|\expgen{G}{X}|$.  To this end, we first sort the nonterminals
  of $G$ topologically (this is possible, since $G$ is an SLG). For
  each $X \in V$, we then compute $|\expgen{G}{X}|$. This lets us
  determine $|w|$, and hence also $k$.  Note that since we assumed
  that $|u|$ fits in $\bigO(1)$ machine words and by the above
  discussion, we have $|\expgen{G}{X}| \leq |u|$ for every $X \in V$,
  the computation of $k$ takes $\bigO(|G|)$ time.  Once $k$ is
  computed, using \cref{lm:add-prefix} for $\Gamma_2$ and $k$, we
  compute $\Gamma_3 = \Gamma'$ in $\bigO(|\Gamma_2| + |\Sigma''| +
  \log k) = \bigO(|\Gamma| + |\Sigma''| + \log |u|)$ time. In total,
  we thus spend $\bigO(|\Gamma| + |\Sigma''| + |G| + \log |u|) =
  \bigO(|\Gamma| + |\Sigma| + |G|)$ time.
\end{proof}

\begin{theorem}\label{th:cfg-sequential}
  Let \algname{Alg} be either the \algname{Sequitur}, \algname{Sequential},
  or any of the global algorithms (e.g., \algname{RePair},
  \algname{Greedy}, or \algname{LongestMatch}).
  Let $\delta \in (0, 1]$.  Assuming the $k$-Clique Conjecture
  (resp.\ Combinatorial $k$-Clique Conjecture), there is no algorithm
  (resp.\ combinatorial algorithm) that, given $u \in \Sigma^{n}$ and
  a CFG $\Gamma$ such that $|\Gamma| = \bigO(n^{\delta})$, $|\Sigma| =
  \bigO(n^{\delta})$, and $|\algname{Alg}(u)| =
  \bigO(n^{\delta})$, determines if $u \in L(\Gamma)$ in
  $\bigO(n^{\omega-\epsilon})$ $($resp.\ $\bigO(n^{3 - \epsilon})$$)$
  time, for any $\epsilon > 0$.
\end{theorem}
\begin{proof}
  We prove the claim by contraposition. Assume that there exists some
  $\epsilon > 0$ such that for every $u \in \Sigma^{n}$ and a CFG
  $\Gamma$ such that $|\Gamma| = \bigO(n^{\delta})$, $|\Sigma| =
  \bigO(n^{\delta})$, and $|\algname{Alg}(u)| =
  \bigO(n^{\delta})$, we can in $\bigO(n^{\omega - \epsilon})$
  (resp.\ $\bigO(n^{3 - \epsilon})$) time determine if $u \in
  L(\Gamma)$. We will prove that this implies that the $k$-Clique
  Conjecture (resp.\ Combinatorial $k$-Clique Conjecture) is false.

  Assume that we are given an undirected graph $G = (V, E)$. Let $k =
  \max(\lceil \tfrac{3}{\delta} \rceil, \lceil \tfrac{30}{\epsilon}
  \rceil)$.  We execute the following algorithm:
  \begin{enumerate}
  \item First, using \cref{lm:kclique-small-slg} in $\bigO(|V|^3)$
    time we build a CFG $\Gamma$ and an SLG $H = (V_H, \Sigma, R_H,
    S_H)$ such that letting $L(H) = \{u\}$, it holds:
    \begin{itemize}[itemsep=0pt,parsep=0pt]
    \item $|\Sigma| = \bigO(1)$,
    \item $|\Gamma| = \bigO(\log |V|)$,
    \item $|H| = \bigO(|V|^3)$,
    \item $|u| = \Theta(|V|^{k+2})$,
    \item $\sum_{X \in V_H}|\expgen{H}{X}| = \bigO(|V|^{k+5})$, and
    \item $u \in L(\Gamma)$ holds if and only if $G$ has a
      $3k$-clique.
    \end{itemize}

  \item Next, using \cref{lm:admissible}, in $\bigO(|H|) =
    \bigO(|V|^3)$ time we build an admissible SLG $H' = (V_{H'},
    \Sigma, R_{H'}, S_{H'})$ such that
    \begin{itemize}[itemsep=0pt,parsep=0pt]
    \item $L(H') = L(H) = \{u\}$,
    \item $|H'| = \bigO(|H|) = \bigO(|V|^3)$,
    \item $\sum_{X \in V_{H'}}|\expgen{H'}{X}| = \bigO(\log |H| \cdot
      \sum_{X \in V_H}|\expgen{H}{X}|) = \bigO(|V|^{k+5} \log |V|)$.
    \end{itemize}

  \item Next, using \cref{lm:construction-alpha} in $\bigO(\sum_{X \in
    V_{H'}}|\expgen{H'}{X}|) = \bigO(|V|^{k+5} \log |V|)$ time we
    construct $w \in \alpha(H')$ (\cref{def:alpha}). Note that by \cref{lm:alpha-length},
    it holds $|w| = \bigO(\sum_{X \in V_{H'}}|\expgen{H'}{X}|) =
    \bigO(|V|^{k+5} \log |V|)$. Note also that $w \in \Sigma''$, where
    $|\Sigma''| = \bigO(|\Sigma| + |V|) = \bigO(|V|) =
    \bigO(|u|^{\delta/3}) = \bigO(|u|^{\delta})$, where we used that
    $|V| = \Theta(|u|^{1/(k+2)})$ and $3/\delta \leq k$ (which implies
    $1/(k+2) \leq \delta/3$).

  \item Next, applying \cref{lm:alpha-construction-of-gamma-prim} to
    $\Gamma$ and $H'$, in $\bigO(|\Gamma| + |\Sigma| + |H'|) =
    \bigO(|V|^3)$ time we construct a CFG $\Gamma'$ such that:
    \begin{itemize}[itemsep=0pt,parsep=0pt]
    \item $|\Gamma'| = \bigO(|\Gamma| + |\Sigma''| + \log |u|) =
      \bigO(|V|) = \bigO(|u|^{\delta})$,
    \item $u \in L(\Gamma)$ holds if and only if $w \in L(\Gamma')$.
    \end{itemize}
    Combining the above with the earlier observation, we thus have
    that $w \in L(\Gamma')$ holds if and only if $G$ has a
    $3k$-clique. Observe also that by
    \cref{lm:sequitur,lm:sequential,lm:global-output-size},
    it holds $|\algname{Alg}(w)| = \bigO(|H'|) = \bigO(|V|^3) =
    \bigO(|u|^{\delta})$, where the last inequality follows by $|V|^3
    = \Theta(|u|^{3/(k+2)})$ and $3/\delta \leq k$ (which implies
    $3/(k+2) \leq \delta$).

  \item Finally, we apply the hypothetical algorithm for CFG parsing
    to $w$ and $\Gamma'$. More precisely, we check if $w \in
    L(\Gamma')$ in $\bigO(|w|^{\omega-\epsilon}) =
    \bigO(|V|^{(k+5)(\omega-\epsilon)} \log^{\omega-\epsilon} |V|) =
    \bigO(|V|^{k(\omega-\epsilon) + 5\omega}) =\allowbreak
    \bigO(|V|^{\omega k(1 - \epsilon/(2\omega)) - (k\epsilon/2 -
    5\omega)}) = \bigO(|V|^{\omega k(1 - \epsilon/(2\omega))})$
    (resp.\ $\bigO(|w|^{3 - \epsilon}) \,{=}\, \bigO(|V|^{(k+5)(3 -
    \epsilon)} \log^{3-\epsilon} |V|) \allowbreak=
    \bigO(|V|^{k(3-\epsilon) + 15)}) \allowbreak=
    \bigO(|V|^{3k(1 - \epsilon/6) - (k\epsilon/2 - 15)}) =
    \bigO(|V|^{3k(1 - \epsilon/6)})$) time, where in the last
    inequality, we used that $k \geq \lceil \tfrac{30}{\epsilon}
    \rceil$ (which implies $k\epsilon/2 \geq 15$).
  \end{enumerate}
  We have thus checked if $G$ contains a $3k$-clique in
  $\bigO(|V|^{k+5} \log |V| + |V|^{\omega k(1 - \epsilon/(2\omega))})
  = \bigO(|V|^{\omega k(1 - \epsilon')})$ (resp.\
  $\bigO(|V|^{k+5} \log |V| + |V|^{3k(1 - \epsilon/6)})
  = \bigO(|V|^{3k(1 - \epsilon')})$) time, where $\epsilon' > 0$ is
  some constant (note that we used that $k \geq 6$). This implies
  that the $k$-Clique Conjecture (resp.\ Combinatorial $k$-Clique
  Conjecture) is false.
\end{proof}

\subsection{Analysis of \algname{LZD}}\label{sec:cfg-lzd}

\begin{lemma}\label{lm:construction-beta}
  Let $G = (V, \Sigma, R, S)$ be an admissible SLG and assume that
  $6\sum_{X \in V}|\expgen{G}{X}|$ fits into a machine word.  Given
  $G$, we can compute some $w \in \beta(G)$ (\cref{def:beta}) in
  $\bigO(\sum_{X \in V}|\expgen{G}{X}|)$ time.
\end{lemma}
\begin{proof}
  Assume that $G$ is given using an encoding in which nonterminals are
  identified with consecutive positive integers. The construction of
  $w \in \beta(G)$ consist of three steps:
  \begin{enumerate}[itemsep=0pt,parsep=0pt]
  \item First, as in the proof of \cref{lm:construction-alpha}, in
    $\bigO(|G|)$ time we compute $|\expgen{G}{X}|$ for every $X \in V$
    and sort all $X \in V$ using $|\expgen{G}{X}|$ as the key (with
    ties resolved arbitrarily). Let $\{N_i\}_{i \in [1 \dd |V|]}$
    denote the resulting sequence.
  \item Second, we construct $G'$ defined as in \cref{def:beta}. This
    is easily done in $\bigO(|G|)$ time.
  \item Lastly, we compute the output string $w =
    \bigodot_{i=1,\dots,|V|} \expgen{G'}{N_{i,0}}
    \expgen{G'}{N_{i,0}}$.  Using $G'$ this is easily done in
    $\bigO(|w|)$ time, which is $\bigO(\sum_{X \in V}|\expgen{G}{X}|)$
    by \cref{lm:beta-length}\eqref{lm:beta-length-it-2}.  \qedhere
  \end{enumerate}
\end{proof}

\begin{observation}\label{ob:beta}
  Let $G = (V, \Sigma, R, S)$ be an admissible SLG such that $\Sigma
  \cap \{\dol_i\}_{i \in [1 \dd 2|V|]} = \emptyset$. Let $u \in
  \Sigma^{+}$ be such that $L(G) = \{u\}$ and let $w \in \beta(G)$
  (\cref{def:beta}).
  Let $u'$ be a suffix of $w$ of length $3|u| - 2$. Then, erasing all
  occurrences of the symbols from the set $\{\dol_i\}_{i \in [1 \dd
  2|V|]}$ in $u'$ results in the string $u$.
\end{observation}
\begin{proof}
  Let $G'$ be the SLG and $(N_i)_{i \in [1 \dd |V|]}$ be a sequence of
  nonterminals corresponding to $w$ in \cref{def:beta}. The assumption
  about $\mathcal{T}(S)$ implies that for every $X \in V \setminus
  \{S\}$, it holds $|\expgen{G}{X}| < |\expgen{G}{S}|$. Thus, $N_{|V|}
  = S$. By \cref{def:beta}, the string $\expgen{G'}{N_{|V|,0}}$ is
  thus a suffix of $w$. Observe now that by
  \cref{lm:beta-length}\eqref{lm:beta-length-it-1}, it holds
  $|\expgen{G'}{N_{|V|,0}}| = 3|\expgen{G}{N_{|V|}}| - 2 =
  3|\expgen{G}{S}| - 2 = 3|u| - 2$. Consequently, $u' =
  \expgen{G'}{N_{|V|,0}}$.  It remains to observe that by
  \cref{def:beta}, $\expgen{G}{S} = u$ is a subsequence of
  $\expgen{G'}{N_{|V|,0}} = u'$. The claim thus follows by noting that
  all the extra symbols in $\expgen{G'}{N_{|V|,0}}$ are from the set
  $\{\dol_i\}_{i \in [1 \dd 2|V|]}$.
\end{proof}

\begin{lemma}\label{lm:beta-construction-of-gamma-prim}
  Let $u \in \Sigma^{+}$ and $G = (V, \Sigma, R, S)$ be an admissible
  SLG such that $L(G) = \{u\}$ and $\Sigma \cap \{\dol_i\}_{i \in [1
  \dd 2|V|]} = \emptyset$. Denote $\Sigma' = \Sigma \cup
  \{\dol_i\}_{i \in [1 \dd 2|V|]}$. Let $\Gamma$ be a CFG such that
  $L(\Gamma) \subseteq \Sigma^{*}$. There exists a CFG $\Gamma'$ such
  that $|\Gamma'| = \bigO(|\Gamma| + |\Sigma'| + \log |u|)$, and for
  every $w \in \beta(G)$ (\cref{def:beta}),
  $u \in L(\Gamma)$ holds if and only if $w
  \in L(\Gamma')$. Moreover, given $G$ and $\Gamma$, and assuming that
  $|u|$ fits in $\bigO(1)$ machine words, we can construct $\Gamma'$
  in $\bigO(|\Gamma| + |\Sigma| + |G|)$ time.
\end{lemma}
\begin{proof}
  Let $w \in \beta(G)$ and $k = |w| - 3|u| + 2$. We define
  \begin{itemize}[itemsep=0pt,parsep=0pt]
  \item $L_1 = L(\Gamma)$,
  \item $L_2 = \{x \in \Sigma'^{*}: {\rm erase}(x, \{\dol_i\}_{i \in
    [1 \dd 2|V|]}) \in L_1\}$ (where ${\rm erase}$ is defined as in
    \cref{lm:extra-dollars}),
  \item $L_3 = \{xy : x \in \Sigma'^k\text{ and }y \in L_2\}$.
  \end{itemize}

  By definition of $L_2$, for every pair of strings $x \in \Sigma^{*}$
  and $y \in \Sigma'^{*}$ such that ${\rm erase}(y, \{\dol_i\}_{i \in
  [1 \dd 2|V|]}) = x$, $x \in L_1$ if and only if $y \in L_2$.  By
  \cref{ob:beta}, letting $u'$ be suffix of $w$ of length $3|u| - 2$,
  we thus obtain that $u \in L_1$ if and only if $u' \in L_2$. On the
  other hand, by definition of $L_3$ and $k = |w| - 3|u| + 2$, $u' \in
  L_2$ holds if and only if $w \in L_3$. We have thus proved that $u
  \in L(\Gamma)$ holds if and only if $w \in L_3$.

  Our next goal is therefore to prove that there exists an CFG
  $\Gamma'$ of the desired size such that $L(\Gamma') = L_3$.  First,
  by \cref{lm:extra-dollars} applied to $\Gamma$ (recall that $L_1 =
  L(\Gamma)$), there exists a CFG $\Gamma_2$ satisfying $|\Gamma_2| =
  \bigO(|\Gamma| + |\Sigma'|)$ such that $L(\Gamma_2) = L_2$.  Second,
  by \cref{lm:add-prefix} applied to $\Gamma_2$ and $k$, there exists
  a CFG $\Gamma_3$ satisfying $|\Gamma_3| = \bigO(|\Gamma_2| +
  |\Sigma'| + \log k) = \bigO(|\Gamma| + |\Sigma'| + \log |u|)$ such
  that $L(\Gamma_3) = L_3$. Letting $\Gamma' = \Gamma_3$, we thus
  obtain the sought CFG. Note that, by
  \cref{lm:beta-length}\eqref{lm:beta-length-it-2}, it holds $|w| \geq
  6|\expgen{G}{S}| - 4 = 6|u| - 4$, and hence $k \geq 1$.  To show
  that $\log k = \bigO(\log |u|)$, first note that since we assumed
  that $G$ is admissible, it follows that $\sum_{X \in
  V}|\expgen{G}{X}| \leq |V| \cdot |u| \leq |u|^2$. By
  \cref{lm:beta-length}\eqref{lm:beta-length-it-2}, we thus obtain $k
  = |w| - 3|u| + 2 = 6\sum_{X \in V}|\expgen{G}{X}| - 4|V| - 3|u| + 2
  \leq 6|u|^2$, and hence $\log k = \bigO(\log |u|)$.

  We construct $\Gamma'$ as follows. First, applying
  \cref{lm:interleave} to $\Gamma$, we construct $\Gamma_2$ in
  $\bigO(|\Gamma| + |\Sigma'|) = \bigO(|\Gamma| + |\Sigma| + |G|)$
  time. Next, we compute $k$. This requires computing $|w| = 6\sum_{X
  \in V}|\expgen{G}{X}| - 4|V|$.  To this end, we first sort the
  nonterminals of $G$ topologically (this is possible, since $G$ is an
  SLG). For each $X \in V$, we then compute $|\expgen{G}{X}|$. This
  lets us determine $|w|$, and hence also $k$.  Note that since we
  assumed that $|u|$ fits in $\bigO(1)$ machine words and by the above
  discussion, we have $|\expgen{G}{X}| \leq |u|$ for every $X \in V$,
  the computation of $k$ takes $\bigO(|G|)$ time. Once $k$ is
  computed, applying \cref{lm:add-prefix} to $\Gamma_2$ and $k$, we
  compute $\Gamma_3 = \Gamma'$ in $\bigO(|\Gamma_2| + |\Sigma'| + \log
  k) = \bigO(|\Gamma| + |\Sigma'| + \log |u|) = \bigO(|\Gamma| +
  |\Sigma| + |G|)$ time. In total, we thus spend $\bigO(|\Gamma| +
  |\Sigma| + |G|)$ time.
\end{proof}

\begin{theorem}\label{th:cfg-lzd}
  Let $\delta \in (0, 1]$.  Assuming the $k$-Clique Conjecture
  (resp.\ Combinatorial $k$-Clique Conjecture), there is no algorithm
  (resp.\ combinatorial algorithm) that, given $u \in \Sigma^{n}$ and
  a CFG $\Gamma$ such that $|\Gamma| = \bigO(n^{\delta})$, $|\Sigma| =
  \bigO(n^{\delta})$, and $|\algname{LZD}(u)| = \bigO(n^{\delta})$,
  determines if $u \in L(\Gamma)$ in $\bigO(n^{\omega-\epsilon})$
  $($resp.\ $\bigO(n^{3 - \epsilon})$$)$ time, for any $\epsilon > 0$.
\end{theorem}
\begin{proof}
  We prove the claim by contraposition. Assume that there exists some
  $\epsilon > 0$ such that for every $u \in \Sigma^{n}$ and a CFG
  $\Gamma$ such that $|\Gamma| = \bigO(n^{\delta})$, $|\Sigma| =
  \bigO(n^{\delta})$, and $|\algname{LZD}(u)| = \bigO(n^{\delta})$, we
  can in $\bigO(n^{\omega - \epsilon})$ (resp.\ $\bigO(n^{3 -
  \epsilon})$) time determine if $u \in L(\Gamma)$. We will prove
  that this implies that the $k$-Clique Conjecture
  (resp.\ Combinatorial $k$-Clique Conjecture) is false.

  Assume that we are given an undirected graph $G = (V, E)$. Let $k =
  \max(\lceil \tfrac{3}{\delta} \rceil, \lceil \tfrac{30}{\epsilon}
  \rceil)$.  We execute the following algorithm:

  \begin{enumerate}
  \item First, using \cref{lm:kclique-small-slg} in $\bigO(|V|^3)$
    time we build a CFG $\Gamma$ and an SLG $H = (V_H, \Sigma, R_H,
    S_H)$ such that letting $L(H) = \{u\}$, it holds $|\Sigma| =
    \bigO(1)$, $|\Gamma| = \bigO(\log |V|)$, $|H| = \bigO(|V|^3)$,
    $|u| = \Theta(|V|^{k+2})$, $\sum_{X \in V_H}|\expgen{H}{X}|
    = (|V|^{k+5})$, and $u \in L(\Gamma)$ holds if and only if $G$
    has a $3k$-clique.

  \item Next, using \cref{lm:admissible}, in $\bigO(|H|) =
    \bigO(|V|^3)$ time we build an admissible SLG $H' = (V_{H'},
    \Sigma, R_{H'}, S_{H'})$ that satisfies $L(H') = L(H) = \{u\}$,
    $|H'| = \bigO(|H|) = \bigO(|V|^3)$, and $\sum_{X \in
    V_{H'}}|\expgen{H'}{X}| = \bigO(\log |H| \cdot \sum_{X \in
    V_H}|\expgen{H}{X}|) = \bigO(|V|^{k+5} \log |V|)$.

  \item Next, using \cref{lm:construction-beta} in $\bigO(\sum_{X \in
    V_{H'}}|\expgen{H'}{X}|) = \bigO(|V|^{k+5} \log |V|)$ time we
    construct $w \in \beta(H')$ (\cref{def:beta}).
    Note that by \cref{lm:beta-length},
    it holds $|w| = \bigO(\sum_{X \in V_{H'}}|\expgen{H'}{X}|) =
    \bigO(|V|^{k+5} \log |V|)$. Note also that $w \in \Sigma'$, where
    $|\Sigma'| = \bigO(|\Sigma| + |V|) = \bigO(|V|) =
    \bigO(|u|^{\delta/3}) = \bigO(|u|^{\delta})$, where we used that
    $|V| = \Theta(|u|^{1/(k+2)})$ and $3/\delta \leq k$ (which implies
    $1/(k+2) \leq \delta/3$).

  \item Next, applying \cref{lm:beta-construction-of-gamma-prim} to
    $\Gamma$ and $H'$, in $\bigO(|\Gamma| + |\Sigma| + |H'|) =
    \bigO(|V|^3)$ time we construct a CFG $\Gamma'$ such that
    $|\Gamma'| = \bigO(|\Gamma| + |\Sigma'| + \log |u|) = \bigO(|V|) =
    \bigO(|u|^{\delta})$ and $u \in L(\Gamma)$ holds if and only if $w
    \in L(\Gamma')$.  Combining the above with the earlier
    observation, we thus have that $w \in L(\Gamma')$ holds if and
    only if $G$ has a $3k$-clique. Observe also that by
    \cref{lm:lzd}, it holds $|\algname{LZD}(w)| = \bigO(|H'|)
    = \bigO(|V|^3) = \bigO(|u|^{\delta})$, where the last inequality
    follows by $|V|^3 = \Theta(|u|^{3/(k+2)})$ and $3/\delta \leq k$
    (which implies $3/(k+2) \leq \delta$).

  \item Finally, we apply the hypothetical algorithm for CFG parsing
    to $w$ and $\Gamma'$. More precisely, we check if $w \in
    L(\Gamma')$ in $\bigO(|w|^{\omega-\epsilon}) =
    \bigO(|V|^{(k+5)(\omega-\epsilon)} \log^{\omega-\epsilon} |V|) =
    \bigO(|V|^{k(\omega-\epsilon) + 5\omega}) =\allowbreak
    \bigO(|V|^{\omega k(1 - \epsilon/(2\omega)) - (k\epsilon/2 -
    2\omega)}) = \bigO(|V|^{\omega k(1 - \epsilon/(2\omega))})$
    (resp.\ $\bigO(|w|^{3 - \epsilon}) \,{=}\, \bigO(|V|^{(k+5)(3 -
    \epsilon)} \log^{3-\epsilon} |V|) \allowbreak=
    \bigO(|V|^{k(3-\epsilon) + 15)}) \allowbreak= \bigO(|V|^{3k(1 -
    \epsilon/6)})$) time, where in the last inequality we used
    that $k \geq \lceil \tfrac{30}{\epsilon} \rceil$ (which implies
    $k \epsilon/2 \geq 15$).
  \end{enumerate}
  We have thus checked if $G$ contains a $3k$-clique in
  $\bigO(|V|^{k+5} \log |V| + |V|^{\omega k(1 - \epsilon/(2\omega))})
  = \bigO(|V|^{\omega k(1 - \epsilon')})$ (resp.\ 
  $\bigO(|V|^{k+5} \log |V| + |V|^{3k(1 - \epsilon/6)})
  = \bigO(|V|^{3k(1 - \epsilon')})$) time, where $\epsilon' > 0$ is some
  constant (note that we used that $k \geq 6$). This implies
  that the $k$-Clique Conjecture (resp.\ Combinatorial $k$-Clique
  Conjecture) is false.
\end{proof}

\subsection{Analysis of \algname{Bisection}}\label{sec:cfg-bisection}

\begin{definition}\label{def:dyadic-slg}
  An SLG $G = (V, \Sigma, R, S)$ is \emph{dyadic} if it is admissible
  (\cref{def:admissible}) and for every $X \in V$, it holds:
  \begin{itemize}[itemsep=0pt,parsep=0pt]
  \item $|\expgen{G}{A}| = 2^k$ for some $k \in \Zn$,
  \item $|\expgen{G}{A}| \geq |\expgen{G}{B}|$,
  \end{itemize}
  where $A, B \in V \cup \Sigma$ are such that $\rhsgen{G}{X} = AB$.
\end{definition}

\begin{lemma}\label{lm:opt-dyadic}
  Let $u \in \Sigma^{*}$. For every dyadic SLG $G = (V, \Sigma, R, S)$
  such that $L(G) = \{u\}$, it holds $|\algname{Bisection}(u)| \leq
  |G|$.
\end{lemma}
\begin{proof}
  Denote $B = \algname{Bisection}(u)$. Observe, that by definition of
  \algname{Bisection} (\cref{sec:algs-nonglobal}), $B$ is dyadic.
  Next, observe that by \cref{def:dyadic-slg}, for every $X \in V$,
  letting $A, B \in V \cup \Sigma$ be such that $\rhsgen{G}{X} = AB$,
  the length of $\expgen{G}{A}$ depends only on
  $|\expgen{G}{X}|$. This implies that the parse trees of two dyadic
  SLGs encoding the same string have the same shape (the proof follows
  by a simple inductive argument on depth of a node in the parse
  tree). In particular, this holds for $\mathcal{T}(B)$ and
  $\mathcal{T}(G)$. Let $\mathcal{S}$ denote the set of substrings
  corresponding to nodes of the common parse tree of $B$ and
  $G$. Observe that:
  \begin{itemize}[itemsep=0pt,parsep=0pt]
  \item By definition of \algname{Bisection}, we have $|B| =
    |\mathcal{S}|$, since this is precisely how the algorithm
    operates: We first create the set of nonterminals corresponding
    the parse tree, and then add one nonterminal for every distinct
    substrings corresponding to some node in this parse tree.
  \item On the other hand, if for non-leaf nodes $v_1$ and $v_2$ in a
    parse tree of some SLG, their corresponding substrings $u_1$ and
    $u_2$ satisfy $u_1 \neq u_2$, the nonterminals $s(v_1)$ and
    $s(v_2)$ (see \cref{sec:prelim}) must satisfy $s(v_1) \neq
    s(v_2)$.  This implies that $|\mathcal{S}| \leq |G|$.
  \end{itemize}
  Putting everything together we thus obtain $|\algname{Bisection}(u)|
  = |B| \leq |\mathcal{S}| \leq |G|$.
\end{proof}

\begin{observation}\label{obs:kclique-small-slg}
  The SLG $H = (V_H, \Sigma, R_H, S_H)$ in \cref{lm:kclique-small-slg}
  satisfies the following property. For every $X \in V_H$, it holds:
  \begin{itemize}[itemsep=0pt,parsep=0pt]
  \item If $|\expgen{H}{X}| > |V|^{k}$, then $\rhsgen{H}{X} = X_1 X_2
    \cdots X_k$, where $|\expgen{H}{X_i}| = |V|^{k}$ for $i \in [1 \dd
    k]$,
  \item Otherwise, there exists $i \in [0 \dd k]$ such that
    $|\expgen{H}{X}| = |V|^i$ and either $\rhsgen{H}{X} = c \in
    \Sigma$ (when $i = 0$), or $\rhsgen{H}{X} = Y_1 Y_2 \cdots
    Y_{|V|}$, where $|\expgen{H}{Y_i}| = |V|^{i-1}$ holds for $i \in
    [1 \dd |V|]$.
  \end{itemize}
\end{observation}

\begin{lemma}\label{lm:kclique-small-cfg-bisection}
  Let $k \geq 1$ be a constant and let $G = (V, E)$ be an undirected
  graph satisfying $|V| = 2^p$, where $p \geq 2$. Then, it holds
  $|\algname{Bisection}(u)| = \bigO(|V|^3)$, where $u \in \Sigma^{*}$
  and $\Sigma = \{{\tt 0}, {\tt 1}, {\rm \#}, {\rm \$}, x, y, z\}$ are
  defined as in \cref{lm:kclique-small-slg}.
\end{lemma}
\begin{proof}
  The key idea of the proof is to modify the SLG $H$ from
  \cref{lm:kclique-small-slg} into a dyadic SLG
  (\cref{def:dyadic-slg}) of size $\bigO(|V|^3)$ whose language is
  $\{u\}$, and then apply \cref{lm:opt-dyadic}.  The SLG is obtained
  in three steps:
  \begin{enumerate}
  \item Let $H = (V_H, \Sigma, R_H, S_H)$ be the grammar defined in
    \cref{lm:kclique-small-slg}.  By $|u| = \Theta(|V|^{k+2})$ (see
    \cref{lm:kclique-small-slg}) and \cref{obs:kclique-small-slg},
    there exists $s = \Theta(|V|^2)$ such that $\rhsgen{H}{S_H} = X_1
    X_2 \cdots X_s$ and for every $i \in [1 \dd s]$, it holds
    $|\expgen{H}{X_i}| = |V|^{k}$.  Denote $\Sigma_{\rm top} = \{X_1,
    \dots, X_s\}$.  We define $H_{\rm top} = (V_{\rm top}, \Sigma_{\rm
    top}, R_{\rm top}, S_{\rm top})$ to be any dyadic SLG such that
    $L(H_{\rm top}) = \rhsgen{H}{S}$. By the trivial upper bound on
    the size of any admissible SLG, we have $|H_{\rm top}| = \bigO(s)
    = \bigO(|V|^2)$.

  \item Let $H' = (V'_H, \Sigma, R'_H, S_H)$ be the SLG obtained from
    $H$ (defined above) by removing all nonterminals $X \in V_H$
    satisfying $|\expgen{H}{X}| = 1$, and replacing all their
    occurrences in the grammar (i.e., in the right-hand sides of
    productions) with their definitions. This does not increase the
    size of the grammar, i.e., by \cref{lm:kclique-small-slg}, $|H'|
    \leq |H| = \bigO(|V|^3)$.  Denote $V'_H \setminus \{S_H\} =
    \{X_{1,1}, X_{2,1}, \dots, X_{q,1}\}$. Let $\{X_{1,2}, X_{1,3},
    \dots, X_{1,|V|-1}, \dots, X_{q,2}, X_{q,3}, \dots, X_{q,|V|-1}\}$
    be a set of fresh $q \cdot (|V| - 2)$ variables.  We define the
    SLG
    \[
      H_{\rm bottom} = (\{S_H\} \cup V_{\rm bottom}, \Sigma, \{(S_H,
        \rhsgen{H'}{S_H})\} \cup R_{\rm bottom}, S_H),
    \]
    where
    \begin{itemize}[itemsep=0pt,parsep=0pt]
    \item $V_{\rm bottom} = \bigcup_{i=1}^{q} \{X_{i,1}, X_{i,2},
      \dots, X_{i,|V|-1}\}$,
    \item $R_{\rm bottom}$ is the smallest set such that for every $i
      \in [1 \dd q]$, denoting $\rhsgen{H'}{X_{i,1}} = X_{i,|V|}
      X_{i,|V|+1} \cdots X_{i,2|V|-1}$, $R_{\rm bottom}$ contains the
      rule $X_{i,j}\rightarrow X_{i,2j} X_{i,2j+1}$ for every $j \in
      [1 \dd |V|)$. By \cref{obs:kclique-small-slg}, for every $j \in
      [|V| \dd 2|V|)$, it holds $|\expgen{H'}{X_{i,j}}| =
      \tfrac{1}{|V|}|\expgen{H'}{X_{i,1}}|$.  Thus, for all $j \in [1
      \dd |V|)$, $|\expgen{H_{\rm bottom}}{X_{i,2j}}| =
      |\expgen{H_{\rm bottom}}{X_{i,2j+1}}|$, and both lengths are
      powers of two.
    \end{itemize}
    Informally, the above construction replaces every length-$|V|$
    definition with $\Theta(|V|)$ length-$2$ definitions of
    nonterminals arranged into a perfect binary tree of height
    $\Theta(\log |V|)$. Observe that $H_{\rm bottom}$ is almost dyadic
    (the only problem is the nonterminal $S_H$). Note also that since in
    $H'$ the definition of every nonterminal in $\{X_{1,1}, \dots,
    X_{q,1}\}$ has length $|V|$ (\cref{obs:kclique-small-slg}), it
    holds $q = \bigO(|V|^2)$. Thus, $|H_{\rm bottom}| = \bigO(q \cdot
    |V|) = \bigO(|V|^3)$.
  \item We define $H_{\rm dyadic} = (V_{\rm top} \cup V_{\rm bottom},
    \Sigma, R_{\rm top} \cup R_{\rm bottom}, S_{\rm top})$.
  \end{enumerate}
  It follows by the construction of $H_{\rm top}$ and $H_{\rm
  bottom}$, that $|H_{\rm dyadic}| = \bigO(|H_{\rm top}| + |H_{\rm
  bottom}|) = \bigO(|V|^3)$. Moreover, by the above discussion,
  $H_{\rm dyadic}$ is dyadic (\cref{def:dyadic-slg}), and $L(H_{\rm
  dyadic}) = \{u\}$.  Thus, by \cref{lm:opt-dyadic}, we obtain
  $|\algname{Bisection}(u)| \leq |H_{\rm dyadic}| = \bigO(|V|^3)$.
\end{proof}

\begin{theorem}\label{th:cfg-bisection}
  Let $\delta \in (0, 1]$.  Assuming the $k$-Clique Conjecture
  (resp.\ Combinatorial $k$-Clique Conjecture), there is no algorithm
  (resp.\ combinatorial algorithm) that, given $u \in \Sigma^{n}$ and
  a CFG $\Gamma$ such that $|\Gamma| = \bigO(\log n)$, $|\Sigma| =
  \bigO(1)$, and $|\algname{Bisection}(u)| = \bigO(n^{\delta})$,
  determines if $u \in L(\Gamma)$ in $\bigO(n^{\omega-\epsilon})$
  $($resp.\ $\bigO(n^{3 - \epsilon})$$)$ time, for any $\epsilon > 0$.
\end{theorem}
\begin{proof}
  We prove the claim by contraposition. Assume that there exists some
  $\epsilon > 0$ such that for every $u \in \Sigma^{n}$ and a CFG
  $\Gamma$ such that $|\Gamma| = \bigO(\log n)$, $|\Sigma| =
  \bigO(1)$, and $|\algname{Bisection}(u)| = \bigO(n^{\delta})$, we
  can in $\bigO(n^{\omega - \epsilon})$ (resp.\ $\bigO(n^{3 -
  \epsilon})$) time determine if $u \in L(\Gamma)$. We will prove
  that this implies that the $k$-Clique Conjecture
  (resp.\ Combinatorial $k$-Clique Conjecture) is false.

  Assume that we are given an undirected graph $G = (V, E)$.  Let $s =
  \lceil \log |V| \rceil$ and $G' = (V', E)$ be the graph obtained by
  adding $2^s - |V|$ isolated vertices into $G$. Then:
  \begin{itemize}[itemsep=0pt,parsep=0pt]
  \item $|V'| = 2^s$,
  \item For every $k \geq 1$, $G$ has a $3k$-clique if and only if
    $G'$ has a $3k$-clique,
  \item For every constant $t > 0$, $|V'|^t \leq (2|V|)^t \leq 2^t
    |V|^t$.  Thus, $|V'|^t = \Theta(|V|^t)$.
  \end{itemize}
  Let $k = \max(\lceil \tfrac{3}{\delta} \rceil, \lceil
  \tfrac{12}{\epsilon} \rceil)$.  We execute the following algorithm:

  \begin{enumerate}
  \item First, using \cref{lm:kclique-small-slg} in $\bigO(|V'|^3) =
    \bigO(|V|^3)$ time we build a CFG $\Gamma$ and an SLG $H = (V_H,
    \Sigma, R_H, S_H)$ such that letting $L(H) = \{u\}$, it holds
    $|\Sigma| = \bigO(1)$, $|\Gamma| = \bigO(\log |V'|) = \bigO(\log
    |V|) = \bigO(\log |u|)$, $|H| = \bigO(|V'|^3) = \bigO(|V|^3)$,
    $|u| = \Theta(|V'|^{k+2}) = \Theta(|V|^{k+2})$, and $u \in
    L(\Gamma)$ holds if and only if $G'$ has a $3k$-clique. Using $H$,
    we generate $u$ in $\bigO(|V|^{k+2})$ time. Observe that by
    \cref{lm:kclique-small-cfg-bisection}, it holds
    $|\algname{Bisection}(u)| = \bigO(|V'|^3) = \bigO(|V|^3) =
    \bigO(|u|^{\delta})$, where the last inequality follows by $|V|^3
    = \Theta(|u|^{3/(k+2)})$ and $3/\delta \leq k$ (which implies
    $3/(k+2) \leq \delta$).

  \item We apply the hypothetical algorithm for CFG parsing to $u$ and
    $\Gamma$. More precisely, we check if $u \in L(\Gamma)$ in
    $\bigO(|u|^{\omega-\epsilon}) =
    \bigO(|V|^{(k+2)(\omega-\epsilon)}) =
    \bigO(|V|^{k(\omega-\epsilon) + 2\omega}) =\allowbreak
    \bigO(|V|^{\omega k(1 - \epsilon/(2\omega)) - (k\epsilon/2 -
      2\omega)}) = \bigO(|V|^{\omega k(1 - \epsilon/(2\omega))})$
    (resp.\ $\bigO(|u|^{3 - \epsilon}) \,{=}\, \bigO(|V|^{(k+2)(3 -
      \epsilon)}) \allowbreak= \bigO(|V|^{k(3-\epsilon) + 6)})
    \allowbreak= \bigO(|V|^{3k(1 - \epsilon/6)})$) time. Since this
    check is equivalent to checking if $G$ contains a $3k$-clique, we
    thus obtain that the $k$-Clique Conjecture (resp.\ Combinatorial
    $k$-Clique Conjecture) is false.  \qedhere
  \end{enumerate}
\end{proof}

\section{RNA Folding}\label{sec:rna-folding}

\subsection{Problem Definition}\label{sec:rna-problem}

\begin{definition}\label{def:match}
  We say that an alphabet $\Sigma$ is augmented with a \emph{match}
  operation, if for every $a \in \Sigma$, there exists a matching symbol
  $\match{a} \in \Sigma$ such that $\match{a} \neq a$ and
  $\match{\match{a}} = a$.
\end{definition}

\begin{definition}\label{def:crossing}
  For every $(i, j), (i', j') \in \Zp^2$, such that $i < j$ and
  $i' < j'$, we say that $(i,j)$ and $(i',j')$ are \emph{non-crossing}
  if the intervals $[i \dd j]$ and $[i' \dd j']$ are disjoint (i.e., $j
  < i'$ or $j' < i$) or properly nested (i.e., $i < i' < j' < j$ or $i'
  < i < j < j'$).  Otherwise they are \emph{crossing}.
\end{definition}

\setlength{\FrameSep}{1.7ex}
\begin{framed}
  \noindent
  \textsc{RNA Folding}\\
  \textbf{Input}: A string $u \in \Sigma^n$ over an alphabet $\Sigma$
  augmented with a match operation.\\
  \textbf{Output:} The cardinality $|R|$ of the largest set
  $R \subseteq [1 \dd n]^2$ such that:
  \vspace{-1ex}
  \begin{enumerate}[itemsep=0pt,parsep=0pt]
  \item For every $(i,j) \in R$,
    it holds $i < j$ and $\match{u[i]} = u[j]$,
  \item $R$ does not contain crossing pairs.
  \end{enumerate}
  \vspace{-1ex}
  We denote this maximum cardinality by $\RNA(u)$.
\end{framed}

We also define the weighted variant of problem.

\setlength{\FrameSep}{1.7ex}
\begin{framed}
  \noindent
  \textsc{Weighted RNA Folding}\\
  \textbf{Input}: A string $u \in \Sigma^n$ over an alphabet $\Sigma$
  augmented with a match operation and a weight function $w : \Sigma
  \rightarrow \Zp$ such that for every $a \in \Sigma$, it holds
  $w(a) = w(\match{a})$.\\
  \textbf{Output:} The largest value $\sum_{(i,j) \in R} w(u[i])$
  over all $R \subseteq [1 \dd n]^2$ such that:
  \vspace{-1ex}
  \begin{enumerate}[itemsep=0pt,parsep=0pt]
  \item For every $(i,j) \in R$,
    it holds $i < j$ and $\match{u[i]} = u[j]$,
  \item $R$ does not contain crossing pairs.
  \end{enumerate}
  \vspace{-1ex}
  We denote this largest value by $\WRNA(u)$.
\end{framed}

Abboud, Backurs, and Vassilevska Williams proved the
following reduction from the weighted to unweighted variant.

\begin{lemma}[\cite{AbboudBW18}]\label{lm:rna-reduction}
  Let $\Sigma$ be an alphabet augmented with a match operation and let
  $w : \Sigma \rightarrow \Zp$ be a weight function. Then, for every
  $u \in \Sigma^n$, it holds $\WRNA(u) = \RNA(u')$, where
  $u' = u[1]^{w(u[1])} \cdots u[n]^{w(u[n])}$.
\end{lemma}

\subsection{Prior Work}\label{sec:rna-prior}

Similarly as in \cref{sec:cfg-prior}, we first recall the main idea
of the technique for proving conditional lower bound on the runtime
of compressed algorithms for the RNA folding problem developed by
Abboud, Backurs, Bringmann, and K{\"{u}}nnemann.

\begin{theorem}[\cite{AbboudBBK17}]\label{th:rna-main-lower-bound}
  Let $\delta \in (0, 1]$. Assuming the $k$-Clique Conjecture
  (resp.\ Combinatorial $k$-Clique Conjecture), there is no algorithm
  (resp.\ combinatorial algorithm) that, given $u \in \Sigma^{n}$
  satisfying $g^{*}(u) = \bigO(n^{\delta})$ computes $\RNA(u)$ in
  $\bigO(n^{\omega-\epsilon})$ $($resp.\ $\bigO(n^{3 - \epsilon})$$)$
  time, for any $\epsilon > 0$.
\end{theorem}

In other words, assuming the $k$-Clique Conjecture
(resp.\ Combinatorial $k$-Clique Conjecture), the RNA folding for any
length-$n$ string $u$ cannot be performed in
$\bigO(n^{\omega-\epsilon})$ (resp.\ $\bigO(n^{3-\epsilon})$) time,
even restricted to highly compressible $u$. This implies the following
result.

\begin{corollary}[\cite{AbboudBBK17}]
  Assuming the $k$-Clique Conjecture (resp.\ Combinatorial $k$-Clique
  Conjecture), there is no algorithm (resp.\ combinatorial algorithm)
  that, given any SLG $G$ such that $L(G) = \{u\}$
  for some $u \in \Sigma^n$ $($where $|\Sigma| = \bigO(1)$$)$, computes
  $\RNA(u)$ in $\bigO(\poly(|G|) \cdot n^{\omega-\epsilon})$
  $($resp.\ $\bigO(\poly(|G|) \cdot n^{3 - \epsilon})$$)$ time, for
  any $\epsilon > 0$.
\end{corollary}

The key idea in the proof of the above result is as follows.
Suppose that the algorithm in
question exists and runs in $\bigO(|G|^c \cdot n^{\omega - \epsilon})$
(resp.\ $\bigO(|G|^c \cdot n^{3 - \epsilon})$) time, where $\epsilon >
0$ and $c > 0$. Let $\delta = \tfrac{\epsilon}{3c}$. Consider any $u
\in \Sigma^{n}$ such that $|\Sigma| = \bigO(1)$ and $g^{*}(u) =
\bigO(n^{\delta})$. Run the following algorithm:

\begin{enumerate}
\item First, using any grammar compression algorithm with an
  approximation ratio of $\bigO(\log n)$ (such
  as~\cite{Rytter03,Charikar05,Jez16}), compute an SLG $G$ such that
  $L(G) = \{u\}$ and $|G| = \bigO(g^{*}(u) \log n) =
  \bigO(n^{\delta} \log n)$. Using for example~\cite{Jez16}, this
  takes $\bigO(n)$ time.
\item Second, run the above hypothetical algorithm algorithm for RNA
  folding. It takes $\bigO(|G|^c \cdot n^{\omega-\epsilon}) =
  \bigO((n^{\delta} \log n)^c \cdot n^{\omega-\epsilon}) =
  \bigO(n^{\omega - 2\epsilon/3} \log^{c} n) = \bigO(n^{\omega -
  \epsilon/3})$ (resp. $\bigO(|G|^{c} \cdot n^{3-\epsilon}) =
  \bigO(n^{3-\epsilon/3})$) time. We have thus obtained a fast RNA
  folding algorithm for highly compressible strings. By
  \cref{th:rna-main-lower-bound}, this implies that the $k$-Clique
  Conjecture (resp.\ Combinatorial $k$-Clique Conjecture) is false.
\end{enumerate}

\subsection{The Main Challenge}\label{sec:rna-challenge}

The above result shows that it is hard to solve RNA parsing in
compressed time assuming the input has been constructed using a
grammar compression algorithm with a small approximation factor.

Similarly as in \cref{sec:cfg-challenge}, this raises the question
about the role of the approximation ratio in the hardness. Let us thus
again redo the above analysis for an algorithm with the approximation
factor as a parameter.

\begin{lemma}[\cite{AbboudBBK17}]\label{lm:rna-small-slg}
  Let $k \geq 1$ be a constant. For every undirected graph $G = (V,
  E)$, there exists a string $u \in \Sigma^{*}$ $($where $|\Sigma| =
  48$$)$ satisfying $|u| = \bigO(|V|^{k+2})$ and $|u| =
  \Omega(|V|^{k})$, a weight function $w : \Sigma \rightarrow [1 \dd
  m]$ $($where $m = \bigO(|V|^2)$$)$, and an integer $\lambda \geq
  0$ such that $\WRNA(u) \geq \lambda$ holds if and only if $G$ has a
  $3k$-clique. Moreover, given $G$, in $\bigO(|V|^3)$ time we can
  compute the weight function $w$, the integer $\lambda$, and an SLG
  $H = (V_H, \Sigma, R_H, S_H)$ such that $L(H) = \{u\}$, $|H| =
  \bigO(|V|^3)$, and $\sum_{X \in V_H}|\expgen{H}{X}| =
  \bigO(|V|^{k+5})$.
\end{lemma}

Using the above lemma, we can prove the following result.

\begin{theorem}[Based
    on~\cite{AbboudBBK17}]\label{th:rna-lower-bound-simple-generalization}
  Consider a grammar compression algorithm \algname{Alg} with an
  approximation ratio $\bigO(n^{\alpha})$ (where $\alpha \in (0, 1)$).
  Let $\delta \in (\alpha, 1]$.  Assuming the $k$-Clique Conjecture
  (resp.\ Combinatorial $k$-Clique Conjecture), there is no algorithm
  (resp.\ combinatorial algorithm) that, given $u \in \Sigma^{n}$
  satisfying $|\Sigma| = \bigO(1)$, and $|\algname{Alg}(u)| =
  \bigO(n^{\delta})$, computes $\RNA(u)$ in
  $\bigO(n^{\omega-\epsilon})$ $($resp.\ $\bigO(n^{3 - \epsilon})$$)$
  time, for any $\epsilon > 0$.
\end{theorem}
\begin{proof}
  Suppose the above theorem is not true, and let $\epsilon > 0$ be
  such that the algorithm in question runs in
  $\bigO(n^{\omega-\epsilon})$ (resp.\ $\bigO(n^{3-\epsilon})$)
  time. Assume that we are given an undirected graph $G = (V, E)$. Let
  $k = \max(\lceil \tfrac{3}{\delta - \alpha} \rceil + 1, \lceil
  \tfrac{24}{\epsilon} \rceil)$.  We execute the following algorithm:
  \begin{enumerate}
  \item Using \cref{lm:rna-small-slg}, in $\bigO(|V|^3)$ time we
    compute a weight function $w : \Sigma \rightarrow [1 \dd m]$
    (where $m = \bigO(|V|^2)$), an integer $\lambda$, and an SLG $H =
    (V_{H}, \Sigma, R_{H}, S_{H})$ such that letting $L(H) = \{u\}$,
    it holds $|\Sigma| = \bigO(1)$, $|H| = \bigO(|V|^3)$, $|u| =
    \bigO(|V|^{k+2})$, $|u| = \Omega(|V|^{k})$, and $\WRNA(u) \geq
    \lambda$ holds if and only if $G$ has a $3k$-clique. Using $H$, we
    generate $u$ in $\bigO(|V|^{k+2})$ time. Using $u$ and $w$, we
    then compute the string $u'$ defined by $u' = u[1]^{w(u[1])}
    \cdots u[|u|]^{w(u[|u|])}$.  By $m = \bigO(|V|^2)$, we obtain
    $|u'| = \bigO(|u| \cdot |V|^2) = \bigO(|V|^{k+4})$. The
    construction of $u'$ is easily done in $\bigO(|u'|) =
    \bigO(|V|^{k+4})$ time.  Note that by \cref{lm:rna-reduction}, it
    holds $\RNA(u') = \WRNA(u)$.  Thus, $\RNA(u') \geq \lambda$ holds
    if and only if $G$ has a $3k$-clique.  Observe now that the
    existence of $H$ implies that $g^{*}(u) = \bigO(|V|^3)$. Since a
    unary string of length $k$ has an SLG of size $\bigO(\log k)$, we
    thus have $g^{*}(u') = \bigO(|V|^3 \log |V|)$.  Consequently,
    $|\algname{Alg}(u')| = \bigO(g^{*}(u') \cdot |u'|^{\alpha}) =
    \bigO(|V|^3 \log |V| \cdot |u'|^{\alpha}) = \bigO(|u'|^{3/k +
      \alpha} \log |V|) = \bigO(|u|^{\delta})$, where we used that
    $|u| = \Omega(|V|^k)$ (and hence $|u'| = \Omega(|V|^k)$) and
    $\tfrac{3}{\delta - \alpha} < k$ (which implies $3/k + \alpha <
    \delta$).

  \item We apply the above hypothetical RNA folding algorithm to $u'$.
    More precisely, we compute $\RNA(u')$ in
    $\bigO(|u'|^{\omega-\epsilon}) \,{=}\,
    \bigO(|V|^{(k+4)(\omega-\epsilon)}) =
    \bigO(|V|^{k(\omega-\epsilon) + 4\omega}) =\allowbreak
    \bigO(|V|^{\omega k(1 - \epsilon/(2\omega)) - (k\epsilon/2 -
    4\omega)}) = \bigO(|V|^{\omega k(1 - \epsilon/(2\omega))})$
    (resp.\ $\bigO(|u'|^{3 - \epsilon}) \,{=}\, \bigO(|V|^{(k+4)(3 -
    \epsilon)}) \allowbreak= \bigO(|V|^{k(3-\epsilon) + 12)})
    \allowbreak= \bigO(|V|^{3k(1 - \epsilon/6)})$) time, where in the
    last equality we used that $k \geq \lceil \tfrac{24}{\epsilon}
    \rceil$ (which implies $k\epsilon/2 \geq 12$).
  \end{enumerate}
  We have thus checked if $G$ contains a $3k$-clique in
  $\bigO(|V|^{k+4} + |V|^{\omega k(1 - \epsilon/(2\omega))}) =
  \bigO(|V|^{\omega k(1 - \epsilon')})$ (resp.\ $\bigO(|V|^{k+4} +
  |V|^{3k(1 - \epsilon/6)}) = \bigO(|V|^{3k(1 - \epsilon')})$) time,
  where $\epsilon' > 0$ is some constant (and we used that $k \geq
  5$). This implies that the $k$-Clique Conjecture
  (resp.\ Combinatorial $k$-Clique Conjecture) is false.
\end{proof}

\begin{corollary}[Based
    on~\cite{AbboudBBK17}]\label{cor:rna-lower-bound-simple-generalization}
  Consider a grammar compression algorithm \algname{Alg} that runs in
  $\bigO(n^q)$ time, where $q < \omega$ (resp.\ $q < 3$) is a constant,
  and has an approximation ratio $\bigO(n^{\alpha})$ (for a constant
  $\alpha \in (0, 1)$). Assuming the $k$-Clique Conjecture (resp.\
  Combinatorial $k$-Clique Conjecture), there is no algorithm (resp.\
  combinatorial algorithm) that, given the SLG $G = \algname{Alg}(u)$ of
  a string $u \in \Sigma^{n}$ $($where $|\Sigma| = \bigO(1)$$)$,
  computes $\RNA(u)$ in
  $\bigO(|G|^c \cdot n^{\omega-\epsilon})$ $($resp.\ $\bigO(|G|^c
  \cdot n^{3 - \epsilon})$$)$ time, for any constants $\epsilon > 0$
  and $c \in (0, \epsilon/\alpha)$.
\end{corollary}
\begin{proof}
  Suppose that such algorithm exists. Let $\delta = \alpha/2 +
  \epsilon/(2c)$.  Observe that $c \in (0, \epsilon/\alpha)$ implies
  $\epsilon/(2c) > \alpha/2$.  Thus, $\delta > \alpha$. On the other
  hand, we also have $c \delta = c\alpha/2 + \epsilon/2 < \epsilon$.
  Suppose now that we are given a string $u \in \Sigma$ such that
  $|\Sigma| = \bigO(1)$, and $|\algname{Alg}(u)| = \bigO(n^{\delta})$.
  We execute the following algorithm:
  \begin{enumerate}
  \item Compute $G = \algname{Alg}(u)$. This takes
    $\bigO(n^q)$ time.  By the assumption, we have $|G| =
    \bigO(n^{\delta})$.
  \item Apply the above hypothetical algorithm to $G$, i.e.,
    we compute $\RNA(u)$ in $\bigO(|G|^c n^{\omega-\epsilon}) =
    \bigO(n^{\delta c} n^{\omega-\epsilon}) =
    \bigO(n^{\omega - (\epsilon - \delta c)})$ (resp.\ $\bigO(|G|^c
    n^{3-\epsilon}) = \bigO(n^{\delta c} n^{3-\epsilon}) = \bigO(n^{3
    - (\epsilon - \delta c)})$) time.
  \end{enumerate}
  We have thus computed $\RNA(u)$ in $\bigO(n^q + n^{\omega -
  (\epsilon - \delta c)}) = \bigO(n^{\omega - \epsilon'})$
  (resp.\ $\bigO(n^q + n^{3 - (\epsilon - \delta c)}) = \bigO(n^{3 -
  \epsilon'})$) time, where $\epsilon' > 0$ is some constant.  By
  \cref{th:rna-lower-bound-simple-generalization}, this implies that
  the $k$-Clique Conjecture (resp.\ Combinatorial $k$-Clique
  Conjecture) is false. Note that
  \cref{th:rna-lower-bound-simple-generalization} requires that
  $\delta > \alpha$, which we showed above.
\end{proof}

\begin{remark}
  Similarly as in \cref{sec:cfg-parsing}, the above result thus shows
  that applying the techniques from~\cite{AbboudBBK17} prevents some
  but not all compressed algorithms due to the dependence of the lower
  bound on the approximation ratio. In the following sections, we
  describe methods for eliminating this dependence from the analysis.
\end{remark}

\subsection{Preliminaries}\label{sec:rna-prelim}

\begin{lemma}\label{lm:decomposition}
  Let $w : \Sigma \rightarrow \Zp$ be a weight function and let $u = x a
  y b z$, where $x, y, z \in \Sigma^{*}$ and $a, b \in \Sigma$. Assume
  $\match{a} = b$, and that $a$ and $b$ do not occur in $x$, $y$, and
  $z$. Then, $w(a) > \sum_{i=1}^{|y|} w(y[i])$ implies $\WRNA(u) =
  w(a) + \WRNA(xz) + \WRNA(y)$.
\end{lemma}
\begin{proof}
  In the proof of this lemma along with the next two lemmas, assume $w(X)$ for some $X \subset [1 \dd |u|]^2$ is defined to be $\sum_{(i,j) \in X}w(u[i])$.
  
  Let $R \subset [1 \dd |u|]^2$ be a an optimal solution to Weighted RNA Folding with value $\WRNA(u)$. We claim that there is no matching pair $(i,j) \in R, i < j$ which goes between $x$ and $y$ or $y$ and $z$, i.e. $|x| + 1 < j \leq |x| + 1 + |y| \implies i > |x| + 1 $, similarly $|x| + 1 < i \leq |x| + 1 + |y| \implies j \leq |x| + 1 + |y|$. This follows from the fact that any solution which has such a crossing pair between $x$ and $y$ or $y$ and $z$ prevents the matching of characters $a$ and $b$ with $i = |x| + 1, j = |x| + |y| + 2$ from appearing in in $R$. Let $R'$ be another solution obtained by removing all such $x,y$ and $y,z$ crossing pairs from $R$ and putting in $(i = |x| + 1, j = |x| + |y| + 2)$ in $R'$. Since $R' \setminus R$ consists of pairs with exactly one endpoint in $y$, $w(R' \setminus R) \leq \sum_{i=1}^{|y|} w(y[i])$, and $R' \setminus R = \{(i = |x| + 1, j = |x| + |y| + 2)\}$. Since $w(a) > \sum_{i=1}^{|y|} w(y[i])$, this follows that $w(R' \setminus R) - w(R \setminus R') \geq w(a) - \sum_{i=1}^{|y|} w(y[i]) > 0$, hence $w(R') > w(R)$ contradicting the optimality of $R$.
  
  Now, assuming there are no crossing pairs between $x$ and $y$ or $y$ and $z$, any optimal solution $R$ to weighted RNA folding on $u$ must satisfy the following properties for any pair $(i,j) \in R$ with $i<j$:
  
  \begin{enumerate}[itemsep=0pt,parsep=0pt]
  	\item $1 \leq i < j \leq |x|$ (matching within $x$),
  	\item $1 + |x| < i < j \leq 1 + |x| + |y|$ (matching within $y$),
  	\item $|x| + |y| + 2 < i < j \leq |u|$ (matching within $z$),
  	\item $i = |x| + 1, j = |x| + |y| + 2$ ($a$,$b$ matching),
  	\item $1 \leq i \leq |x|$ and $|x| + |y| + 2 < j \leq |u|$ (matching from $x$ to $y$).
  \end{enumerate}
	Therefore, we can partition the pairs of any optimal solution on $u$ into an RNA folding of $xz$ and of $y$ and the pair ($a,b$).
	Hence we obtain $\WRNA(u) = \WRNA(xz) + \WRNA(y) + w(a)$.

\end{proof}

\begin{lemma}\label{lm:rev}
  For every string $u$, it holds $\WRNA(u) = \WRNA(\revstr{u})$.
\end{lemma}
\begin{proof}
  Let $R'$ be an optimal solution to Weighted RNA folding on $u$, $w(R') = \WRNA(u)$. Let $R'' = \{(i,j) | (j,i) \in R'\}$. Then $w(R'') = w(R')$ and $R''$ is a valid RNA folding solution of $u^R$ (this follows from the fact that if $(i,j)$ and $(i',j')$ are illegal crossing pairs in $R''$, then $(j,i)$ and $(j',i')$ are illegal crossing pairs in $R'$). Hence $\WRNA(u^R) \geq \WRNA(u)$. Similarly, applying the above argument to $u^R$ and $(u^R)^R = u$, we obtain $\WRNA(u) \geq \WRNA(u^R)$. Hence we conclude $\WRNA(u) = \WRNA(u^R)$. 
\end{proof}

\begin{lemma}\label{lm:match}
  For every string $u$, it holds $\WRNA(u) = \WRNA(u')$, where
  $u' = \match{u[1]} \cdots \match{u[n]}$.
\end{lemma}
\begin{proof}
  Let $R$ be an optimal solution to Weighted RNA folding on $u$, $w(R) = \WRNA(u)$, then $R$ is also a valid solution for RNA folding on $u'$. This follows from the fact that if $(i,j) \in R$, then $\match{u_i} = u_j$ and $\match{u_j} = u_i$, hence $u'_i = \match{u_i} = u_j$ and $u'_j = \match{u_j} = u_i$, therefore $\match{u'_i} = \match{u_j} = u_i = u'_j$. Hence $\WRNA(u') \geq \WRNA(u)$. But since for any $a \in \Sigma, \match{\match{a}} = a$, we must have $(u')' = u$, hence applying the above argument to $u'$, we obtain $\WRNA(u) = \WRNA((u')') \geq \WRNA(u')$. Therefore, we conclude $\WRNA(u) = \WRNA(u')$ 
\end{proof}

\subsection{Analysis of Global Algorithms}\label{sec:rna-global}

\begin{definition}\label{def:rna-alpha}
  Let $G = (V, \Sigma, R, S)$ be an admissible SLG and assume that
  $\Sigma$ is augmented with a match operation (\cref{def:match}).
  Let also $w : \Sigma \rightarrow \Zp$ be a weight function.
  Let $G' = (V, \Sigma, R', S)$ be an SLG defined by replacing every
  occurrence of $a \in \Sigma$ on the right-hand of side of $G$ with
  $\match{a}$, for every $a \in \Sigma$. Assume that the sets $\Sigma$,
  $\{\dol_i\}_{i \in [1 \dd |V|]}$,
  $\{\dol'_i\}_{i \in [1 \dd |V|]}$,
  $\{\hash_i\}_{i \in [1 \dd 2|V|]}$, and
  $\{\hash'_i\}_{i \in [1 \dd 2|V|]}$ are
  pairwise disjoint. Denote
  \begin{itemize}[itemsep=0pt,parsep=0pt]
  \item $\Sigma_{\rm aux} = \Sigma \cup \{\dol_i\}_{i \in [1 \dd |V|]}$,
  \item $\Sigma'_{\rm aux} = \Sigma \cup \{\dol'_i\}_{i \in [1 \dd |V|]}$,
  \item $\Sigma_{\rm all} =
    \Sigma \cup
    \{\dol_i, \dol'_i\}_{i \in [1 \dd |V|]} \cup
    \{\hash_i, \hash'_i\}_{i \in [1 \dd 2|V|]}$.
  \end{itemize}
  By $\alpha(G)$, we denote the subset of $\Sigma_{\rm all}^{*}$ such
  that for every $v \in \Sigma_{\rm all}^{*}$, $v \in \alpha(G)$ holds
  if and only if there exists a sequence $(N_i)_{i \in [1 \dd |V|]}$
  such that:
  \begin{itemize}[itemsep=0pt,parsep=0pt]
    \item $\{N_i : i \in [1 \dd |V|]\} = V$,
    \item $|\expgen{G}{N_i}| \leq |\expgen{G}{N_{i+1}}|$ holds for $i
      \in [1 \dd |V|)$, and
    \item $v = v'_{\rm aux} \cdot v_{\rm aux}$, where
    \begin{align*}
      v_{\rm aux}
        &= \textstyle\bigodot_{i=1,\dots,|V|}
           \expgen{G_{\rm aux}}{N_i} \cdot
           \hash_{2i-1} \cdot
           \expgen{G_{\rm aux}}{N_i} \cdot
           \hash_{2i},\\
      v'_{\rm aux}
        &= \textstyle\bigodot_{i=|V|,\dots,i}
           \hash'_{2i} \cdot
           \expgen{G'_{\rm aux}}{N_i} \cdot
           \hash'_{2i-1} \cdot
          \expgen{G'_{\rm aux}}{N_i},
    \end{align*}
  \end{itemize}
  where $G_{\rm aux} = (V, \Sigma_{\rm aux}, R_{\rm aux}, S)$
  (resp.\ $G'_{\rm aux} = (V, \Sigma'_{\rm aux}, R'_{\rm aux}, S)$) is
  defined so that for every $i \in [1 \dd |V|]$, it holds
  $\rhsgen{G_{\rm aux}}{N_i} = A \cdot \dol_{i} \cdot B$
  (resp.\ $\rhsgen{G'_{\rm aux}}{N_i} = B \cdot \dol'_{i} \cdot A$),
  where $A, B \in V \cup \Sigma$ are such that $\rhsgen{G}{N_i} = AB$
  (resp.\ $\rhsgen{G'}{N_i} = AB$).

  We also extend the match operation from $\Sigma$ to $\Sigma_{\rm all}$
  so that:
  \begin{itemize}[itemsep=0pt,parsep=0pt]
  \item For every $i \in [1 \dd |V|]$, $\match{\dol_i} = \dol'_i$,
  \item For every $i \in [1 \dd 2|V| - 2]$, $\match{\hash_i} =
    \hash'_i$,
  \item $\match{\hash_{2|V|-1}} = \hash_{2|V|}$ and
    $\match{\hash'_{2|V|-1}} = \hash'_{2|V|}$.
  \end{itemize}

  Finally, we extend the weight function $w$ from $\Sigma$ to
  $\Sigma_{\rm all}$, so that:
  \begin{itemize}[itemsep=0pt,parsep=0pt]
  \item For every $i \in [1 \dd |V|]$, $w(\dol_i) = w(\dol'_i) = 1$,
  \item For every $i \in [1 \dd 2|V|]$, $w(\hash_i) = w(\hash'_i) =
    (2|u| - 1) + 2\sum_{j \in [1 \dd |u|]} w(u[j])$, where
    $u = \expgen{G}{S}$.
  \end{itemize}
\end{definition}

\begin{lemma}\label{lm:global-concat}
  Let $x, y \in \Sigma^{+}$. Assume that there exist integers $m_x$
  and $m_y$ such that for every global algorithm \algname{Alg}, it
  holds $|\algname{Alg}(x)| \leq m_x$ and $|\algname{Alg}(y)| \leq
  m_y$.  Assume also that if some $s \in \Sigma^{+}$ satisfying
  $|s| \geq 2$ has two
  nonoverlapping occurrences in $xy$, then either both
  these occurrences are contained in $x$, or both are contained in
  $y$. Then, for every global algorithm
  \algname{Alg}, it holds $|\algname{Alg}(xy)| \leq m_x + m_y$.
\end{lemma}
\begin{proof}
  First, observe that by definition of a global algorithm (see
  \cref{sec:algs-global}), for every global algorithm \algname{Alg}
  running on a string $u$, there exists a finite sequence $(G_0, G_1,
  \dots, G_k)$ of SLGs such that:
  \begin{itemize}[itemsep=0pt,parsep=0pt]
  \item $G_0$ has only a single nonterminal whose definition is the
    string $u$,
  \item For every $i > 0$, $G_i$ is obtained from $G_{i-1}$ by first
    selecting a maximal string $s$ (\cref{def:maximal-string}) with
    respect to $G_{i-1}$, and then replacing all occurrences of $s$ on
    the right-hand side of $G_{i-1}$ (by scanning the definition of
    every nonterminal left-to-right) with a new nonterminal $N$, and
    then adding the nonterminals $N$ into $G_{i-1}$ (with $s$ as its
    definition),
  \item There is no maximal string with respect to $G_k$.
  \end{itemize}
  Note that the converse also holds, i.e., every sequence satisfying
  the above three conditions corresponds to the execution of some
  global algorithm \algname{Alg} on $u$.

  Let $(G_0^{xy}, G_1^{xy}, \ldots, G_k^{xy})$ be the sequence
  corresponding to the execution of a global algorithm \algname{Alg}
  on $xy$. For every $i \in [0 \dd k]$, denote $G_i^{xy} = (V_i^{xy},
  \Sigma, R_i^{xy}, S)$. Observe that for every $i \in [0 \dd k]$,
  \begin{itemize}[itemsep=0pt,parsep=0pt]
  \item We can write $V_i^{xy} \setminus \{S\} = V_i^{x} \cup
    V_i^{y}$, where $V_i^{x} \cap V_i^{y} = \emptyset$, and for every
    $N \in V_i^{x}$ (resp.\ $N \in V_i^{y}$), $\expgen{G_i^{xy}}{N}$
    has length at least two, has at least two disjoint occurrences in
    $x$ (resp.\ $y$), and no occurrences in $y$ (resp.\ $x$). To see
    this, note that once $N$ is created, it by definition has two
    occurrences on the right-hand side of the current grammar. This
    corresponds to two disjoint occurrences of its expansion (which is
    of length at least two) in the string $xy$.  On the other hand,
    note that once a nonterminal is created by a global algorithm, it
    is never deleted and its expansion never changes (though its
    definition can). Consequently, by the assumption about $x$ and $y$
    in the claim, it follows that the expansion of every nonterminal
    (except $S$) must have occurrences only in either $x$ or $y$.
  \item We can write $\rhsgen{G_i^{xy}}{S} = s_i^{x} \cdot s_i^{y}$,
    so that $\expgen{G_i^{xy}}{s_i^{x}} = x$ and
    $\expgen{G_i^{xy}}{s_i^{y}} = y$. This is because every maximal
    string selected during the execution of \algname{Alg} occurs at
    least twice in $xy$ and is of length at least two.
    Thus, a maximal string that falsified the existence of the
    above partition for the first time would contradict the assumption
    from the claim.
  \end{itemize}

  For every $i \in [0 \dd k]$, let $G_i^{x} = (V_i^{x} \cup \{S\},
  \Sigma, R_i^{x}, S)$ (resp.\ $G_i^{y} = (V_i^{y} \cup \{S\}, \Sigma,
  R_i^{y}, S)$) be an SLG defined such that:
  \begin{itemize}[itemsep=0pt,parsep=0pt]
  \item $\rhsgen{G_i^{x}}{S} = s_i^{x}$ (resp.\ $\rhsgen{G_i^{y}}{S} =
    s_i^{y}$),
  \item For every $N \in V_i^{x}$ (resp.\ $N \in V_i^{y}$),
    $\rhsgen{G_i^{x}}{N} = \rhsgen{G_i^{xy}}{N}$
    (resp.\ $\rhsgen{G_i^{y}}{N} = \rhsgen{G_i^{xy}}{N}$).
  \end{itemize}

  Observe that by the above discussion, for every $i \in [0 \dd k]$,
  it holds $|G_i^{xy}| = |G_i^{x}| + |G_i^{y}|$. On the other hand,
  note that removing duplicates from each sequence $(G_0^{x}, \dots,
  G_k^{x})$ (resp.\ $(G_0^{y}, \dots, G_k^{y})$) results in a valid
  sequence of SLGs corresponding to an execution of some global
  algorithm on $x$ (resp.\ $y$).  In particular, $|G_k^{x}| \leq m_x$
  and $|G_k^{y}| \leq m_y$, and hence:
  \[
    |\algname{Alg}(xy)| = |G_k^{xy}| = |G_k^{x}| + |G_k^{y}| \leq m_x
    + m_y.
    \qedhere
  \]
\end{proof}

\begin{lemma}\label{lm:rna-length-alpha}
  Let $G = (V, \Sigma, R, S)$
  be an admissible SLG. Assume that $\Sigma$ is augmented with a match
  operation (\cref{def:match}).  For every $v \in \alpha(G)$
  (\cref{def:rna-alpha}), it holds $|v| = 8\sum_{X \in V}|\expgen{G}{X}|$.
\end{lemma}
\begin{proof}
  Letting $v_{\rm aux}$ and $v'_{\rm aux}$ be as in \cref{def:rna-alpha}
  (i.e., so that $v = v'_{\rm aux} \cdot v_{\rm aux}$), it follows
  by \cref{lm:alpha-length}, that $|v_{\rm aux}| = |v'_{\rm aux}| =
  4\sum_{X \in V}|\expgen{G}{X}|$, and hence we have $|v| = |v_{\rm aux}|
  + |v'_{\rm aux}| = 8\sum_{X \in V}|\expgen{G}{X}|$.
\end{proof}

\begin{lemma}\label{lm:rna-construction-alpha}
  Let $G = (V, \Sigma, R, S)$ be an admissible SLG. Assume that
  $\Sigma$ is augmented with a match operation (\cref{def:match}) and
  that $8\sum_{X \in V}|\expgen{G}{X}|$ fits into a machine word.
  Let also $w : \Sigma \rightarrow \Zp$ be a weight function.
  Given $G$, we can compute some $v \in \alpha(G)$ and augment the
  weight function $w$ as described by \cref{def:rna-alpha} in
  $\bigO(\sum_{X \in V}|\expgen{G}{X}|)$ time.
\end{lemma}
\begin{proof}
  Assume that $G$ is given using an encoding in which nonterminals are
  identified with consecutive positive integers. Assume also that for
  every $c \in \Sigma$, we can compute $\match{c}$ in $\bigO(1)$ time.
  The construction of $v \in \alpha(G)$ proceeds as follows:
  \begin{enumerate}[itemsep=0pt,parsep=0pt]
  \item Construct the grammar $G'$ as in \cref{def:rna-alpha}.  Given
    $G$, this is easily done in $\bigO(|G|)$ time.
  \item In $\bigO(|G|)$ time we sort the nonterminals of the implicit
    grammar DAG of $G$ (defined so that there is an edge connecting
    nonterminals $X$ and $Y$ when $Y$ appears in $\rhsgen{G}{X}$)
    topologically. In $\bigO(|G|)$ time we then compute
    $|\expgen{G}{X}|$ for every $X \in V$. We then sort all $X \in V$
    using $|\expgen{G}{X}|$ as the key (with ties resolved
    arbitrarily).  Using radix sort, this can be done in
    $\bigO(\sum_{X \in V}|\expgen{G}{X}|)$ time. Let $\{N_i\}_{i \in
    [1 \dd |V|]}$ denote the resulting sequence.
  \item Construct $G_{\rm aux}$ and $G'_{\rm aux}$ as in
    \cref{def:rna-alpha}. This is easily done in $\bigO(|G|)$ time.
  \item Compute $v_{\rm aux} = \bigodot_{i=1,\dots,|V|} \expgen{G_{\rm
    aux}}{N_i} \cdot \hash_{2i-1} \cdot \expgen{G_{\rm aux}}{N_i}
    \cdot \hash_{2i}$ and $v'_{\rm aux}$ (defined symmetrically; see
    \cref{def:rna-alpha}).  Using $G_{\rm aux}$ and $G'_{\rm aux}$
    this is easily done in $\bigO(|v_{\rm aux}| + |v'_{\rm aux}|)$
    time, which is $\bigO(\sum_{X \in V}|\expgen{G}{X}|)$ by
    \cref{lm:rna-length-alpha}.
  \item Output $v = v'_{\rm aux} \cdot v_{\rm aux}$.
  \end{enumerate}
  In total, the construction takes $\bigO(\sum_{X \in
  V}|\expgen{G}{X}|)$ time.

  To augment $w$ as in \cref{def:rna-alpha}, for every $X \in V$, we
  compute $|x|$ and $\sum_{i \in [1 \dd |x|]} w(x[i])$, where $x =
  \expgen{G}{X}$. Given the ordering of nonterminals $X \in V$ by
  $|\expgen{G}{X}|$ (computed above), this is easily done in
  $\bigO(|G|)$ time. We then compute $(2|\expgen{G}{S}| - 1) +
  2\sum_{i \in [1 \dd |u|]}w(u[i])$, which is the value needed to
  augment $w$ as in \cref{def:rna-alpha}.
\end{proof}

\begin{lemma}\label{lm:rna-boosting-global}
  Let \algname{Alg} be a global algorithm.  Let $G = (V, \Sigma, R,
  S)$ be an admissible SLG. Assume that $\Sigma$ is augmented with a
  match operation (\cref{def:match}).  For every $v \in \alpha(G)$
  (\cref{def:rna-alpha}), it holds $|\algname{Alg}(v)| = \bigO(|G|)$.
\end{lemma}
\begin{proof}
  Let $v_{\rm aux}$ and $v'_{\rm aux}$ be as in \cref{def:rna-alpha},
  i.e., such that it holds $v = v'_{\rm aux} \cdot v_{\rm aux}$.
  By \cref{lm:global-output-size}, $|\algname{Alg}(v_{\rm aux})|
  = \tfrac{7}{2}|G|$ and $|\algname{Alg}(v'_{\rm aux})| =
  \tfrac{7}{2}|G'|$ (where $G'$ is as in \cref{def:rna-alpha}). Note
  also that $|G| = |G'|$. Finally, observe that every second symbol in
  $v_{\rm aux}$ (resp.\ $v'_{\rm aux}$) belongs to $\{\hash_i\}_{i \in
  [1 \dd 2|V|]} \cup \{\dol_i\}_{i \in [1 \dd |V|]}$
  (resp.\ $\{\hash'_{i}\}_{i \in [1 \dd 2|V|]} \cup \{\dol'_i\}_{i \in
  [1 \dd |V|]}$).  This implies that if a substring $s$ satisfying
  $|s| \geq 2$ has two nonoverlapping occurrences in $v'_{\rm aux}
  \cdot v_{\rm aux}$, then they are both either in $v'_{\rm aux}$ or
  $v_{\rm aux}$.  Consequently, it follows by \cref{lm:global-concat},
  that
  \[
    |\algname{Alg}(v)| = |\algname{Alg}(v'_{\rm aux} \cdot v_{\rm
    aux})| \leq \tfrac{7}{2}(|G| + |G'|) = \bigO(|G|).  \qedhere
  \]
\end{proof}

\begin{lemma}\label{lm:wrna-alpha}
  Let $G = (V, \Sigma, R, S)$ be an admissible SLG and assume that
  $\Sigma$ is augmented with a match operation (\cref{def:match}).
  Let also $w : \Sigma \rightarrow \Zp$ be a weight function. Denote
  $u = \expgen{G}{S}$. For every $X \in V$, denote $q_X = |x| - 1 +
  \sum_{i \in [1 \dd |x|]} w(x[i])$, where $x = \expgen{G}{X}$.
  For
  every $v \in \alpha(G)$ (\cref{def:rna-alpha}), it holds
  \begin{align*}
    \WRNA(v) = 2\WRNA(u) + \delta,
  \end{align*}
  where $\delta = 2|V| \cdot (2q_S + 1) + q_S + 2\sum_{X \in V
  \setminus \{S\}} q_X$, and the match operation of $\Sigma$ and the
  weight function $w$ have been extended to the alphabet of $v$ as in
  \cref{def:rna-alpha}. Moreover, given $G$ and the weight function $w$,
  we can compute $\delta$ in $\bigO(|G|)$ time.
\end{lemma}
\begin{proof}
  Let $G' = (V, \Sigma, R', S)$, $\{\dol_i\}_{i \in [1 \dd |V|]}$,
  $\{\dol'_i\}_{i \in [1 \dd |V|]}$, $\{\hash_i\}_{i \in [1 \dd
  2|V|]}$, $\{\hash'_i\}_{i \in [1 \dd 2|V|]}$, $\Sigma_{\rm
  aux}$, $\Sigma'_{\rm aux}$, and $\Sigma_{\rm all}$ be as in
  \cref{def:rna-alpha}. Let $(N_i)_{i \in [1 \dd |V|]}$ be a sequence
  corresponding to $v$ in \cref{def:rna-alpha}. Then, let $v_{\rm
  aux}$, $v'_{\rm aux}$, $G_{\rm aux} = (V, \Sigma_{\rm aux}, R_{\rm
  aux}, S)$ and $G'_{\rm aux} = (V, \Sigma'_{\rm aux}, R'_{\rm aux},
  S)$ be as in \cref{def:rna-alpha} and correspond to $(N_i)_{i \in [1
  \dd |V|]}$.  For every $i \in [1 \dd |V|]$, denote
  \begin{align*}
    v_{{\rm aux},i}
      &= (\textstyle\bigodot_{j=i,\dots,|V|-1}
            \expgen{G_{\rm aux}}{N_j} \cdot
            \hash_{2j-1} \cdot
            \expgen{G_{\rm aux}}{N_j} \cdot
            \hash_{2j}
         ) \cdot \expgen{G_{\rm aux}}{N_{|V|}},\\
    v'_{{\rm aux},i}
      &= \expgen{G'_{\rm aux}}{N_{|V|}} \cdot
         (\textstyle\bigodot_{j=|V|-1,\dots,i}
            \hash'_{2j} \cdot
            \expgen{G'_{\rm aux}}{N_j} \cdot
            \hash'_{2j-1} \cdot
            \expgen{G'_{\rm aux}}{N_j}
         ).
  \end{align*}
  Note that $v_{\rm aux} = v_{{\rm aux},1} \cdot \hash_{2|V|-1} \cdot
  \expgen{G_{\rm aux}}{N_{|V|}} \cdot \hash_{2|V|}$, and a symmetric
  equality holds for $v'_{\rm aux}$. Thus, $v = \hash'_{2|V|} \cdot
  \expgen{G'_{\rm aux}}{N_{|V|}} \cdot \hash'_{2|V|-1} \cdot v'_{{\rm
  aux},1} \cdot v_{{\rm aux},1} \cdot \hash_{2|V|-1} \cdot
  \expgen{G_{\rm aux}}{N_{|V|}} \cdot \hash_{2|V|}$.

  The proof proceeds in three steps:

  1. First, we prove that for every $i \in [1 \dd |V|]$, it holds $
  \WRNA(\expgen{G'_{\rm aux}}{N_i} \cdot \expgen{G_{\rm aux}}{N_i}) =
  q_{N_i} $.  Denote $t = \expgen{G_{\rm aux}}{N_i}$. Observe the
  three differences between $G_{\rm aux}$ and $G'_{\rm aux}$:
  \begin{itemize}[itemsep=0pt,parsep=0pt]
  \item The order of nonterminals in the definition of $G'_{\rm aux}$
    is reversed (compared to $G_{\rm aux}$),
  \item $\dol_{i}$ in $G_{\rm aux}$ is replaced by $\dol'_i$ in
    $G'_{\rm aux}$,
  \item $G'_{\rm aux}$ uses $G'$ instead of $G$, which replaces
    symbols from $\Sigma$ with their matching symbols.
  \end{itemize}
  Combined, these three facts imply that, letting $t' =
  \expgen{G'_{\rm aux}}{N_i}$, we have $|t'| = |t|$ and for every $j
  \in [1 \dd |t|]$, it holds $t'[j] = \match{t[|t|-j+1]}$ (note that
  we used the fact that for every $i \in [1 \dd |V|]$, we defined
  $\match{\dol_i} = \dol'_i$).  Consequently, there exists a set of
  matching pairs for the string $t' \cdot t$ such that every symbol is
  matched. This immediately implies that $\WRNA(t' \cdot t) = \sum_{j
  \in [1 \dd |t|]} w(t[j])$. Observe that letting $x =
  \expgen{G}{N_i}$, for every $j \in [1 \dd |x|]$, we have $x[j] =
  t[2j-1]$.  On the other hand, every symbol at an even position in
  $t$ is from the set $\{\dol_i\}_{i \in [1 \dd |V|]}$. Consequently,
  since for every $i \in [1 \dd |V|]$, we defined $w(\dol_i) = 1$, it
  follows that $\sum_{i \in [1 \dd |t|]} w(t[i]) = |x| - 1 + \sum_{i
    \in [1 \dd |x|]} w(x[i]) = q_{N_i}$. We have thus proved that
  $\WRNA(t' \cdot t) = q_{N_i}$, i.e., the claim.

  2. Next, we prove by induction on $|V|-i$, that for $i \in [1 \dd
  |V|]$, it holds $\WRNA(v'_{{\rm aux},i} \cdot v_{{\rm aux},i}) =
  2(|V|-i) \cdot (2q_S + 1) + q_S + 2\sum_{j=i}^{|V|-1}q_{N_j}$.

  Let $i = |V|$. Then, $v_{{\rm aux},i} = \expgen{G_{\rm
  aux}}{N_{|V|}}$ and $v'_{{\rm aux},i} = \expgen{G'_{\rm
  aux}}{N_{|V|}}$. By the above, we thus have $\WRNA(v'_{{\rm
  aux},i} \cdot v_{{\rm aux},i}) = \WRNA(\expgen{G'_{\rm
  aux}}{N_{|V|}} \cdot \expgen{G_{\rm aux}}{N_{|V|}}) =
  q_{N_{|V|}}$. Since every nonterminal in $V$ occurs in the parse
  tree of $G$, it follows that for every $i \in [1 \dd |V|)$, it holds
  $|\expgen{G}{N_i}| < |\expgen{G}{S}|$.  Thus, by
  $|\expgen{G}{N_1}| \leq \cdots \leq |\expgen{G}{N_{|V|}}|$ we must
  have $N_{|V|} = S$. Hence, $\WRNA(v'_{{\rm aux},i} \cdot v_{{\rm
  aux},i}) = q_{S}$, i.e., we have proved the induction base.

  Let $i < |V|$. Denote $t = \expgen{G_{\rm aux}}{N_i}$ and $t' =
  \expgen{G'_{\rm aux}}{N_i}$, and observe that then $v'_{{\rm aux},i}
  \cdot v_{{\rm aux},i} = v'_{{\rm aux},i+1} \cdot \hash'_{2i} \cdot
  t' \cdot \hash'_{2i-1} \cdot t' \cdot t \cdot \hash_{2i-1} \cdot t
  \cdot \hash_{2i} \cdot v_{{\rm aux},i+1}$. Recall now from the
  above, that it holds $|t| = |t'|$ and for every $j \in [1 \dd |t|]$,
  we have $t'[j] = \match{t[|t| - j + 1]}$. Thus, $\sum_{j \in [1 \dd
  |t|]}w(t'[j]) = \sum_{j \in [1 \dd |t|]} w(t[j])$.  On the other
  hand, above we also proved that $\sum_{j \in [1 \dd |t|]} w(t[i]) =
  q_{N_i}$.  Finally, note that since $t$ is a substring of
  $\expgen{G_{\rm aux}}{S}$, we have $q_{N_i} \leq q_S$. Thus, the sum
  of weights for all symbols in $t' \cdot t$ is at most $2q_{S}$.
  Recall now that by \cref{def:rna-alpha}, we have $w(\hash_{2i-1}) =
  w(\hash_{2i}) = w(\hash'_{2i-1}) = w(\hash'_{2i}) = 2q_{S} + 1$,
  $\match{\hash_{2i-1}} = \hash'_{2i-1}$, and $\match{\hash_{2i}} =
  \hash'_{2i}$.  By applying \cref{lm:decomposition} twice, utilizing
  the inductive assumption, and noting that $\WRNA(t' \cdot t) =
  q_{N_i}$ (proved above), we thus obtain
  \begin{align*}
    \WRNA(v'_{{\rm aux},i} \cdot v_{{\rm aux},i})
      &= \WRNA(v'_{{\rm aux},i+1} \cdot
               \hash'_{2i} \cdot
               t' \cdot
               \hash'_{2i-1} \cdot
               t' \cdot
               t \cdot
               \hash_{2i-1} \cdot
               t \cdot
               \hash_{2i} \cdot
               v_{{\rm aux},i+1})\\
      &= w(\hash'_{2i-1}) +
         \WRNA(v'_{{\rm aux},i+1} \cdot
               \hash'_{2i} \cdot
               t' \cdot
               t \cdot
               \hash_{2i} \cdot
               v_{{\rm aux},i+1}) +
         \WRNA(t' \cdot t)\\
      &= (2q_S + 1) +
         \WRNA(v'_{{\rm aux},i+1} \cdot
               \hash'_{2i} \cdot
               t' \cdot
               t \cdot
               \hash_{2i} \cdot
               v_{{\rm aux},i+1}) +
         q_{N_i}\\
      &= 2(2q_s + 1) +
         \WRNA(v'_{{\rm aux},i+1} \cdot
               v_{{\rm aux},i+1}) +
         2q_{N_i}\\
      &= 2(2q_S + 1) +
         2(|V| - (i+1)) \cdot (2q_S + 1) +
         q_S +
         2\textstyle\sum_{j=i+1}^{|V|-1} q_{N_j} +
         2q_{N_i}\\
      &= 2(|V|-i) \cdot (2q_S + 1) +
         q_S +
         2\textstyle\sum_{j=i}^{|V|-1} q_{N_j}.
  \end{align*}

  By recalling that $N_{|V|} = S$ and applying the above for $i = 1$,
  we thus obtain:
  \[
    \WRNA(v'_{{\rm aux},1} \cdot v_{{\rm aux},1}) = 2(|V| - 1)(2q_S +
    1) + q_S + 2\textstyle\sum_{X \in V \setminus \{S\}} q_X.
  \]

  3. Denote $t = \expgen{G_{\rm aux}}{S}$ and $t' = \expgen{G'_{\rm
  aux}}{S}$.  Recall that $v = \hash'_{2|V|} \cdot t' \cdot
  \hash'_{2|V|-1} \cdot v'_{{\rm aux},1} \cdot v_{{\rm aux},1} \cdot
  \hash_{2|V|-1} \cdot t \cdot \hash_{2|V|}$. To complete the proof,
  we recall that $|t| = |t'|$ and for every $j \in [1 \dd |t|]$, it
  holds $t'[j] = \match{t[|t| - j + 1]}$. By \cref{lm:rev,lm:match},
  we thus have $\WRNA(t) = \WRNA(t')$. Moreover, note that by
  \cref{def:rna-alpha}, removing all symbols from $t$ that belong to
  the set $\{\dol_i\}_{i \in [1 \dd |V|]}$ results in the string
  $\expgen{G}{S} = u$.  Since none of these symbols have matching
  characters in $t$, we thus obtain $\WRNA(t) = \WRNA(u)$.  Recall now
  that $\sum_{j \in [1 \dd |t|]} w(t[j]) = q_{N_{|V|}} = q_S$.  On the
  other hand, by \cref{def:rna-alpha}, we have $w(\hash'_{2|V|}) =
  w(\hash'_{2|V|-1}) = w(\hash_{2|V|-1}) = w(\hash_{2|V|}) = 2q_S +
  1$.  Note also that $\match{\hash'_{2|V|-1}} = \hash'_{2|V|}$ and
  $\match{\hash_{2|V|-1}} = \hash_{2|V|}$. By applying
  \cref{lm:decomposition} and combining with the above observations,
  we thus obtain:
  \begin{align*}
    \WRNA(v)
      &= \WRNA(\hash'_{2|V|} \cdot
               t' \cdot
               \hash'_{2|V|-1} \cdot
               v'_{{\rm aux},1} \cdot
               v_{{\rm aux},1} \cdot
               \hash_{2|V|-1} \cdot
               t \cdot
               \hash_{2|V|}
              )\\
      &= w(\hash'_{2|V|}) +
         \WRNA(v'_{{\rm aux},1} \cdot
               v_{{\rm aux},1} \cdot
               \hash_{2|V|-1} \cdot
               t \cdot
               \hash_{2|V|}
              ) +
          \WRNA(t')\\
      &= w(\hash'_{2|V|}) +
         w(\hash_{2|V|}) +
         \WRNA(v'_{{\rm aux},1} \cdot
               v_{{\rm aux},1}) +
         \WRNA(t') + \WRNA(t)\\
      &= 2(2q_S + 1) +
         \WRNA(v'_{{\rm aux},1} \cdot
               v_{{\rm aux},1}) +
         2\WRNA(u)\\
      &= 2(2q_S + 1) +
         2(|V| - 1)(2q_S + 1) +
         q_S +
         2\textstyle\sum_{X \in V \setminus \{S\}} q_X +
         2\WRNA(u)\\
      &= \delta + 2\WRNA(u).
  \end{align*}

  We now explain how to compute $\delta$ in $\bigO(|G|)$ time.  First,
  in $\bigO(|G|)$ time we sort the nonterminals of the implicit
  grammar DAG of $G$ (defined so that there is an edge connecting
  nonterminals $X$ and $Y$ when $Y$ appears in $\rhsgen{G}{X}$)
  topologically. For every $X \in V$, we then compute $|x|$ and
  $\sum_{j \in [1 \dd |x|]} w(x[j])$, where $x = \expgen{G}{X}$.
  Given the above ordering, this is easily done in $\bigO(|G|)$ time.
  Given these values, we can immediately determine $q_X$ for every $X
  \in V$.  The value of $\delta$ is then easily deduced in
  $\bigO(|G|)$ time.
\end{proof}

\begin{theorem}\label{th:rna-global}
  Let \algname{Alg} be a global algorithm (e.g., \algname{RePair},
  \algname{Greedy}, or \algname{LongestMatch}). Let $\delta \in (0,
  1]$.  Assuming the $k$-Clique Conjecture (resp.\ Combinatorial
  $k$-Clique Conjecture), there is no algorithm (resp.\ combinatorial
  algorithm) that, given $x \in \Sigma^{n}$ (where $\Sigma$ is
  augmented with a match operation) and a weight function $w : \Sigma
  \rightarrow [0 \dd m]$ such that $|\Sigma| = \bigO(n^{\delta})$, $m
  = \bigO(\poly(n))$, and $|\algname{Alg}(x)| = \bigO(n^{\delta})$,
  computes $\WRNA(x)$ in $\bigO(n^{\omega - \epsilon})$
  (resp.\ $\bigO(n^{3 - \epsilon})$) time, for any $\epsilon > 0$.
\end{theorem}
\begin{proof}
  We prove the claim by contraposition. Assume that there exists some
  $\epsilon > 0$ such that for every $x \in \Sigma^{n}$ (where
  $\Sigma$ is augmented with a match operation) and for every weight
  function $w : \Sigma \rightarrow [1 \dd m]$ satisfying $|\Sigma| =
  \bigO(n^{\delta})$, $m = \bigO(\poly(n))$, and $|\algname{Alg}(x)| =
  \bigO(n^{\delta})$, we can in $\bigO(n^{\omega - \epsilon})$
  (resp.\ $\bigO(n^{3 - \epsilon})$) time compute $\WRNA(x)$. We will
  prove that this implies that the $k$-Clique Conjecture
  (resp.\ Combinatorial $k$-Clique Conjecture) is false.

  Assume that we are given an undirected graph $G = (V, E)$. Let $k =
  \max(\lceil \tfrac{3}{\delta} \rceil, \lceil \tfrac{30}{\epsilon}
  \rceil)$.  we execute the following algorithm:
  \begin{enumerate}
  \item Using \cref{lm:rna-small-slg}, in $\bigO(|V|^3)$ time we
    compute an SLG $H = (V_H, \Sigma, R_H, S_H)$, a weight function
    $w' : \Sigma \rightarrow [1 \dd m']$, and an integer $\lambda$,
    such that letting $L(H) = \{u\}$, it holds:
    \begin{itemize}[itemsep=0pt,parsep=0pt]
    \item $|\Sigma| = \bigO(1)$,
    \item $|V_H| = \bigO(|V|^3)$,
    \item $|u| = \bigO(|V|^{k+2})$ and $|u| = \Omega(|V|^k)$,
    \item $m' = \bigO(|V|^2)$,
    \item $\sum_{X \in V_H}|\expgen{H}{X}| = \bigO(|V|^{k+5})$,
    \item $\WRNA(u) \geq \lambda$ holds if and only if $G$ contains a
      $3k$-clique.
    \end{itemize}
  \item Using \cref{lm:admissible}, in $\bigO(|H|) =
    \bigO(|V|^3)$ time we construct an admissible SLG $H' = (V_{H'},
    \Sigma, R_{H'}, S_{H'})$ such that:
    \begin{itemize}[itemsep=0pt,parsep=0pt]
    \item $L(H') = L(H) = \{u\}$,
    \item $|H'| = \bigO(|H|) = \bigO(|V|^3)$,
    \item
      $\sum_{X \in V_{H'}}|\expgen{H'}{X}|
        = \bigO(\log |H| \cdot \sum_{X \in V_H}|\expgen{H}{X}|)
        = \bigO(|V|^{k+5} \log |V|)$.
    \end{itemize}
  \item Using \cref{lm:rna-construction-alpha}, in $\bigO(\sum_{X \in
    V_{H'}}|\expgen{H'}{X}|) = \bigO(|V|^{k+5} \log |V|)$ time we
    construct $v \in \alpha(H')$, and augment $w'$ into a weight
    function $w : \Sigma_{\rm all} \rightarrow [0 \dd m]$ described in
    \cref{def:rna-alpha}. Observe that the following properties hold.
    \begin{itemize}[itemsep=0pt,parsep=0pt]
    \item First, we show that $|\Sigma_{\rm all}| =
      \bigO(|v|^{\delta})$.  To this end, first note that $|v| \geq
      |u| = \Omega(|V|^{k})$.  On the other hand, by $3/\delta \leq k$
      we have $1/k \leq \delta/3$.  Consequently, $|\Sigma_{\rm all}|
      = |\Sigma| + 6|V| = \bigO(|V|) = \bigO(|v|^{1/k}) =
      \bigO(|v|^{\delta/3}) = \bigO(|v|^{\delta})$.
    \item Next, we prove that $m = \bigO(\poly(|v|))$. For this, note
      that by \cref{def:rna-alpha}, we have $m \leq 2|u| + 2\sum_{i
        \in [1 \dd |u|]}w'(u[i])$ By $|u| \leq |v|$ and $m' =
      \bigO(|V|^2)$, we thus have $m \leq 2|u| + m'|u| = \bigO(|V|^2
      \cdot |u|) = \bigO(\poly(|u|)) = \bigO(\poly(|v|))$.
    \item Finally, we show that $|\algname{Alg}(v)| =
      \bigO(|v|^{\delta})$.  By \cref{lm:rna-boosting-global}, it
      holds $|\algname{Alg}(v)| = \bigO(|H'|) = \bigO(|V|^3)$.  Since
      above we observed that $|V| = \bigO(|v|^{\delta/3})$, it follows
      that $|\algname{Alg}(v)| = \bigO(|V|^3) = \bigO(|v|^{\delta})$.
    \end{itemize}
  \item Using \cref{lm:wrna-alpha}, in $\bigO(|H'|)$ time we compute
    $\delta$ satisfying $\WRNA(v) = \WRNA(u) + \delta$.
  \item We apply the hypothetical algorithm for RNA folding to
    $v$. More precisely, we compute $\WRNA(v)$ in
    \begin{align*}
      \bigO(|v|^{\omega - \epsilon})
        &= \bigO(|V|^{(k+5)(\omega - \epsilon)}
                 \log^{\omega-\epsilon} |V|)\\
        &= \bigO(|V|^{k(\omega - \epsilon) + 5\omega})\\
        &= \bigO(|V|^{k\omega(1 - \epsilon/(2\omega)) +
                 (k\epsilon/2 - 5\omega)})\\
        &= \bigO(|V|^{k\omega(1 - \epsilon/(2\omega))})
    \end{align*}
    (resp.\
    $\bigO(|v|^{3 - \epsilon})
      =\allowbreak \bigO(|V|^{(k+5)(3 - \epsilon)} \log^{3-\epsilon} |V|)
      =\allowbreak \bigO(|V|^{k(3 - \epsilon) + 15})
      =\allowbreak \bigO(|V|^{3k(1 - \epsilon/6) + (k\epsilon/2 - 15)})
      =\allowbreak \bigO(|V|^{3k(1 - \epsilon/6)})$)
    time. Note that we used that $k \geq \tfrac{30}{\epsilon}$, which
    implies $k\epsilon/2 \geq 15$. By \cref{lm:wrna-alpha}, we thus
    obtain $\WRNA(u) = \WRNA(v) - \delta$. Recall that above we noted
    that $\WRNA(u) \geq \lambda$ holds if and only if $G$ contains a
    $3k$-clique.
  \end{enumerate}
  We have thus checked if $G$ contains a $3k$-clique in
  $\bigO(|V|^{k+5} \log |V| + |V|^{k\omega(1 - \epsilon/(2\omega))})
  = \bigO(|V|^{k\omega(1 - \epsilon')})$ (resp.\
  $\bigO(|V|^{k+5} \log |V| + |V|^{3k(1 - \epsilon/6)})
  = \bigO(|V|^{3k(1 - \epsilon')})$) time, where $\epsilon' > 0$ is
  some constant (note that we used that $k \geq 6$). This implies
  that the $k$-Clique Conjecture (resp.\ Combinatorial $k$-Clique
  Conjecture) is false.
\end{proof}

\subsection{Analysis of \algname{Sequential}}\label{sec:rna-sequential}

\begin{definition}\label{def:rna-beta}
  Let $G = (V, \Sigma, R, S)$ be an admissible SLG and assume that
  $\Sigma$ is augmented with a match operation (\cref{def:match}).
  Let also $w : \Sigma \rightarrow \Zp$ be a weight function.
  Let $G' = (V, \Sigma, R', S)$ be an SLG defined by replacing every
  occurrence of $a \in \Sigma$ on the right-hand of side of $G$ with
  $\match{a}$, for every $a \in \Sigma$. Assume that the sets $\Sigma$,
  $\{\dol_{i}\}_{i \in [1 \dd |V|]}$,
  $\{\dol'_{i}\}_{i \in [1 \dd |V|]}$,
  $\{\hash_{L,i}\}_{i \in [1 \dd 2|V|]}$,
  $\{\hash_{R,i}\}_{i \in [1 \dd 2|V|]}$,
  $\{\hash'_{L,i}\}_{i \in [1 \dd 2|V|]}$, and
  $\{\hash'_{R,i}\}_{i \in [1 \dd 2|V|]}$ are
  pairwise disjoint. Denote
  \begin{itemize}[itemsep=0pt,parsep=0pt]
  \item $\Sigma_{\rm aux} = \Sigma \cup \{\dol_{i}\}_{i \in [1 \dd |V|]}$,
  \item $\Sigma'_{\rm aux} = \Sigma \cup \{\dol'_{i}\}_{i \in [1 \dd |V|]}$,
  \item $\Sigma_{\rm all} =
    \Sigma \cup
    \{\dol_i, \dol'_i\}_{i \in [1 \dd |V|]} \cup
    \{\hash_{L,i}, \hash_{R,i},
      \hash'_{L,i}, \hash'_{R,i}\}_{i \in [1 \dd 2|V|]}$.
  \end{itemize}
  By $\beta(G)$, we denote the subset of $\Sigma_{\rm all}^{*}$ such
  that for every $v \in \Sigma_{\rm all}^{*}$, $v \in \beta(G)$ holds
  if and only if there exists a sequence $(N_i)_{i \in [1 \dd |V|]}$
  such that:
  \begin{itemize}[itemsep=0pt,parsep=0pt]
    \item $\{N_i : i \in [1 \dd |V|]\} = V$,
    \item $|\expgen{G}{N_i}| \leq |\expgen{G}{N_{i+1}}|$ holds for $i
      \in [1 \dd |V|)$, and
    \item $v = v_{\rm L} \cdot v_{\rm R} \cdot
               v'_{\rm L} \cdot v'_{\rm R}$, where
    \begin{align*}
      v_{\rm L}
        &= \textstyle\bigodot_{i=1,\dots,|V|}
           \expgen{G_{\rm aux}}{N_i} \cdot
           \hash_{L,2i-1} \cdot
           \expgen{G_{\rm aux}}{N_i} \cdot
           \hash_{L,2i},\\
      v_{\rm R}
        &= \textstyle\bigodot_{i=|V|,\dots,1}
           \expgen{G_{\rm aux}}{N_i} \cdot
           \hash_{R,2i-1} \cdot
           \expgen{G_{\rm aux}}{N_i} \cdot
           \hash_{R,2i},\\
      v'_{\rm L}
        &= \textstyle\bigodot_{i=1,\dots,|V|}
           \hash'_{L,2i} \cdot
           \expgen{G'_{\rm aux}}{N_i} \cdot
           \hash'_{L,2i-1} \cdot
          \expgen{G'_{\rm aux}}{N_i},\\
      v'_{\rm R}
        &= \textstyle\bigodot_{i=|V|,\dots,i}
           \hash'_{R,2i} \cdot
           \expgen{G'_{\rm aux}}{N_i} \cdot
           \hash'_{R,2i-1} \cdot
          \expgen{G'_{\rm aux}}{N_i},
    \end{align*}
  \end{itemize}
  and $G_{\rm aux} = (V, \Sigma_{\rm aux}, R_{\rm aux}, S)$
  (resp.\ $G'_{\rm aux} = (V, \Sigma'_{\rm aux}, R'_{\rm aux}, S)$) is
  defined so that for every $i \in [1 \dd |V|]$, it holds
  $\rhsgen{G_{\rm aux}}{N_i} = A \cdot \dol_{i} \cdot B$
  (resp.\ $\rhsgen{G'_{\rm aux}}{N_i} = B \cdot \dol'_{i} \cdot A$),
  where $A, B \in V \cup \Sigma$ are such that $\rhsgen{G}{N_i} = AB$
  (resp.\ $\rhsgen{G'}{N_i} = AB$).

  We also extend the match operation from $\Sigma$ to $\Sigma_{\rm all}$
  so that:
  \begin{itemize}[itemsep=0pt,parsep=0pt]
  \item For every $i \in [1 \dd |V|]$, $\match{\dol_i} = \dol'_i$,
  \item For every $i \in [1 \dd 2|V| - 2]$, $\match{\hash_{L,i}} =
    \hash'_{R,i}$ and $\match{\hash_{R,i}} = \hash'_{L,i}$,
  \item $\match{\hash_{L,2|V|-1}} = \hash_{L,2|V|}$,
    $\match{\hash_{R,2|V|-1}} = \hash_{R,2|V|}$,
    $\match{\hash'_{L,2|V|-1}} = \hash'_{L,2|V|}$,
    $\match{\hash'_{R,2|V|-1}} = \hash'_{R,2|V|}$.
  \end{itemize}

  Finally, we extend the weight function $w$ from $\Sigma$ to
  $\Sigma_{\rm all}$, so that:
  \begin{itemize}[itemsep=0pt,parsep=0pt]
  \item For every $i \in [1 \dd |V|]$, $w(\dol_i) = w(\dol'_i) = 1$,
  \item For every $i \in [1 \dd 2|V|]$, we define
    $w(\hash_{L,i}) = w(\hash_{R,i}) =
    w(\hash'_{L,i}) = w(\hash'_{R,i}) =
    (4|u| - 3) +\allowbreak 4\sum_{j \in [1 \dd |u|]} w(u[j])$, where
    $u = \expgen{G}{S}$.
  \end{itemize}
\end{definition}

\begin{lemma}\label{lm:sequential-concat}
  Let $x, y \in \Sigma^{+}$. Assume that if some $s \in \Sigma^{+}$
  satisfying $|s| \geq 2$ has two nonoverlapping occurrences in $xy$,
  then either both these occurrences are contained in $x$, or both are
  contained in $y$.  Then, it holds $|\algname{Sequential}(xy)| =
  |\algname{Sequential}(x)| + |\algname{Sequential}(y)|$.
\end{lemma}
\begin{proof}
  Denote $\algname{Sequential}(xy) = G^{xy} = (V^{xy}, \Sigma, R^{xy}, S)$.
  Let $k$ denote the number of steps performed by \algname{Sequential}
  when processing $xy$, and let $G^{xy}_i = (V^{xy}_i, \Sigma,
  R^{xy}_i, S)$ (where $i \in [1 \dd k]$) denote the intermediate SLG
  computed after $i$ steps.

  By definition of \algname{Sequential} (\cref{sec:algs-nonglobal}),
  for every $i \in [1 \dd k]$, $\expgen{G^{xy}_i}{S}$ is a prefix of
  $xy$.  Let $k_1 = \min\{i \in [1 \dd k] : |\expgen{G^{xy}_i}{S}|
  \geq |x|\}$.  Observe that it holds $|\expgen{G^{xy}_{k_1}}{S}| =
  |x|$. Otherwise, the last symbol $X$ of $\rhsgen{G^{xy}_{k_1}}{S}$
  would be a nonterminal from $V^{xy}_{k_1}$ occurring twice on the
  right-hand side of $G^{xy}_{k_1}$. These two occurrences would
  correspond to two disjoint occurrences of $s =
  \expgen{G^{xy}_{k_1}}{X}$ in $xy$.  By $|s| \geq 2$, this would
  contradict the assumption from the claim.  We thus obtain that
  $G^{xy}_{k_1}$ is isomorphic with $G^{x}$, where
  $\algname{Sequential}(x) = G^{x} = (V^{x}, \Sigma, R^{x}, S)$.
 
  Denote $k_2 = k - k_1$. Let $k_y$ denote the number of steps
  performed by \algname{Sequential} when processing $y$, and let
  $G^{y}_{i} = (V^{y}_i, \Sigma, R^{y}_i, S)$ (where $i \in [1 \dd
  k_y]$) denote the intermediate SLG computed after the first $i$
  steps. Let also $\algname{Sequential}(y) = G^{y} = (V^{y}, \Sigma,
  R^{y}, S)$. Assume that for every $i \in [1 \dd k_y]$, it holds
  $V^{y}_i \cap V^{x} = \{S\}$.  We prove by the induction on $i$
  that for every $i \in [1 \dd k_2]$, $G^{xy}_{k_1+i}$ is isomorphic
  with $H_{i} = (V_{i}, \Sigma, R_{i}, S)$ defined as follows:
  \begin{itemize}[itemsep=0pt,parsep=0pt]
  \item $V_{i} = V^{x} \cup V^{y}_{i}$,
  \item $\rhsgen{H_i}{S} = \rhsgen{G^{x}}{S} \cdot
    \rhsgen{G^{y}_{i}}{S}$,
  \item For every $X \in V^{x} \setminus \{S\}$, $\rhsgen{H_i}{X} =
    \rhsgen{G^{x}}{X}$,
  \item For every $X \in V^{y}_i \setminus \{S\}$, $\rhsgen{H_i}{X} =
    \rhsgen{G^{y}_i}{X}$.
  \end{itemize}
  We also denote $H = (V, \Sigma, R, S) = H_{k_2}$.

  Let $i = 1$. To compute $G^{xy}_{k_1+1}$, we first determine the
  longest prefix of the remaining suffix of the input string (which in
  this case is $y$) that is equal to the expansion of some secondary
  nonterminal existing in the current grammar (i.e.,
  $G^{xy}_{k_1}$). No such nonterminal can exist in this case, since
  for every $X \in V^{xy}_{k_1} \setminus \{S\}$, we have
  $|\expgen{G^{xy}_{k_1}}{X}| \geq 2$.  Thus, such nonterminal would
  contradict the assumption from the claim. Consequently, the
  algorithm simply appends $y[1]$ to the definition of the current
  starting nonterminal.  Since we also have $\rhsgen{G^{y}_1}{S} =
  y[1]$, we thus obtain that $G^{xy}_{k_1+1}$ is isomorphic to $H_1$.

  Let us now assume $i > 1$. By the inductive assumption,
  $G^{xy}_{k_1+i-1}$ is isomorphic with $H_{i-1}$. Let $f : V_{i-1}
  \cup \Sigma \rightarrow V^{xy}_{k_1+i-1} \cup \Sigma$ be the
  corresponding bijection (see \cref{sec:prelim}). We extend $f$
  so that for every $s \in (V_{i-1} \cup \Sigma)^{+}$, it holds
  $f(s) = \bigodot_{j=1,\dots,|s|} f(s[j])$. Let $w$ be the
  remaining suffix of $y$ to be processed, i.e., such that
  $\expgen{G^{xy}_{k_1+i-1}}{S} \cdot w = xy$. Let $A$ (resp.\ $B$) be
  the last symbol of $\rhsgen{G^{y}_{i-1}}{S}$
  (resp.\ $\rhsgen{G^{xy}_{k_1+i-1}}{S}$). Note, that $f(A) = B$.  We
  prove that $G^{xy}_{k_1+i}$ is isomorphic with $H_i$ in two steps:
  \begin{enumerate}

  \item Let $A' \in (V^{y}_{i-1} \setminus \{S\}) \cup \Sigma$
    (resp.\ $B' \in (V^{xy}_{k_1+i-1} \setminus \{S\}) \cup \Sigma$)
    be the symbol initially appended to the definition of the starting
    nonterminal in $G^{y}_{i-1}$ (resp.\ $G^{xy}_{k_1+i-1}$) when
    executing the $i$th (resp.\ $(k_1+i)$th) step of
    \algname{Sequential} for $y$ (resp.\ $xy$), and let $G^{y}_{\rm
    tmp} = (V^{y}_{i-1}, \Sigma, R^{y}_{\rm tmp}, S)$
    (resp.\ $G^{xy}_{\rm tmp} = (V^{xy}_{k_1+i-1}, \Sigma, R^{xy}_{\rm
    tmp}, S)$) be the resulting SLG.  We will prove that $B' =
    f(A')$ (recall that $V^{y}_{i-1} \cup \Sigma \subseteq V_{i-1}
    \cup \Sigma$, and hence $f(A')$ is well-defined).  Consider two
    cases:
    \begin{itemize}
    \item Assume $A' \in \Sigma$, i.e., $A' = w[1]$. By definition of
      $A'$, this implies that for every prefix $w'$ of $w$, it holds
      $w' \not\in \{\expgen{G^{y}_{i-1}}{X} : X \in V^{y}_{i-1}
      \setminus \{S\}\}$.  On the other hand, for every prefix $w'$ of
      $w$, we also have $w' \not\in \{\expgen{G^{x}}{X} : X \in V^{x}
      \setminus \{S\}\}$, since otherwise, we would have $|w'| \geq 2$
      and $w'$ would occur both in $x$ and $y$, contradicting the main
      assumption from the claim.  Combining the above observations
      with $V_{i-1} = V^{x} \cup V^{y}_{i-1}$, we have thus proved
      that for every prefix $w'$ of $w$, it holds $w' \not\in
      \{\expgen{H_{i-1}}{X} : X \in V_{i-1} \setminus \{S\}\}$.  Since
      $\{f(X) : X \in V_{i-1} \setminus \{S\}\} = V^{xy}_{k_1+i-1}
      \setminus \{S\}$, and since isomorphism preserves the
      nonterminal expansion (see \cref{sec:prelim}), this is
      equivalent to $w' \not\in \{\expgen{G^{xy}_{k_1+i-1}}{X} : X \in
      V^{xy}_{k_1+i-1} \setminus \{S\}\}$.  Hence, $B' = w[1] = A' =
      f(A')$.
    \item Assume $A' \in V^{y}_{i-1} \setminus \{S\}$.  Denote $s =
      \expgen{G^{y}_{i-1}}{A'}$, $s' = \expgen{G^{xy}_{k_1+i-1}}{B'}$,
      and $A'' = f^{-1}(B')$.  Note that $|s| \geq 2$.  By $A' \in
      V^{y}_{i-1} \subseteq V_{i-1}$, we thus have $f(A') \in
      V^{xy}_{k_1+i-1} \setminus \{S\}$. Consequently, by definition
      of \algname{Sequential}, $s$ is a prefix of $s'$. Thus, $B'
      \not\in \Sigma$. Using the same argument as above, $s' \not\in
      \{\expgen{G^{x}}{X} : X \in V^{x} \setminus \{S\}\}$. Thus, $A''
      \in V^{y}_{i-1} \setminus \{S\}$. Suppose now that $|s'| >
      |s|$. Since isomorphism preserves the expansion, and since
      $f(A'') = B'$, we have $\expgen{G^{y}_{i-1}}{A''} =
      \expgen{G^{xy}_{k_1+i-1}}{B'}$. Therefore, $A''$ is a strictly
      better candidate (compared to $A'$) to be chosen by
      \algname{Sequential} when processing $w$, a contradiction.
      Thus, $s = s'$, i.e., $\expgen{G^{y}_{i-1}}{A'} =
      \expgen{G^{y}_{i-1}}{A''}$. By \cref{lm:seq-irreducible}, this
      implies $A' = A''$, and hence $B' = f(A'') = f(A')$.
    \end{itemize}
    In both cases, we obtain $B' = f(A')$.

  \item Next, observe that for every $s \in (V^{y}_{i-1} \cup
    \Sigma)^{+}$ that occurs on the right-hand size of $G^{y}_{\rm
    tmp}$ and satisfies $|\expgen{G^{y}_{\rm tmp}}{s}| \geq 2$, the
    number of occurrences of $s$ on the right-hand side of $G^{y}_{\rm
    tmp}$ is equal to the number of occurrences of $f(s)$ on the
    right-hand side of $G^{xy}_{\rm tmp}$, and moreover, any two
    nonoverlapping occurrences of $s$ on the right-hand side of
    $G^{y}_{\rm tmp}$ correspond to two nonoverlapping occurrences of
    $f(s)$ on the right-hand side of $G^{xy}_{\rm tmp}$. Otherwise, we
    would obtain a contradiction with the main assumption in the
    claim. By applying this observation to $AA'$, and recalling that
    $f(AA') = BB'$, we obtain that the pair $AA'$ has at least two
    nonoverlapping occurrences on the right-hand side of $G^{y}_{\rm
    tmp}$ if and only if $BB'$ has at least two nonoverlapping
    occurrences on the right-hand side of $G^{xy}_{\rm tmp}$.
    \begin{itemize}
    \item If no such repetition occurs, then the $i$th step of
      \algname{Sequential} is completed, i.e., $G^{y}_{i} = G^{y}_{\rm
      tmp}$ and $G^{xy}_{k_1+i} = G^{xy}_{\rm tmp}$. It then follows
      by the inductive assumption, and the definition of $H_{i-1}$
      that indeed $G^{xy}_{k_1+i}$ is isomorphic with $H_{i}$, since
      we now see that $H_{i}$ (resp.\ $G^{xy}_{k_1+i}$) is obtained by
      appending $A'$ (resp.\ $B'$) to the definition of the starting
      nonterminal in $H_{i-1}$ (resp.\ $G^{xy}_{k_1+i-1}$).
    \item Let us assume that there exist at least two nonoverlapping
      occurrences of $AA'$ (resp.\ $BB'$) on the right-hand side of
      $G^{y}_{\rm tmp}$ (resp.\ $G^{xy}_{\rm tmp}$). Let $H^{y}_{\rm
      tmp}$ (resp.\ $H^{xy}_{\rm tmp}$) be the SLG obtained by
      replacing all occurrences of $AA'$ (resp.\ $BB'$) on the
      right-hand side of $G^{y}_{\rm tmp}$ (resp.\ $G^{xy}_{\rm
      tmp}$), and let $C$ (resp.\ $C'$) be the newly created
      nonterminal with the definition $AA'$ (resp.\ $BB'$).  By the
      above discussion, there is a one-to-one correspondence between
      occurrences of $AA'$ in $G^{x}_{\rm tmp}$ and $BB'$ in
      $G^{xy}_{\rm tmp}$. And hence the number of occurrences of $C$
      after the replacement is equal to the number of occurrences of
      $C'$. Moreover, these occurrences are matching such that
      combining $H^{y}_{\rm tmp}$ with $G^{x}$, analogously to how
      $G^{y}_{i-1}$ is combined with $G^{x}$ in the definition of
      $H_{i-1}$, yields an SLG isomorphic with $H^{xy}_{\rm
      tmp}$. Finally, observe that after the replacement, a
      nonterminal $X$ in $G^{y}_{\rm tmp}$ occurs only once on the
      right-hand side of $H^{y}_{\rm tmp}$ if and only if $f(X)$
      occurs only once on the right-hand side of $H^{xy}_{\rm tmp}$.
      Since all these changes are done symmetrically to $H^{y}_{\rm
      tmp}$ and $H^{xy}_{\rm tmp}$, $G^{xy}_{k_1+i-1}$ after the
      modifications is isomorphic to the modified $G^{y}_{\rm tmp}$
      combined with $G^{x}$ (as in the definition of $H_i$). In other
      words, $G^{xy}_{k_1+i}$ is isomorphic with $H_i$.
    \end{itemize}
  \end{enumerate}

  By the above, it holds $k_2 = k_y$ and the final grammar
  $G^{xy}$ is isomorphic with $H$. Consequently,
  \begin{align*}
    |\algname{Sequential}(xy)|
      &= \textstyle\sum_{X \in V}|\rhsgen{H}{X}|\\
      &= \textstyle\sum_{X \in V^{x} \setminus \{S\}}|\rhsgen{H}{X}| +
         \textstyle\sum_{X \in V^{y} \setminus \{S\}}|\rhsgen{H}{X}| +
         |\rhsgen{H}{S}|\\
      &= \textstyle\sum_{X \in V^{x} \setminus \{S\}}|\rhsgen{G^{x}}{X}| +
         \textstyle\sum_{X \in V^{y} \setminus \{S\}}|\rhsgen{G^{y}}{X}| +
         |\rhsgen{G^{x}}{S}| + |\rhsgen{G^{y}}{S}|\\
      &= \textstyle\sum_{X \in V^{x}}|\rhsgen{G^{x}}{X}| +
         \textstyle\sum_{X \in V^{y}}|\rhsgen{G^{y}}{X}|\\
      &= |\algname{Sequential}(x)| + |\algname{Sequential}(y)|.
      \qedhere
  \end{align*}
\end{proof}

\begin{lemma}\label{lm:rna-length-beta}
  Let $G = (V, \Sigma, R, S)$
  be an admissible SLG. Assume that $\Sigma$ is augmented with a match
  operation (\cref{def:match}).  For every $v \in \beta(G)$
  (\cref{def:rna-beta}), it holds $|v| = 16\sum_{X \in V}|\expgen{G}{X}|$.
\end{lemma}
\begin{proof}
  Letting $v_{L}$, $v_{R}$, $v'_{L}$, and $v'_{R}$ be as in
  \cref{def:rna-beta} (i.e., so that $v = v_{L} \cdot v_{R} \cdot
  v'_{L} \cdot v'_{R}$), it follows
  by \cref{lm:alpha-length}
  that $|v_{L}| = |v_{R}| = |v'_{L}| = |v'_{R}| =
  4\sum_{X \in V}|\expgen{G}{X}|$, and hence we have $|v| =
  16\sum_{X \in V}|\expgen{G}{X}|$.
\end{proof}

\begin{lemma}\label{lm:rna-construction-beta}
  Let $G = (V, \Sigma, R, S)$ be an admissible SLG. Assume that
  $\Sigma$ is augmented with a match operation (\cref{def:match}) and
  that $16\sum_{X \in V}|\expgen{G}{X}|$ fits into a machine word.
  Let also $w : \Sigma \rightarrow \Zp$ be a weight function.  Given
  $G$, we can compute some $v \in \beta(G)$ and augment the weight
  function $w$ as described by \cref{def:rna-beta} in $\bigO(\sum_{X
  \in V}|\expgen{G}{X}|)$ time.
\end{lemma}
\begin{proof}
  The construction proceeds analogously as in
  \cref{lm:rna-construction-alpha}, except rather than $v_{\rm aux}$
  and $v'_{\rm aux}$, we compute $v_{L}$, $v_{R}$, $v'_{L}$, and
  $v'_{R}$ (see \cref{def:rna-beta}), and then output $v = v_{L} \cdot
  v_{R} \cdot v'_{L} \cdot v'_{R}$. The weight function $w$ is also
  augmented similarly: first, for every $X \in V$, we compute $|x|$
  and $\sum_{i \in [1 \dd |x|]} w(x[i])$, where $x = \expgen{G}{X}$,
  which is used to calculate the value $(4|\expgen{G}{S}| - 3) +
  4\sum_{i \in [1 \dd |u|]}w(u[i])$, and then we assign the weights.
\end{proof}

\begin{lemma}\label{lm:rna-boosting-sequential}
  Let $G = (V, \Sigma, R, S)$ be an admissible SLG. Assume that
  $\Sigma$ is augmented with a match operation (\cref{def:match}).
  For every $v \in \beta(G)$ (\cref{def:rna-beta}), it holds
  $|\algname{Sequential}(v)| = \bigO(|G|)$.
\end{lemma}
\begin{proof}
  Let $v_{L}$, $v_{R}$, $v'_{L}$, and $v'_{R}$ be as in
  \cref{def:rna-beta}.

  By \cref{lm:sequential}, it holds
  $|\algname{Sequential}(v_{L})| = \tfrac{7}{2}|G|$, and for
  every $X \in V$, $\algname{Sequential}(v_{L})$ contains a
  nonterminal expanding to $\expgen{G_{\rm aux}}{X}$ (where $G_{\rm
  aux}$ is as in \cref{def:rna-beta}). Consider now the computation
  of $\algname{Sequential}(v_{L} \cdot v_{R})$ and observe that the
  last symbol of $v_{L}$ is unique in $v_{L}$, and does not occur in
  $v_{R}$. Thus, when processing $v_{L} \cdot v_{R}$,
  $\algname{Sequential}$ will at some point be left with $v_{R}$ as
  the remaining suffix. Note also that $v_{R}$ is composed of $2|V|$
  substrings, each of which is a concatenation of $\expgen{G_{\rm
  aux}}{X}$ for some $X \in V$, and a symbol from
  $\{\hash_{R,i}\}_{i \in [1 \dd 2|V|]}$. Since each of the symbols in
  the latter set has only a single occurrence in $v_{R}$, and does not
  occur in $v_{L}$, any of the above $2|V|$ substrings will take
  exactly two steps to be processed by \algname{Sequential}. Moreover,
  this will only extend the definition of the starting nonterminal by
  two, and will not create or remove any nonterminals.  We thus obtain
  $|\algname{Sequential}(v_{L} \cdot v_{R})| =
  |\algname{Sequential}(v_{L})| + 4|V| = \tfrac{15}{2}|G|$.

  Let us now consider the computation of
  $\algname{Sequential}(v'_{L})$.  Observe that the structure of
  $v'_{L}$ is nearly identical to $v_{L}$, except the last unique
  symbol is moved at the beginning. This does not change the size of
  the output grammar, or the invariant that for every $X \in V$,
  $\algname{Sequential}(v'_{L})$ contains a nonterminal expanding to
  $\expgen{G'_{\rm aux}}{X}$ (where $G'_{\rm aux}$ is as in
  \cref{def:rna-beta}). Thus, $|\algname{Sequential}(v'_{L})| =
  \tfrac{7}{2}|G|$. Let us now consider the computation of
  $\algname{Sequential}(v'_{L} \cdot v'_{R})$. Observe that the first
  symbol of $v'_{R}$ is unique in $v'_{R}$ and does not occur in
  $v'_{L}$. Thus, when processing $v'_{L} \cdot v'_{R}$,
  $\algname{Sequential}$ will at some point be left with $v'_{R}$ as
  the remaining suffix. It remains to note that, similarly as above,
  $v'_{R}$ is composed of $2|V|$ substrings, each of which is a
  concatenation of $\expgen{G'_{\rm aux}}{X}$ for some $X \in V$, and
  a symbol from $\{\hash'_{R,i}\}_{i \in [1 \dd 2|V|]}$.  We thus
  obtain $|\algname{Sequential}(v'_{L} \cdot v'_{R})| =
  |\algname{Sequential}(v'_{L})| + 4|V| = \tfrac{15}{2}|G|$.

  Observe now that all symbols at even positions in the string $v_{L}
  \cdot v_{R}$ are from the set $\{\dol_i\}_{i \in [1 \dd |V|]} \cup
  \{\hash_{L,i}, \hash_{R,i}\}_{i \in [1 \dd 2|V|]}$, and all symbols
  at odd positions in the string $v'_{L} \cdot v'_{R}$ are from the
  set $\{\dol'_i\}_{i \in [1 \dd |V|]} \cup \{\hash'_{L,i},
  \hash'_{R,i}\}_{i \in [1 \dd 2|V|]}$. Consequently, if a substring
  $s$ satisfying $|s| \geq 2$ has at least two nonoverlapping
  occurrences in $v$, they are either both contained in $v_{L} \cdot
  v_{R}$, or both are contained in $v'_{L} \cdot v'_{R}$. By
  \cref{lm:sequential-concat} and the above, we thus obtain
  \begin{align*}
    |\algname{Sequential}(v)|
      &= |\algname{Sequential}(v_{L} \cdot v_{R} \cdot
                               v'_{L} \cdot v'_{R})|\\
      &= |\algname{Sequential}(v_{L} \cdot v_{R})| +
         |\algname{Sequential}(v'_{L} \cdot v'_{R})|\\
      &\leq 15|V| = \bigO(|G|).
      \qedhere
  \end{align*}
\end{proof}

\begin{lemma}\label{lm:wrna-beta}
  Let $G = (V, \Sigma, R, S)$ be an admissible SLG and assume that
  $\Sigma$ is augmented with a match operation (\cref{def:match}).
  Let also $w : \Sigma \rightarrow \Zp$ be a weight function. Denote
  $u = \expgen{G}{S}$. For every $X \in V$, denote $q_X = |x| - 1 +
  \sum_{i \in [1 \dd |x|]} w(x[i])$, where $x = \expgen{G}{X}$.
  For
  every $v \in \beta(G)$ (\cref{def:rna-beta}), it holds
  \begin{align*}
    \WRNA(v) = 4\WRNA(u) + \delta,
  \end{align*}
  where $\delta = 4|V| \cdot (4q_S + 1) + 2q_S + 4\sum_{X \in V
  \setminus \{S\}} q_X$, and the match operation of $\Sigma$ and the
  weight function $w$ have been extended to the alphabet of $v$ as in
  \cref{def:rna-beta}. Moreover, given $G$ and the weight function $w$,
  we can compute $\delta$ in $\bigO(|G|)$ time.
\end{lemma}
\begin{proof}
  Let
  $G' = (V, \Sigma, R', S)$ and the sets
  $\{\dol_i\}_{i \in [1 \dd |V|]}$,
  $\{\dol'_i\}_{i \in [1 \dd |V|]}$,
  $\{\hash_{L,i}\}_{i \in [1 \dd 2|V|]}$,
  $\{\hash_{R,i}\}_{i \in [1 \dd 2|V|]}$,
  $\{\hash'_{L,i}\}_{i \in [1 \dd 2|V|]}$,
  $\{\hash'_{R,i}\}_{i \in [1 \dd 2|V|]}$,
  $\Sigma_{\rm aux}$,
  $\Sigma'_{\rm aux}$, and
  $\Sigma_{\rm all}$
  be as in
  \cref{def:rna-beta}. Let $(N_i)_{i \in [1 \dd |V|]}$ be a sequence
  corresponding to $v$ in \cref{def:rna-beta}. Then, let
  $v_{L}$,
  $v_{R}$,
  $v'_{L}$,
  $v'_{R}$, 
  $G_{\rm aux} = (V, \Sigma_{\rm aux}, R_{\rm aux}, S)$, and
  $G'_{\rm aux} = (V, \Sigma'_{\rm aux}, R'_{\rm aux}, S)$
  be as in \cref{def:rna-beta} and correspond to $(N_i)_{i \in [1 \dd
  |V|]}$.  Since every nonterminal in $V$ occurs in the parse tree
  of $G$, it follows that for every $i \in [1 \dd |V|)$, it holds
  $|\expgen{G}{N_i}| < |\expgen{G}{S}|$.  Thus, by
  $|\expgen{G}{N_1}| \leq \cdots \leq |\expgen{G}{N_{|V|}}|$ we must
  have $N_{|V|} = S$ For every $i \in [1 \dd |V|]$, denote
  \begin{align*}
    v_{L,i}
      &= (\textstyle\bigodot_{j=i,\dots,|V|-1}
            \expgen{G_{\rm aux}}{N_j} \cdot
            \hash_{L,2j-1} \cdot
            \expgen{G_{\rm aux}}{N_j} \cdot
            \hash_{L,2j}
         ) \cdot \expgen{G_{\rm aux}}{N_{|V|}},\\
    v_{R,i}
      &= (\textstyle\bigodot_{j=|V|-1,\dots,i}
            \expgen{G_{\rm aux}}{N_j} \cdot
            \hash_{R,2j-1} \cdot
            \expgen{G_{\rm aux}}{N_j} \cdot
            \hash_{R,2j}
         ) \cdot \expgen{G_{\rm aux}}{N_{|V|}},\\
    v'_{L,i}
      &= \expgen{G'_{\rm aux}}{N_{|V|}} \cdot
         (\textstyle\bigodot_{j=i,\dots,|V|-1}
            \hash'_{L,2j} \cdot
            \expgen{G'_{\rm aux}}{N_j} \cdot
            \hash'_{L,2j-1} \cdot
            \expgen{G'_{\rm aux}}{N_j}
         )\\
    v'_{R,i}
      &= \expgen{G'_{\rm aux}}{N_{|V|}} \cdot
         (\textstyle\bigodot_{j=|V|-1,\dots,i}
            \hash'_{R,2j} \cdot
            \expgen{G'_{\rm aux}}{N_j} \cdot
            \hash'_{R,2j-1} \cdot
            \expgen{G'_{\rm aux}}{N_j}
         ).
  \end{align*}
  Note that $v_L = v_{L,1} \cdot \hash_{L,2|V|-1} \cdot \expgen{G_{\rm
  aux}}{N_{|V|}} \cdot \hash_{L,2|V|}$, and analogous or symmetric
  equalities hold for $v_R$, $v'_L$, and $v'_R$. Thus,
  \begin{align*}
    v &=
        v_{L,1} \cdot \hash_{L,2|V|-1} \cdot \expgen{G_{\rm aux}}{S} \cdot \hash_{L,2|V|} \cdot
        \hash_{R,2|V|} \cdot \expgen{G_{\rm aux}}{S} \cdot \hash_{R,2|V|-1} \cdot v_{R,1} \cdot\\
      & \hspace{0.5cm}
        v'_{L,1} \cdot \hash'_{L,2|V|-1} \cdot \expgen{G'_{\rm aux}}{S} \cdot \hash'_{L,2|V|} \cdot
        \hash'_{R,2|V|} \cdot \expgen{G'_{\rm aux}}{S} \cdot \hash'_{R,2|V|-1} \cdot v'_{R,1}.
  \end{align*}
  The proof proceeds in five steps:

  1. Denote $t = \expgen{G_{\rm aux}}{S}$ and $t' = \expgen{G'_{\rm
  aux}}{S}$. Observe that $|t| = |t'|$ and for every $j \in [1 \dd
  |t|]$, it holds $t'[j] = \match{t[|t| - j + 1]}$. By
  \cref{lm:rev,lm:match}, we thus have $\WRNA(t) =
  \WRNA(t')$. Moreover, note that by \cref{def:rna-beta}, removing all
  symbols from $t$ that belong to the set $\{\dol_i\}_{i \in [1 \dd
  |V|]}$ results in the string $\expgen{G}{S} = u$.  Since none of
  these symbols have matching characters in $t$, we thus obtain
  $\WRNA(t) = \WRNA(u)$. Recall now that $\sum_{j \in [1 \dd |t|]}
  w(t[j]) = q_{N_{|V|}} = q_S$.  On the other hand, by
  \cref{def:rna-beta}, we have
  $w(\hash_{L,2|V|-1}) = w(\hash_{L,2|V|}) =
  w(\hash_{R,2|V|-1}) = w(\hash_{R,2|V|}) =
  w(\hash'_{L,2|V|-1}) = w(\hash'_{L,2|V|}) =
  w(\hash'_{R,2|V|-1}) = w(\hash'_{R,2|V|}) =
  4q_S + 1$.
  Note also that
  $\match{\hash_{L,2|V|-1}} = \hash_{L,2|V|}$,
  $\match{\hash_{R,2|V|-1}} = \hash_{R,2|V|}$,
  $\match{\hash'_{L,2|V|-1}} = \hash'_{L,2|V|}$, and
  $\match{\hash'_{R,2|V|-1}} = \hash'_{R,2|V|}$.
  By applying
  \cref{lm:decomposition} four times, we thus obtain:
  \begin{align*}
    \WRNA(v)
      &= 4w(\hash_{L,2|V|}) +
         2\WRNA(t) +
         2\WRNA(t') +
         \WRNA(v_{L,1} \cdot
               v_{R,1} \cdot
               v'_{L,1} \cdot
               v'_{R,1})\\
      &= 4(4q_S + 1) +
         4\WRNA(u) +
         \WRNA(v_{L,1} \cdot
               v_{R,1} \cdot
               v'_{L,1} \cdot
               v'_{R,1}).
  \end{align*}

  2. Next, we observe that for every $i \in [1 \dd |V|]$, it holds $
  \WRNA(\expgen{G_{\rm aux}}{N_i} \cdot \expgen{G'_{\rm aux}}{N_i}) =
  q_{N_i} $. The proof of this fact is analogous as in
  \cref{lm:wrna-alpha}.

  3. Next, we prove by induction on $|V|-i$, that for $i \in [1 \dd
  |V|]$, it holds $\WRNA(v_{L,1} \cdot v_{R,i} \cdot v'_{L,i} \cdot
  v'_{R,1}) = \WRNA(v_{L,1} \cdot t \cdot t' \cdot v'_{R,1}) +
  2(|V|-i) \cdot (4q_S + 1) + 2\sum_{j=i}^{|V|-1}q_{N_j}$, where $t =
  \expgen{G_{\rm aux}}{S}$ and $t' = \expgen{G'_{\rm aux}}{S}$.

  Let $i = |V|$. Then, $v_{R,i} = t$ and $v'_{L,i} = t'$. By the
  above, we thus have $\WRNA(v_{L,1} \cdot v_{R,i} \cdot v'_{L,i}
  \cdot v'_{R,1}) = \WRNA(v_{L,1} \cdot t \cdot t' \cdot v_{R,1})$,
  i.e., we have proved the induction base.

  Let $i < |V|$. Denote $s = \expgen{G_{\rm aux}}{N_i}$ and $s' =
  \expgen{G'_{\rm aux}}{N_i}$, and observe that then $v_{R,i} \cdot
  v'_{L,i} = v_{R,i+1} \cdot \hash_{R,2i} \cdot s \cdot \hash_{R,2i-1}
  \cdot s \cdot s' \cdot \hash'_{L,2i-1} \cdot s \cdot \hash'_{L,2i}
  \cdot v'_{L,i+1}$. Recall now from the above, that it holds $|s| =
  |s'|$ and for every $j \in [1 \dd |s|]$, we have $s'[j] =
  \match{s[|s| - j + 1]}$. Thus, $\sum_{j \in [1 \dd |s|]}w(s'[j]) =
  \sum_{j \in [1 \dd |s|]} w(s[j])$.  On the other hand, note that
  $\sum_{j \in [1 \dd |s|]} w(s[i]) = q_{N_i}$.  Finally, note that
  since $s$ is a substring of $\expgen{G_{\rm aux}}{S}$, we have
  $q_{N_i} \leq q_S$. Thus, the sum of weights for all symbols in $s
  \cdot s'$ is at most $2q_{S}$.  Recall now that by
  \cref{def:rna-beta}, we have $w(\hash_{R,2i-1}) = w(\hash_{R,2i}) =
  w(\hash'_{L,2i-1}) = w(\hash'_{L,2i}) = 4q_{S} + 1$,
  $\match{\hash_{L,2i-1}} = \hash'_{R,2i-1}$, and
  $\match{\hash_{R,2i}} = \hash'_{L,2i}$.  By applying
  \cref{lm:decomposition} twice, utilizing the inductive assumption,
  we thus obtain
  \begin{align*}
      &\hspace{-0.5cm}
         \WRNA(v_{L} \cdot v_{R,i} \cdot v'_{L,i} \cdot v'_{R})\\
      &= \WRNA(v_{L} \cdot
               v_{R,i+1} \cdot
               \hash_{R,2i} \cdot
               s \cdot
               \hash_{R,2i-1} \cdot
               s \cdot
               s' \cdot
               \hash'_{L,2i-1} \cdot
               s' \cdot
               \hash'_{L,2i} \cdot
               v'_{L,i+1} \cdot
               v_{R})\\
      &= w(\hash_{R,2i-1}) +
         \WRNA(v_{L} \cdot
               v_{R,i+1} \cdot
               \hash_{R,2i} \cdot
               s \cdot
               s' \cdot
               \hash'_{L,2i} \cdot
               v'_{L,i+1} \cdot
               v_{R}) +
         \WRNA(s \cdot s')\\
      &= w(\hash_{R,2i-1}) +
         w(\hash_{R,2i}) +
         \WRNA(v_{L} \cdot
               v_{R,i+1} \cdot
               v'_{L,i+1} \cdot
               v_{R}) +
         2\WRNA(s \cdot s')\\
      &= 2(4q_s + 1) +
         \WRNA(v_{L} \cdot
               v_{R,i+1} \cdot
               v'_{L,i+1} \cdot
               v_{R}) +
         2q_{N_i}\\
      &= 2(4q_S + 1) +
         \WRNA(v_{L,1} \cdot t \cdot t' \cdot v'_{R,1}) +
         2(|V| - (i+1)) \cdot (4q_S + 1) +
         2\textstyle\sum_{j=i+1}^{|V|-1} q_{N_j} +
         2q_{N_i}\\
      &= \WRNA(v_{L,1} \cdot t \cdot t' \cdot v'_{R,1}) +
         2(|V|-i) \cdot (4q_S + 1) +
         2\textstyle\sum_{j=i}^{|V|-1} q_{N_j}.
  \end{align*}

  Applying the above for $i = 1$, we thus obtain:
  \[
    \WRNA(v_{L} \cdot v_{R,1} \cdot v'_{L,1} \cdot v'_{R}) =
    \WRNA(v_{L,1} \cdot t \cdot t' \cdot v'_{R,1}) +
    2(|V| - 1)(4q_S + 1) +
    2\textstyle\sum_{X \in V \setminus \{S\}} q_X.
  \]

  4. Next, we observe that by analogous induction as above,
  \[
    \WRNA(v_{L,1} \cdot t \cdot t' \cdot v'_{R,1}) =
    2(|V| - 1)(4q_S + 1) +
    2q_S +
    2\textstyle\sum_{X \in V \setminus \{S\}} q_X.
  \]

  By plugging this into the earlier formula, we thus have
  \[
    \WRNA(v_{L} \cdot v_{R,1} \cdot v'_{L,1} \cdot v'_{R}) =
    4(|V| - 1)(4q_S + 1) +
    2q_S +
    4\textstyle\sum_{X \in V \setminus \{S\}} q_X.
  \]

  5. By putting everything together, we thus obtain
  \begin{align*}
    \WRNA(v)
      &= 4(4q_S + 1) +
         4\WRNA(u) +
         \WRNA(v_{L,1} \cdot
               v_{R,1} \cdot
               v'_{L,1} \cdot
               v'_{R,1})\\
      &= 4(4q_S + 1) +
         4\WRNA(u) +
         4(|V| - 1)(4q_S + 1) +
         2q_S +
         4\textstyle\sum_{X \in V \setminus \{S\}} q_X\\
     &=  4\WRNA(u) +
         4|V|(4q_S + 1) +
         2q_S +
         4\textstyle\sum_{X \in V \setminus \{S\}} q_X\\
      &= 4\WRNA(u) + \delta.
  \end{align*}

  The value $\delta$ is computed analogously as in
  \cref{lm:wrna-alpha}, i.e., in $\bigO(|G|)$ time we first compute
  $q_X$ for every $X \in V$, and then easily deduce $\delta$.
\end{proof}

\begin{theorem}\label{th:rna-sequential}
  Assuming the $k$-Clique Conjecture (resp.\ Combinatorial
  $k$-Clique Conjecture), there is no algorithm (resp.\ combinatorial
  algorithm) that, given $x \in \Sigma^{n}$ (where $\Sigma$ is
  augmented with a match operation) and a weight function $w : \Sigma
  \rightarrow [0 \dd m]$ such that $|\Sigma| = \bigO(n^{\delta})$, $m
  = \bigO(\poly(n))$, and $|\algname{Sequential}(x)| = \bigO(n^{\delta})$,
  computes $\WRNA(x)$ in $\bigO(n^{\omega - \epsilon})$
  (resp.\ $\bigO(n^{3 - \epsilon})$) time, for any $\epsilon > 0$.
\end{theorem}
\begin{proof}
  The proof is analogous to the proof of \cref{th:rna-global}, except
  instead of
  \cref{lm:rna-construction-alpha,lm:rna-boosting-global,lm:wrna-alpha},
  we use
  \cref{lm:rna-construction-beta,lm:rna-boosting-sequential,lm:wrna-beta},
  respectively.
\end{proof}

\begin{remark}
  Note that the above lower bound for compressed computation on
  grammars obtained using \algname{Sequential} can be quite easily
  generalized to the unweighted case (\cref{sec:rna-problem}). To
  achieve this it suffices to note that the weights in our reduction
  are sufficiently small, and then use \cref{lm:rna-reduction}. We
  choose to state our results for the weighted case, however, to keep
  it consistent with the results established, e.g., in
  \cref{sec:rna-global}, where establishing the unweighted
  case is difficult without first understanding the behavior of
  global algorithms on unary strings.
\end{remark}

\subsection{Analysis of \algname{LZD}}\label{sec:rna-lzd}

\begin{definition}\label{def:lzd-rna}
	Let $G = (V, \Sigma, R, S)$ be an admissible SLG generating a string $u$. Assume $\Sigma$ is augmented with a match operation (\cref{def:match}) and let $w : \Sigma \rightarrow \Zp$ be a weight function. Let $\Sigma' = \Sigma \cup \{\dol_i : i \in [1 \dd 2|V|]\} \cup \{\dol'_i : i \in [1 \dd 2|V|]\} \cup \{\hash_1, \hash_2, \hash_3, \hash_4\}$ (assume $\dol_i \notin \Sigma$, $\dol'_i \notin \Sigma$ for every $i \in [1 \dd 2|V|], $ and $\hash_i \notin \Sigma$ for $i \in [1 \dd 4]$).
	
	By $\gamma(G)$, we denote the subset of $\Sigma'^*$ such that for every $v \in \Sigma'^*, v \in \gamma(G)$ holds if and only if there exists a sequence $(N_i)_{i \in [1 \dd |V|]}$ such that:
	\begin{itemize}
		\item $\{N_i : i \in [1 \dd |V|]\} = V$,
		\item $|\expgen{G}{N_i}| \leq |\expgen{G}{N_{i+1}}|$ holds for $i \in [1 \dd |V|)$, and
		\item $v = \hash_1 \hash_2 \bigodot_{i \in [1 \dd |V|)} \expgen{G'}{N_{i,0}}\expgen{G'}{N_{i,0}}\expgen{G''}{N'_{i,0}}\expgen{G''}{N'_{i,0}} \cdot \hash_3 \hash_4 \expgen{G'}{N_{|V|,0}}$.
	\end{itemize}
	
	Where $G' = (V',\Sigma', R', S')$ and $G'' = (V'', \Sigma', R'', S'')$ are defined as:
	
	\begin{itemize}
		\item $V' = \bigcup_{i \in [1 \dd |V|]} \{N_{i,0}, N_{i,1}, N_{i,2}\}$ is a set of $3|V|$ variables,
		\item $V'' = \bigcup_{i \in [1 \dd |V|]} \{N'_{i,0}, N'_{i,1}, N'_{i,2}\}$ is a set of $3|V|$ variables,
		\item for every $i \in [1 \dd |V|]$,
		\begin{align*}
			\rhsgen{G'}{N_{i,1}} &=
			\begin{cases}
				N_{j,0} \dol_{2i-1}
				& \text{if }A = N_j\text{ for }j \in [1 \dd |V|],\\
				A \dol_{2i-1}
				& \text{otherwise},\\
			\end{cases}\\
			\rhsgen{G'}{N_{i,2}} &=
			\begin{cases}
				N_{k,0} \dol_{2i}
				& \text{if }B = N_k\text{ for }k \in [1 \dd |V|],\\
				B \dol_{2i}
				& \text{otherwise},\\
			\end{cases}\\[3ex]
			\rhsgen{G'}{N_{i,0}} &= N_{i,1} N_{i,2},
		\end{align*}
		\begin{align*}
			\rhsgen{G''}{N'_{i,2}} &=
			\begin{cases}
				\dol'_{2i} N'_{k,0}
				& \text{if }B = N_k\text{ for }k \in [1 \dd |V|],\\
				\dol'_{2i} \match{B}
				& \text{otherwise},\\
			\end{cases}\\
			\rhsgen{G''}{N'_{i,1}} &=
			\begin{cases}
				\dol'_{2i - 1} N'_{j,0}
				& \text{if }A = N_j\text{ for }j \in [1 \dd |V|],\\
				\dol'_{2i - 1} \match{A}
				& \text{otherwise},\\
			\end{cases}\\[3ex]
			\rhsgen{G''}{N'_{i,0}} &= N'_{i,2} N'_{i,1},
		\end{align*}
		where $A,B \in V \cup \Sigma$ are such that $\rhsgen{G}{N_i} = AB$.
		
	\end{itemize}
	
	We also extend the match operation from $\Sigma$ to $\Sigma'$ such that:
	\begin{itemize}
		\item For every $i \in [1 \dd |V|], \match{\dol_i} = \dol'_i$,
		\item $\match{\hash_1} = \hash_4$, and
		\item $\match{\hash_2} = \hash_3$
	\end{itemize}
	We also extend the weight function $w$ from $\Sigma$ to $\Sigma'$ such that:
	\begin{itemize}
		\item For every $i \in [1 \dd |V|], w(\dol_i) = w(\dol'_i) = 1$,
		\item $w(\hash_1) = w(\hash_4) = 1$, and
		\item $w(\hash_2) = w(\hash_3) = c_0$
	\end{itemize}
	where $c_0 = 1 + \sum_{i \in [1 \dd |V|)} \sum_{j \in [1 \dd |x_i|]} 2w(x_i[j])$ where $x_i = \expgen{G'}{N_{i,0}}$.
\end{definition}

\begin{remark}\label{rm:lzdrna0}
  Note that by the next lemma, the value of $c_0$ is 1 more than the sum of all weights in the string
  $\bigodot_{i \in [1 \dd |V|)} \expgen{G'}{N_{i,0}}\expgen{G'}{N_{i,0}}\expgen{G''}{N'_{i,0}}\expgen{G''}{N'_{i,0}}$
  divided by 2.
\end{remark}

\begin{lemma}\label{lm:lzdstructure0}
	Given an admissible SLG $G = (V, \Sigma, R, S)$ generating a string $u$, let $v \in \gamma(G)$. Let $x_i = \expgen{G'}{N_{i,0}}$ and $y_i = \expgen{G''}{N'_{i,0}}$, then $|x_i| = |y_i|$ and $\match{x_i[j]} = y_i[n_i+1-j]$ for every $j \in [1 \dd n_i]$ where $n_i = |x_i|$ for every $i \in [1 \dd |V|)$, i.e. $y_i$ is obtained by reversing $x_i$ and replacing each character with its matching counterpart. 
\end{lemma}

\begin{proof}
	We prove this by induction on $i$, let $\rhsgen{G}{N_i} = AB$, if $A,B \in \Sigma$, then $x_i = A\dol_{2i-1}B\dol_{2i}$ and $y_i = \dol'_{2i}\match{B}\dol'_{2i-1}\match{A}$ and hence this follows trivially. If both $A,B$ are nonterminals with $A = N_j, B=N_k$ where $ j,k < i$, then $x_i = x_j \dol_{2i-1} x_k \dol_{2i}$ and $y_i = \dol'_{2i}y_k\dol'_{2i-1}y_j$. Here by induction hypothesis, we have:
	\begin{itemize}
		\item $\match{\dol'_{2i}} = \dol_{2i}$,
		\item $|x_k| = |y_k|$ and $\match{x_k[l]} = y_k[n_k + 1 - l]$ for $l \in [1 \dd n_k]$,
		\item $\match{\dol'_{2i-1}} = \dol_{2i-1}$, and
		\item $|x_j| = |y_j|$ and $\match{x_j[l]} = y_j[n_j + 1 - l]$ for $l \in [1 \dd n_j]$.
	\end{itemize}
	Hence we obtain $|x_i| = |y_i|$ and $\match{x_i[l]} = y_i[n_i+1-l]$ for every $l \in [1 \dd n_i]$.
	In case when $A \in \Sigma$ and $B = N_k$ or $A = N_j$ and $B \in \Sigma$ a similar argument as above follows.
	
	Therefore we conclude that $y_i$ is obtained from $x_i$ by first reversing $x_i$ and then replacing each character by its matching counterpart.
\end{proof}
\begin{lemma}\label{lm:lzdstructure1}
	Given an admissible SLG $G = (V, \Sigma, R, S)$ generating a string $u$, let $v \in \gamma(G)$, then $\WRNA(v) = 1 + c_0 + \WRNA(w') + \WRNA(\expgen{G'}{N_{|V|,0}})$ where
	
	$w' = \bigodot_{i \in [1 \dd |V|)} \expgen{G'}{N_{i,0}}\expgen{G'}{N_{i,0}}\expgen{G''}{N'_{i,0}}\expgen{G''}{N'_{i,0}}$.
\end{lemma}

\begin{proof}
	By definition of $c_0$, we have $w(\hash_2) = w(\hash_3) = c_0 > \frac{\sum_{i \in [1 \dd |w'|]}w(w'[i])}{2}$, we can apply \cref{lm:decomposition} to $v$ with $x = \hash_1, a = \hash_2, y = w', b = \hash_3$, and $z = \hash_4 \cdot \expgen{G'}{N_{|V|,0}}$. This gives us $\WRNA(v) = c_0 + \WRNA(w') + \WRNA(\hash_1 \hash_4 \cdot \expgen{G'}{N_{|V|,0}})$, now applying \cref{lm:decomposition} again to $\hash_1 \hash_4 \cdot \expgen{G'}{N_{|V|,0}}$ with $x = y = \epsilon, a = \hash_1, b = \hash_4$ and $z = \expgen{G'}{N_{|V|,0}}$, we obtain $\WRNA(v) = c_0 + \WRNA(w') + 1 + \WRNA(\expgen{G'}{N_{|V|,0}})$.
\end{proof}

\begin{lemma}\label{lm:lzdstructure2}
	Given an admissible SLG $G = (V, \Sigma, R, S)$ generating a string $u$, let $v \in \gamma(G)$ and let $w = \bigodot_{i \in [1 \dd |V|)} \expgen{G'}{N_{i,0}}\expgen{G'}{N_{i,0}}\expgen{G''}{N'_{i,0}}\expgen{G''}{N'_{i,0}}$. Then we must have $\WRNA(w') = c_0 - 1$.
\end{lemma}

\begin{proof}
	Firstly, we notice that for any string, the maximum value of Weighted RNA Folding must be at-most half the sum of all weights in the string. Hence $\WRNA(w') \leq \frac{1}{2} \sum_{i \in [1 \dd |V|)} \sum_{j \in [1 \dd |w'_i|]}w(w'_i[j])$ where $w'_i = \expgen{G'}{N_{i,0}}\expgen{G'}{N_{i,0}}\expgen{G''}{N'_{i,0}}\expgen{G''}{N'_{i,0}}$. Now we show that this inequality as actually tight. To show that, we prove that there is a RNA folding for $w'$ where every character is matched, moreover, we show that there exists a solution where every character within $w'_i$ is matched within $w'_i$ itself. This follows from the fact that $w'_i = x_i \cdot x_i \cdot y_i \cdot y_i$ and hence $\match{w'_i[j]} = w'_i[|w'_i| + 1 - j]$ for every $j \in [1 \dd |w'_i|]$ where $x_i,y_i$ are defined as in \cref{lm:lzdstructure0}. 
	
	Therefore $\WRNA(w') = (\sum_{i \in [1 \dd |V| - 1)} \sum_{j \in [1 \dd n_i]}2w(x_i[j])) = c_0 - 1$.
\end{proof}

\begin{lemma}\label{lm:lzdstructure3}
	Given an admissible SLG $G = (V, \Sigma, R, S)$ generating a string $u$, let $v \in \gamma(G)$, then $\WRNA(v) = 2c_0 + \WRNA(u)$.
\end{lemma}

\begin{proof}
	Combining \cref{lm:lzdstructure1} and \cref{lm:lzdstructure2}, we have $\WRNA(v) = 1 + c_0 + c_0 - 1 + \WRNA(z)$ where $z = \expgen{G'}{N_{|V|,0}}$. The string $z$ only consists of characters from $\Sigma$ and $\{\dol_i | i \in [1 \dd |V|]\}$, in particular, there is no $\dol'_i$ in $z$. Hence every $\dol_i$ in $z$ goes unmatched in any RNA folding of $z$. let $z'$ be the subsequence of $z$ obtained by deleting every $\dol_i, i \in [1 \dd |V|]$. Now $\WRNA(z) = \WRNA(z')$ since the only characters matched in an optimal RNA folding for $z$ are the non $\dol$ symbols. Next we prove that $z' = u$. This follows by an induction argument, i.e. for any $i \in [1 \dd |V|]$, the subsequence of $\expgen{G'}{N_{i,0}}$ obtained by deleting all $\dol$ symbols is identical to $\expgen{G}{N_i}$. $N_{|V|}$ is the largest nonterminal in $G$, hence it must be the starting nonterminal, therefore we can conclude $z' = u$ and hence $\WRNA(v) = 2c_0 + \WRNA(u)$.
\end{proof}

Thus we have proved how to calculate $\WRNA(u)$ if we can calculate $\WRNA(v)$. Now we prove that the grammar \algname{LZD} produces for $v \in \gamma(G)$ has size $O(|G|)$ and hence we can also calculate values like $c_0$ fast enough.

\begin{lemma}\label{lm:lzdrnasize}
	Given an admissible SLG $G = (V, \Sigma, R, S)$ generating a string $u$, let $v \in \gamma(G)$. \algname{LZD} outputs a grammar of size at-most $9|G|$.
\end{lemma}

\begin{proof}
	This proof proceeds exactly like the proof for \cref{lm:lzd}, hence extremely formal details and edge-cases are omitted. %

	The first step of the \algname{LZD} algorithm looks at $\hash_1\hash_2$ and creates one nonterminal corresponding to it.
	
	Now we show that after every $1 + 6i$ steps for $i \leq |V|-1$, the grammar produced has size $3 + 18i$. In particular, lets denote $G_i$ as the first $6i$ nonterminals in the combination of $G'$ and $G''$ defined in \cref{def:lzd-rna}, i.e. $G_i = (V_i,\Sigma',R_i,S_i)$ where:
	\begin{itemize}
		\item $V_i = \{S_i\} \cup \{N_{i,0},N_{i,1},N_{i,2},N'_{i,0},N'_{i,1},N'_{i,2}\} \cup N_0$,
		\item for every $j \in [1 \dd i]$ and $b \in \{0,1,2\}$, it holds that $\rhsgen{G_i}{N_{j,b}} = \rhsgen{G'}{N_{j,b}}$ and $\rhsgen{G_i}{N'_{j,b}} = \rhsgen{G''}{N'_{j,b}}$,
		\item $\rhsgen{G_i}{N_0} = \hash_1\hash_2$
		\item $\rhsgen{G_i}{S_i} = N_0 \bigodot_{j \in [1 \dd i]} N_{j,1}N_{j,2}N_{j,0}N'_{j,2}N'_{j,1}N'_{j,0}$
	\end{itemize}
	We show using induction that after $1 + 6i$ steps, \algname{LZD} produces a grammar isomorphic to $G_i$ for every $1 \leq i \leq |V|-1$.
	
	After $1$ step, we have a grammar $G_0$ which generates $\hash_1\hash_2$.
	Assume that after $1 + 6(i-1)$ steps, \algname{LZD} produced a grammar isomorphic to $G_{i-1}$ generating the string:
	
	$\hash_1 \hash_2 \bigodot_{j \in [1 \dd i-1]} \expgen{G'}{N_{j,0}}\expgen{G'}{N_{j,0}}\expgen{G''}{N'_{j,0}}\expgen{G''}{N'_{j,0}}$.
	
	Now a prefix of the unprocessed string looks like $\expgen{G'}{N_{i,0}}\expgen{G'}{N_{i,0}}\expgen{G''}{N'_{i,0}}\expgen{G'}{N'_{i,0}} \dd$
	
	Recall that \algname{LZD} at any step picks the largest factor $f_l$ as a which is a prefix of the unprocessed string and then the largest factor $f_r$ on the remaining prefix and introduces a nonterminal $f = f_lf_r$ to the grammar. Also assume $\rhsgen{G}{N_i} = N_jN_k$ (The case when $\rhsgen{G}{N_i} = AB, A,B \in \Sigma$ is analogous).
	
	Hence the next $6$ steps of \algname{LZD} are as follows:
	\begin{itemize}
		\item New nonterminal $f_1$ introduced with $\rhs{f_1} = N_{j,0}\dol_{2i-1}$,
		\item New nonterminal $f_2$ introduced with $\rhs{f_2} = N_{k,0}\dol_{2i}$,
		\item New nonterminal $f_3$ introduced with $\rhs{f_3} = f_1f_2$,
		\item New nonterminal $f_4$ introduced with $\rhs{f_4} = \dol'_{2i}N'_{k,0}$,
		\item New nonterminal $f_5$ introduced with $\rhs{f_5} = \dol'{2i-1}N'_{j,0}$, and
		\item New nonterminal $f_6$ introduced with $\rhs{f_6} = f_4f_5$.
	\end{itemize}
	
	Mapping the nonterminals $f_1,f_2,f_3,f_4,f_5,f_6$ to $N_{i,1},N_{i,2},N_{i,0},N'_{i,2},N'_{i,1},N'_{i,0}$ respectively now gives us the grammar $G_i$. A total of 6 nonterminals were added to the right-hand side of the starting nonterminal, since each of them has a definition of size 2, a total increment of 18 in size of the grammar was observed.
	
	Now at the end of $1 + 6(|V|-1)$, the remaining prefix of the string which is unprocessed is $\hash_3\hash_4 \cdot \expgen{G'}{N_{|V|,0}}$. Assume $\rhsgen{G}{N_{|V|}} = N_jN_k$.
	
	This is processed as follows:
	\begin{itemize}
		\item New nonterminal $f_1$ introduced with $\rhs{f_1} = \hash_3\hash_4$,
		\item New nonterminal $f_2$ introduced with $\rhs{f_2} = N_{j,0}\dol_{2|V|-1}$, and
		\item New nonterminal $f_3$ introduced with $\rhs{f_3} = N_{k,0}\dol_{2|V|}$.
	\end{itemize}

	Hence the whole algorithm terminated in $6|V| - 2$ steps producing a grammar of size $18|V|-6 = 9|G|-6$.
\end{proof}

\begin{lemma}\label{lm:lzdrnatime}
	The value $c_0$ from \cref{lm:lzdstructure3} can be calculated in $O(|G|)$ time.
\end{lemma}

\begin{proof}
	Recall that $c_0 = 1 + 2\times \sum_{i \in [1 \dd |V|)} \sum_{j \in [1 \dd |x_i|]} w(x_i[j])$ where $x_i = \expgen{G'}{N_{i,0}}$. Lets denote $q_i = \sum_{j \in [1 \dd |x_i|]}w(x_i[j])$ where $x_i = \expgen{G'}{N_{i,0}}$. To calculate $q_i$, lets assume $\rhsgen{G}{N_i} = N_jN_k$ (In case $\rhsgen{G}{N_i} = AB, A,B \in \Sigma$, $q_i$ can be calculated in $O(1)$ time by just looking at $w(A)$ and $w(B)$). Now since $\expgen{G'}{N_{i,0}} = \expgen{G'}{N_{j,0}} \dol_{2i-1} \expgen{G'}{N_{k,0}} \dol_{2i}$, we must have $q_i = q_j + 1 + q_k + 1$ (Recall $w(\dol_{2i}) = w(\dol_{2i-1}) = 1$). Hence processing the nonterminals in increasing order of size and maintaining the $q_i$ values along with the fact that $|\expgen{G'}{N_{i,0}}| \leq 3|\expgen{G}{N_i}|$(\cref{lm:beta-length}) implying the summation doesn't blow up lets us compute $c_0$ in $O(|G|)$ time.
	
\end{proof}

\begin{lemma}\label{lm:lzdrnasize2}
	Given an admissible SLG $G = (V, \Sigma, R, S)$ generating a string $u$, $\Sigma$ is augmented with a match operation along with a weight function $w : \Sigma \rightarrow \Zp$. We can compute a string $v \in \gamma(G)$ and extend the weight function $w$ as described by \cref{def:lzd-rna} in $O(\sum_{X \in V}(\expgen{G}{X}))$ time.
\end{lemma}

\begin{proof}
	Since $|\expgen{G'}{N_{i,0}}| \leq 3|\expgen{G}{N_i}|$ and $|\expgen{G''}{N'_{i,0}}| = |\expgen{G'}{N_{i,0}}|$, we conclude $|v| \leq 12\times \sum_{X \in V}|\expgen{G}{X}|$. The value $c_0$ was calculated in $O(|G|)$ time hence extending $w$ and the matching to $\Sigma'$ also took $O(|G|)$ time.
\end{proof}

\begin{theorem}\label{th:rna-lzd}
	Assuming the $k$-Clique Conjecture (resp.\ Combinatorial
	$k$-Clique Conjecture), there is no algorithm (resp.\ combinatorial
	algorithm) that, given $x \in \Sigma^{n}$ (where $\Sigma$ is
	augmented with a match operation) and a weight function $w : \Sigma
	\rightarrow [0 \dd m]$ such that $|\Sigma| = \bigO(n^{\delta})$, $m
	= \bigO(\poly(n))$, and $|\algname{LZD}(x)| = \bigO(n^{\delta})$,
	computes $\WRNA(x)$ in $\bigO(n^{\omega - \epsilon})$
	(resp.\ $\bigO(n^{3 - \epsilon})$) time, for any $\epsilon > 0$.
\end{theorem}
\begin{proof}
	The proof is analogous to the proof of \cref{th:rna-sequential}, except
	instead of
	\cref{lm:rna-construction-beta,lm:rna-boosting-sequential,lm:wrna-beta},
	we use
	\cref{lm:lzdstructure3,lm:lzdrnasize,lm:lzdrnatime,lm:lzdrnasize2},
	respectively.
\end{proof}

\appendix

\section{Conversion to an Admissible Grammar}\label{sec:convert-admissible}

\begin{lemma}\label{lm:admissible}
  For every SLG $G = (V, \Sigma, R, S)$ such that $|\expgen{G}{S}|
  \geq 2$, there exists an admissible SLG $G' = (V', \Sigma, R', S')$
  satisfying $L(G') = L(G)$, $\sum_{X \in V'} |\expgen{G'}{X}| =
  \bigO(\log |G| \cdot \sum_{X \in V} |\expgen{G}{X}|)$, and $|G'| \leq
  2|G|$. Moreover, given $G$, we can compute $G'$ in $\bigO(|G|)$
  time.
\end{lemma}
\begin{proof}
  Let us assume that $G$ is given using an encoding in which
  nonterminals are identified with consecutive positive integers. The
  construction of $G'$ consists of three steps:
  \begin{enumerate}
  \item Let $V_s \subseteq V$ be the set of nonterminals occurring in
    the parse tree $\mathcal{T}_{G}(S)$ and let $\Sigma_s \subseteq
    \Sigma$ be the set of symbols occurring in $\expgen{G}{S}$. We
    begin by constructing a directed multigraph $D$ containing as
    vertices the set $V_s \cup \Sigma_s$. The set of edges in $D$ is
    defined as follows. For every $X \in V_s$ we add $|\rhsgen{G}{X}|$
    edges into $D$: the $j$th edge connects vertex $X$ to vertex
    $\rhsgen{G}{X}[j]$ and stores the position $j$ as an auxiliary
    attribute. Construction of $D$ takes $\bigO(|G|)$ time. In
    $\bigO(|G|)$ time we then construct its transpose $D'$. Every edge
    in $D$ is linked with its symmetric edge in $D'$.

  \item In the second step, we prune $D$ so that it does not contain
    vertices with out-degree $1$. First, in $\bigO(|G|)$ time we sort
    $D$ topologically. We then scan the list of vertices in
    reverse-topological order. Let $v$ be the current vertex. If the
    out-degree of $v$ is different than $1$, we move onto another
    vertex. Let us thus assume the out-degree of $v$ is 1 and let $v'$
    be the target of the edge from $v$. First, with the help in $D'$,
    we redirect all edges ending in $v$ into $v'$ ($D'$ is updated
    symmetrically). We then delete $v$ from both $D$ and $D'$. During
    the pruning we keep track of the vertex corresponding to
    nonterminal $S$. Whenever such vertex $v$ gets deleted, its target
    $v'$ becomes the new ``starting vertex''. Observe that deleting a
    vertex does not change the out-degree of the remaining vertices,
    and hence cannot create new vertices to be deleted. Since an edge
    is always redirected into a vertex further than the current vertex
    in topological order, no edge will ever be redirected twice. Thus,
    altogether this pruning takes time proportional to the number of
    edges in $D$ which initially is $\bigO(|G|)$.

  \item In the third step, we create the final grammar $G'$. Let $V'_s
    \subseteq V_s$ be the set of nonterminals remaining in $D$ after
    pruning, and let $S' \in V'_s$ be the nonterminal corresponding to
    the ``starting vertex''. Note that since we assumed
    $|\expgen{G}{S}| \geq 2$, it holds $|V'_s| \geq 1$. Consider any
    $X \in V'_s$ and let $u$ be the string obtained obtained by
    enumerating the outgoing edges of $X$ in $D$ in the ascending
    order of their auxiliary ``position'' attribute. Note that $|u|
    \geq 2$. We run the following procedure:
    \begin{itemize}[itemsep=0pt,parsep=0pt]
    \item If $|u| = 2$, then we add the nonterminal $X$ to $G'$, set
      $\rhsgen{G'}{X} = u$, and stop.
    \item Otherwise, we add $k := \lfloor \tfrac{|u|}{2} \rfloor$
      fresh nonterminals $X_1, \dots, X_k$ to $G'$, setting
      $\rhsgen{G'}{X_i} = u[2i-1] u[2i]$, replace the length-$2k$
      prefix of $u$ with $X_1 \dots X_k$, and repeat the procedure.
    \end{itemize}
    Observe that this procedure performs $\lceil \log |u| \rceil$
    rounds and adds exactly $|u| - 1$ new nonterminals into $G'$, each
    with the definition of length $2$. Since the total length of
    strings $u$ over all $X \in V'_s$ is bounded by $|G|$, we thus
    obtain $|G'| \leq 2|G|$. Thus, the final step in the construction
    of $G'$ takes $\bigO(|G|)$ time. Observe also that during the
    above procedure, each update does not change the total expansion
    length of symbols in $u$.  Thus, if the total expansion length of
    symbols in the initial $u$ is $\ell$, then the total expansion
    length of nonterminals created during the whole procedure is
    bounded by $\lceil \log |u| \rceil \cdot \ell$.  Consequently,
    $\sum_{X \in V'}|\expgen{G'}{X}| = \bigO(\log |G| \cdot \sum_{X \in
    V}|\expgen{G}{X}|)$. \qedhere
  \end{enumerate}
\end{proof}

\section{Efficient Implementation of
  \algname{Sequential}}\label{sec:sequential-impl}

\begin{definition}\label{def:irreducible}
  An SLG is \emph{irreducible} if it satisfies the following three
  properties:
  \begin{enumerate}[itemsep=0pt,parsep=0pt]
  \item All non-overlapping pairs of adjacent symbols of the grammar
    are distinct,
  \item Every secondary nonterminal appears twice on the right-hand
    side of the grammar,
  \item No two nonterminals in the grammar have the same expansion.
  \end{enumerate}
\end{definition}

\begin{lemma}[Theorem~1 in~\cite{sequential}]\label{lm:seq-irreducible}
  Every intermediate SLG computed during the execution
  of \algname{Sequential} is irreducible.
\end{lemma}

\begin{proposition}\label{pr:sequential-implementation}
  For any $x \in \Sigma^n$, \algname{Sequential} can be implemented
  in $\bigO(n \log n)$ time.
\end{proposition}
\begin{proof}
  Let $n_i$ denote the length of the prefix of $x$ processed
  by \algname{Sequential} after the first $i$ steps. Let $G_i =
  (V_i,\Sigma,R_i,S_i)$ be the grammar for $x[1 \dd n_i]$ produced
  by \algname{Sequential}. In step $i+1$, the algorithm finds the
  longest prefix of $x[n_i + 1 \dd n]$ such that there is a
  nonterminal in $G_i$ with the same expansion and then appends that
  nonterminal to the definition of $S_i$, i.e. $n_{i+1}$ is the
  largest number such that $x[n_i + 1 \dd n_{i+1}] = \expgen{G_i}{N}$
  for some $N \in V_i$ and $N$ is appended to the definition of
  $S_i$. If no such prefix exists, $x[n_i + 1]$ is appended to
  $\rhsgen{G_i}{S_i}$. To find $N$, we maintain a suffix tree of
  $x$. Each time a new nonterminal is introduced, we mark that
  (potentially implicit) node of the suffix tree, i.e. if $N$ is a new
  nonterminal formed with $\expgen{G_i}{N} = x[l \dd r]$ where $l =
  n_{i-1} + 1$ and $r = n_i$, we mark a node in the tree that
  corresponds to the substring $x[l \dd r]$. Note that $x[l \dd r]$
  might not correspond to an explicit node in the suffix tree. We thus
  employ the following strategy. With each explicit node $v$ of the
  suffix tree, we keep a pair $(N_v, d_v)$, where $N_v \in V_i$ and
  $d_v \in \Zp$ such that the node corresponding to
  $\expgen{G_i}{N_v}$ is either $v$ or an implicit node between $v$
  and its parent, $d_v = |\expgen{G_i}{N_v}|$, and $d_v$ is
  maximized. Then, the longest prefix of $x[n_i + 1 \dd n]$ which has
  a nonterminal expanding to it is $\expgen{G_i}{N_v}$, where $v$ is
  the closest marked ancestor of the leaf node corresponding to the
  suffix $x[n_i + 1 \dd n]$ in the suffix tree. The marked ancestor
  problem (both queries and updates) can be solved in $\bigO(\log n)$
  time per operation~\cite{markedAncestor}.
  
  This solves finding the longest prefix subproblem. Now moving on to
  next part of the algorithm, we need to check if after appending a
  new symbol at the end of $S_i$ introduced some non-overlapping pair
  of consecutive nonterminals appears twice in the right hand side of
  the grammar. Since \cref{lm:seq-irreducible} states that no
  non-overlapping pairs of characters appear twice, we can maintain a
  dictionary with all consecutive pairs of characters on the
  right-hand side of $G_i$ as keys and store their unique address as
  well. Formally, let $xy$ be the last 2 characters in $S_i$ ($y$ is
  the character we just appended). If $xy$ has another non-overlapping
  occurrence (we can use a dictionary to lookup this in $\bigO(\log
  n)$ time), we introduce a new nonterminal $N \rightarrow xy$ and
  replace both of the occurrences of $xy$ with $N$. We then remove
  $xy$ and the pairs intersecting them from the dictionary and add $N$
  and its neighbors in the dictionary along with their new
  address. When removing a nonterminal since it has only one
  occurrence on the right-hand side of the grammar, we remove pair
  corresponding to it from the dictionary and add any newly introduced
  pairs. All of these operations on the dictionary can be performed in
  $\bigO(\log n)$ time. To implement deletions and insertions, we
  store every right-hand side of a nonterminal as a doubly-linked
  list. The dictionary stores a pointer to a linked list node for
  every key in it. This let us perform modifications on the grammar in
  $\bigO(1)$ time.
  
  In total, we spend $\bigO(n \log n)$ time.
\end{proof}

\bibliographystyle{plainurl}
\bibliography{paper}

\begin{thebibliography}{100}

\bibitem{AbboudBBK17}
Amir Abboud, Arturs Backurs, Karl Bringmann, and Marvin K{\"{u}}nnemann.
\newblock Fine-grained complexity of analyzing compressed data: Quantifying
  improvements over decompress-and-solve.
\newblock In {\em FOCS}, pages 192--203, 2017.
\newblock \href {https://doi.org/10.1109/FOCS.2017.26}
  {\path{doi:10.1109/FOCS.2017.26}}.

\bibitem{AbboudBBK20}
Amir Abboud, Arturs Backurs, Karl Bringmann, and Marvin K{\"{u}}nnemann.
\newblock Impossibility results for grammar-compressed linear algebra.
\newblock In {\em NeurIPS}, 2020.
\newblock URL: \url{https://proceedings.neurips.cc/paper/2020/file/
  645e6bfdd05d1a69c5e47b20f0a91d46-Paper.pdf}.

\bibitem{AbboudBW15}
Amir Abboud, Arturs Backurs, and Virginia~Vassilevska Williams.
\newblock Tight hardness results for {LCS} and other sequence similarity
  measures.
\newblock In Venkatesan Guruswami, editor, {\em {IEEE} 56th Annual Symposium on
  Foundations of Computer Science, {FOCS} 2015, Berkeley, CA, USA, 17-20
  October, 2015}, pages 59--78. {IEEE} Computer Society, 2015.
\newblock \href {https://doi.org/10.1109/FOCS.2015.14}
  {\path{doi:10.1109/FOCS.2015.14}}.

\bibitem{AbboudBW18}
Amir Abboud, Arturs Backurs, and Virginia~Vassilevska Williams.
\newblock If the current clique algorithms are optimal, so is valiant's parser.
\newblock {\em {SIAM} J. Comput.}, 47(6):2527--2555, 2018.
\newblock \href {https://doi.org/10.1137/16M1061771}
  {\path{doi:10.1137/16M1061771}}.

\bibitem{markedAncestor}
S.~Alstrup, T.~Husfeldt, and T.~Rauhe.
\newblock Marked ancestor problems.
\newblock In {\em Proceedings 39th Annual Symposium on Foundations of Computer
  Science}, pages 534--543, 1998.
\newblock \href {https://doi.org/10.1109/SFCS.1998.743504}
  {\path{doi:10.1109/SFCS.1998.743504}}.

\bibitem{greedy3}
Alberto Apostolico and Stefano Lonardi.
\newblock Some theory and practice of greedy off-line textual substitution.
\newblock In {\em DCC}, pages 119--128, 1998.
\newblock \href {https://doi.org/10.1109/DCC.1998.672138}
  {\path{doi:10.1109/DCC.1998.672138}}.

\bibitem{greedy1}
Alberto Apostolico and Stefano Lonardi.
\newblock Compression of biological sequences by greedy off-line textual
  substitution.
\newblock In {\em DCC}, pages 143--152, 2000.
\newblock \href {https://doi.org/10.1109/DCC.2000.838154}
  {\path{doi:10.1109/DCC.2000.838154}}.

\bibitem{greedy2}
Alberto Apostolico and Stefano Lonardi.
\newblock Off-line compression by greedy textual substitution.
\newblock {\em Proceedings of the {IEEE}}, 88(11):1733--1744, 2000.
\newblock \href {https://doi.org/10.1109/5.892709}
  {\path{doi:10.1109/5.892709}}.

\bibitem{BadkobehGIKKP17}
Golnaz Badkobeh, Travis Gagie, Shunsuke Inenaga, Tomasz Kociumaka, Dmitry
  Kosolobov, and Simon~J. Puglisi.
\newblock On two {LZ78}-style grammars: Compression bounds and compressed-space
  computation.
\newblock In {\em SPIRE}, pages 51--67, 2017.
\newblock \href {https://doi.org/10.1007/978-3-319-67428-5\_5}
  {\path{doi:10.1007/978-3-319-67428-5\_5}}.

\bibitem{BannaiHHIJLR21}
Hideo Bannai, Momoko Hirayama, Danny Hucke, Shunsuke Inenaga, Artur Jez, Markus
  Lohrey, and Carl~Philipp Reh.
\newblock The smallest grammar problem revisited.
\newblock {\em {IEEE} Trans. Inf. Theory}, 67(1):317--328, 2021.
\newblock \href {https://doi.org/10.1109/TIT.2020.3038147}
  {\path{doi:10.1109/TIT.2020.3038147}}.

\bibitem{blocktree}
Djamal Belazzougui, Manuel C{\'{a}}ceres, Travis Gagie, Pawel Gawrychowski,
  Juha K{\"{a}}rkk{\"{a}}inen, Gonzalo Navarro, Alberto~Ord{\'{o}}{\~{n}}ez
  Pereira, Simon~J. Puglisi, and Yasuo Tabei.
\newblock Block trees.
\newblock {\em Journal of Computer and System Sciences}, 117:1--22, 2021.
\newblock \href {https://doi.org/10.1016/j.jcss.2020.11.002}
  {\path{doi:10.1016/j.jcss.2020.11.002}}.

\bibitem{DCC2015}
Djamal Belazzougui, Travis Gagie, Paweł Gawrychowski, Juha
  K{\"{a}}rkk{\"{a}}inen, Alberto~Ord{\'{o}}{\~{n}}ez Pereira, Simon~J.
  Puglisi, and Yasuo Tabei.
\newblock Queries on {LZ}-bounded encodings.
\newblock In {\em DCC}, pages 83--92, 2015.
\newblock \href {https://doi.org/10.1109/DCC.2015.69}
  {\path{doi:10.1109/DCC.2015.69}}.

\bibitem{BenderF04}
Michael~A. Bender and Martin Farach{-}Colton.
\newblock The level ancestor problem simplified.
\newblock {\em Theoretical Computer Science}, 321(1):5--12, 2004.
\newblock \href {https://doi.org/10.1016/j.tcs.2003.05.002}
  {\path{doi:10.1016/j.tcs.2003.05.002}}.

\bibitem{BergerDY16}
Bonnie Berger, Noah~M. Daniels, and Y.~William Yu.
\newblock Computational biology in the 21st century: Scaling with compressive
  algorithms.
\newblock {\em Communication of the {ACM}}, 59(8):72–80, jul 2016.
\newblock \href {https://doi.org/10.1145/2957324} {\path{doi:10.1145/2957324}}.

\bibitem{BerkmanV94}
Omer Berkman and Uzi Vishkin.
\newblock Finding level-ancestors in trees.
\newblock {\em Journal of Computer and System Sciences}, 48(2):214--230, 1994.
\newblock \href {https://doi.org/10.1016/S0022-0000(05)80002-9}
  {\path{doi:10.1016/S0022-0000(05)80002-9}}.

\bibitem{BilleGP17}
Philip Bille, Inge~Li G{\o}rtz, and Nicola Prezza.
\newblock Space-efficient {R}e-{P}air compression.
\newblock In {\em DCC}, pages 171--180, 2017.
\newblock \href {https://doi.org/10.1109/DCC.2017.24}
  {\path{doi:10.1109/DCC.2017.24}}.

\bibitem{BLRSRW15}
Philip Bille, Gad~M. Landau, Rajeev Raman, Kunihiko Sadakane, Srinivasa~Rao
  Satti, and Oren Weimann.
\newblock Random access to grammar-compressed strings and trees.
\newblock {\em {SIAM} Journal on Computing}, 44(3):513--539, 2015.
\newblock \href {https://doi.org/10.1137/130936889}
  {\path{doi:10.1137/130936889}}.

\bibitem{blumer1987complete}
Anselm Blumer, Janet~A. Blumer, David Haussler, Ross~M. McConnell, and Andrzej
  Ehrenfeucht.
\newblock Complete inverted files for efficient text retrieval and analysis.
\newblock {\em Journal of the {ACM}}, 34(3):578--595, 1987.
\newblock \href {https://doi.org/10.1145/28869.28873}
  {\path{doi:10.1145/28869.28873}}.

\bibitem{BringmannWK19}
Karl Bringmann, Philip Wellnitz, and Marvin K{\"{u}}nnemann.
\newblock Few matches or almost periodicity: Faster pattern matching with
  mismatches in compressed texts.
\newblock In {\em SODA}, pages 1126--1145, 2019.
\newblock \href {https://doi.org/10.1137/1.9781611975482.69}
  {\path{doi:10.1137/1.9781611975482.69}}.

\bibitem{BrisaboaGMP18}
Nieves~R. Brisaboa, Adri{\'{a}}n G{\'{o}}mez{-}Brand{\'{o}}n, Miguel~A.
  Mart{\'{\i}}nez{-}Prieto, and Jos{\'{e}}~R. Param{\'{a}}.
\newblock 3dgract: {A} grammar-based compressed representation of 3d
  trajectories.
\newblock In {\em SPIRE}, pages 102--116, 2018.
\newblock \href {https://doi.org/10.1007/978-3-030-00479-8\_9}
  {\path{doi:10.1007/978-3-030-00479-8\_9}}.

\bibitem{BrisaboaGNP16}
Nieves~R. Brisaboa, Adri{\'{a}}n G{\'{o}}mez{-}Brand{\'{o}}n, Gonzalo Navarro,
  and Jos{\'{e}}~R. Param{\'{a}}.
\newblock Gract: {A} grammar based compressed representation of trajectories.
\newblock In {\em SPIRE}, pages 218--230, 2016.
\newblock \href {https://doi.org/10.1007/978-3-319-46049-9\_21}
  {\path{doi:10.1007/978-3-319-46049-9\_21}}.

\bibitem{BWT}
Michael Burrows and David~J. Wheeler.
\newblock A block-sorting lossless data compression algorithm.
\newblock Technical Report 124, Digital Equipment Corporation, Palo Alto,
  California, 1994.
\newblock URL:
  \url{https://www.hpl.hp.com/techreports/Compaq-DEC/SRC-RR-124.pdf}.

\bibitem{CKW20}
Panagiotis Charalampopoulos, Tomasz Kociumaka, and Philip Wellnitz.
\newblock Faster approximate pattern matching: {A} unified approach.
\newblock In {\em FOCS}, pages 978--989, 2020.
\newblock \href {https://doi.org/10.1109/FOCS46700.2020.00095}
  {\path{doi:10.1109/FOCS46700.2020.00095}}.

\bibitem{Charikar05}
Moses Charikar, Eric Lehman, Ding Liu, Rina Panigrahy, Manoj Prabhakaran, Amit
  Sahai, and Abhi Shelat.
\newblock The smallest grammar problem.
\newblock {\em {IEEE} Transactions on Information Theory}, 51(7):2554--2576,
  2005.
\newblock \href {https://doi.org/10.1109/TIT.2005.850116}
  {\path{doi:10.1109/TIT.2005.850116}}.

\bibitem{ChristiansenEKN21}
Anders~Roy Christiansen, Mikko~Berggren Ettienne, Tomasz Kociumaka, Gonzalo
  Navarro, and Nicola Prezza.
\newblock Optimal-time dictionary-compressed indexes.
\newblock {\em {ACM} Transactions on Algorithms}, 17(1):8:1--8:39, 2021.
\newblock \href {https://doi.org/10.1145/3426473} {\path{doi:10.1145/3426473}}.

\bibitem{ClaudeN11}
Francisco Claude and Gonzalo Navarro.
\newblock Self-indexed grammar-based compression.
\newblock {\em Fundamenta Informaticae}, 111(3):313--337, 2011.
\newblock \href {https://doi.org/10.3233/FI-2011-565}
  {\path{doi:10.3233/FI-2011-565}}.

\bibitem{ClaudeN12a}
Francisco Claude and Gonzalo Navarro.
\newblock Improved grammar-based compressed indexes.
\newblock In {\em SPIRE}, pages 180--192, 2012.
\newblock \href {https://doi.org/10.1007/978-3-642-34109-0_19}
  {\path{doi:10.1007/978-3-642-34109-0_19}}.

\bibitem{ClaudeNP21}
Francisco Claude, Gonzalo Navarro, and Alejandro Pacheco.
\newblock Grammar-compressed indexes with logarithmic search time.
\newblock {\em Journal of Computer and System Sciences}, 118:53--74, 2021.
\newblock \href {https://doi.org/10.1016/j.jcss.2020.12.001}
  {\path{doi:10.1016/j.jcss.2020.12.001}}.

\bibitem{cocke1969programming}
John Cocke.
\newblock {\em Programming languages and their compilers: Preliminary notes}.
\newblock New York University, 1969.

\bibitem{mg}
European Commission.
\newblock {1+ Million Genomes Initiative}.
\newblock
  \url{https://digital-strategy.ec.europa.eu/en/policies/1-million-genomes}.

\bibitem{AGC}
Sebastian Deorowicz, Agnieszka Danek, and Heng Li.
\newblock {AGC}: Compact representation of assembled genomes.
\newblock {\em bioRxiv}, 2022.
\newblock \href {https://doi.org/10.1101/2022.04.07.487441}
  {\path{doi:10.1101/2022.04.07.487441}}.

\bibitem{Diaz-DominguezN21}
Diego D{\'{\i}}az{-}Dom{\'{\i}}nguez, Gonzalo Navarro, and Alejandro Pacheco.
\newblock An {LMS}-based grammar self-index with local consistency properties.
\newblock In {\em SPIRE}, pages 100--113, 2021.
\newblock \href {https://doi.org/10.1007/978-3-030-86692-1\_9}
  {\path{doi:10.1007/978-3-030-86692-1\_9}}.

\bibitem{Dietz91}
Paul~F. Dietz.
\newblock Finding level-ancestors in dynamic trees.
\newblock In {\em WADS}, pages 32--40, 1991.
\newblock \href {https://doi.org/10.1007/BFb0028247}
  {\path{doi:10.1007/BFb0028247}}.

\bibitem{DuttaLRR13}
Akashnil Dutta, Reut Levi, Dana Ron, and Ronitt Rubinfeld.
\newblock A simple online competitive adaptation of {L}empel-{Z}iv compression
  with efficient random access support.
\newblock In {\em DCC}, pages 113--122, 2013.
\newblock \href {https://doi.org/10.1109/DCC.2013.19}
  {\path{doi:10.1109/DCC.2013.19}}.

\bibitem{FerraginaMGKNST22}
Paolo Ferragina, Giovanni Manzini, Travis Gagie, Dominik K{\"{o}}ppl, Gonzalo
  Navarro, Manuel Striani, and Francesco Tosoni.
\newblock Improving matrix-vector multiplication via lossless
  grammar-compressed matrices.
\newblock {\em Proc. {VLDB} Endow.}, 15(10):2175--2187, 2022.
\newblock URL: \url{https://www.vldb.org/pvldb/vol15/p2175-tosoni.pdf}.

\bibitem{LZFG}
Edward~R Fiala and Daniel~H Greene.
\newblock Data compression with finite windows.
\newblock {\em Communications of the ACM}, 32(4):490--505, 1989.
\newblock \href {https://doi.org/10.1145/63334.63341}
  {\path{doi:10.1145/63334.63341}}.

\bibitem{FuruyaTNIBK19}
Isamu Furuya, Takuya Takagi, Yuto Nakashima, Shunsuke Inenaga, Hideo Bannai,
  and Takuya Kida.
\newblock {MR}-{RePair}: Grammar compression based on maximal repeats.
\newblock In {\em DCC}, pages 508--517, 2019.
\newblock \href {https://doi.org/10.1109/DCC.2019.00059}
  {\path{doi:10.1109/DCC.2019.00059}}.

\bibitem{gage1994new}
Philip Gage.
\newblock A new algorithm for data compression.
\newblock {\em {C} Users Journal}, 12(2):23–38, feb 1994.
\newblock URL: \url{https://dl.acm.org/doi/abs/10.5555/177910.177914}.

\bibitem{GagieGKNP14}
Travis Gagie, Pawel Gawrychowski, Juha K{\"{a}}rkk{\"{a}}inen, Yakov Nekrich,
  and Simon~J. Puglisi.
\newblock {LZ77}-based self-indexing with faster pattern matching.
\newblock In {\em LATIN}, pages 731--742, 2014.
\newblock \href {https://doi.org/10.1007/978-3-642-54423-1\_63}
  {\path{doi:10.1007/978-3-642-54423-1\_63}}.

\bibitem{GagieGKNP12}
Travis Gagie, Paweł Gawrychowski, Juha K{\"{a}}rkk{\"{a}}inen, Yakov Nekrich,
  and Simon~J. Puglisi.
\newblock A faster grammar-based self-index.
\newblock In {\em LATA}, pages 240--251, 2012.
\newblock \href {https://doi.org/10.1007/978-3-642-28332-1_21}
  {\path{doi:10.1007/978-3-642-28332-1_21}}.

\bibitem{GagieIMNST19}
Travis Gagie, Tomohiro I, Giovanni Manzini, Gonzalo Navarro, Hiroshi Sakamoto,
  and Yoshimasa Takabatake.
\newblock {Rpair}: Rescaling {RePair} with {Rsync}.
\newblock In {\em SPIRE}, pages 35--44, 2019.
\newblock \href {https://doi.org/10.1007/978-3-030-32686-9\_3}
  {\path{doi:10.1007/978-3-030-32686-9\_3}}.

\bibitem{GNPlatin18}
Travis Gagie, Gonzalo Navarro, and Nicola Prezza.
\newblock On the approximation ratio of {L}empel-{Z}iv parsing.
\newblock In {\em LATIN}, pages 490--503, 2018.
\newblock \href {https://doi.org/10.1007/978-3-319-77404-6_36}
  {\path{doi:10.1007/978-3-319-77404-6_36}}.

\bibitem{Gagie2020}
Travis Gagie, Gonzalo Navarro, and Nicola Prezza.
\newblock Fully functional suffix trees and optimal text searching in
  {BWT}-runs bounded space.
\newblock {\em Journal of the {ACM}}, 67(1):1--54, 2020.
\newblock \href {https://doi.org/10.1145/3375890} {\path{doi:10.1145/3375890}}.

\bibitem{GanardiG22}
Moses Ganardi and Pawel Gawrychowski.
\newblock Pattern matching on grammar-compressed strings in linear time.
\newblock In {\em SODA}, pages 2833--2846, 2022.
\newblock \href {https://doi.org/10.1137/1.9781611977073.110}
  {\path{doi:10.1137/1.9781611977073.110}}.

\bibitem{balancing}
Moses Ganardi, Artur Jez, and Markus Lohrey.
\newblock Balancing straight-line programs.
\newblock {\em Journal of the {ACM}}, 68(4):27:1--27:40, 2021.
\newblock \href {https://doi.org/10.1145/3457389} {\path{doi:10.1145/3457389}}.

\bibitem{GanczorzJ17}
Michal Ganczorz and Artur Jez.
\newblock Improvements on {Re}-{Pair} grammar compressor.
\newblock In {\em DCC}, pages 181--190, 2017.
\newblock \href {https://doi.org/10.1109/DCC.2017.52}
  {\path{doi:10.1109/DCC.2017.52}}.

\bibitem{GaneshKLS22}
Arun Ganesh, Tomasz Kociumaka, Andrea Lincoln, and Barna Saha.
\newblock How compression and approximation affect efficiency in string
  distance measures.
\newblock In {\em SODA}, pages 2867--2919, 2022.
\newblock \href {https://doi.org/10.1137/1.9781611977073.112}
  {\path{doi:10.1137/1.9781611977073.112}}.

\bibitem{Gawrychowski11}
Pawel Gawrychowski.
\newblock Optimal pattern matching in {LZW} compressed strings.
\newblock In {\em Proceedings of the Twenty-Second Annual {ACM}-{SIAM}
  Symposium on Discrete Algorithms, {SODA} 2011, San Francisco, California,
  USA, January 23-25}, pages 362--372. {SIAM}, 2011.
\newblock \href {https://doi.org/10.1137/1.9781611973082.29}
  {\path{doi:10.1137/1.9781611973082.29}}.

\bibitem{dynstr}
Pawel Gawrychowski, Adam Karczmarz, Tomasz Kociumaka, Jakub Lacki, and Piotr
  Sankowski.
\newblock Optimal dynamic strings.
\newblock In {\em SODA}, pages 1509--1528, 2018.
\newblock Full version: \url{arxiv.org/abs/1511.02612}.
\newblock \href {https://doi.org/10.1137/1.9781611975031.99}
  {\path{doi:10.1137/1.9781611975031.99}}.

\bibitem{100k}
{Genomics England}.
\newblock {T}he 100,000 {G}enomes {P}roject.
\newblock
  \url{https://www.genomicsengland.co.uk/about-genomics-england/the-100000-genomes-project/}.

\bibitem{lzd}
Keisuke Goto, Hideo Bannai, Shunsuke Inenaga, and Masayuki Takeda.
\newblock {LZD} factorization: Simple and practical online grammar compression
  with variable-to-fixed encoding.
\newblock In {\em CPM}, Lecture Notes in Computer Science, pages 219--230,
  2015.
\newblock \href {https://doi.org/10.1007/978-3-319-19929-0\_19}
  {\path{doi:10.1007/978-3-319-19929-0\_19}}.

\bibitem{GreenfieldWH2019}
Dan Greenfield, Vaughan Wittorff, and Michael Hultner.
\newblock The importance of data compression in the field of genomics.
\newblock {\em {IEEE} Pulse}, 10(2):20--23, 2019.
\newblock \href {https://doi.org/10.1109/MPULS.2019.2899747}
  {\path{doi:10.1109/MPULS.2019.2899747}}.

\bibitem{HermelinLLW13}
Danny Hermelin, Gad~M. Landau, Shir Landau, and Oren Weimann.
\newblock Unified compression-based acceleration of edit-distance computation.
\newblock {\em Algorithmica}, 65(2):339--353, 2013.
\newblock \href {https://doi.org/10.1007/s00453-011-9590-6}
  {\path{doi:10.1007/s00453-011-9590-6}}.

\bibitem{hernaez2019genomic}
Mikel Hernaez, Dmitri Pavlichin, Tsachy Weissman, and Idoia Ochoa.
\newblock Genomic data compression.
\newblock {\em Annual Review of Biomedical Data Science}, 2:19--37, 2019.
\newblock \href {https://doi.org/10.1146/annurev-biodatasci-072018-021229}
  {\path{doi:10.1146/annurev-biodatasci-072018-021229}}.

\bibitem{tomohiro-lce}
Tomohiro I.
\newblock Longest common extensions with recompression.
\newblock In {\em 28th Annual Symposium on Combinatorial Pattern Matching,
  {CPM} 2017, July 4-6, Warsaw, Poland}, pages 18:1--18:15, 2017.
\newblock \href {https://doi.org/10.4230/LIPIcs.CPM.2017.18}
  {\path{doi:10.4230/LIPIcs.CPM.2017.18}}.

\bibitem{Jez2015}
Artur Je{\.z}.
\newblock Faster fully compressed pattern matching by recompression.
\newblock {\em {ACM} Transactions on Algorithms}, 11(3):20:1--20:43, 2015.
\newblock \href {https://doi.org/10.1145/2631920} {\path{doi:10.1145/2631920}}.

\bibitem{Jez16}
Artur Je{\.z}.
\newblock A really simple approximation of smallest grammar.
\newblock {\em Theoretical Computer Science}, 616:141--150, 2016.
\newblock \href {https://doi.org/10.1016/j.tcs.2015.12.032}
  {\path{doi:10.1016/j.tcs.2015.12.032}}.

\bibitem{kasami1966efficient}
Tadao Kasami.
\newblock An efficient recognition and syntax-analysis algorithm for
  context-free languages.
\newblock {\em Coordinated Science Laboratory Report no. R-257}, 1966.

\bibitem{Kempa19}
Dominik Kempa.
\newblock Optimal construction of compressed indexes for highly repetitive
  texts.
\newblock In {\em SODA}, pages 1344--1357, 2019.
\newblock \href {https://doi.org/10.1137/1.9781611975482.82}
  {\path{doi:10.1137/1.9781611975482.82}}.

\bibitem{resolution}
Dominik Kempa and Tomasz Kociumaka.
\newblock Resolution of the {B}urrows-{W}heeler transform conjecture.
\newblock In {\em FOCS}, pages 1002--1013, 2020.
\newblock \href {https://doi.org/10.1109/FOCS46700.2020.00097}
  {\path{doi:10.1109/FOCS46700.2020.00097}}.

\bibitem{dynsa}
Dominik Kempa and Tomasz Kociumaka.
\newblock Dynamic suffix array with polylogarithmic queries and updates.
\newblock In {\em {STOC}}, pages 1657--1670. {ACM}, 2022.

\bibitem{breaking}
Dominik Kempa and Tomasz Kociumaka.
\newblock Breaking the {O}(n)-barrier in the construction of compressed suffix
  arrays and suffix trees.
\newblock In {\em {SODA}}, pages 5122--5202. {SIAM}, 2023.

\bibitem{dcc2017}
Dominik Kempa and Dmitry Kosolobov.
\newblock {LZ}-{E}nd parsing in compressed space.
\newblock In {\em DCC}, pages 350--359, 2017.
\newblock \href {https://doi.org/10.1109/DCC.2017.73}
  {\path{doi:10.1109/DCC.2017.73}}.

\bibitem{avl}
Dominik Kempa and Ben Langmead.
\newblock Fast and space-efficient construction of {AVL} grammars from the
  {LZ77} parsing.
\newblock In {\em ESA}, pages 56:1--56:14, 2021.
\newblock \href {https://doi.org/10.4230/LIPIcs.ESA.2021.56}
  {\path{doi:10.4230/LIPIcs.ESA.2021.56}}.

\bibitem{attractors}
Dominik Kempa and Nicola Prezza.
\newblock At the roots of dictionary compression: String attractors.
\newblock In {\em STOC}, pages 827--840, 2018.
\newblock \href {https://doi.org/10.1145/3188745.3188814}
  {\path{doi:10.1145/3188745.3188814}}.

\bibitem{KempaS22}
Dominik Kempa and Barna Saha.
\newblock An upper bound and linear-space queries on the {LZ}-end parsing.
\newblock In {\em SODA}, pages 2847--2866, 2022.
\newblock \href {https://doi.org/10.1137/1.9781611977073.111}
  {\path{doi:10.1137/1.9781611977073.111}}.

\bibitem{collage}
Takuya Kida, Tetsuya Matsumoto, Yusuke Shibata, Masayuki Takeda, Ayumi
  Shinohara, and Setsuo Arikawa.
\newblock Collage system: A unifying framework for compressed pattern matching.
\newblock {\em Theoretical Computer Science}, 298(1):253--272, 2003.
\newblock \href {https://doi.org/10.1016/S0304-3975(02)00426-7}
  {\path{doi:10.1016/S0304-3975(02)00426-7}}.

\bibitem{KiefferY00}
John~C. Kieffer and En{-}Hui Yang.
\newblock Grammar-based codes: {A} new class of universal lossless source
  codes.
\newblock {\em {IEEE} Transactions on Information Theory}, 46(3):737--754,
  2000.
\newblock \href {https://doi.org/10.1109/18.841160}
  {\path{doi:10.1109/18.841160}}.

\bibitem{bisection1}
John~C. Kieffer, En{-}Hui Yang, Gregory~J. Nelson, and Pamela~C. Cosman.
\newblock Universal lossless compression via multilevel pattern matching.
\newblock {\em {IEEE} Transactions on Information Theory}, 46(4):1227--1245,
  July 2000.
\newblock \href {https://doi.org/10.1109/18.850665}
  {\path{doi:10.1109/18.850665}}.

\bibitem{KociumakaNO22}
Tomasz Kociumaka, Gonzalo Navarro, and Francisco Olivares.
\newblock Near-optimal search time in {\(\delta\)}-optimal space.
\newblock In {\em LATIN}, pages 88--103, 2022.
\newblock \href {https://doi.org/10.1007/978-3-031-20624-5\_6}
  {\path{doi:10.1007/978-3-031-20624-5\_6}}.

\bibitem{delta}
Tomasz Kociumaka, Gonzalo Navarro, and Nicola Prezza.
\newblock Towards a definitive measure of repetitiveness.
\newblock In {\em LATIN}, volume 12118, pages 207--219, 2020.
\newblock \href {https://doi.org/10.1007/978-3-030-61792-9\_17}
  {\path{doi:10.1007/978-3-030-61792-9\_17}}.

\bibitem{kreft2010navarro}
Sebastian Kreft and Gonzalo Navarro.
\newblock {LZ77}-like compression with fast random access.
\newblock In {\em {DCC}}, pages 239--248, 2010.
\newblock \href {https://doi.org/10.1109/DCC.2010.29}
  {\path{doi:10.1109/DCC.2010.29}}.

\bibitem{longestmatch}
J.~Kevin Lanctot, Ming Li, and En-hui Yang.
\newblock Estimating {DNA} sequence entropy.
\newblock In {\em SODA}, page 409–418, USA, 2000.
\newblock URL: \url{http://dl.acm.org/citation.cfm?id=338219.338586}.

\bibitem{repair}
N.~Jesper Larsson and Alistair Moffat.
\newblock Off-line dictionary-based compression.
\newblock {\em Proceedings of the {IEEE}}, 88(11):1722--1732, 2000.
\newblock \href {https://doi.org/10.1109/5.892708}
  {\path{doi:10.1109/5.892708}}.

\bibitem{Lohrey12}
Markus Lohrey.
\newblock Algorithmics on {SLP}-compressed strings: {A} survey.
\newblock {\em Groups Complexity Cryptology}, 4(2):241--299, 2012.
\newblock \href {https://doi.org/10.1515/gcc-2012-0016}
  {\path{doi:10.1515/gcc-2012-0016}}.

\bibitem{LohreyMM13}
Markus Lohrey, Sebastian Maneth, and Roy Mennicke.
\newblock {XML} tree structure compression using {RePair}.
\newblock {\em Information Systems}, 38(8):1150--1167, 2013.
\newblock \href {https://doi.org/10.1016/j.is.2013.06.006}
  {\path{doi:10.1016/j.is.2013.06.006}}.

\bibitem{sa}
Udi Manber and Eugene~W. Myers.
\newblock Suffix arrays: A new method for on-line string searches.
\newblock {\em {SIAM} J. Comput.}, 22(5):935--948, 1993.

\bibitem{MienoIH22}
Takuya Mieno, Shunsuke Inenaga, and Takashi Horiyama.
\newblock {RePair} grammars are the smallest grammars for {F}ibonacci words.
\newblock In {\em CPM}, pages 26:1--26:17, 2022.
\newblock \href {https://doi.org/10.4230/LIPIcs.CPM.2022.26}
  {\path{doi:10.4230/LIPIcs.CPM.2022.26}}.

\bibitem{navarrobook}
Gonzalo Navarro.
\newblock {\em Compact data structures: A practical approach}.
\newblock Cambridge University Press, Cambridge, UK, 2016.
\newblock \href {https://doi.org/10.1017/cbo9781316588284}
  {\path{doi:10.1017/cbo9781316588284}}.

\bibitem{NavarroMeasures}
Gonzalo Navarro.
\newblock Indexing highly repetitive string collections, part {I}:
  {R}epetitiveness measures.
\newblock {\em {ACM} Comput. Surv.}, 54(2):29:1--29:31, 2021.
\newblock \href {https://doi.org/10.1145/3434399} {\path{doi:10.1145/3434399}}.

\bibitem{NavarroIndexes}
Gonzalo Navarro.
\newblock Indexing highly repetitive string collections, part {II}:
  {C}ompressed indexes.
\newblock {\em {ACM} Comput. Surv.}, 54(2):26:1--26:32, 2021.
\newblock \href {https://doi.org/10.1145/3432999} {\path{doi:10.1145/3432999}}.

\bibitem{bisection2}
Greg Nelson, John Kieffer, and Pamela Cosman.
\newblock An interesting hierarchical lossless data compression algorithm.
\newblock In {\em {IEEE} Information Theory Society Workshop}, 1995.

\bibitem{sequitur}
Craig~G. Nevill{-}Manning and Ian~H. Witten.
\newblock Identifying hierarchical structure in sequences: {A} linear-time
  algorithm.
\newblock {\em Journal of Artificial Intelligence Research}, 7:67--82, 1997.
\newblock \href {https://doi.org/10.1613/jair.374}
  {\path{doi:10.1613/jair.374}}.

\bibitem{estimate}
National Human Genome Research~Institute (NIH).
\newblock Genomic data science.
\newblock
  \url{https://www.genome.gov/about-genomics/fact-sheets/Genomic-Data-Science}.

\bibitem{NishimotoMFCS}
Takaaki Nishimoto, Tomohiro I, Shunsuke Inenaga, Hideo Bannai, and Masayuki
  Takeda.
\newblock Fully dynamic data structure for {LCE} queries in compressed space.
\newblock In {\em MFCS}, pages 72:1--72:15, 2016.
\newblock \href {https://doi.org/10.4230/LIPIcs.MFCS.2016.72}
  {\path{doi:10.4230/LIPIcs.MFCS.2016.72}}.

\bibitem{NishimotoDAM}
Takaaki Nishimoto, Tomohiro I, Shunsuke Inenaga, Hideo Bannai, and Masayuki
  Takeda.
\newblock Dynamic index and {LZ} factorization in compressed space.
\newblock {\em Discret. Appl. Math.}, 274:116--129, 2020.
\newblock \href {https://doi.org/10.1016/j.dam.2019.01.014}
  {\path{doi:10.1016/j.dam.2019.01.014}}.

\bibitem{nishimototabei}
Takaaki Nishimoto and Yasuo Tabei.
\newblock Optimal-time queries on {BWT}-runs compressed indexes.
\newblock In {\em {ICALP}}, pages 101:1--101:15, 2021.
\newblock \href {https://doi.org/10.4230/LIPIcs.ICALP.2021.101}
  {\path{doi:10.4230/LIPIcs.ICALP.2021.101}}.

\bibitem{nishimototabei2}
Takaaki Nishimoto and Yasuo Tabei.
\newblock R-enum: {E}numeration of characteristic substrings in {BWT}-runs
  bounded space.
\newblock In {\em {CPM}}, pages 21:1--21:21, 2021.
\newblock \href {https://doi.org/10.4230/LIPIcs.CPM.2021.21}
  {\path{doi:10.4230/LIPIcs.CPM.2021.21}}.

\bibitem{OchoaN19}
Carlos Ochoa and Gonzalo Navarro.
\newblock {RePair} and all irreducible grammars are upper bounded by high-order
  empirical entropy.
\newblock {\em {IEEE} Transactions on Information Theory}, 65(5):3160--3164,
  2019.
\newblock \href {https://doi.org/10.1109/TIT.2018.2871452}
  {\path{doi:10.1109/TIT.2018.2871452}}.

\bibitem{OhnoSTIS18}
Tatsuya Ohno, Kensuke Sakai, Yoshimasa Takabatake, Tomohiro I, and Hiroshi
  Sakamoto.
\newblock A faster implementation of online {RLBWT} and its application to
  {LZ77} parsing.
\newblock {\em J. Discrete Alg.}, 52-53:18--28, 2018.
\newblock \href {https://doi.org/10.1016/j.jda.2018.11.002}
  {\path{doi:10.1016/j.jda.2018.11.002}}.

\bibitem{Patrascu07}
Mihai Patrascu.
\newblock Lower bounds for 2-dimensional range counting.
\newblock In {\em STOC}, pages 40--46, 2007.
\newblock \href {https://doi.org/10.1145/1250790.1250797}
  {\path{doi:10.1145/1250790.1250797}}.

\bibitem{PatrascuT06}
Mihai Patrascu and Mikkel Thorup.
\newblock Time-space trade-offs for predecessor search.
\newblock In {\em STOC}, pages 232--240, 2006.
\newblock \href {https://doi.org/10.1145/1132516.1132551}
  {\path{doi:10.1145/1132516.1132551}}.

\bibitem{PereiraNB17}
Alberto~Ord{\'{o}}{\~{n}}ez Pereira, Gonzalo Navarro, and Nieves~R. Brisaboa.
\newblock Grammar compressed sequences with rank/select support.
\newblock {\em Journal of Discrete Algorithms}, 43:54--71, 2017.
\newblock \href {https://doi.org/10.1016/j.jda.2016.10.001}
  {\path{doi:10.1016/j.jda.2016.10.001}}.

\bibitem{PolicritiP17}
Alberto Policriti and Nicola Prezza.
\newblock From {LZ77} to the run-length encoded {B}urrows-{W}heeler transform,
  and back.
\newblock In {\em CPM}, pages 17:1--17:10, 2017.
\newblock \href {https://doi.org/10.4230/LIPIcs.CPM.2017.17}
  {\path{doi:10.4230/LIPIcs.CPM.2017.17}}.

\bibitem{Prezza19}
Nicola Prezza.
\newblock Optimal rank and select queries on dictionary-compressed text.
\newblock In {\em CPM}, pages 4:1--4:12, 2019.
\newblock \href {https://doi.org/10.4230/LIPIcs.CPM.2019.4}
  {\path{doi:10.4230/LIPIcs.CPM.2019.4}}.

\bibitem{Przeworski2000}
Molly Przeworski, Richard~R. Hudson, and Anna {Di Rienzo}.
\newblock Adjusting the focus on human variation.
\newblock {\em Trends in Genetics}, 16(7):296--302, 2000.
\newblock \href {https://doi.org/10.1016/S0168-9525(00)02030-8}
  {\path{doi:10.1016/S0168-9525(00)02030-8}}.

\bibitem{RaskhodnikovaRRS13}
Sofya Raskhodnikova, Dana Ron, Ronitt Rubinfeld, and Adam~D. Smith.
\newblock Sublinear algorithms for approximating string compressibility.
\newblock {\em Algorithmica}, 65(3):685--709, 2013.
\newblock \href {https://doi.org/10.1007/s00453-012-9618-6}
  {\path{doi:10.1007/s00453-012-9618-6}}.

\bibitem{Rytter03}
Wojciech Rytter.
\newblock Application of {L}empel--{Z}iv factorization to the approximation of
  grammar-based compression.
\newblock {\em Theoretical Computer Science}, 302(1--3):211--222, 2003.
\newblock \href {https://doi.org/10.1016/S0304-3975(02)00777-6}
  {\path{doi:10.1016/S0304-3975(02)00777-6}}.

\bibitem{stephens2015big}
Zachary~D Stephens, Skylar~Y Lee, Faraz Faghri, Roy~H Campbell, Chengxiang
  Zhai, Miles~J Efron, Ravishankar Iyer, Michael~C Schatz, Saurabh Sinha, and
  Gene~E Robinson.
\newblock Big data: astronomical or genomical?
\newblock {\em PLoS biology}, 13(7):e1002195, 2015.
\newblock \href {https://doi.org/10.1371/journal.pbio.1002195}
  {\path{doi:10.1371/journal.pbio.1002195}}.

\bibitem{storer1978macro}
James~A. Storer and Thomas~G. Szymanski.
\newblock The macro model for data compression.
\newblock In {\em SODA}, pages 30--39, 1978.
\newblock \href {https://doi.org/10.1145/800133.804329}
  {\path{doi:10.1145/800133.804329}}.

\bibitem{LZSS}
James~A. Storer and Thomas~G. Szymanski.
\newblock Data compression via textual substitution.
\newblock {\em Journal of the {ACM}}, 29(4):928--951, 1982.
\newblock \href {https://doi.org/10.1145/322344.322346}
  {\path{doi:10.1145/322344.322346}}.

\bibitem{Tiskin15}
Alexander Tiskin.
\newblock Fast distance multiplication of unit-monge matrices.
\newblock {\em Algorithmica}, 71(4):859--888, 2015.
\newblock \href {https://doi.org/10.1007/s00453-013-9830-z}
  {\path{doi:10.1007/s00453-013-9830-z}}.

\bibitem{VerbinY13}
Elad Verbin and Wei Yu.
\newblock Data structure lower bounds on random access to grammar-compressed
  strings.
\newblock In {\em CPM}, volume 7922, pages 247--258, 2013.
\newblock \href {https://doi.org/10.1007/978-3-642-38905-4_24}
  {\path{doi:10.1007/978-3-642-38905-4_24}}.

\bibitem{st}
Peter Weiner.
\newblock Linear pattern matching algorithms.
\newblock In {\em {SWAT}/{FOCS}}, pages 1--11, 1973.

\bibitem{LZW}
Terry~A. Welch.
\newblock A technique for high-performance data compression.
\newblock {\em Computer}, 17(6):8--19, 1984.
\newblock \href {https://doi.org/10.1109/MC.1984.1659158}
  {\path{doi:10.1109/MC.1984.1659158}}.

\bibitem{sequential}
En{-}Hui Yang and John~C. Kieffer.
\newblock Efficient universal lossless data compression algorithms based on a
  greedy sequential grammar transform - part one: Without context models.
\newblock {\em {IEEE} Transactions on Information Theory}, 46(3):755--777,
  2000.
\newblock \href {https://doi.org/10.1109/18.841161}
  {\path{doi:10.1109/18.841161}}.

\bibitem{younger1967recognition}
Daniel~H Younger.
\newblock Recognition and parsing of context-free languages in time $n^3$.
\newblock {\em Information and control}, 10(2):189--208, 1967.

\bibitem{LZ77}
Jacob Ziv and Abraham Lempel.
\newblock A universal algorithm for sequential data compression.
\newblock {\em {IEEE} Transactions on Information Theory}, 23(3):337--343,
  1977.
\newblock \href {https://doi.org/10.1109/TIT.1977.1055714}
  {\path{doi:10.1109/TIT.1977.1055714}}.

\bibitem{LZ78}
Jacob Ziv and Abraham Lempel.
\newblock Compression of individual sequences via variable-rate coding.
\newblock {\em {IEEE} Transactions on Information Theory}, 24(5):530--536,
  1978.
\newblock \href {https://doi.org/10.1109/TIT.1978.1055934}
  {\path{doi:10.1109/TIT.1978.1055934}}.

\end{thebibliography}

\end{document}